\documentclass[11pt]{article}

\counterwithout{figure}{section}
\usepackage{amsmath}
\usepackage{bm}
\usepackage[authoryear]{natbib}
\usepackage[colorlinks,citecolor=blue,linkcolor=blue,urlcolor=blue]{hyperref}
\usepackage{latexsym, epsfig, amssymb, amsmath, amsthm, graphicx, mathrsfs, booktabs, amsbsy}
\usepackage{rotating, tabularx, setspace,paralist}
\usepackage{color}
\usepackage{graphicx}
\usepackage[OT1]{fontenc}
\usepackage{lscape}
\usepackage{comment}
\usepackage{multirow,multicol}
\usepackage{float}
\usepackage{anysize}
\renewcommand{\baselinestretch}{1.2}
\marginsize{1.1in}{1.1in}{0.55in}{0.8in} %

\makeatletter
\def\singlespace{\def\baselinestretch{1}\@normalsize}

\newtheorem{condition}{Condition}
\newtheorem{lemma}{Lemma}
\newtheorem{proposition}{Proposition}
\newtheorem{theorem}{Theorem}
\newtheorem{remark}{Remark}
\renewcommand{\theequation}{
\arabic{equation}%
}

\makeatother

\newcommand{\bZ}{\mbox{\bf Z}}

\newcommand{\bzero}{\mbox{\bf 0}}

\newcommand{\bbeta}{\mbox{\boldmath $\beta$}}

\newcommand{\bgamma}{\mbox{\boldmath $\gamma$}}

\newcommand{\bmu}{\mbox{\boldmath $\mu$}}

\newcommand{\bSig}{\mbox{\boldmath $\Sigma$}}

\newcommand{\Var}{\mathrm{Var}}

\newcommand{\argmin}{\mbox{argmin}}

\newcommand{\BS}{\boldsymbol}

\def\independenT#1#2{\mathrel{\setbox0\hbox{$#1#2$}%
\copy0\kern-\wd0\mkern4mu\box0}}

\renewcommand{\hat}{\widehat}
\renewcommand{\tilde}{\widetilde}
\renewcommand{\dots}{\cdots}
\renewcommand{\ldots}{\cdots}

\renewcommand{\qquad}{\quad}
\renewcommand{\subseteq}{\subset}

\begin{document}

\title{HSCI: Neyman-Orthogonal Causal Inference under High-Dimensional Proportional Hazards%
\date{June 9, 2026}
\author{Yingying Fan$^1$, Lan Gao$^2$, Daoji Li$^3$ and Jinchi Lv$^1$
\medskip\\
University of Southern California$^1$, The University of Tennessee, Knoxville$^2$\\ and California State University, Fullerton$^3$
\\
} %
}

\maketitle

\begin{abstract}
Valid treatment effect inference in survival studies is fundamental yet challenging when the treatment assignments and outcomes are confounded by many baseline covariates. To this end, in this paper we propose a high-dimensional survival causal inference (HSCI) framework that delivers valid inference under a sparse high-dimensional Cox proportional hazards outcome model and a high-dimensional logistic propensity score working model. To mitigate the nuisance estimation bias, we develop a Neyman near-orthogonal score for the treatment effect and implement it with cross-fitting. Under doubly robust nuisance-rate conditions, we establish the root-$n$ asymptotic normality and consistent variance estimation. We also extend the framework to inference on high-dimensional survival covariate effects. Simulation examples confirm that HSCI reduces sharply the bias relative to the regularized Cox estimators and maintains valid confidence interval coverage across different dimensionality, censoring, and misspecified propensity-model settings. An application to diffuse large-B-cell lymphoma data further showcases its value for high-dimensional biomedical survival studies.
\end{abstract}
 
\textit{Running title}: HSCI


\section{Introduction} \label{new.sec.intro}

Treatment effect inference for survival outcomes is central in many modern applications. In biomedical studies, investigators often compare alternative treatments for right-censored outcomes, such as the overall survival after cancer therapies, time to hospital readmissions, and post-transplant mortalities \citep{hernan2000marginal, petito2020estimates, hadley2010comparative, strohmaier2022survival}. Related duration-outcome questions arise in social sciences, including evaluations of active labor-market programs, unemployment insurance policies, and other interventions on the unemployment duration and job-finding hazards \citep{lalive2008impact, abbring2003nonparametric}. These analyses target causal questions instead of individual-level prediction: \textit{How} would survival have changed under one treatment rather than another after adjusting for patient characteristics? Such questions are challenging since survival outcomes are often right censored, the treatment assignment may be confounded by high-dimensional covariates, and regularization can induce non-negligible bias. These features motivate survival causal inference methods that address the censoring, confounding, and regularization bias \textit{simultaneously}.

A growing literature has studied treatment effect estimation and inference for survival outcomes. Under the additive hazards models, \citet{hou2023treatment} developed orthogonal-score methods for treatment effect estimation with high-dimensional confounding and established the asymptotic normality. \citet{cui2023estimating} proposed nonparametric causal survival forests for heterogeneous treatment effects with right-censored outcomes through adaptive local weighting. \citet{westling2024inference} introduced cross-fitted doubly robust estimators for treatment-specific survival curves using machine learning, whereas \citet{luo2025doubly} studied doubly robust estimation under a possibly misspecified marginal structural Cox model. These works advance survival causal inference, but they do \textit{not} directly address simultaneous treatment-effect and covariate-effect inference under sparse high-dimensional proportional hazards models. Such gap is important since the Cox model remains central for censored time-to-event analysis, while modern biomedical and social sciences studies often involve many confounders and prognostic variables. We address this gap by developing inference for both causal treatment effect and high-dimensional covariate effects under the Cox model.

To this end, in this paper we propose high-dimensional survival causal inference with proportional hazards (HSCI), a framework for joint treatment effect and covariate effect inference under a sparse high-dimensional Cox model. We exploit the Cox model with high-dimensional covariates for the potential failure time and a high-dimensional propensity score working model for treatment assignment. A key methodological contribution is the construction of a Neyman near-orthogonal score for the treatment-effect parameter, designed to remove first-order bias from high-dimensional nuisance parameter estimation. The construction combines double machine learning \citep{chernozhukov2018double} with martingale and partial-likelihood techniques for censored survival data. This combination is nontrivial since the Cox partial likelihood involves time-varying risk sets and is not a simple independent sum of smooth losses. We further employ cross-fitting to reduce the overfitting bias and establish the asymptotic normality for the treatment effect estimator under doubly robust rate conditions on the nuisance estimation errors.

We also \textit{distinguish} the rate-double-robust property studied here from classical double robustness. The classical double robustness refers typically to consistency under model misspecification: an estimator remains consistent if at least one of two nuisance models, such as the outcome model or the treatment/censoring model, is correctly specified. In contrast, the rate double robustness in high-dimensional orthogonal-score analysis requires the product, or suitable combinations, of nuisance estimation errors to be sufficiently small for valid root-$n$ inference. Consequently, the double robustness considered here is an \textit{inference-rate} condition instead of a model-misspecification guarantee.

\subsection{Our contributions}

The major contributions of our paper are threefold. First, we develop the HSCI framework by constructing a Neyman near-orthogonal score equation for survival treatment-effect parameter estimation. The score is tailored to the proportional hazards structure and removes the first-order bias from high-dimensional nuisance-parameter estimation. Its orthogonalization parameter is defined through the population counterparts of the Cox information quantities and estimated by incorporating the propensity score working model. This \textit{distinguishes} our approach from classical debiased Cox inference, which does \textit{not} model the treatment assignment or causal identification, and from standard double machine learning, since the Cox partial likelihood depends on time-varying risk sets as opposed to independent smooth losses. We show that the proposed score satisfies the Neyman near-orthogonality, with a negligible approximation error under explicit sparsity and rate conditions. 

Second, we establish the asymptotic normality for the proposed cross-fitted treatment-effect estimator and provide a consistent sandwich-type variance estimator. These results enable valid inference for the causal log-hazard treatment effect in the high-dimensional Cox proportional hazards model. The theory establishes root-$n$ inference under doubly robust rate conditions requiring the product of the relevant nuisance-parameter estimation errors to be sufficiently small. Our analysis controls the censoring, martingale variation, time-dependent risk sets, sparse high-dimensional nuisance estimation, and cross-fitting errors.

Third, we extend the inference procedure from the causal treatment effect to high-dimensional survival covariate effects. Building on the debiased survival causal effect estimators, 
we further construct a debiased estimator for the high-dimensional covariate coefficient vector and establish its asymptotic normality. Such an extension makes the proposed HSCI framework a joint inference method for both causal treatment effect and covariate effects in high-dimensional censored survival models.

\subsection{Related works}

Our method builds on the Neyman-orthogonal score framework of \citet{chernozhukov2018double}, which reduces the first-order impact of nuisance-parameter estimation errors on inference for low-dimensional target parameters. In non-survival settings, \citet{bradic2019sparsity} studied sparsity-based doubly robust inference for average treatment effects, while \citet{bradic2024high} developed high-dimensional inference for dynamic treatment effects.

Recent works have extended the orthogonal-score and doubly robust ideas to survival causal inference. For example, under the additive hazards models, \citet{hou2023treatment} introduced an orthogonal-score method for treatment-effect estimation with high-dimensional confounders, using the double machine learning framework of \citet{chernozhukov2018double}. They showed that the resulting score coincides with that of \citet{dukes2019doubly} from classical semiparametric efficiency theory and established the asymptotic normality under doubly robust nuisance-estimation rates. \citet{rava2023doubly} developed doubly robust orthogonal scores for causal hazard-difference estimation with right-censored competing-risks data and established root-$n$ asymptotic normality under rate double robustness. These works are closest in spirit to ours, but the additive hazards models differ fundamentally from the Cox models: treatment enters the hazard linearly, whereas the Cox models incorporate an exponential link and time-varying risk sets. More recently, \citet{westling2024inference} studied machine-learning-based doubly robust inference for treatment-specific survival curves under fixed-dimensional covariates, and \citet{luo2025doubly} developed doubly robust estimation for marginal structural Cox models under possible misspecification. These papers target \textit{marginal} survival-curve or \textit{marginal} structural Cox estimands, whereas we study debiased inference for both the treatment effect and high-dimensional covariate effects under a sparse high-dimensional Cox proportional hazards model.

Another related line of work develops flexible nonparametric methods for heterogeneous treatment effects. \citet{wager2018estimation} introduced causal forests for non-survival outcomes, and \citet{cui2023estimating} extended it to right-censored survival data through causal survival forests. Their approach uses adaptive random forest weights and orthogonal estimating equations to adjust for selection and censoring under unconfoundedness. These methods are flexible for heterogeneous treatment-effect estimation, but they do \textit{not} target high-dimensional proportional-hazards inference for low-dimensional causal effects together with high-dimensional covariate effects.

High-dimensional Cox regression inference has also been studied in non-causal or partially causal settings. \citet{fang2017testing} proposed the decorrelated score, Wald, and partial-likelihood-ratio tests for low-dimensional components of high-dimensional Cox models. \citet{kong2021high} studied desparsified Lasso inference for potentially misspecified Cox models. \citet{zhong2022deep} proposed a deep partially linear Cox model that estimates nonlinear covariate effects by a deep neural network while retaining interpretable, root-$n$ inference for treatment effects. \citet{yu2021confidence} developed confidence intervals for high-dimensional Cox coefficients via one-step debiasing. These works handle Cox partial likelihood and high-dimensional nuisance parameters, and we build on several of their technical tools. However, they do \textit{not} jointly incorporate a treatment-assignment mechanism, causal identification assumptions, and inference for both treatment and covariate effects. In our setting, the treatment indicator is linked to a logistic propensity score working model and a sparse high-dimensional Cox proportional hazards model for the potential outcomes.

Finally, synthetic-control and matrix-completion ideas have begun to enter causal inference with survival outcomes. \citet{curth2024cautionary} studied adaptations of synthetic-control methods to time-to-event settings, emphasizing challenges from censoring, nonlinear survival dynamics, and hazard-based estimands. \citet{han2025synthetic} introduced synthetic survival control methods for estimating counterfactual hazard trajectories in panel data under low-rank structure. These methods target panel or aggregate intervention settings, and are \textit{complementary} to our individual-level treatment-effect inference problem with high-dimensional confounding.

The rest of the paper is organized as follows. Section \ref{Sec2} introduces the model setting. Section \ref{Sec3} presents the proposed HSCI procedure for valid inference on both survival treatment effect and covariate effects. Section \ref{new.Sec.theory} establishes the asymptotic normality of the proposed estimators. Section \ref{Sec5} reports finite-sample performance in simulations. Section \ref{Sec6} applies the method to diffuse large-B-cell lymphoma data. Section \ref{Sec7} concludes with some discussions and extensions. All the proofs with technical details and some additional simulation results are provided in the Supplementary Material.

\medskip
\noindent\textit{Notation}. Throughout the paper, let $\mathcal{P}$ be the collection of probability distributions for the population $\{\mathcal{O}\}=\{X, \delta, D, \BS{Z}\}$. For $P \in \mathcal{P}$, denote by $\mathbb{E}$ the expectation under $P$, $\mathbb{E}_n$ the empirical expectation, and $\mathbb{G}_n$ the empirical process. Hence, for any measurable function $g$, $\mathbb{E} g(\mathcal{O}):= P g(\mathcal{O})$, $\mathbb{E}_n g(\mathcal{O}):= \sum_{i=1}^n g(\mathcal{O}_i)/n$, and $\mathbb{G}_n (g(\mathcal{O})) =n^{-1/2}\sum_{i=1}^n [g(\mathcal{O}_i) -\mathbb{E} g(\mathcal{O}_i)]$.

For any set $S$, let $|S|$ be its cardinality, and $S^c$ its complement. For a vector $\BS{v}=(v_1, \dots, v_m)^{\top} \in \mathbb{R}^m$, denote by $\| \BS{v} \|_q$ its $l_q$-norm, $\BS{v}^{\otimes 2} = \BS{v}\BS{v}^{\top}$, and $\BS{v}_S :=(v_j)_{j \in S} \in \mathbb{R}^{|S|}$ for $S \subseteq \{1, \dots, m\}$. For a matrix $\BS{B}$, let $\| \BS{B}\|_{\max}=\max_{i,j} |\BS{B}_{i,j}|$ be the entrywise maximum norm, and $\| \BS{B}\|_{\infty} =\sup_{\BS{v} \neq \BS{0}} \frac{\| \BS{B} \BS{v}\|_{\infty}}{\| \BS{v}\|_{\infty}}$ and $\| \BS{B}\|_{1} =\sup_{\BS{v} \neq \BS{0}} \frac{\| \BS{B} \BS{v}\|_{1}}{\| \BS{v}\|_{1}}$ the operator $l_{\infty}$- and $l_1$-norms, respectively. For real sequences $\{a_n\}$ and $\{b_n\}$, write $a_n \asymp b_n$ if there exist constants $c,C>0$ and $n_0 > 0$ such that $c < |{a_n}/{b_n}| < C$ for all $n \geq n_0$.

For any measurable function $g$, let $\|g\|_{P,q}$ be the $L^q(P)$ norm $\|g \|_{P,q} =(\int |g|^q dP)^{1/q}$. For a differentiable function $f(x)$ with $x$ being a scalar or vector, write $\dot f_{x} (x)$ and $\ddot f_{x} (x)$ for $\frac{\partial f(x)}{\partial x}$ and $\frac{\partial^2 f(x)}{\partial x \partial x}$, respectively; write $\dot f_{x} (x_0)$ and $\ddot f_{x} (x_0)$ for $\dot f_{x} (x) \big|_{x=x_0}$ and $\ddot f_{x} (x) \big|_{x=x_0}$, respectively. For all quantities defined explicitly or implicitly below that depend on $T$, superscripts $^{(1)}$ and $^{(0)}$ represent the counterparts with $T$ replaced by $T^{(1)}$ and $T^{(0)}$, respectively. A subscript $i$ stands for the $i$th observation from the original population for $i = 1, \dots, n$.

\section{Model setting} \label{Sec2}

Let $D\in \{0, 1\}$ be the random treatment indicator, where $D=1$ and $0$ represent the assignments to the treatment and control groups, respectively.
We consider an observational survival study, and let $T$ be an individual's potential failure time. Throughout the paper, we adopt the stable unit treatment value assumption (SUTVA; \cite{rubin1980randomization}), under which each individual may be assigned to either the treatment or control group and there is no interference across individuals. Denote by $T^{(1)}$ and $T^{(0)}$ the potential failure times under the treatment and control, respectively.
Under the SUTVA assumption, only one of these potential outcomes is observed, and the observed failure time satisfies that $T=DT^{(1)}+(1-D)T^{(0)}$.

In survival analysis, failure time $T$ is often right-censored by a censoring time $C$. We thus observe only $X =\min(T,C)$ and the censoring indicator $\delta=I_{\{T \le C\}}$. Let $\BS{Z}$ be a $p$-dimensional covariate vector. Assume that $T$ and $C$ are independent conditional on $(D, \BS{Z})$. The observed data $\{\mathcal{O}_i\}_{i=1}^n=\{X_i, \delta_i, D_i, \BS{Z}_i\}_{i=1}^n$ consists of $n$ independent and identically distributed (i.i.d.) observations from the population $\{X, \delta, D, \BS{Z}\}$.

We focus on the joint causal and survival inference in the high-dimensional setting. To this end, we consider the following survival treatment effect model
\begin{eqnarray}
	\lambda_T(t \mid D, \BS{Z}) &=&\lambda_0(t) \exp\{D \theta_0+\BS{Z}^{\top} \BS{\beta}_0\}, \label{eqmodelcox}\\
	P(D=1 \mid \BS{Z}) &=& \frac{\exp\{\BS{Z}^{\top} \BS{\gamma}_0\}}{1+\exp\{ \BS{Z}^{\top} \BS{\gamma}_0\}} := m_0(\BS{Z}) \label{eqmodellog},
\end{eqnarray}
where $\lambda_0(t)$ is an unknown baseline hazard function, and $\theta_0 \in \Theta \subset \mathbb{R}$, $\BS{\beta}_0 \in \mathcal{B} \subset {\mathbb{R}}^{p}$, and $\BS{\gamma}_0 \in \BS{\Gamma} \subset {\mathbb{R}}^{p}$ are unknown true parameters. In addition, we assume sparsity in both $\BS{\beta}_0$ and $\BS{\gamma}_0$, and denote the corresponding active sets as $S_{\BS{\beta}_0}=\{j: \BS{\beta}_{0,j} \neq 0, 1 \le j \le p\}$ and $S_{\BS{\gamma}_0}=\{j: \BS{\gamma}_{0,j} \neq 0, 1 \le j \le p\}$, respectively, with cardinalities $s_{\BS{\beta}_0}=| S_{\BS{\beta}_0}|$ and $s_{\BS{\gamma}_0}=|S_{\BS{\gamma}_0} |$.
The covariate dimensionality $p$ is allowed to grow exponentially with sample size $n$ in the sense that $\log p=o(n^{\alpha})$ for some constant $\alpha>0$. 
The propensity score model in (\ref{eqmodellog}) serves as a \textit{working model}, and our method and theory accommodate the misspecified propensity score model later. 

We next introduce the counting-process notation used in survival analysis. Let $Y(t)=I_{\{X \ge t\}}$ be the at-risk indicator, and $N(t)= I_{\{X \le t, \,\delta=1\}}$ the counting process for the failure time. Denote by $N^{(1)}(t)$ and $N^{(0)}(t)$ the potential counting processes under the treatment and control groups, respectively. We define a right continuous filtration $\{\mathcal{F}_t: t \in [0, \tau]\}$ with $$\mathcal{F}_t=\sigma \{Y_i(s+), N_i(s), D_i, \BS{Z}_i, \, 0 \le s \le t, \, i=1, \dots , n \},$$
where $\tau$ is the study ending time. Under the Cox proportional hazards outcome model (\ref{eqmodelcox}), $N(t)$ satisfies the conditions for the Doob--Meyer decomposition \citep{fleming2011counting}, and has predictable compensator
  \begin{align*}
      A(t)=\int_0^t Y(u) \exp\{D \theta_0+\BS{Z}^{\top} \BS{\beta}_0 \} \lambda_0(u) du
  \end{align*}
such that $M(t)= N(t)- A(t)$ is a zero-mean, locally square-integrable martingale with respect to $\{\mathcal{F}_t:  t \in [0, \tau]\}$. For technical simplicity, assume that for $i \neq j$, the processes $N_i(t)$ and $N_j(t)$ never jump at the same time. The martingale $M_i(t)$ has predictable variation process $\langle M_i, M_i \rangle (t)=A_i(t)$ for $i= 1, \cdots, n$,  and the covariation process satisfies that $\langle M_i, M_j \rangle (t) = 0$ for any $j=1, \cdots,n, \, j \neq i$.

The log-partial likelihood function (scaled by $n$) for the Cox model (\ref{eqmodelcox}) is defined as
\begin{eqnarray}\label{eqlogparlik}
	l(\{\mathcal{O}_i\}_{i=1}^n;\theta, \BS{\beta})=\frac{1}{n} \sum_{i=1}^n \int_0^{\tau}\left[ D_i \theta+\BS{Z}_i^{\top} \BS{\beta} - \log \Big\{\sum_{j=1}^n  \tilde w_j(t, \theta, \BS{\beta})  \Big\} \right] dN_i(t),
\end{eqnarray}
where
\begin{equation} \label{eq-def-tilde-w}
\tilde w_i(t, \theta, \BS{\beta})=Y_i(t) \exp\{D_i \theta+\BS{Z}_i^{\top} \BS{\beta} \}.
\end{equation}
Standard differentiation gives the score functions with respect to $\theta$ and $\BS{\beta}$ as
\begin{eqnarray*}
	\dot l_{ \theta}(\{\mathcal{O}_i\}_{i=1}^n; \theta , \BS{\beta})& :=&\frac{\partial l(\{\mathcal{O}_i\}_{i=1}^n; \theta, \BS{\beta})}{\partial {\theta}} = \frac{1}{n} \sum_{i=1}^n \int_0^{\tau} (D_i - \bar D(t,\theta, \BS{\beta}) ) dN_i(t),\\
	\dot l_{ \BS{\beta}}(\{\mathcal{O}_i\}_{i=1}^n; \theta , \BS{\beta}) &:=& \frac{\partial l(\{\mathcal{O}_i\}_{i=1}^n;\theta, \BS{\beta})}{\partial {\BS{\beta}}} = \frac{1}{n} \sum_{i=1}^n \int_0^{\tau}( \BS{Z}_i -\bar {\BS{Z}} (t,\theta, \BS{\beta}) ) dN_i(t),
\end{eqnarray*}
respectively, where
\begin{equation} \label{eq-barD-barZ}
\bar D(t,\theta, \BS{\beta})= \sum_{i=1}^n D_i w_i(t,  \theta, \BS{\beta}), ~~ \bar {\BS{Z}} (t,\theta, \BS{\beta})= \sum_{i=1}^n \BS{Z}_i w_i(t,  \theta, \BS{\beta})
\end{equation}
are the weighted averages of the covariate processes with weights $w_i(t,  \theta, \BS{\beta})$ given by 
\[
 w_i(t,\theta, \BS{\beta})=\frac{\tilde w_i(t, \theta, \BS{\beta})}{ \sum_{j=1}^n \tilde w_j(t, \theta, \BS{\beta})}.
\]

The true parameter vector $(\theta_0, \BS{\beta}_0^{\top})^{\top}$ satisfies the following population moment conditions
\begin{eqnarray}\label{eqpartiallogparlik}
	\mathbb{E} [\dot l_{ \theta}(\{\mathcal{O}_i\}_{i=1}^n;\theta_0, \BS{\beta}_0)]=0 , ~~
	\mathbb{E} [\dot l_{ \BS{\beta}}(\{\mathcal{O}_i\}_{i=1}^n;\theta_0, \BS{\beta}_0)]=0.
\end{eqnarray}
In practice, estimation and inference rely on the empirical counterparts of these moment conditions and 
estimated nuisance quantities. Substituting estimated nuisance parameters can introduce non-negligible bias, which renders valid inference difficult. Right censoring and the intrinsic dependence structure of survival data, viewed through continuous-time counting processes, further complicate the joint inference for causal and survival effects in high-dimensional settings.

Our work develops debiased estimation procedures and valid inference for both the causal effect parameter $\theta_0$ and the covariate coefficient vector $\bbeta_0$, using the Cox model in \eqref{eqmodelcox} and the propensity score working model in \eqref{eqmodellog}. We first treat the high-dimensional covariate coefficient vector $\bbeta_0$ as a nuisance parameter, and construct a debiased estimator for $\theta_0$ through a Neyman near-orthogonal estimating equation \citep{chernozhukov2018double}, which removes the first-order effect of estimating the association between $D$ and $\BS{Z}$ under working model (\ref{eqmodellog}). We then extend the framework to obtain valid inference for the high-dimensional survival covariate effects $\bbeta_0$.

\section{High-dimensional survival treatment effect inference with HSCI}\label{Sec3}

\subsection{Neyman orthogonality for HSCI} \label{Sec3.2}

We are interested in the scalar causal parameter $\theta_0$, and thus regard $\BS{\beta}_0$ as a high-dimensional nuisance parameter vector. Based on (\ref{eqpartiallogparlik}), a natural estimator for $\theta_0$ is obtained by solving the following plug-in estimating equation
\begin{eqnarray*}
	\dot l_{ \theta}(\{\mathcal{O}_i\}_{i=1}^n;\theta, \hat{\BS{\beta}})=0,
\end{eqnarray*}
where $\hat {\BS{\beta}}$ is an initial estimator of $\BS{\beta}_0$. In high-dimensional settings, regularized estimators are needed to control variance inflation, but the penalty induces regularization bias. Consequently, $\hat{\BS{\beta}}$ converges typically at a slower rate, and the resulting plug-in estimator of $\theta_0$ generally fails to attain the $n^{-1/2}$ rate. Moreover, model (\ref{eqmodellog}) shows that the treatment indicator $D$ is associated with high-dimensional covariate vector $\BS{Z}$ through the nuisance parameter vector $\BS{\gamma}_0$, which provides additional information but also complicates inference for $\theta_0$. To address these challenges, we exploit the framework in \cite{chernozhukov2018double} and construct a Neyman-orthogonal estimating equation for $\theta_0$ whose first-order sensitivity to the nuisance estimation error is reduced. We also employ data splitting to mitigate the overfitting bias.

Following Section 2.2.1 of \cite{chernozhukov2018double} and using the moment conditions in (\ref{eqpartiallogparlik}), we construct a debiased estimator of $\theta_0$ by solving the following Neyman-orthogonal estimating equation for $\theta$
\begin{eqnarray}\label{eqortho}
	\Phi (\{\mathcal{O}_i\}_{i=1}^n;\theta, \BS{\eta}) :=\dot l_{ \theta}(\{\mathcal{O}_i\}_{i=1}^n;\theta, \BS{\beta})- {\BS{\mu}}^{\top} \dot l_{ \BS{\beta}}(\{\mathcal{O}_i\}_{i=1}^n;\theta, \BS{\beta}) = 0,
\end{eqnarray}
where $\BS{\eta}=(\BS{\beta}^{\top},{\BS{\mu}}^{\top} )^{\top} \in \mathbb{R}^{2p}$ is the nuisance parameter vector, and $\BS{\mu} = \BS{\mu}(\theta, \BS{\beta}, \BS{\gamma})  \in \mathbb{R}^{p}$ is a deterministic orthogonalization parameter determined by
\begin{eqnarray}\label{eqnusipara}
	\BS{J}_{\BS{\beta} \theta} - \BS{J}_{\BS{\beta}\BS{\beta}} \BS{\mu} =0
\end{eqnarray}
with
\[
\BS{J}= \left(\begin{array}{cc}
{J}_{\theta \theta} & \BS{J}_{ \theta \BS{\beta}} \\
\BS{J}_{\BS{\beta} \theta} & \BS{J}_{\BS{\beta}\BS{\beta}}
\end{array} \right)
=-  \frac{\partial}{\partial {(\theta,{\BS{\beta}}^{\top})}} \Big\{ \mathbb{E} \Big[\frac{\partial  l(\{\mathcal{O}_i\}_{i=1}^n;\theta, \BS{\beta})}{\partial_{(\theta,{\BS{\beta}}^{\top})^{\top}}} \Big] \Big\}\Big|_{\theta=\theta_0, \,\BS{\beta}=\BS{\beta}_0}.
\]
Provided that $\BS{J}$ exists and $\BS{J}_{\BS{\beta}\BS{\beta}}$ is invertible, an application of Lemma 2.1 of \cite{chernozhukov2018double} shows that $\Phi (\{\mathcal{O}_i\}_{i=1}^n;\theta, \BS{\eta})$ is Neyman orthogonal at $(\theta_0,\BS{\eta}_0)$, where $\BS{\eta}_0=(\BS{\beta}_0^{\top},{\BS{\mu}}_0^{\top} )^{\top}$ is the true nuisance value and ${\BS{\mu}}_0={\BS{J}}^{-1}_{\BS{\beta}\BS{\beta}} \BS{J}_{\BS{\beta} \theta}$.

Constructing an estimator of $\theta_0$ from \eqref{eqortho} thus requires estimates of ${\bbeta}$ and 
the expected information matrix $\BS{J}$. Calculating the second derivatives, interchanging the differentiation and expectation, and applying the Doob--Meyer decomposition yield that 
\begin{eqnarray*}
	{J}_{\theta \theta} &=&  \mathbb{E} \left\{  \frac{1}{n} \sum_{i=1}^n \int_0^{\tau}  ( D_i- \bar D (t,\theta_0, \BS{\beta}_0) )^{2} \tilde w_i(t,  \theta_0, \BS{\beta}_0)  \lambda_0(t) dt  \right\}, \\
	\BS{J}_{\BS{\beta} \theta} &=&  \mathbb{E} \left\{ \frac{1}{n} \sum_{i=1}^n \int_0^{\tau}  ( \BS{Z}_i - \bar {\BS{Z}} (t,\theta_0, \BS{\beta}_0)) (D_i - \bar D(t,\theta_0, \BS{\beta}_0) )  \tilde w_i(t,  \theta_0, \BS{\beta}_0)  \lambda_0(t) dt \right\}   = \BS{J}_{\theta \BS{\beta} }^{\top},  \\
    \BS{J}_{\BS{\beta}\BS{\beta}} &=&  \mathbb{E} \left\{  \frac{1}{n} \sum_{i=1}^n \int_0^{\tau}  ( \BS{Z}_i- \bar {\BS{Z}} (t,\theta_0, \BS{\beta}_0) )^{\otimes 2} \tilde w_i(t,  \theta_0, \BS{\beta}_0)  \lambda_0(t) dt  \right\} .
\end{eqnarray*}
Since terms $\bar D(t,\theta_0, \BS{\beta}_0)$ and $\bar {\BS{Z}} (t,\theta_0, \BS{\beta}_0)$ depend on the full risk set, the summands in ${J}_{\theta\theta}$, $\BS{J}_{\BS{\beta}\theta}$, and $\BS{J}_{\BS{\beta}\BS{\beta}}$ are not independent. Such dependence complicates direct estimation of $\BS{J}$.

To alleviate this dependence, we exploit the population counterparts of the risk-set averages in \eqref{eq-barD-barZ} and approximate ${J}_{\theta \theta}$, $\BS{J}_{\BS{\beta} \theta}$, and $\BS{J}_{\BS{\beta}\BS{\beta}}$ with the simpler population quantities ${\Sigma}_{\theta \theta}$, $\BS{\Sigma}_{\BS{\beta} \theta}$, and $\BS{\Sigma}_{\BS{\beta}\BS{\beta}}$ defined as 
\begin{eqnarray*}
	{\Sigma}_{\theta \theta}
	&=&  \mathbb{E} \left\{ \int_0^{\tau}  ( D- U_D (t,\theta_0, \BS{\beta}_0))^{2} \tilde w(t, \theta_0, \BS{\beta}_0) \lambda_0(t) dt  \right\},\\
	\BS{\Sigma}_{\BS{\beta} \theta}
	&=&  \mathbb{E} \left\{ \int_0^{\tau}  ( \BS{Z} - \BS{U}_{\BS{Z}} (t,\theta_0, \BS{\beta}_0)) (D - U_D(t,\theta_0, \BS{\beta}_0))  \tilde w(t, \theta_0, \BS{\beta}_0) \lambda_0(t) dt \right\} = \BS{\Sigma}_{ \theta \BS{\beta}}^{\top}, \\
	\BS{\Sigma}_{\BS{\beta}\BS{\beta}}
	&=&  \mathbb{E} \left\{ \int_0^{\tau}  ( \BS{Z} - \BS{U}_{\BS{Z}} (t,\theta_0, \BS{\beta}_0))^{\otimes 2} \tilde w(t, \theta_0, \BS{\beta}_0) \lambda_0(t) dt  \right\},
\end{eqnarray*}
where $U_D(t,\theta, \BS{\beta})$ and $\BS{U}_{\BS{Z}} (t,\theta, \BS{\beta})$ are the population counterparts of $\bar D(t,\theta_0, \BS{\beta}_0)$ and $\bar {\BS{Z}} (t,\theta_0, \BS{\beta}_0)$ given by
\[
U_D(t,\theta, \BS{\beta})=\frac{\mathbb{E}[D  \tilde w(t, \theta, \BS{\beta}) ]}{\mathbb{E}[\tilde w(t, \theta, \BS{\beta}) ]},~~\BS{U}_{\BS{Z}} (t,\theta, \BS{\beta})=\frac{\mathbb{E}[\BS{Z} \tilde w(t, \theta, \BS{\beta}) ]}{\mathbb{E}[\tilde w(t, \theta, \BS{\beta}) ]}.
\]

Under the SUTVA assumption, it holds that $Y(t)= D Y^{(1)}(t) + (1-D)Y^{(0)}(t)$ and $N(t)= D N^{(1)}(t) + (1-D)N^{(0)}(t)$.
Using the definition of $\tilde w(t, \theta, \BS{\beta})$ in \eqref{eq-def-tilde-w} along with the propensity score model (\ref{eqmodellog}), we can obtain the following equivalent representation of ${\Sigma}_{\theta \theta}$, $\BS{\Sigma}_{\BS{\beta}\BS{\beta}}$, and $\BS{\Sigma}_{\BS{\beta} \theta}$.
\begin{lemma} \label{lemma-equi-Sigma}
Under the strong ignorability assumption \citep{rubin1978bayesian} that $(Y^{(1)}(t), Y^{(0)}(t))$ and $D$ are independent conditional on $\BS{Z}$, we have that 
\begin{eqnarray}
    {\Sigma}_{\theta \theta}&=&  \mathbb{E} \left\{ \int_0^{\tau}  ( 1 - U_D (t,\theta_0, \BS{\beta}_0))^{2} m_0(\BS{Z}) dN^{(1)}(t)  \right\} \nonumber\\
	& & + \mathbb{E} \left\{ \int_0^{\tau} U_D (t,\theta_0, \BS{\beta}_0)^{2}  (1-m_0(\BS{Z})) dN^{(0)}(t)  \right\}, \label{eq-Sigma-2}\\
	\BS{\Sigma}_{\BS{\beta}\BS{\beta}} &=&  \mathbb{E} \left\{ \int_0^{\tau}  ( \BS{Z} -  \BS{U}_{\BS{Z}} (t,\theta_0, \BS{\beta}_0))^{\otimes 2} m_0(\BS{Z}) dN^{(1)}(t)  \right\} \nonumber\\
	& & + \mathbb{E} \left\{ \int_0^{\tau}(\BS{Z} -  \BS{U}_{\BS{Z}} (t,\theta_0, \BS{\beta}_0))^{\otimes 2}  (1-m_0(\BS{Z})) dN^{(0)}(t)  \right\},  \label{eq-Sigma-3} \\
     \BS{\Sigma}_{\BS{\beta} \theta}
	&=& \mathbb{E} \left\{   \int_0^{\tau}  ( \BS{Z} -  \BS{U}_{\BS{Z}} (t,\theta_0, \BS{\beta}_0)) (1 - U_D(t,\theta_0, \BS{\beta}_0)) m_0(\BS{Z}) dN^{(1)}(t)  \right\} \nonumber \\
	&  & - \mathbb{E} \left\{   \int_0^{\tau}  ( \BS{Z} -  \BS{U}_{\BS{Z}} (t,\theta_0, \BS{\beta}_0)) U_D(t,\theta_0, \BS{\beta}_0) (1-m_0(\BS{Z})) dN^{(0)}(t) \right\}. \label{eq-Sigma-1}
\end{eqnarray}
\end{lemma}

Lemma \ref{lemma-equi-Sigma} above links the information quantities employed in the orthogonal score to observable treatment and control components through the propensity score, rewriting the relevant population covariance matrices in forms that separate the treated and untreated potential counting processes. Such representation clarifies how the propensity model provides the weights needed to approximate the population information matrix and gives each population quantity a direct sample counterpart, thereby connecting the causal identification assumptions to the practical estimation of the orthogonalization parameter vector $\BS{\mu}_{a0}$ and making the subsequent estimators and error bounds transparent.

If $\BS{\Sigma}_{\BS{\beta}\BS{\beta}}$ is invertible, a natural approximation to $\BS{\mu}_0$ is $\BS{\mu}_{a0}=\BS{\mu}_{a}(\theta_0, \BS{\beta}_0, \BS{\gamma}_0):=\BS{\Sigma}_{\BS{\beta}\BS{\beta}}^{-1} \BS{\Sigma}_{\BS{\beta} \theta}$. Denote by  $\BS{\eta}_{a0}=(\BS{\beta}_0^{\top}, \BS{\mu}_{a0}^{\top})^{\top}$.
Proposition \ref{lmnearortho} in Section \ref{new.Sec.theory.2} later reveals that estimating equation (\ref{eqortho}) satisfies the Neyman near-orthogonality at $(\theta_0, \BS{\eta}_{a0})$ over a neighborhood $\mathcal{T}_N$ of the nuisance parameter vector $\BS{\eta}= (\bbeta^{\top}, \BS{\mu}^{\top})^{\top}$. Consequently, a debiased estimator of $\theta_0$ can be obtained by solving (\ref{eqortho}) after plugging in estimators of nuisance parameter vector $\BS{\eta}_{a0}$.

\subsection{Initial parameter estimators} \label{Sec: initial-estimators}

We next specify the initial estimators for $\BS{\beta}_0$ and $\BS{\mu}_{a0}$. In view of Lemma \ref{lemma-equi-Sigma}, $\BS{\mu}_{a0}$ is a deterministic function of parameters $\theta_0$, $\BS{\beta}_0$, and $\BS{\gamma}_0$. Hence, it suffices to construct initial estimators for these parameters. Under the high-dimensional survival treatment effect model in \eqref{eqmodelcox} and \eqref{eqmodellog}, regularization-based methods such as the Lasso \citep{tibshirani1996regression} 
can be used to obtain initial coefficient estimates. For concreteness, we apply the Lasso to estimate $\theta_0$, $\boldsymbol{\beta}_0$, and $\boldsymbol{\gamma}_0$.

\textit{Initial estimators of \texorpdfstring{$\theta_0$}{theta0} and \texorpdfstring{${\bbeta}_0$}{beta0}}. Under the Cox model (\ref{eqmodelcox}), the $L_1$-penalized maximum partial likelihood estimators \citep{huang2013oracle} for $\theta_0$ and $\BS{\beta}_0$ are given by
\begin{eqnarray}\label{eqBetaL1}
	(\hat \theta, {\hat{\BS{\beta}}}^{\top})^{\top}= \argmin_{\theta \in \mathbb{R}, \,\BS{\beta} \in \mathbb{R}^p} \big\{ -l(\{\mathcal{O}_i\}_{i=1}^n ; \theta, \BS{\beta}) + \lambda_{\theta\BS{\beta} } (|\theta| + \| \BS{\beta} \|_{1}) \big\},
\end{eqnarray}
where $ \lambda_{\theta\BS{\beta}} $ is a nonnegative regularization parameter.

\textit{Initial estimator of \texorpdfstring{${\bgamma}_0$}{gamma0} and $m_0(\BS{z})$}. Under the logistic model (\ref{eqmodellog}), the log-likelihood function divided by $n$ is

\[
l_{\BS{\gamma}}(\{D_i, \BS{Z}_i\}_{i=1}^n; \BS{\gamma}) = -\mathbb{E}_n \rho_{\BS{\gamma}}(\BS{Z}_i, D_i),
\]
where $\rho_{\BS{\gamma}}(\BS{z}, d) =\rho({\BS{z}}^{\top}\BS{\gamma}, d)$, $\rho(u, d)=-du +\log(1+e^u)$ is the logistic loss function, and $\mathbb{E}_n[X_i] := n^{-1} \sum_{i=1} X_i$ represents the empirical mean of $\{X_i\}_{i=1}^n$.
Thus, the Lasso estimator for $\BS{\gamma}$ is defined as
\begin{eqnarray}\label{eqGamL1}
	\hat {\BS{\gamma}}= \argmin_{\BS{\gamma} \in \mathbb{R}^p} ( \mathbb{E}_n \rho_{\BS{\gamma}}(\BS{Z}_i, D_i)+\lambda_{\BS{\gamma}} \| \BS{\gamma} \|_{1}),
\end{eqnarray}
where $\lambda_{\BS{\gamma}}$ is a nonnegative regularization parameter. The resulting estimated propensity score is
\begin{equation} \label{eq-propen-score-logistic}
\hat{m}(\BS{Z}_i) = {\exp\{{\BS{Z}_i}^{\top} \hat{\BS{\gamma}}\}}/(1+\exp\{{\BS{Z}_i}^{\top} \hat{\BS{\gamma}}\}), \quad i=1, \dots, n.
\end{equation}

\textit{Initial estimator of \texorpdfstring{${\bmu}_{a0}$}{mua0}}. We proceed to estimate the nuisance parameter vector ${\bmu}_{a0} =\BS{\Sigma}_{\BS{\beta}\BS{\beta}}^{-1} \BS{\Sigma}_{\BS{\beta} \theta} $. In light of Lemma \ref{lemma-equi-Sigma}, we first estimate $ \BS{\Sigma}_{\BS{\beta} \theta}$ and $ \BS{\Sigma}_{\BS{\beta}\BS{\beta}}$ as 
\begin{eqnarray}
	\hat {\BS{\Sigma}}_{\BS{\beta} \theta} &=&  \frac{1}{n_1} \sum_{i=1}^n  D_i \delta_i  ( \BS{Z}_i - \bar {\BS{Z}} (X_i,\hat \theta, \hat {\BS{\beta}}) ) (1 - \bar D(X_i,\hat \theta, \hat {\BS{\beta}})) \hat{m}(\BS{Z}_i)  \nonumber  \\
	& & - \frac{1}{n_0} \sum_{i=1}^n  (1-D_i) \delta_i ( \BS{Z}_i - \bar {\BS{Z}} (X_i,\hat \theta, \hat {\BS{\beta}})) \bar D(X_i,\hat \theta, \hat {\BS{\beta}}) (1-\hat{m}(\BS{Z}_i) ),  \label{eq-hatSigma_2} \\
	\hat {\BS{\Sigma}}_{\BS{\beta}\BS{\beta}} &=&  \frac{1}{n_1} \sum_{i=1}^n  D_i \delta_i  ( \BS{Z}_i - \bar {\BS{Z}}(X_i,\hat \theta, \hat {\BS{\beta}}))^{\otimes 2}  \hat{m}(\BS{Z}_i)  \nonumber  \\
	& & + \frac{1}{n_0} \sum_{i=1}^n  (1-D_i) \delta_i  ( \BS{Z}_i- \bar {\BS{Z}} (X_i,\hat \theta, \hat {\BS{\beta}}))^{\otimes 2}  (1-\hat{m}(\BS{Z}_i) ), \label{eq-hatSigma_3}
\end{eqnarray}
where $n_0+n_1=n$, and $n_1$ and $n_0$ stand for the sample sizes of the treatment and control groups, respectively.

Consequently, the estimator $\hat {\BS{\mu}}_a$ of $\BS{\mu}_{a0}$ can be obtained by the constrained $l_1$ minimization method of \cite{cai2011constrained}
\begin{eqnarray}\label{eqClime1}
	\hat {\BS{\mu}}_a:= \hat {\BS{\mu}}_a( \hat \theta, \hat {\BS{\beta}}, \hat{\BS{\gamma}})=\argmin_{ \BS{b} \in \mathbb{R}^p}  \|\BS{b}\|_1 \ \text{ such that } \ \| \hat {\BS{\Sigma}}_{\BS{\beta}\BS{\beta}} \BS{b} -\hat {\BS{\Sigma}}_{\BS{\beta} \theta} \|_{\infty} \le \lambda_n,
\end{eqnarray}
where $\lambda_n$ is a nonnegative tuning parameter.

\subsection{HSCI estimator for survival treatment effect inference} \label{Sec3.3}

Motivated by the Neyman near-orthogonality result in Proposition \ref{lmnearortho}, we construct an asymptotic Neyman near-orthogonal estimating equation for $\theta_0$ by applying \eqref{eqortho} with the plug-in nuisance estimator $\boldsymbol{\eta}_{a} = (\widehat{\bbeta}^{\top}, \widehat{\bmu}_{a}^{\top})^{\top}$.
We now introduce the HSCI inference procedure for the survival treatment effect $\theta_0$.

If the same data is used both to estimate the nuisance parameter vector $\boldsymbol{\eta}_{a}$ and to construct the orthogonal estimating equation \eqref{eqortho}, the overfitting bias may arise, as discussed in \cite{chernozhukov2018double}. We thus adopt the cross-fitting, which reduces the overfitting bias and improves the stability of the estimating equations. Specifically, let $K$ be a fixed positive integer and assume, for simplicity, that $n$ is divisible by $K$. Randomly partition the $n$ observations into $K$ disjoint blocks $\{I_k\}_{k=1}^K$, each of size $m=n/K$. For each $k=1, \dots,K$, we use the observations indexed by $I_k^c :=\{1,2,\dots, n\} \setminus I_k$ to construct initial estimators of $\theta_0$, $\BS{\beta}_0$, $\BS{\gamma}_0$, and $\BS{\mu}_{a0}$ following the procedure described in Section \ref{Sec: initial-estimators}.
Denote the resulting estimators as $\hat {\theta}_k=\hat {\theta}( (\mathcal{O}_i)_{i \in I_k^c} )$, $\hat {\BS{\beta}}_k=\hat {\BS{\beta}}( (\mathcal{O}_i)_{i \in I_k^c} )$, $\hat {\BS{\gamma}}_k=\hat {\BS{\gamma}}( (\mathcal{O}_i)_{i \in I_k^c} )$, $\hat {\BS{\mu}}_{ak}=\hat {\BS{\mu}}_a(\hat {\theta}_k, \hat {\BS{\beta}}_k,\hat {\BS{\gamma}}_k)$, and set $\hat {\BS{\eta}}_k=( {\hat {\BS{\beta}}_k}^{\top}, {\hat {\BS{\mu}}_{ak}}^{\top})^{\top}$. By construction, $\hat{\BS{\eta}}_k$ depends only on the observations in $I_k^c$. The estimator of $\theta_0$ is then obtained by one of the following two ways of combining the $K$ blocks.

{\it Approach 1: averaging over solutions}. For any $k \in \{1, \cdots, K\}, \theta \in \mathbb{R}, \BS{\beta} \in \mathbb{R}^{p}$, and $\BS{\eta} \in \mathbb{R}^{2p}$, let us define 
\begin{align*}
	w_{k,j} (t, \theta, \BS{\beta})& = \frac{Y_j(t) \exp\{D_j \theta +\BS{Z}_j^{\top} \BS{\beta}\}}{\sum_{j \in I_k} Y_j(t) \exp\{D_j \theta +\BS{Z}_j^{\top} \BS{\beta}\}}, \ j \in I_k, \\
	\bar D_k (t, \theta, \BS{\beta})& = \sum_{j \in I_k} D_j w_{k,j} (t, \theta, \BS{\beta}), \quad \bar {\BS{Z}}_k (t, \theta, \BS{\beta}) =\sum_{j \in I_k} \BS{Z}_j w_{k,j} (t, \theta, \BS{\beta}),
\end{align*}
and
\begin{align*}
	\Phi (\{\mathcal{O}_i\}_{i \in I_k}; \theta, \BS{\eta})
	= \dot l_{ \theta}(\{\mathcal{O}_i\}_{i \in I_k}; \theta, \BS{\beta})-  \BS{\mu}_a^{\top} \dot l_{ \BS{\beta}}(\{\mathcal{O}_i\}_{i \in I_k}; \theta, \BS{\beta}).
\end{align*}
Averaging the solutions to the orthogonal estimating equations across the $K$ blocks yields the first HSCI survival treatment effect estimator $\check \theta_1$
  \begin{align}\label{eq: check-theta1}
     \check \theta_1 = K^{-1}\sum_{k=1}^K \check{\theta}_{1,k},
  \end{align}
where $\check{\theta}_{1,k}$ solves the equation
\begin{eqnarray}\label{eqesttheta1}
	\Phi (\{\mathcal{O}_i\}_{i \in I_k};\check{\theta}_{1,k}, \hat {\BS{\eta}}_k) = 0
\end{eqnarray}
for $k=1, \dots, K$.

{\it Approach 2: solving the averaged orthogonal equation}. The second HSCI survival treatment effect estimator $\check \theta_2$ is obtained by solving the averaged orthogonal equation across the $K$ blocks
\begin{eqnarray}\label{eqesttheta2}
	K^{-1}\sum_{k=1}^K \Phi (\{\mathcal{O}_i\}_{i \in I_k};\check \theta_2, \hat {\BS{\eta}}_k) = 0.
\end{eqnarray}
In general, the exact solutions to equations (\ref{eqesttheta1}) and (\ref{eqesttheta2}) may not be analytically tractable. Instead, we consider approximate solutions satisfying that 
\begin{eqnarray}\label{eqapptheta1}
    \big| \Phi (\{\mathcal{O}_i\}_{i \in I_k};\check{\theta}_{1,k}, \hat {\BS{\eta}}_k)\big| \le  \inf_{\theta \in  \Theta} \big|\Phi (\{\mathcal{O}_i\}_{i \in I_k}; \theta, \hat {\BS{\eta}}_k)\big| +o(n^{-1/2}) \ \mbox{ for } k = 1, \ldots, K,
\end{eqnarray}
and
\begin{eqnarray}\label{eqapptheta2}
 \Big| K^{-1}\sum_{k=1}^K \Phi (\{\mathcal{O}_i\}_{i \in I_k};\check \theta_2, \hat {\BS{\eta}}_k) \Big|
     \le    \inf_{\theta \in  \Theta} \Big| K^{-1}\sum_{k=1}^K \Phi (\{\mathcal{O}_i\}_{i \in I_k}; \theta, \hat {\BS{\eta}}_k) \Big|+o(n^{-1/2}),
\end{eqnarray}
respectively.

\subsection{HSCI estimator for survival covariate effect inference} \label{Sec3.4}

Section \ref{Sec3.3} constructs two debiased HSCI estimators $\check \theta = \check \theta_i$, $i \in \{1,2\}$, for the survival causal parameter $\theta_0$; Section \ref{new.Sec.theory.2} later unveils that they are both $\sqrt{n}$ consistent. We now turn to inference for high-dimensional survival covariate effect vector $\BS{\beta}_0$. Focusing on the Cox model (\ref{eqmodelcox}), we exploit a debiasing technique related to \cite{yu2021confidence}; this idea is also connected to the debiased estimation approaches of \cite{van2014asymptotically} and \cite{zhang2014confidence} for the linear and generalized linear models. Based on the initial estimator $\hat{\BS{\beta}}$ for $\BS{\beta}_0$ and a $\sqrt{n}$-consistent estimator $\check \theta$ for $\theta_0$, we construct the debiased estimator $\check {\BS{\beta}}$ as
\begin{align}\label{eqestbeta}
	\check {\BS{\beta}} = \hat{\BS{\beta}} + \hat {\BS{\Xi}} \dot l_{ \BS{\beta}}(\{\mathcal{O}_i\}_{i=1}^n; \check \theta ,\hat{\BS{\beta}}),
\end{align}
where $\hat {\BS{\Xi}} =( {\hat {\BS{\Xi}}}_1, \dots, {\hat {\BS{\Xi}}}_p)^{\top}$ is the CLIME sparse precision matrix estimator defined as
\begin{eqnarray}\label{eqClime2}
	{\hat {\BS{\Xi}}}_j =\argmin_{ \BS{d}_j \in \mathbb{R}^p} \| \BS{d}_j\|_1 \ \text{ such that } \ \|\check {\BS{\Sigma}}_{\BS{\beta}\BS{\beta}} \BS{d}_j - \BS{e}_j \|_{\infty} \le \lambda_n^*.
\end{eqnarray}
Here, $\BS{e}_j$ is the unit vector whose $j$th component is one and whose remaining components are zero, $\lambda_n^*$ is a nonnegative tuning parameter, and $\check {\BS{\Sigma}}_{\BS{\beta}\BS{\beta}}$ is an improved estimator of $\BS{\Sigma}_{\BS{\beta}\BS{\beta}}$ 
 \begin{align*}
   \check {\BS{\Sigma}}_{\BS{\beta}\BS{\beta}} := \hat{\BS{\Sigma}}_{\BS{\beta}\BS{\beta}}(\check{\theta}, \hat{\bbeta}) =&  \frac{1}{n_1} \sum_{i=1}^n  D_i \delta_i  ( \BS{Z}_i - \bar {\BS{Z}}(X_i,\check \theta, \hat {\BS{\beta}}))^{\otimes 2}  \hat{m}(\BS{Z}_i)  \\
	 & + \frac{1}{n_0} \sum_{i=1}^n  (1-D_i) \delta_i  ( \BS{Z}_i- \bar {\BS{Z}} (X_i,\check \theta, \hat {\BS{\beta}}))^{\otimes 2}  (1-\hat{m}(\BS{Z}_i) )
\end{align*}
obtained by plugging in $\check{\theta}$ and $\hat{\bbeta}$.

\section{Asymptotic properties of HSCI} \label{new.Sec.theory}

\subsection{Technical conditions} \label{Sec3.1}

We impose some regularity conditions to facilitate the technical analysis of the suggested HSCI method.

\begin{condition}\label{coninicox}
	(i) The covariate vector $\BS{Z} \in \mathbb{R}^p $ satisfies that $\| \BS{Z} \|_{\infty} \le K_{\BS{Z}}$, where $K_{\BS{Z}} > 0$ is a finite constant.
    
	\noindent(ii) Let $-\ddot l(\{\mathcal{O}_i\}_{i=1}^n; \theta_0, \BS{\beta}_0)$ be the negative Hessian matrix of the log-partial likelihood function (\ref{eqlogparlik}) evaluated at $(\theta_0, \BS{\beta}_0)$, i.e., $-\ddot l(\{\mathcal{O}_i\}_{i=1}^n; \theta_0, \BS{\beta}_0):=-\ddot l(\{\mathcal{O}_i\}_{i=1}^n; \theta, \BS{\beta})  |_{\theta=\theta_0, \BS{\beta}=\BS{\beta}_0}$.
	Assume that $1/\kappa_0=O_p(1)$, where
      \begin{align*}
	    \kappa_0 := \inf_{ \{\BS{v} \in \mathbb{R}^{p+1} \setminus \{0\}:  \|\BS{v}_{S_{cox}^C}\|_1 \le 2 \|\BS{v}_{S_{cox}}\|_1   \} } \frac{(s_{\BS{\beta}_0}+1)^{1/2} \{\BS{v}^{\top} (-\ddot l(\{\mathcal{O}_i\}_{i=1}^n; \theta_0, \BS{\beta}_0)) \BS{v}\}^{1/2}}{\|\BS{v}_{S_{cox}}\|_1}
	 \end{align*}
	is the compatibility factor of $-\ddot l(\{\mathcal{O}_i\}_{i=1}^n; \theta_0, \BS{\beta}_0)$ with $S_{cox} :=\{1\}  \cup \{j+1 : j \in S_{\BS{\beta}_0}, 1 \le j \le p\}$.

	\noindent(iii) $s_{\BS{\beta}_0} = o( n^{1/4}/\sqrt{\log (np)})$.

    \noindent(iv) $\| \BS{\beta}_0 \|_1 \leq C$.
\end{condition}

\begin{condition}\label{coninilog}
 (i) (Strict overlap) There exists some constant $c_0 > 0$ such that $c_0 < P(D=1 | \BS{Z}=\BS{z}) =m_0(\BS{z}) < 1-c_0 $ for all $\BS{z}$ satisfying that $\|\BS{z}\|_{\infty} \leq K_{\BS{Z}}$.

 \noindent(ii) There exists a sequence $\tau_n > 0 $ such that the propensity score estimator $\hat{m}(\BS{z})$ satisfies that 
 $$ \mathbb{P} \Big(\sup_{\BS{z}: \|\BS{z}\|_{\infty} \leq K_{\BS{Z}}} | \hat{m}(\BS{z}) - m_0(\BS{z}) | > \tau_n \Big) \to 0.
 $$
\end{condition}

\begin{condition}\label{coninimu}
	(i) $n_0 \asymp n_1$.

	\noindent(ii) The failure time $T$ has distribution function $F_T(\cdot)$ and bounded density function $f_T(\cdot)$.

	\noindent(iii) $\int_0^{\tau} \lambda_0(t)dt < \infty$ and $P(Y(\tau)=1) >0$.

	\noindent(iv) The matrix $\BS{\Sigma}$ is invertible, and the eigenvalues $0 < \Lambda_1 \le \cdots \le \Lambda_p \leq \Lambda_{p+1 } < \infty$  of $\BS{\Sigma}$ are bounded from below and above.

   \noindent(v)   $  {\|\BS{\Sigma}}^{-1}_{\BS{\beta}\BS{\beta}} \|_{\infty}  s_{{\BS{\mu}}_{a0}} s_{\BS{\beta}_0}^2    \log (np)  = o(\sqrt{n}) $ and $    \tau_n  {\|\BS{\Sigma}}^{-1}_{\BS{\beta}\BS{\beta}} \|_{\infty}  s_{{\BS{\mu}}_{a0}} s_{\BS{\beta}_0} \sqrt{\log (np)}   = o(1)$.

	\noindent(vi) Denote the active set of ${\BS{\mu}}_{a0}$ as $S_{{\BS{\mu}}_{a0}}=\{j: {\BS{\mu}}_{a0,j} \neq 0, 1 \le j \le p\}$, with 
    cardinality $s_{{\BS{\mu}}_{a0}}=| S_{{\BS{\mu}}_{a0}}|=o(\sqrt{n}/\log (np))$.   Moreover, assume that $\|{\BS{\mu}}_{a0}\|_1 \le K_{\BS{\mu}}$ for a finite constant $K_{\BS{\mu}} >0 $, and $\min_{j \in S_{{\BS{\mu}}_{a0}}} |\bmu_{a0, j}| \geq C {\|\BS{\Sigma}}^{-1}_{\BS{\beta}\BS{\beta}} \|_{\infty}  s_{{\BS{\mu}}_{a0}} (s_{\BS{\beta}_0} \sqrt{\log (np) /n} + \tau_n)$ for  a sufficiently large constant $C > 0$ if $S_{{\BS{\mu}}_{a0}} \neq \emptyset$.
\end{condition}

\begin{condition}\label{contheta}
  Assume that $\mathcal{N}_{\theta_0}^*:= \mathcal{N}_{\theta_0} (2C_0 s_{\BS{\beta}_0} \sqrt{\log (np)/{n}} ) \subset \Theta$, where $C_0 >0$ is a sufficiently large constant.

\end{condition}

\begin{condition}\label{condebeta}
	 Let $h_j := \sum_{i=1}^p I_{\{ ({\BS{\Sigma}}_{\BS{\beta}\BS{\beta}}^{-1})_{ij} \neq 0\}}$, $j=1,\dots, p$, and assume that 
	\[
   \|{\Sigma}_{\BS{\beta}\BS{\beta}}^{-1} \|_{\infty}^2     \big(s_{\BS{\beta}_0}  \log (np) /\sqrt{n} + \tau_n \sqrt{\log(np) } \big) \max_{j=1,\dots, p} h_j = o(1).
	\]
\end{condition}

Condition \ref{coninicox} above collects the assumptions needed for the initial high-dimensional Cox estimator. Part (i) bounds the covariates, which simplifies the empirical process control; it can be relaxed to suitable tail assumptions at the cost of more technical arguments. Part (ii) is a compatibility condition for the negative Hessian of the partial likelihood, ensuring that the sparse directions of the parameter space are identifiable, and the $l_1$ estimation error can be controlled \citep{van2009conditions,huang2013oracle}. Part (iii) limits the sparsity of $\BS{\beta}_0$ so that the regularization bias remains asymptotically negligible after debiasing. Part (iv) controls the overall signal strength of the Cox regression coefficients and rules out unstable risk scores.

Condition \ref{coninilog} above governs the treatment assignment model. Strict overlap in part (i) prevents the propensity score from approaching zero or one, ensuring that both treatment groups remain informative over the covariate region under consideration. Part (ii) is stated as a high-level uniform convergence requirement for the estimated propensity score. In particular, the proposition below verifies such requirement in Condition \ref{coninilog} for the sparse logistic Lasso and provides the concrete rate $\tau_n = O\big(s_{\BS{\gamma}_0} \sqrt{\log (np)/n}\big)$ under the compatibility and sparsity assumptions.

\begin{proposition}\label{lminilog}
	Assume that Condition \ref{coninicox}(i) is satisfied, and there exists some constant $\kappa_1 >0$ such that
	\[
	\kappa_1^2 \| \BS{v}_{S_{\BS{\gamma}_0}}\|_1^2 \le s_{\BS{\gamma}_0} \BS{v}^{\top} \mathbb{E}[\BS{Z} {\BS{Z}}^{\top}] \BS{v}
	\]
	for any $\BS{v} \in \mathcal{C}(S_{\BS{\gamma}_0})$, where $\mathcal{C}(S_{\BS{\gamma}_0})=\{ \BS{v} \in \mathbb{R}^{p}: \|\BS{v}_{S_{\BS{\gamma}_0}^C}\|_1 \le 3 \|\BS{v}_{S_{\BS{\gamma}_0}}\|_1\}$. Assume further that $s_{\BS{\gamma}_0} =o(\sqrt{n}/\log (np))$ and  $\lambda_{ \BS{\gamma}}  \asymp  \sqrt{\log (np)/n}$ in (\ref{eqGamL1}). Then for the estimator defined in \eqref{eq-propen-score-logistic}, we have that 
    $$
     \sup_{\BS{z}: \|\BS{z}\|_{\infty} \leq K_{\BS{Z}}} | \hat{m}(\BS{z}) - m_0(\BS{z}) | = O_p(s_{\BS{\gamma}_0} \sqrt{\log (np)/n} ) .
    $$
\end{proposition}

Proposition \ref{lminilog} above shows that the Lasso logistic regression estimator provides the uniform propensity-score approximation required by Condition \ref{coninilog}, thereby translating standard sparsity and compatibility assumptions for the logistic model into the general error sequence $\tau_n$ exploited throughout the later theory. Such rate controls the discrepancy between the estimated propensity score $\hat m(\BS{Z})$ and its population counterpart $m_0(\BS{Z})$, making the propensity-score component of the HSCI procedure operational.

Condition \ref{coninimu} above controls the estimation of the orthogonalization vector $\BS{\mu}_{a0}$ and the population information quantities. Comparable group sizes in part (i) avoid degeneracy in the treatment and control contributions. Parts (ii) and (iii) impose standard survival-analysis regularity conditions on the event-time distribution, the cumulative baseline hazard, and the amount of follow-up through $P(Y(\tau)=1)>0$. Part (iv) requires the population information matrix $\BS{\Sigma}$ to be nondegenerate. Parts (v) and (vi) are high-dimensional rate and sparsity requirements for $\BS{\mu}_{a0}$; they ensure that the \textit{products} of the Cox-model and propensity-score nuisance errors are sufficiently small. In particular, the second requirement in part (v) has a \textit{rate-double-robust} interpretation because it permits one nuisance component to be estimated less accurately when the other component is sufficiently accurate.

Condition \ref{contheta} above is a local-parameter-space condition. It ensures that the neighborhood in which the debiased estimator is shown to lie remains inside the parameter space $\Theta$. This avoids the boundary effects and allows the local expansion of the orthogonal estimating equation around $\theta_0$.

Condition \ref{condebeta} above is employed for inference on the high-dimensional covariate effect vector $\BS{\beta}_0$. It restricts the sparsity of the inverse information matrix through the column sparsity levels $h_j$, and controls the interaction between this sparsity, Cox-model estimation error, and propensity-score estimation error. This condition is the analog of sparse-precision-matrix assumptions adopted in the debiased high-dimensional regression.

\subsection{Asymptotic distributions of HSCI estimators} \label{new.Sec.theory.2}

We first establish the convergence rate of the initial estimators $\hat \theta$ and $\hat {\BS{\beta}}$ introduced in Section \ref{Sec3.2}.

\begin{proposition}\label{lminicox}
	Assume that Condition \ref{coninicox} holds, and set $\lambda_{\theta \BS{\beta}} \asymp \sqrt{\log (np)/{n}}$ in (\ref{eqBetaL1}). Then we have that 
	\begin{eqnarray*}
		| \hat \theta - \theta_0 | +\| \hat {\BS{\beta}} - \BS{\beta}_0 \|_1=O_p(  s_{\BS{\beta}_0} \sqrt{\log (np)/{n}}).
	\end{eqnarray*}
\end{proposition}

Proposition \ref{lminicox} above establishes the baseline high-dimensional convergence rate for the penalized Cox estimator $(\hat\theta,\hat{\BS{\beta}})$, providing the preliminary control needed to construct the final orthogonal estimators in a local neighborhood of the initial estimator. Although this Cox-Lasso rate is slower than $n^{-1/2}$, it is sufficiently sharp for the orthogonal score construction and underscores the need for debiasing: the initial estimator controls the nuisance complexity but does not yield valid root-$n$ inference for $\theta_0$ by itself.

We next derive the convergence rates for $\hat {\BS{\mu}}_a$ in both the $l_\infty$-norm and  $l_1$-norm.

\begin{proposition}\label{lminimut}
    Assume that Conditions \ref{coninicox}--\ref{coninimu} are satisfied. Let $\lambda_{\theta \BS{\beta}} \asymp \sqrt{\log (np) /n}$ in (\ref{eqBetaL1}),
    and $\lambda_{n} \asymp (s_{\BS{\beta}_0} \sqrt{\log (np) /n} + \tau_n)$ in (\ref{eqClime1}). Then we have that 
	\begin{eqnarray}
		\| \hat {\BS{\mu}}_{a} - {\BS{\mu}}_{a0} \|_{\infty} & = &  O_p({\|\BS{\Sigma}}^{-1}_{\BS{\beta}\BS{\beta}} \|_{\infty}  (s_{\BS{\beta}_0} \sqrt{\log (np) /n} + \tau_n) ), \label{eq-lminimut-1}\\
		\| \hat {\BS{\mu}}_{a} - {\BS{\mu}}_{a0} \|_{1} &= & O_p({\|\BS{\Sigma}}^{-1}_{\BS{\beta}\BS{\beta}} \|_{\infty}  s_{{\BS{\mu}}_{a0}} (s_{\BS{\beta}_0} \sqrt{\log (np) /n} + \tau_n) ). \label{eq-lminimut-2}
	\end{eqnarray}
\end{proposition}

Proposition \ref{lminimut} above controls the estimation error of the Riesz-representer-type orthogonalization vector $\BS{\mu}_{a0}$ through its estimator $\hat{\BS{\mu}}_a$, the component that orthogonalizes the estimating equation for $\theta_0$ against the high-dimensional Cox nuisance parameter vector. The resulting $l_\infty$ bound ensures the entrywise stability of the CLIME-type solution, while the $l_1$ bound controls the expansions of the orthogonal score against high-dimensional nuisance errors, unveiling how the Cox and propensity-score nuisance rates affect jointly the final treatment effect inference.

We aim to establish the asymptotic theory for the HSCI estimators $\check{\theta} = \check{\theta}_i$, $i \in \{1,2\}$, introduced in Section \ref{Sec3.3}. These estimators are asymptotically equivalent. Since the Neyman orthogonality is central to inference on survival treatment effect $\theta_0$, we first verify the near-orthogonality of estimating equation \eqref{eqortho} at the population parameter vector $(\theta_0, \BS{\eta}_{a0})$, where $\BS{\eta}_{a0} = (\BS{\beta}_0^{\top}, \BS{\mu}_{a0}^{\top})^{\top}$.

\begin{proposition}\label{lmnearortho}
	Assume that Conditions \ref{coninicox}(i), \ref{coninicox}(iii), and \ref{coninimu}(ii)--(iv) and (vi) are satisfied. Denote by $a_n  = C s_{\BS{\beta}_0} \log (np)/n$ with a sufficiently large constant $C > 0$.
	Then the estimating equation $\Phi (\{\mathcal{O}_i\}_{i=1}^n;\theta, \BS{\eta}_a)$ obeys the $a_n$-Neyman near-orthogonality at $(\theta_0, \BS{\eta}_{a0})$ with respect to the nuisance realization set $\mathcal{T}_N$, where $\mathcal{T}_N =\{\BS{\eta}_a= (\BS{\beta}^{\top}, \BS{\mu}_a^{\top})^{\top} \in \mathcal{B} \times \mathcal{U} \subset  \mathbb{R}^{2p}: \, \|\BS{\beta}-\BS{\beta}_0\|_1 \le \tilde C s_{\BS{\beta}_0} \sqrt{\log (np)/n}, \, 
     \|\BS{\mu}_a-\BS{\mu}_{a0} \|_1 \le \tilde C {\|\BS{\Sigma}}^{-1}_{\BS{\beta}\BS{\beta}} \|_{\infty}  s_{{\BS{\mu}}_{a0}} (s_{\BS{\beta}_0} \sqrt{\log (np) /n} + \tau_n) \}
    $ for some sufficiently large constant $\tilde C >0$.
\end{proposition}

Proposition \ref{lmnearortho} above verifies the Neyman near-orthogonality property of the proposed score, providing the key mechanism behind the HSCI estimator. While exact orthogonality would eliminate first-order nuisance effects entirely, such result quantifies the remaining approximation error by $a_n$, which is sufficiently small under the stated sparsity conditions. Consequently, the first-stage nuisance estimation does \textit{not} contribute to the first-order distribution of the final estimator, justifying the plug-in score equation even when the nuisance parameters are estimated at slower high-dimensional rates.

Building on Propositions \ref{lminicox}--\ref{lmnearortho} above, we are now ready to establish the asymptotic normality of the HSCI estimators $\check \theta = \check \theta_i$ with $i \in \{1,2\}$ for survival treatment effect inference in the theorem below.

\begin{theorem}\label{ththeta}
	For any $P \in \mathcal{P}$, assume that Conditions \ref{coninicox}--\ref{contheta} are satisfied. Then with asymptotic probability one, there exist estimators $\check \theta_1$ and $\check \theta_2$ lying in $\mathcal{N}_{\theta_0}^*$ that solve \eqref{eqesttheta1} and \eqref{eqesttheta2}, respectively. Further, the estimators $\check \theta_1$ and $\check \theta_2$ in neighborhood $\mathcal{N}_{\theta_0}^*$ satisfy that 
	\begin{align}\label{eq: limiting-dist}
		{\sigma}^{-1} \sqrt{n} (\check \theta_l-\theta_0)  \stackrel{\mathcal D} {\rightarrow} N(0, 1), \ l=1,2
	\end{align}
	uniformly over $P \in \mathcal{P}$ as $n \rightarrow \infty$, where $\stackrel{\mathcal D}{\rightarrow} $ represents convergence in distribution, and
	\begin{align*}
       {\sigma}^2=J_0^{-1} \mathbb{E}[ n \Phi (\{\mathcal{O}_i\}_{i=1}^n;\theta_0, \BS{\eta}_{a0})^2] J_0^{-1}
	\end{align*}
    with $J_0 :=  -\Sigma_{\theta \theta} + \BS{\mu}_{a0}^{\top} \BS{\Sigma}_{\BS{\beta} \theta} $.
	The result in \eqref{eq: limiting-dist} also holds when ${\sigma}^2$ is replaced with ${\hat \sigma}_l^2$, where
	\begin{eqnarray*}
		{\hat \sigma}_l^2=\hat {J}_{0,l}^{-1} \left\{\frac{1}{n} \sum_{k = 1}^K \sum_{i \in I_k} \int_0^{\tau} [(D_i - \bar D_k(t, \check \theta_l,\hat { \BS{\beta}}_k))  - \hat{\BS{\mu}}_{ak}^{\top} ( \BS{Z}_i -\bar {\BS{Z}}_k (t,\check \theta_l, \hat { \BS{\beta}}_k) ) ]^2 dN_i(t) \right\}\hat{J}_{0,l}^{-1}
	\end{eqnarray*}
	with
	$\hat{J}_{0,l} = -\frac{1}{K} \sum_{k=1}^K( \hat {\Sigma}_{\theta \theta, k,l} - {\hat {\BS{\mu}}_{ak}}^{\top} \hat {\BS{\Sigma}}_{\BS{\beta} \theta, k,l})$, and 
	\begin{eqnarray*}
		\hat {\Sigma}_{\theta \theta, k,l} &=& \frac{1}{m} \sum_{j \in I_k} \delta_j (D_j - \bar D_k (X_j,\check \theta_l, \hat {\BS{\beta}}_k) )^{ 2} , \\
		\hat {\BS{\Sigma}}_{\BS{\beta} \theta, k,l} &= &  \frac{1}{m} \sum_{j \in I_k} \delta_j ( \BS{Z}_j - \bar {\BS{Z}}_k (X_j,\check \theta_l, \hat {\BS{\beta}}_k)) (D_j - \bar D_k (X_j,\check \theta_l, \hat {\BS{\beta}}_k) )
	\end{eqnarray*}
    for $l=1, 2$ and $k=1, \cdots, K$. 
	Moreover, the uniform asymptotic confidence interval for $\theta_0$ based on $\check \theta_l$ with $l=1, 2$ is given by
	\[
	\left( \check \theta_l - \Phi^{-1}(1-\alpha/2) \hat \sigma_l /\sqrt{n}, \, \check \theta_l + \Phi^{-1}(1-\alpha/2) \hat \sigma_l /\sqrt{n} \right),
	\]
	where $\alpha$ is the significance level and $ \Phi^{-1}(1-\alpha/2)$ is the upper $\alpha/2$ quantile of $N(0,1)$.
\end{theorem}

Theorem \ref{ththeta} above is the main inference result for the survival treatment effect, establishing that both cross-fitted debiased estimators, $\check\theta_1$ and $\check\theta_2$, are asymptotically normal after standardization. By showing that the orthogonalized estimating equation removes the first-order impact of high-dimensional nuisance estimation while retaining the Cox-model structure, the theorem justifies the Wald-type confidence intervals for $\theta_0$ and provides a consistent plug-in variance estimator for use in practical implementation.

\begin{remark}
In this paper, we restrict our attention to the case when the survival causal parameter $\theta_0$ is a scalar for ease of presentation. Nevertheless, the proposed method can be extended readily to accommodate a finite-dimensional survival causal parameter vector $\boldsymbol{\theta}_0$.
\end{remark}

We also derive the asymptotic normality of the HSCI estimator $\check {\BS{\beta}}$ introduced in Section \ref{Sec3.4} for survival covariate effect inference in the following theorem.

\begin{theorem}\label{thbeta}
	Assume that Conditions \ref{coninicox}--\ref{condebeta} are satisfied. Let $\BS{c} \in \mathbb{R}^p$ with $\|\BS{c}\|_1 =1$, and assume further that $${\BS{c}}^{\top} \big\{\BS{\Sigma}_{\BS{\beta}\BS{\beta}}^{-1} + \BS{\Sigma}_{\BS{\beta}\BS{\beta}}^{-1}  \BS{\Sigma}_{\BS{\beta} \theta} (  \BS{\Sigma}_{\theta \theta} -  \BS{\Sigma}_{\theta\BS{\beta}} \BS{\Sigma}_{\BS{\beta}\BS{\beta}}^{-1}   \BS{\Sigma}_{\BS{\beta}\theta } )^{-1} \BS{\Sigma}_{\theta \BS{\beta}} \BS{\Sigma}_{\BS{\beta}\BS{\beta}}^{-1}   \big\} {\BS{c}} \rightarrow \varrho^2 \in (0, \infty).$$
    Set regularization parameters
	\[
	\lambda_n^* = C  \|{\BS{\Sigma}}_{\BS{\beta}\BS{\beta}}^{-1} \|_{\infty} (s_{\BS{\beta}_0} \sqrt{\log (np) /n} + \tau_n)
	\]
	in (\ref{eqClime2}) for a large constant $C> 0$ and $\lambda_{\theta \BS{\beta}} \asymp \sqrt{\log (np) /n}$ in (\ref{eqBetaL1}). Then for the debiased estimator $\check {\BS{\beta}}$ given in (\ref{eqestbeta}), we have that 
	\[
	\sqrt{n} {\BS{c}}^{\top} (\check {\BS{\beta}} -\BS{\beta}_{0}) \stackrel{\mathcal D} {\rightarrow} N(0, \varrho^2)
	\]
	as $n \rightarrow \infty$. The result also holds after replacing $\varrho^2$ with a consistent estimator $\hat \varrho^2$; that is,
	\[
	\sqrt{n} {\BS{c}}^{\top} (\check {\BS{\beta}} -\BS{\beta}_{0})/{\hat \varrho} \stackrel{\mathcal D} {\rightarrow} N(0, 1)
	\]
as $n \rightarrow \infty$.
\end{theorem}

Theorem \ref{thbeta} above extends the HSCI framework from the scalar treatment effect to valid inference for low-dimensional linear functionals of the high-dimensional regression coefficient vector $\BS{\beta}_0$. Complementing Theorem \ref{ththeta}, it accounts for uncertainty from estimating $\theta_0$ and 
approximating the inverse high-dimensional information matrix, revealing that the bias-corrected estimator based on the estimated precision matrix has an asymptotic Gaussian distribution even when the ambient dimensionality is large. The result therefore provides a basis for confidence intervals and hypothesis tests for covariate effects or low-dimensional contrasts of $\BS{\beta}_0$.

Together, Theorems \ref{ththeta} and \ref{thbeta} provide the distributional approximations needed for valid inference on both the survival treatment effect $\theta_0$ and the survival covariate effect vector $\BS{\beta}_0$.

\section{Simulation studies} \label{Sec5}

This section describes the implementation of the proposed method and then uses simulation studies to assess the finite-sample performance of $\check{\theta}_1$ and $\check{\theta}_2$ relative to the theoretical results.

\subsection{Implementation of HSCI}\label{subsec: HCIS-Implementation}

For data splitting, all observations are partitioned randomly into $K$ blocks $I_1,  \ldots, I_K$. In the simulation studies, we set $K=5$. For each $k=1, \ldots, K$, we use observations in $I_k^c$ and the R package \texttt{glmnet}, with regularization parameters selected by the $10$-fold cross-validation, to obtain the initial estimators $\hat{\theta}_k$, $\hat{\bbeta}_k$, and $\hat{\bgamma}_k$.
The estimator $\hat{\bmu}_{ak}$ is computed from \eqref{eqClime1}, with $\hat{\theta}$, $\hat{\BS{\beta}}$, and $\hat{\BS{\gamma}}$ replaced by $\hat{\theta}_k$, $\hat{\bbeta}_k$, and $\hat{\bgamma}_k$, respectively. The tuning parameter $\lambda_n$ is selected by 
the generalized information criterion (GIC) in \cite{fan2013tuning}; specifically, $\lambda_n$ in \eqref{eqClime1} is chosen to minimize
   \begin{align}\label{eq: tune_muak}
       \| \hat {\BS{\Sigma}}_{\BS{\beta}\BS{\beta}}\hat{\bmu}_{ak}(\lambda_n) -\hat {\BS{\Sigma}}_{\BS{\beta} \theta} \|_{\infty} + c_{\tilde{n}, p}\|\hat{\bmu}_a(\lambda_n)\|_0,
   \end{align}
where $c_{\tilde{n}, p}=\log\{\log(\tilde{n})\}\log(p)/\tilde{n}$, $\tilde{n}$ is the number of observations in $I_k^c$,
and $\|\hat{\bmu}_a(\lambda_n)\|_0$ is the number of nonzero components of $\hat{\bmu}_{ak}(\lambda_n)$.
Once the initial estimators $\hat{\theta}_k$, $\hat{\bbeta}_k$, $\hat{\bgamma}_k$, and $\hat{\bmu}_{ak}$ are available, the two debiased estimators $\check{\theta}_1$ and $\check{\theta}_2$ of $\theta$ are obtained from \eqref{eq: check-theta1} and \eqref{eqesttheta2}, respectively.

Given a debiased estimator of $\theta$, the corresponding debiased estimator $\check{\bbeta}$ of $\bbeta$ can be obtained from \eqref{eqestbeta}. In particular, estimator $\hat{\bbeta}$ used in \eqref{eqestbeta} is from the penalized maximum partial likelihood estimator $(\hat{\theta},\, \hat{\bbeta})$ defined in \eqref{eqBetaL1}, which can be obtained using all observations and the R package \texttt{glmnet} with regularization parameter $\lambda_{\theta\BS{\beta}}$ selected by the $10$-fold cross-validation. The CLIME-type estimator $\hat {\BS{\Xi}}$ used in \eqref{eqestbeta} can be obtained using the R package \texttt{flare}~\citep{li2015flare} with regularization parameter tuned by 
the GIC \citep{fan2013tuning}, which minimizes
   \begin{align}\label{eq: tune_Xi}
       \|\check {\BS{\Sigma}}_{\BS{\beta}\BS{\beta}}\hat {\BS{\Xi}}- I_p\|_{\infty} + c_{n, p}df,
   \end{align}
where $I_p$ is the identity matrix of size $p$, $\|\cdot\|_{\infty}$ represents the supremum norm, $c_{n, p}=\log\{\log(n)\}\log(p)/n$, and $df$ is the number of nonzero entries in the upper triangular part of $\hat {\BS{\Xi}}$.

\subsection{Simulation examples}\label{sec:examples}

\textbf{\textit{Study 1}}. Our first simulation study verifies the theoretical results in Theorem \ref{ththeta} by investigating the empirical distributions of our estimators $\check{\theta}_1$ and $\check{\theta}_2$. We generate data from the joint models \eqref{eqmodelcox} and \eqref{eqmodellog}, 
where the baseline hazard function $\lambda_0(t)=1$, the true coefficients $\theta_0=0.5$,
$\bbeta_0=(1, -1, 1, -1, 1, \bzero_{p-5})^T$, $\bgamma_0=(0, 0, 0, 0, 1, -1, \bzero_{p-6})^T$, and the covariate vector $\bZ=(Z_1, \cdots, Z_p)^T$ is from a multivariate normal distribution $N_p(\bzero, \bSig^Z)$ with
 \begin{equation*}
    \Sigma^Z_{ij}
    = \begin{cases}
      1, & \mbox{if } i = j,\\
      0.5, & \mbox{if } i \neq j, \, \beta_{0i} \neq 0, \beta_{0j} \neq 0, \\
      0, & \mbox{if } i \neq j, \beta_{0i}\beta_{0j} = 0, |\beta_{0i}| + |\beta_{0j}| > 0,\\
      0.5^{|i-j|}, & \mbox{if } i \neq j, \beta_{0i} = 0, \beta_{0j} = 0.
      \end{cases}
 \end{equation*}
Here, $\beta_{0i}$ and $\beta_{0j}$ are the $i$th and $j$th components of $\bbeta_0$, respectively. The censoring time $C$ has a uniform distribution $U(0, 2)$.
We consider the settings of $(n, p)=(500, 10)$ and $(n, p)=(500, 100)$.

\begin{figure}[h!]
 \begin{center}
      \includegraphics[scale=0.5]{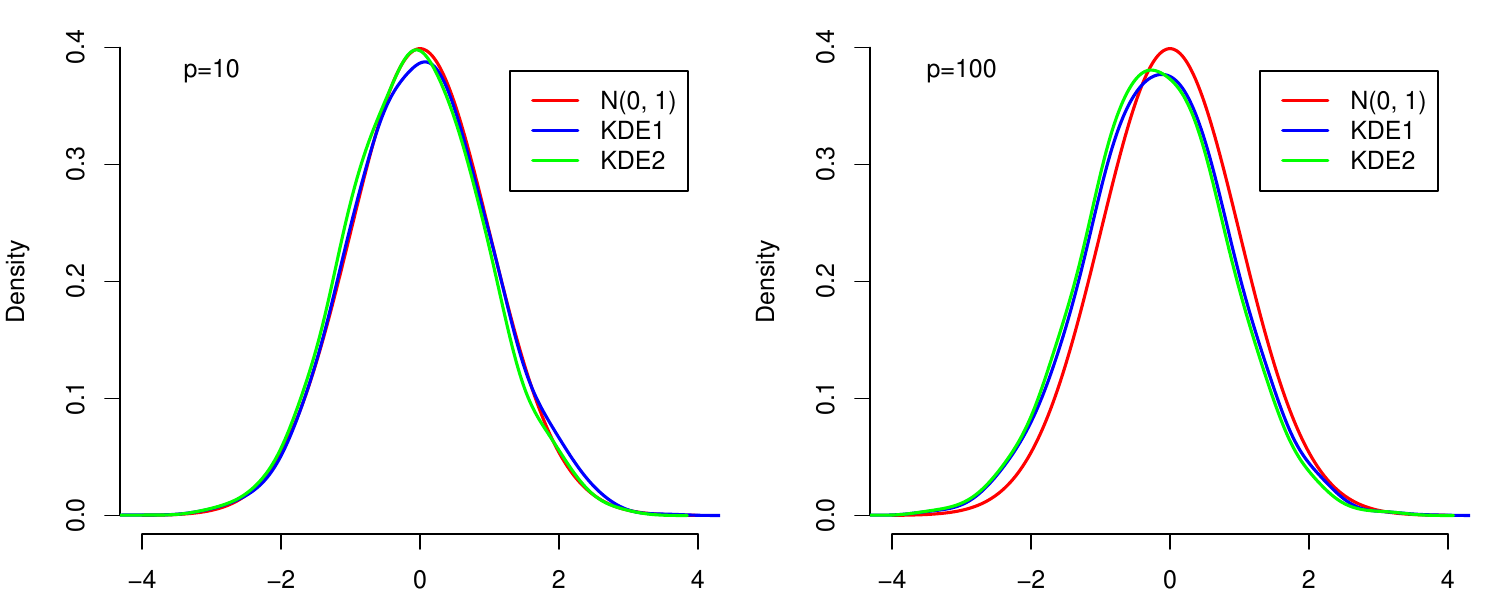}
      \caption{The empirical distributions of $T_{n1}=\hat{\sigma}_1^{-1}\sqrt{n}(\check{\theta}_1-\theta_0)$ and $T_{n2}=\hat{\sigma}_2^{-1}\sqrt{n}(\check{\theta}_2-\theta_0)$ in Study 1. KDE1 and KDE2 are the kernel density estimates of $T_{n1}$ and $T_{n2}$, respectively.}\label{fig1new}
 \end{center}
\end{figure}

We run $5000$ Monte Carlo simulations and examine the empirical distributions of $T_{n1}=\hat{\sigma}_1^{-1}\sqrt{n}(\check{\theta}_1-\theta_0)$ and $T_{n2}=\hat{\sigma}_2^{-1}\sqrt{n}(\check{\theta}_2-\theta_0)$. Figure \ref{fig1new} shows the kernel density estimates (KDEs) of $T_{n1}$ and $T_{n2}$ along with the standard normal density. The distributions of both $T_{n1}$ and $T_{n2}$ track closely the target standard normal distribution across different settings, supporting the asymptotic normality stated in Theorem \ref{ththeta}.

\smallskip
\textbf{\textit{Study 2}}. Our second simulation study examines the performance of the proposed estimators when $p$ is relatively small. We use the same model setup as in Study 1, except for the censoring time $C$ and covariance structure for $\bSig^Z$. In this study, censoring time $C$ follows a uniform distribution $U(0, c)$, where $c$ is a parameter controlling the censoring rate. We adopt two different values $c=10$ and $c=2$.
A larger value of $c$ yields a smaller censoring rate (CR).
We also consider two different covariance structures for $\bSig^Z$ with $\bSig^Z=I_p$ and $\bSig^Z=(\Sigma^Z_{ij})_{p\times p}$, respectively, where $I_p$ is the identity matrix of size $p$ and $\bSig^Z=(\Sigma^Z_{ij})_{p\times p}$ is the same as in Study 1.  Hence, we end up with the following four different settings:
  \begin{itemize}
    \itemsep0em
  	\item Setting 1: $\bSig^Z=I_p$ and $c=10$,
  	\item Setting 2: $\bSig^Z=I_p$ and $c=2$,
  	\item Setting 3: $\bSig^Z=(\Sigma^Z_{ij})$ and $c=10$,
  	\item Setting 4: $\bSig^Z=(\Sigma^Z_{ij})$ and $c=2$.
  \end{itemize}
For each setting, we set $(n, p)=(500, 10)$ and run $1000$ repetitions.

Recall that there are two different estimators, $\check{\theta}_1$ and $\check{\theta}_2$, of $\theta$ in our method. We compare these two proposed estimators to the Lasso estimator $\hat{\theta}$ and oracle estimator $\hat{\theta}_{\mathrm{or}}$.
In particular, as mentioned in Section \ref{subsec: HCIS-Implementation}, the Lasso estimator $\hat{\theta}$ is from the penalized maximum partial likelihood estimator $(\hat{\theta}, \hat{\bbeta})$, which is obtained from \eqref{eqBetaL1} using all observations and the R package \texttt{glmnet} with regularization parameter $\lambda_{\theta\BS{\beta}}$ selected by the $10$-fold cross-validation. The oracle estimator $\hat{\theta}_{\mathrm{or}}$ is from $(\hat{\theta}_{\mathrm{or}}, \hat{\bbeta}_{\mathrm{or}})$, which assumes that the true underlying sparse model is known in advance, and is obtained using the \texttt{coxph} function in the R package \texttt{survival} with all the observations.

\begin{table}[t]
	\caption{Simulation results for treatment effect $\theta$ in Study 2. }
	\vspace{2mm}
	\centering
	\scalebox{0.85}{
	\begin{tabular}{l rrcc rrcc}
		\toprule
		\multirow{3}{*}{Approach}  & \multicolumn{4}{c}{Setting 1 (CR: 0.1325)}     & \multicolumn{4}{c}{Setting 2 (CR: 0.3937)} \\[-6pt]
		&\multicolumn{4}{c@{}}{\hrulefill} &\multicolumn{4}{c@{}}{\hrulefill} \\
		& \multicolumn{1}{c}{$\mbox{BIAS}$} & \multicolumn{1}{c}{$\mbox{SE.T}$} & \multicolumn{1}{c}{$\mbox{SE}$} & \multicolumn{1}{c}{$\mbox{CP}$} & \multicolumn{1}{c}{$\mbox{BIAS}$} & \multicolumn{1}{c}{$\mbox{SE.T}$} & \multicolumn{1}{c}{$\mbox{SE}$} & \multicolumn{1}{c}{$\mbox{CP}$} \\
		\midrule
        $\hat{\theta}$       & -0.0451 & 0.1119 & -   & -    & -0.0529 & 0.1333 & - & -\\
        $\check{\theta}_{1}$ & -0.0116 & 0.1205 & 0.1197 & 0.951 & -0.0102 & 0.1421 & 0.1418 & 0.953\\
        $\check{\theta}_{2}$ & -0.0174 & 0.1205 & 0.1186 & 0.954 & -0.0167 & 0.1421 & 0.1398 & 0.955\\
        $\hat{\theta}_{\mathrm{or}}$ & 0.0025 & 0.1059 & 0.1068 & 0.941 & 0.0032 & 0.1264 & 0.1276 & 0.950\\
		\bottomrule
		\multirow{3}{*}{Approach}  & \multicolumn{4}{c}{Setting 3 (CR: 0.1149)}     & \multicolumn{4}{c}{Setting 4 (CR: 0.3845)} \\[-6pt]
		&\multicolumn{4}{c@{}}{\hrulefill} &\multicolumn{4}{c@{}}{\hrulefill} \\
		& \multicolumn{1}{c}{$\mbox{BIAS}$} & \multicolumn{1}{c}{$\mbox{SE.T}$} & \multicolumn{1}{c}{$\mbox{SE}$} & \multicolumn{1}{c}{$\mbox{CP}$} & \multicolumn{1}{c}{$\mbox{BIAS}$} & \multicolumn{1}{c}{$\mbox{SE.T}$} & \multicolumn{1}{c}{$\mbox{SE}$} & \multicolumn{1}{c}{$\mbox{CP}$} \\
		\midrule
        $\hat{\theta}$               & -0.0269 & 0.1112   & - & - & -0.0293 & 0.1319  & - & - \\
        $\check{\theta}_{1}$         & 0.0018 & 0.1194 & 0.1187 & 0.959 &  0.0035 & 0.1412 & 0.1431 & 0.954\\
        $\check{\theta}_{2}$         & -0.0062 & 0.1194 & 0.1168 & 0.966 &  -0.0056 & 0.1412 & 0.1402 & 0.957\\
        $\hat{\theta}_{\mathrm{or}}$ &  0.0084 & 0.1049 & 0.1064 & 0.946 &  0.0091 & 0.1253 & 0.1260 & 0.951\\
		\bottomrule
	\end{tabular}}\label{Table-theta-n500-p10-N1000}
\end{table}

Table \ref{Table-theta-n500-p10-N1000} summarizes the simulation results for $\theta$ over $1000$ repetitions under different settings in Study 2. In this table, CR represents the censoring rate, BIAS is the bias of the estimator, SE.T denotes the estimate of the asymptotic standard deviation of the estimator, and SE stands for the empirical standard error of the estimator. Specifically, $\mbox{SE.T}(\check{\theta}_1)=\hat{\sigma}_1/\sqrt{n}$ and $\mbox{SE.T}(\check{\theta}_2)=\hat{\sigma}_2/\sqrt{n}$, where $\hat{\sigma}_1$ and
$\hat{\sigma}_2$ are given in Theorem \ref{ththeta}.
The simulation results for $\bbeta$ in Study 2 are reported in Section \ref{SecA} of the Supplementary Material. Table \ref{Table-theta-n500-p10-N1000} shows that the Lasso estimator has larger bias, while the proposed debiasing method reduces substantially this bias. In addition, the estimated asymptotic standard deviations (SE.T) for the proposed estimators $\check{\theta}_1$ and $\check{\theta}_2$ 
agree closely with the corresponding
empirical standard errors (SE), indicating that the proposed asymptotic variance estimation procedure performs well in finite samples. The difference between the two estimators, $\check{\theta}_1$ and $\check{\theta}_2$, is negligible. Both $\check{\theta}_1$ and $\check{\theta}_2$ achieve reasonably good coverage rates of the nominal $95\%$ confidence intervals.

\smallskip
\textbf{\textit{Study 3}}. Our third simulation study further assesses the performance of the proposed estimators when $p$ is large. We consider the same settings as in Study 2, except for $(n, p)=(500, 100)$. Recall that in our data-splitting method, all the observations in our data are partitioned randomly into $K=5$ parts. 
For each $k \in \{1, \dots, K\}$, we use all observations not in the $k$th block to estimate nuisance parameter $\eta$ in~\eqref{eqortho}, and use only the data in the $k$th block to evaluate the orthogonal estimating equation. 
In this sense, the choice of $p=100$ is already in a high-dimensional regime for such setting. It is worth pointing out that both the Lasso estimator $\hat{\theta}$ and oracle estimator $\hat{\theta}_{\mathrm{or}}$ use all observations to estimate $\theta$.

\begin{table}[h!]
	\caption{Simulation results for treatment effect $\theta$ in Study 3 over 1000 repetitions.}
	\vspace{2mm}
	\centering
	\scalebox{0.85}{
	\begin{tabular}{l rrcc rrcc}
		\toprule
		\multirow{3}{*}{Approach}  & \multicolumn{4}{c}{Setting 1 (CR: 0.1333)}     & \multicolumn{4}{c}{Setting 2 (CR: 0.3952)} \\[-6pt]
		&\multicolumn{4}{c@{}}{\hrulefill} &\multicolumn{4}{c@{}}{\hrulefill} \\
		& \multicolumn{1}{c}{$\mbox{BIAS}$} & \multicolumn{1}{c}{$\mbox{SE.T}$} & \multicolumn{1}{c}{$\mbox{SE}$} & \multicolumn{1}{c}{$\mbox{CP}$} & \multicolumn{1}{c}{$\mbox{BIAS}$} & \multicolumn{1}{c}{$\mbox{SE.T}$} & \multicolumn{1}{c}{$\mbox{SE}$} & \multicolumn{1}{c}{$\mbox{CP}$} \\
		\midrule
		$\hat{\theta}$               & -0.1279 & 0.0998 & - & -  & -0.1426 & 0.1224 & - & -\\
        $\check{\theta}_{1}$         & -0.0281 & 0.1132 & 0.1062 & 0.953 & -0.0194 & 0.1333 & 0.1299 & 0.957\\
        $\check{\theta}_{2}$         & -0.0334 & 0.1132 & 0.1055 & 0.950 & -0.0247 & 0.1333 & 0.1288 & 0.954\\
        $\hat{\theta}_{\mathrm{or}}$ &  0.0064 & 0.1060 & 0.1029 & 0.964 &  0.0082 & 0.1267 & 0.1260 & 0.949\\
		\bottomrule
		\multirow{3}{*}{Approach}  & \multicolumn{4}{c}{Setting 3 (CR: 0.1149)}     & \multicolumn{4}{c}{Setting 4 (CR: 0.3871)} \\[-6pt]
		&\multicolumn{4}{c@{}}{\hrulefill} &\multicolumn{4}{c@{}}{\hrulefill} \\
		& \multicolumn{1}{c}{$\mbox{BIAS}$} & \multicolumn{1}{c}{$\mbox{SE.T}$} & \multicolumn{1}{c}{$\mbox{SE}$} & \multicolumn{1}{c}{$\mbox{CP}$} & \multicolumn{1}{c}{$\mbox{BIAS}$} & \multicolumn{1}{c}{$\mbox{SE.T}$} & \multicolumn{1}{c}{$\mbox{SE}$} & \multicolumn{1}{c}{$\mbox{CP}$} \\
		\midrule
        $\hat{\theta}$               & -0.1123 & 0.1108  & - & - & -0.1257 & 0.1298 & - & - \\
        $\check{\theta}_{1}$         & -0.0307 & 0.1117 & 0.1202 & 0.920 & -0.0235 & 0.1316 & 0.1402 & 0.939 \\
        $\check{\theta}_{2}$         & -0.0375 & 0.1117 & 0.1188 & 0.919 & -0.0305 & 0.1316 & 0.1385 & 0.939\\
        $\hat{\theta}_{\mathrm{or}}$ &  0.0041 & 0.1049 & 0.1126 & 0.926 &  0.0041 & 0.1255 & 0.1318 & 0.944\\
		\bottomrule
	\end{tabular}}\label{Table-theta-n500-p100-N1000}
\end{table}

The simulation results for $\theta$ in Study 3 are reported in Table \ref{Table-theta-n500-p100-N1000}. We continue to see that in all settings, the Lasso estimator has larger bias, whereas our method can reduce significantly the estimation bias for $\theta$. The estimated asymptotic standard deviations (SE.T) are close to the empirical standard errors (SE) for both $\check{\theta}_1$ and $\check{\theta}_2$. The performance of $\check{\theta}_1$ is similar to that of $\check{\theta}_2$.
The coverage rates for both $\check{\theta}_1$ and $\check{\theta}_2$ are close to the nominal coverage of $95\%$. As expected, the oracle estimator $\hat{\theta}_{or}$, which is unavailable in practice, has the smallest overall bias, typically close to zero, as it relies on 
the true underlying sparsity structure.

\smallskip
\textbf{\textit{Study 4}}. The fourth simulation study assesses \textit{robustness} when the propensity model~\eqref{eqmodellog} is \textit{misspecified}. The baseline hazard function $\lambda_0(t)$, along with the true coefficients $\theta_0$, $\bbeta_0$, and $\bgamma_0$, are the same as in Study 3. The covariate vector $\bZ=(Z_1,\ldots,Z_p)^T$ and the censoring time $C$ are also generated in the same manner as in Study 3. We consider the following two scenarios of \textit{model misspecification}:
%
1) Scenario S1:
the propensity model~\eqref{eqmodellog} is misspecified and given by $P(D=1 \mid \BS{Z}) =\exp\{\BS{Z}^{\top} \BS{\gamma}_0 + 0.25Z_1^2\}/\left[1+\exp\{ \BS{Z}^{\top} \BS{\gamma}_0 + 0.25Z_1^2\}\right]$;
2) Scenario S2:
the propensity model~\eqref{eqmodellog} is misspecified and given by
$P(D=1 \mid \BS{Z}) = \Phi(\BS{Z}^{\top} \BS{\gamma}_0)$,
where $\Phi(\cdot)$ is the cumulative distribution function (CDF) of the standard normal distribution. In each scenario, the hazards model~\eqref{eqmodelcox} for the event time is correctly specified.

\begin{table}[h!]
	\caption{Simulation results for treatment effect $\theta$ in Study 4 over 1000 repetitions.}
	\vspace{4mm}
	\centering
	\scalebox{0.85}{
	\begin{tabular}{c|l rrcc rrcc}
		\toprule
		\multirow{14}{*}{S1} & \multirow{3}{*}{Approach}  & \multicolumn{4}{c}{Setting 1 (CR: 0.1311)}     & \multicolumn{4}{c}{Setting 2 (CR: 0.3906)} \\[-6pt]
		& &\multicolumn{4}{c@{}}{\hrulefill} &\multicolumn{4}{c@{}}{\hrulefill} \\
		& & \multicolumn{1}{c}{$\mbox{BIAS}$} & \multicolumn{1}{c}{$\mbox{SE.T}$} & \multicolumn{1}{c}{$\mbox{SE}$} & \multicolumn{1}{c}{$\mbox{CP}$} & \multicolumn{1}{c}{$\mbox{BIAS}$} & \multicolumn{1}{c}{$\mbox{SE.T}$} & \multicolumn{1}{c}{$\mbox{SE}$} & \multicolumn{1}{c}{$\mbox{CP}$} \\ [-6pt]
        & \multicolumn{9}{c@{}}{\hrulefill}  \\
		& $\hat{\theta}$        & -0.1283 & 0.1049 & -      & -     & -0.1463 & 0.1273 & -      & - \\
        & $\check{\theta}_{1}$  & -0.0283 & 0.1136 & 0.1102 & 0.954 & -0.0205 & 0.1344 & 0.1352 & 0.951 \\
        & $\check{\theta}_{2}$  & -0.0328 & 0.1136 & 0.1097 & 0.951 & -0.0260 & 0.1344 & 0.1337 & 0.949 \\
        & $\hat{\theta}_{\mathrm{or}}$ & 0.0062 & 0.1065 & 0.1087 & 0.947 & 0.0032 & 0.1279 & 0.1307 & 0.941  \\ [-6pt]
		  & \multicolumn{9}{c@{}}{\hrulefill}  \\
		& \multirow{3}{*}{Approach}  & \multicolumn{4}{c}{Setting 3 (CR: 0.1126)}     & \multicolumn{4}{c}{Setting 4 (CR: 0.3826)} \\[-6pt]
		& &\multicolumn{4}{c@{}}{\hrulefill} &\multicolumn{4}{c@{}}{\hrulefill} \\
		& & \multicolumn{1}{c}{$\mbox{BIAS}$} & \multicolumn{1}{c}{$\mbox{SE.T}$} & \multicolumn{1}{c}{$\mbox{SE}$} & \multicolumn{1}{c}{$\mbox{CP}$} & \multicolumn{1}{c}{$\mbox{BIAS}$} & \multicolumn{1}{c}{$\mbox{SE.T}$} & \multicolumn{1}{c}{$\mbox{SE}$} & \multicolumn{1}{c}{$\mbox{CP}$} \\ [-6pt]
		 & \multicolumn{9}{c@{}}{\hrulefill}  \\
		& $\hat{\theta}$        & -0.1138 & 0.1065 & -  & -  & -0.1257 & 0.1306  & -  & - \\
        & $\check{\theta}_{1}$  & -0.0302 & 0.1118 & 0.1151 & 0.940  & -0.0198 & 0.1326 & 0.1409 & 0.936\\
        & $\check{\theta}_{2}$  & -0.0364 & 0.1118 & 0.1144 & 0.937 & -0.0274 & 0.1326 & 0.1399 & 0.929 \\
        & $\hat{\theta}_{\mathrm{or}}$ & 0.0038 & 0.1052 & 0.1065 & 0.953 & 0.0047 & 0.1269 & 0.1308 & 0.947\\
		 \midrule
		\multirow{14}{*}{S2} & \multirow{3}{*}{Approach}  & \multicolumn{4}{c}{Setting 1 (CR: 0.1345)}     & \multicolumn{4}{c}{Setting 2 (CR: 0.3955)} \\[-6pt]
		& &\multicolumn{4}{c@{}}{\hrulefill} &\multicolumn{4}{c@{}}{\hrulefill} \\
		& & \multicolumn{1}{c}{$\mbox{BIAS}$} & \multicolumn{1}{c}{$\mbox{SE.T}$} & \multicolumn{1}{c}{$\mbox{SE}$} & \multicolumn{1}{c}{$\mbox{CP}$} & \multicolumn{1}{c}{$\mbox{BIAS}$} & \multicolumn{1}{c}{$\mbox{SE.T}$} & \multicolumn{1}{c}{$\mbox{SE}$} & \multicolumn{1}{c}{$\mbox{CP}$} \\ [-6pt]
        & \multicolumn{9}{c@{}}{\hrulefill}  \\
		& $\hat{\theta}$      &  -0.1306   & 0.1083  & -  & -  & -0.1503 & 0.1340 & -  & - \\
        & $\check{\theta}_{1}$      & -0.0354   & 0.1221  & 0.1215 & 0.944 & -0.0323 & 0.1433 & 0.1488 & 0.938\\
        & $\check{\theta}_{2}$       & -0.0422 &0.1221  & 0.1201 & 0.939 & -0.0401 & 0.1433 & 0.1467 & 0.936\\
        & $\hat{\theta}_{\mathrm{or}}$ & -0.0010 & 0.1111 & 0.1124  & 0.953 & -0.0059 & 0.1328 & 0.1360 & 0.949\\ [-6pt]
		  & \multicolumn{9}{c@{}}{\hrulefill}  \\
		& \multirow{3}{*}{Approach}  & \multicolumn{4}{c}{Setting 3 (CR: 0.1160)}     & \multicolumn{4}{c}{Setting 4 (CR: 0.3875)} \\[-6pt]
		& &\multicolumn{4}{c@{}}{\hrulefill} &\multicolumn{4}{c@{}}{\hrulefill} \\
		& & \multicolumn{1}{c}{$\mbox{BIAS}$} & \multicolumn{1}{c}{$\mbox{SE.T}$} & \multicolumn{1}{c}{$\mbox{SE}$} & \multicolumn{1}{c}{$\mbox{CP}$} & \multicolumn{1}{c}{$\mbox{BIAS}$} & \multicolumn{1}{c}{$\mbox{SE.T}$} & \multicolumn{1}{c}{$\mbox{SE}$} & \multicolumn{1}{c}{$\mbox{CP}$} \\ [-6pt]
		 & \multicolumn{9}{c@{}}{\hrulefill}  \\
        & $\hat{\theta}$            & -0.1062 & 0.1085 & - & - & -0.1199 & 0.1302  & - & -\\
        & $\check{\theta}_{1}$      & -0.0240 & 0.1202 & 0.1206 & 0.942 & -0.0184 & 0.1412 & 0.1440 & 0.938\\
        & $\check{\theta}_{2}$      & -0.0332 & 0.1202 & 0.1181 & 0.943 & -0.0290 & 0.1412 & 0.1404 & 0.951\\
        & $\hat{\theta}_{\mathrm{or}}$ & 0.0066 & 0.1100 & 0.1091 & 0.948 &  0.0051 & 0.1316 & 0.1308 & 0.956\\
		\bottomrule
	\end{tabular}}\label{Example3-Table-theta-n500-p100-N100}
\end{table}


In scenario S1, the true propensity model includes a quadratic term $0.25Z_1^2$ in the exponent, but the estimator assumes a linear form in the exponent. In scenario S2, the true propensity model uses a probit link function (the normal CDF), whereas the estimator assumes a logistic link function. 
The simulation results for $\theta$ under these two misspecification scenarios and  four simulation settings with different censoring rates are detailed in Table \ref{Example3-Table-theta-n500-p100-N100}. It can be seen that the Lasso estimator $\hat{\theta}$ performs the worst across all settings, exhibiting substantial bias in estimating the treatment effect. In contrast, both proposed estimators, $\check{\theta}_1$ and $\check{\theta}_2$, improve substantially over the Lasso estimator, with markedly reduced bias in each setting.  
In addition, the estimated asymptotic standard deviations (SE.T) for proposed estimators $\check{\theta}_1$ and $\check{\theta}_2$ match closely the empirical standard errors (SE), indicating that the asymptotic variance estimation procedure can still perform well in finite samples even when the propensity score model is misspecified.
Moreover, the coverage rates are generally close to the nominal $95\%$ level. 
Comparing the results for scenarios S1 and S2, the standard errors under scenario S2 become slightly larger in magnitude than S1 across corresponding settings, reflecting the greater difficulty induced by the probit link misspecification in scenario S2.

\section{Real data application} \label{Sec6}

In this section, we 
further illustrate the performance of our method by an application to the diffuse large-B-cell lymphoma (DLBCL) data, which was analyzed previously in \cite{rosenwald2002use}, \cite{zhu2011model}, and \cite{lin2013high}. The data set consists of gene expression measurements for 7399 genes, gene-expression subgroups, and survival outcomes after chemotherapy for 240 patients. There are three gene-expression subgroups in this data set: germinal-center B-cell-like (GCB), activated B-cell-like (ABC), and type 3 diffuse large-B-cell lymphoma (Type III). Our main focus is the treatment effect of the gene-expression subgroup on the overall survival time of the patients. To apply our method, we exclude all the observations from group Type III, resulting in a sample of size $n=188$ with 115 observations from group GCB and 73 from group ABC.
We set $D=1$ for observations from group GCB (i.e., treatment), and $D=0$ for observations from group ABC (i.e., control). In our sample, the median follow-up time is 2.95 years, and 105 patients died during the study period. Thus, the censoring rate is about $44\%$. Figure \ref{fig2} depicts the Kaplan--Meier curves for the two groups GCB and ABC. The curves suggest that group GCB is associated with superior survival, implying a lower hazard rate for patients in group GCB. This finding has been widely confirmed  in the existing literature; see, e.g.,
\cite{rosenwald2002use}, \cite{wright2003gene}, and \cite{hans2004confirmation}.

\begin{figure}[h]
 \begin{center}
      \includegraphics[scale=0.5]{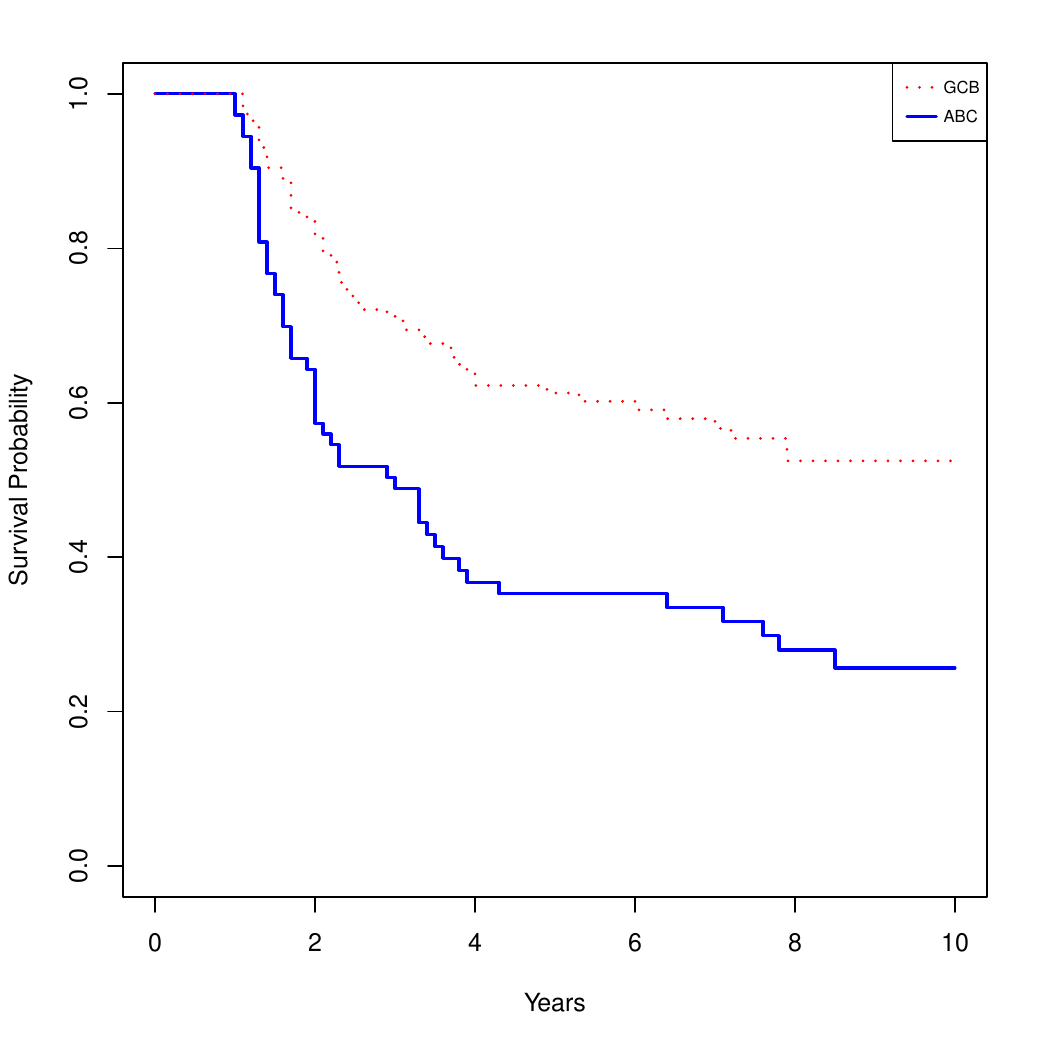}
      \caption{The Kaplan--Meier curves for GCB (dotted red) and ABC (solid blue). }\label{fig2}
 \end{center}
\end{figure}

Since gene expression measurements are highly correlated, we first preprocess the data using principal component analysis (PCA). Specifically, we apply PCA to the data matrix for the 7399 genes and 188 observations in groups GCB and ABC, and exploit the top 50 sample principal components as the covariate design matrix $\BS{Z}$ in the joint models \eqref{eqmodelcox} and \eqref{eqmodellog}. Similar dimension-reduction strategies have been used in related works \citep{zheng2021nonsparse}.

We compare the proposed estimators $\check{\theta}_1$ and $\check{\theta}_2$ to the Lasso estimator $\hat{\theta}$ only because the oracle estimator $\hat{\theta}_{\mathrm{or}}$ is not available in this case. Since the sample size is $n = 188$, we adopt $K=4$ for our data-splitting procedure for simplicity.  All other tuning parameters are selected using the approaches detailed in Section \ref{subsec: HCIS-Implementation}.

\begin{table}[h]
	\caption{Estimated treatment effect for the DLBCL data.}
	\vspace{2mm}
	\centering
	\begin{tabular}{lrcc}
		\toprule
        Approach & Estimate  & SE & $95\%$ Confidence Interval \\
        \midrule
        $\hat{\theta}$ & -0.3120 & - & - \\
        $\check{\theta}_1$ & -0.5360 & 0.2045 &  [-0.9368, -0.1352]\\
        $\check{\theta}_2$ & -0.5270 & 0.2045 & [-0.9278, -0.1261]\\
		\bottomrule
	\end{tabular}\label{Tab: RealData}
\end{table}

Table \ref{Tab: RealData} reports the empirical results. We present two estimates of treatment effect $\theta$ from the proposed method and their corresponding confidence intervals. Inference for the Lasso estimate $\hat{\theta}$ is not directly available, so only its point estimate is reported. The two proposed estimates, $\check{\theta}_1$ and $\check{\theta}_2$, are very similar. The Lasso approach estimates a log-hazard effect of -0.3120, whereas the proposed method estimates log-hazard effects of -0.5360 and -0.5270, respectively. In addition, our proposed method provides uncertainty quantification for the estimated treatment effect.
The $95\%$ confidence intervals from the proposed method are entirely negative, indicating that group GCB has a significantly lower hazard rate than group ABC. Equivalently, the overall survival rate for patients in group GCB is significantly higher than that for patients in group ABC, consistent with the Kaplan--Meier curves in Figure \ref{fig2}.

\section{Discussion} \label{Sec7}

We have developed in this paper a high-dimensional survival causal inference framework, HSCI, for valid inference on the treatment and covariate effects under a sparse Cox proportional hazards model. The proposed method combines the propensity-score adjustment, Neyman near-orthogonal score construction, and cross-fitting to remove the first-order impact of high-dimensional nuisance estimation. We have established root-$n$ asymptotic normality under doubly robust rate conditions, and further extended the analysis to inference on high-dimensional covariate effects. The simulation and real data examples confirm that HSCI can reduce substantially the bias of regularized Cox estimators and achieve confidence-interval coverage close to the nominal level across different dimensionality, censoring, and misspecified propensity-model settings.

It would be interesting to extend the HSCI framework to more general survival models, nonlinear covariate effects, and machine-learning-based nuisance estimators. Also, it would be important to adapt the suggested framework to more complex survival-causal settings, such as time-varying treatments and time-dependent covariates, for longitudinal applications. These problems are beyond the scope of the current paper and will be interesting topics for future research.

\section{Data availability statement}
The data that support the findings of this study are openly available at
\url{https://hastie.su.domains/StatLearnSparsity_files/DATA/} and
\url{https://llmpp.ccr.cancer.gov/DLBCL/DLBCL_patient_data_NEW.txt}.

\bibliographystyle{chicago}
\bibliography{references}

\newpage
\appendix
\setcounter{page}{1}
\setcounter{section}{0}
\renewcommand{\theequation}{A.\arabic{equation}}
\setcounter{equation}{0}

\begin{center}{\bf \Large Supplementary Material to  ``HSCI: Neyman-Orthogonal Causal Inference under High-Dimensional Proportional Hazards"}

\bigskip

Yingying Fan, Lan Gao, Daoji Li and Jinchi Lv
\end{center}

This Supplementary Material contains the proofs of the main results, some technical lemmas and their proofs, and some additional simulation results. For notational simplicity, we will use $C_0$ and $\tilde{C}$ to denote generic positive constants whose values may vary from line to line. All other notation is the same as in the main text.

\section{Proofs of Theorems \ref{ththeta}--\ref{thbeta} and Propositions \ref{lminilog}--\ref{lmnearortho}} \label{SecB}

\subsection{Proof of Theorem \ref{ththeta}} \label{SecB.1}

The proof of this theorem will be divided into three parts. First, we show that the estimating equation has a solution in the local neighborhood $\mathcal{N}_{\theta_0}^*$ by checking the signs of the score at the two endpoints. Second, we linearize the orthogonal score around $(\theta_0,\BS{\eta}_{a0})$ and exploit the Neyman near-orthogonality to make the nuisance-estimation terms negligible. Third, we verify that the sandwich variance estimator is consistent, so that an application of Slutsky's lemma yields the desired limiting distribution. A major technical challenge is to control the empirical process and Taylor remainder terms uniformly over a shrinking high-dimensional nuisance neighborhood.

\textbf{Step 1}. We begin with showing that with probability tending to one, there exists an estimator $\check \theta_{1,k}$ within  neighborhood $\mathcal{N}_{\theta_0}^*$ that solves the estimating equation \eqref{eqesttheta1}. Let $\Delta_n = C_0 s_{\BS{\beta}_0} \sqrt{\log (np) /n } $ be half of the radius of neighborhood $\mathcal{N}_{\theta_0}^*$. We claim that when $n$ is sufficiently large, the following sign conditions 
\begin{align}
    \Phi(\{\mathcal{O}_{i}\}_{i \in I_k} ; \theta_0 + \Delta_n, \hat{\BS{\eta}}_k) & < 0 \label{eq-pf-thm1-neww-goal1}, \\
    \Phi(\{\mathcal{O}_{i}\}_{i \in I_k} ; \theta_0 - \Delta_n, \hat{\BS{\eta}}_k) & > 0 \label{eq-pf-thm1-neww-goal2}
\end{align}
hold with probability tending to one. Then it follows from the continuity of function $\Phi(\{\mathcal{O}_{i}\}_{i \in I_k} ; \theta, \hat{\BS{\eta}}_k) $ with respect to $\theta$ that within the range $ (\theta_0 - \Delta_n, \theta_0 + \Delta_n) \subset \mathcal{N}_{\theta_0}^*$, there exists a solution $\check{\theta}_{1,k} \in \mathcal{N}_{\theta_0}^* $ to \eqref{eqesttheta1}.

Since the proofs of \eqref{eq-pf-thm1-neww-goal1} and \eqref{eq-pf-thm1-neww-goal2} are similar, it suffices to show \eqref{eq-pf-thm1-neww-goal1}. To this end, observe that $\mathbb{E} [ \Phi(\{\mathcal{O}_{i}\}_{i \in I_k}; \theta_0, {\BS{\eta}}_{a0})] = 0$ and $\Phi(\{\mathcal{O}_{i}\}_{i \in I_k} ; \theta_0 + \Delta_n, \hat{\BS{\eta}}_k) $ can be decomposed into the following three terms
\begin{align} \label{eq-pf-thm1-neww-1}
     \Phi(\{\mathcal{O}_{i}\}_{i \in I_k} ; \theta_0 + \Delta_n, \hat{\BS{\eta}}_k) = R_1 + R_2 + R_3,
\end{align}
where the three components are defined as
\begin{align}
    R_1 &= \Phi(\{\mathcal{O}_{i}\}_{i \in I_k} ; \theta_0 + \Delta_n, \hat{\BS{\eta}}_k) - \mathbb{E} [\Phi(\{\mathcal{O}_{i}\}_{i \in I_k} ; \theta_0 + \Delta_n, \hat{\BS{\eta}}_k)] \nonumber\\
    & \quad  - \big( \Phi(\{\mathcal{O}_{i}\}_{i \in I_k}; \theta_0, {\BS{\eta}}_{a0}) -  \mathbb{E} [ \Phi(\{\mathcal{O}_{i}\}_{i \in I_k}; \theta_0, {\BS{\eta}}_{a0})] \big), \nonumber \\
    R_2 &= \Phi(\{\mathcal{O}_{i}\}_{i \in I_k}; \theta_0, {\BS{\eta}}_{a0}), \nonumber\\
    R_3 & = \mathbb{E} [\Phi(\{\mathcal{O}_{i}\}_{i \in I_k} ; \theta_0 + \Delta_n, \hat{\BS{\eta}}_k)] \nonumber.
\end{align}
We will analyze the three terms $R_1$, $R_2$, and $R_3$ above separately. Recall that $m = |I_k| \asymp n$. For term $R_1$, we establish in Lemma \ref{lmentropy0-4} in Section \ref{SecB.9-new} that with probability tending to one,
\begin{equation} \label{eq-pf-thm1-neww-2}
\begin{split}
    |R_1| & \leq m^{-1/2} \sup_{ \theta \in \mathcal{N}_{\theta_0}^* } \big| \mathbb{G}_m [\Phi (\{\mathcal{O}_i\}_{i \in I_k};\theta, \hat{\BS{\eta}}_k) -\Phi (\{\mathcal{O}_i\}_{i \in I_k};\theta_0, \BS{\eta}_{a0})] \big| \\
    & \leq  O_p\big(  (s_{\BS{\beta}_0}^2 + {\|\BS{\Sigma}}^{-1}_{\BS{\beta}\BS{\beta}} \|_{\infty} s_{{\BS{\mu}}_{a0}}  s_{\BS{\beta}_0}  ) \log(np) /n  +  \tau_n  {\|\BS{\Sigma}}^{-1}_{\BS{\beta}\BS{\beta}} \|_{\infty} s_{{\BS{\mu}}_{a0}}   \sqrt{\log (np)} /\sqrt{n} \big) \\
    & = o_p( s_{\BS{\beta}_0} \sqrt{\log (np) /n }),
\end{split}
\end{equation}
where the last step above follows since $ s_{\BS{\beta}_0} \sqrt{\log (np)} = o(\sqrt {n})$, and ${\|\BS{\Sigma}}^{-1}_{\BS{\beta}\BS{\beta}} \|_{\infty} s_{{\BS{\mu}}_{a0}} \sqrt{\log (np)} = o(\sqrt n)$ and $ \tau_n {\|\BS{\Sigma}}^{-1}_{\BS{\beta}\BS{\beta}} \|_{\infty} s_{{\BS{\mu}}_{a0}} = o(1)$ by Condition \ref{coninimu}(v).

Next, for term $R_2$, let us recall the Doob--Meyer decomposition for the counting process $ \{N_i(t)\}_{t \in [0, \tau]} $
\begin{equation*}
    N_i(t) = M_i(t) + \int_{0}^ {t} \widetilde{w}_i(s, \theta_0, \BS{\beta}_0) \lambda_0(s) ds,
\end{equation*}
where $ \{M_i(t)\}_{t \in [0, \tau]} $ is a mean-zero martingale. Substituting such decomposition into the oracle score leads to 
	\begin{align} \label{eq-neww-expand-Phi-0}
		& \Phi (\{\mathcal{O}_i\}_{i \in I_k}; \theta_0, \BS{\eta}_{a0}) \nonumber \\
		&= \dot l_{ \theta}(\{\mathcal{O}_i\}_{i \in I_k}; \theta_0, \BS{\beta}_0)-  \BS{\mu}_{a0}^{\top} \dot l_{ \BS{\beta}}(\{\mathcal{O}_i\}_{i \in I_k}; \theta_0, \BS{\beta}_0) \nonumber\\
		&=   \frac{1}{m} \sum_{i \in I_k} \int_0^{\tau} (D_i - \bar D_k(t,  \theta_0, \BS{\beta}_0)) dN_i(t)  - \BS{\mu}_{a0}^{\top} \frac{1}{m} \sum_{i \in I_k} \int_0^{\tau}( \BS{Z}_i -\bar {\BS{Z}}_k (t, \theta_0, \BS{\beta}_0) ) dN_i(t) \nonumber\\
		&= \frac{1}{m} \sum_{i \in I_k}  \int_0^{\tau} [(D_i - \bar D_k(t, \theta_0, \BS{\beta}_0))  - {\BS{\mu}}_{a0}^{\top} ( \BS{Z}_i -\bar {\BS{Z}}_k (t,\theta_0, \BS{\beta}_0) ) ]dM_i(t).
	\end{align}
	The martingale isometry then gives the following variance representation
	\begin{align} \label{eq-neww-sigma-oracle}
		&\sigma_{\Phi}^2 := m \mathbb{E}[ \Phi (\{\mathcal{O}_i\}_{i \in I_k}; \theta_0, \BS{\eta}_{a0})^{2}] \nonumber\\
		&= \mathbb{E} \bigg\{\frac{1}{\sqrt{m}} \sum_{i \in I_k} \int_0^{\tau} [(D_i - \bar D_k(t, \theta_0, \BS{\beta}_0))  - {\BS{\mu}}_{a0}^{\top} ( \BS{Z}_i -\bar {\BS{Z}}_k (t,\theta_0, \BS{\beta}_0) )]dM_i(t) \bigg \}^2 \nonumber \\
		&=  \mathbb{E} \bigg \{ \frac{1}{m} \sum_{i \in I_k} \int_0^{\tau} [(D_i - \bar D_k(t, \theta_0, \BS{\beta}_0))  - {\BS{\mu}}_{a0}^{\top} ( \BS{Z}_i -\bar {\BS{Z}}_k (t,\theta_0, \BS{\beta}_0) ) ]^2 d \langle M_i, M_i \rangle (t) \bigg \} \nonumber \\
		&=  \mathbb{E} \bigg \{ \frac{1}{m} \sum_{i \in I_k} \int_0^{\tau} [(D_i - \bar D_k(t, \theta_0, \BS{\beta}_0))  - {\BS{\mu}}_{a0}^{\top} ( \BS{Z}_i -\bar {\BS{Z}}_k (t,\theta_0, \BS{\beta}_0) ) ]^2 \tilde w_i(t, \theta_0, \BS{\beta}_0) \lambda_0(t) dt \bigg \} \nonumber \\
		&=  \mathbb{E} \bigg \{ \frac{1}{m} \sum_{i \in I_k} \int_0^{\tau} [(D_i - \bar D_k(t, \theta_0, \BS{\beta}_0))  - {\BS{\mu}}_{a0}^{\top} ( \BS{Z}_i -\bar {\BS{Z}}_k (t,\theta_0, \BS{\beta}_0) ) ]^2 dN_i(t) \bigg \}.
	\end{align}
	In view of Conditions \ref{coninicox}(i) and \ref{coninimu}(vi), it holds that $ \sigma_{\Phi}^2 \le 4(1+K_{\BS{Z}} K_{\BS{\mu}})^2$.
Consequently, an application of the Markov inequality results in $\Phi (\{\mathcal{O}_i\}_{i \in I_k}; \theta_0, \BS{\eta}_{a0}) = O_p(\frac{1}{\sqrt m})$, and thus the stochastic order of $R_2$ is
\begin{equation} \label{eq-pf-thm1-neww-R1}
    |R_2| = O_p( {1}/{\sqrt m}) = O_p( {1}/{\sqrt n} ) = o_p(s_{\BS{\beta}_0} \sqrt{\log (np) /n }).
\end{equation}

Further, for term $R_3$, a Taylor expansion around $(\theta_0, \BS{\eta}_{a0})$ shows that there exists some $r_0 \in (0, 1)$ such that
\begin{equation*}
    \begin{split}
        R_3 & = \Delta_n \cdot \frac{\partial \mathbb{E}[\Phi(\{\mathcal{O}_{i}\}_{i \in I_k} ; \theta, \BS{\eta}_{a0})] }{\partial \theta} \Big|_{\theta = \theta_0} + \partial_{\BS{\eta}} \mathbb{E}\Phi (\{\mathcal{O}_i\}_{i=1}^n;\theta_0, \BS{\eta}_{a0}) [\hat{\BS{\eta}}_k -\BS{\eta}_{a0}] \\
        & \quad +   \frac{\partial^2 \mathbb{E} [\Phi (\{\mathcal{O}_i\}_{i \in I_k};\theta_0+r \Delta_n, \BS{\eta}_{a0}+r(\hat{\BS{\eta}}_k-\BS{\eta}_{a0}))]}{\partial r^2 }  \Big|_{r = r_0}.
    \end{split}
\end{equation*}
By resorting to Lemma \ref{lmidentification} in Section \ref{SecC.6}, Proposition \ref{lmnearortho}, and Lemma \ref{lmderivative} in Section \ref{SecB.12}, we can deduce that 
\begin{equation} \label{eq-pf-thm1-neww-3}
\begin{split}
    R_3 & \leq - c_1 C_0 s_{\BS{\beta}_0} \sqrt{\log (np) /n } + o( n^{-1/2} ) \\
    & \quad + O \big( {\|\BS{\Sigma}}^{-1}_{\BS{\beta}\BS{\beta}} \|_{\infty}  s_{{\BS{\mu}}_{a0}} s_{\BS{\beta}_0}^2   \log (np)/n  +  \tau_n  {\|\BS{\Sigma}}^{-1}_{\BS{\beta}\BS{\beta}} \|_{\infty}  s_{{\BS{\mu}}_{a0}} s_{\BS{\beta}_0} \sqrt{\log (np) /n}  \big) \\
    & = - c_1   C_0 s_{\BS{\beta}_0} \sqrt{\log (np) /n }  + o(s_{\BS{\beta}_0} \sqrt{\log (np) /n} ),
\end{split}
\end{equation}
where in the last step above we have used the facts that $ {\|\BS{\Sigma}}^{-1}_{\BS{\beta}\BS{\beta}} \|_{\infty} s_{{\BS{\mu}}_{a0}} s_{\BS{\beta}_0} \sqrt{\log (np)} = o(\sqrt{n}) $ and $ \tau_n {\|\BS{\Sigma}}^{-1}_{\BS{\beta}\BS{\beta}} \|_{\infty} s_{{\BS{\mu}}_{a0}} = o(1) $ by Condition \ref{coninimu}(v).
Substituting the bounds in \eqref{eq-pf-thm1-neww-2}, \eqref{eq-pf-thm1-neww-R1}, and \eqref{eq-pf-thm1-neww-3} into decomposition \eqref{eq-pf-thm1-neww-1} yields that when $n$ is sufficiently large,
\begin{equation}
\begin{split}
    \Phi(\{\mathcal{O}_{i}\}_{i \in I_k} ; \theta_0 + \Delta_n, \hat{\BS{\eta}}_k) & \leq - c_1 C_0 s_{\BS{\beta}_0} \sqrt{\log (np) /n } + o(s_{\BS{\beta}_0} \sqrt{\log (np) /n} ) \\
    & \leq - c_1   C_0 s_{\BS{\beta}_0} \sqrt{\log (np) /n }  /2 < 0,
\end{split}
\end{equation}
which establishes the claim in \eqref{eq-pf-thm1-neww-goal1}. An application of similar arguments gives \eqref{eq-pf-thm1-neww-goal2}.

\medskip

\textbf{Step 2}. We aim to establish the asymptotic normality of estimator $\check{\theta}_1$ lying in the neighborhood $ \mathcal{N}_{\theta_0}^*$. Recall that $|I_k| = m$. A key step is to show that $\check{\theta}_{1,k} - \theta_0$ admits the following asymptotically linear representation
	\begin{align} \label{eq-pf-thm1-asym-normal-1}
		\sqrt{m}( \check{\theta}_{1,k} - \theta_0)
		 = - \sqrt{m} J_0^{-1} \Phi (\{\mathcal{O}_i\}_{i \in I_k}; \theta_0, \BS{\eta}_{a0})+ o_p(1).
	\end{align}
    It follows from Lemma \ref{lemma-asymp-norm-main-term} in Section \ref{new.sec.lem6} that $\sqrt{m}( \check{\theta}_{1,k} - \theta_0) \stackrel{\mathcal D} {\rightarrow} N(0, J_0^{-1} \sigma_{\Phi}^2 J_0^{-1}) $.
    Averaging over the fixed number of folds, we can derive the limiting distribution of $\check{\theta}_1$
    \begin{equation} \label{eq-neww-check_theta1-normal}
        \sqrt{n}(\check{\theta}_1 - \theta_0) = \frac{1}{\sqrt {K}} \sum_{k = 1}^K \sqrt{m}( \check{\theta}_{1,k} - \theta_0) \stackrel{\mathcal D} {\rightarrow} N(0, J_0^{-1} \sigma_{\Phi}^2 J_0^{-1}).
    \end{equation}
This establishes the desired asymptotic normality of $\check{\theta}_1$ in \eqref{eq: limiting-dist}.

It remains to justify the representation in \eqref{eq-pf-thm1-asym-normal-1} above. For each $\theta \in \mathcal{N}_{\theta_0}^*$ and $\BS{\eta} \in \mathcal{T}_{N}$, a Taylor expansion gives that there exists some $( \tilde{\theta}, \tilde{\BS{\eta}} )$ lying on the line segment connecting $( \theta, \BS{\eta})$ and $ ( \theta_0, \BS{\eta}_{a0})$ such that
\begin{align*}
		& \sqrt{m} \Phi (\{\mathcal{O}_i\}_{i \in I_k}; \theta, \BS{\eta}) \nonumber  \\
		 & = \sqrt{m} \Phi (\{\mathcal{O}_i\}_{i \in I_k}; \theta_0, \BS{\eta}_{a0}) + \sqrt{m} [\Phi (\{\mathcal{O}_i\}_{i \in I_k}; \theta, \BS{\eta})-\Phi (\{\mathcal{O}_i\}_{i \in I_k}; \theta_0, \BS{\eta}_{a0})] \nonumber\\
     &= \sqrt{m} \Phi (\{\mathcal{O}_i\}_{i \in I_k}; \theta_0, \BS{\eta}_{a0}) + \sqrt{m} \partial_{\theta}  \Phi (\{\mathcal{O}_i\}_{i \in I_k}; \theta_0, \BS{\eta}_{a0})   ( \theta - \theta_0 ) \nonumber\\
        &\quad + \sqrt{m} \partial_{\BS{\eta}} \Phi (\{\mathcal{O}_i\}_{i \in I_k}; \theta_0, \BS{\eta}_{a0}) [ \BS{\eta} - \BS{\eta}_{a0} ] \nonumber \\
        & \quad + \sqrt{m}\int_{0}^{1}(1 - r) \frac{\partial^2  \Phi (\{\mathcal{O}_i\}_{i \in I_k}; \theta_0+r(\theta-\theta_0), \BS{\eta}_{a0}+r( {\BS{\eta}}-\BS{\eta}_{a0})) }{{\partial r}^2} dr.
 \end{align*}
In light of $ J_0 = - \Sigma_{\theta \theta} + \BS{\mu}_{a0}^{\top} \BS{\Sigma}_{\BS{\beta} \theta}$, applying the expansion above to $\BS{\eta} = \hat{\BS{\eta}}_{k} $ and 
solution $ \theta = \check{{\theta}}_{1,k} $ of the estimating equation \eqref{eqesttheta1} yields that 
	\begin{align} \label{eqdecompose}
		&  \sqrt{m} J_0 (\check{\theta}_{1, k} - \theta_0) + \sqrt{m} \Phi (\{\mathcal{O}_i\}_{i \in I_k}; \theta_0, \BS{\eta}_{a0}) \nonumber \\
       & =  \sqrt{m}  \big(J_0 -   \partial_{\theta}  \Phi (\{\mathcal{O}_i\}_{i \in I_k}; \theta_0, \BS{\eta}_{a0})\big)   ( \check{\theta}_{1, k}- \theta_0 )   - \sqrt{m} \partial_{\BS{\eta}}  \mathbb{E}\Phi (\{\mathcal{O}_i\}_{i \in I_k}; \theta_0, \BS{\eta}_{a0}) [ \hat{\BS{\eta}}_k - \BS{\eta}_{a0} ] \nonumber  \\
       & \quad -   \sqrt{m} \big( \partial_{\BS{\eta}}   \Phi (\{\mathcal{O}_i\}_{i \in I_k}; \theta_0, \BS{\eta}_{a0}) -  \partial_{\BS{\eta}}  \mathbb{E}\Phi (\{\mathcal{O}_i\}_{i \in I_k}; \theta_0, \BS{\eta}_{a0})\big)^{\top} ( \hat{\BS{\eta}}_k - \BS{\eta}_{a0} ) \nonumber \\
         & \quad - \sqrt{m}\int_{0}^{1}(1 - r) \frac{\partial^2  \Phi (\{\mathcal{O}_i\}_{i \in I_k}; \theta_0+r(\theta-\theta_0), \BS{\eta}_{a0} +r(\hat{\BS{\eta}}_k-\BS{\eta}_{a0})) }{{\partial r}^2} dr \nonumber\\
         &:= I_1 + I_2 + I_3 + I_4 .
 \end{align}
 
 The four terms on the right-hand side of (\ref{eqdecompose}) above are dealt with by Lemma \ref{lmentropy1}, Proposition \ref{lmnearortho}, Lemma \ref{newlemma1} in Section \ref{Sec:pf-newLemma1}, and Lemma \ref{lmderivative}.  Specifically, we have that 
\begin{align*}
     |I_1| & = O_p \big( s_{\BS{\beta}_0} n^{-1/2} {\log (np)} \big), \\
      |I_2| & = O_p(s_{\BS{\beta}_0} n^{-1/2} \log (np) ), \\
    |I_3|
    & = O_p \Big( {\|\BS{\Sigma}}^{-1}_{\BS{\beta}\BS{\beta}} \|_{\infty} s_{{\BS{\mu}}_{a0}} \big(s_{\BS{\beta}_0} n^{-1/2} \log (np) + \tau_n \sqrt{\log (np)} \big) \Big),
\end{align*}
and
\begin{align*}
    |I_4|
    & \leq \sup_{r \in (0, 1), \theta \in \mathcal{N}_{\theta_0}^*, \BS{\eta} \in \mathcal{T}_{N}}  \sqrt{m}  \bigg| \frac{\partial^2 \mathbb{E} [\Phi (\{\mathcal{O}_i\}_{i \in I_k};\theta_0+r(\theta-\theta_0), \BS{\eta}_{a0}+r(\BS{\eta}-\BS{\eta}_{a0}))]}{\partial r^2 }  \bigg| \\
    &\leq O \big( {\|\BS{\Sigma}}^{-1}_{\BS{\beta}\BS{\beta}} \|_{\infty}  s_{{\BS{\mu}}_{a0}} s_{\BS{\beta}_0}^2   \log (np)/\sqrt{n}  +  \tau_n  {\|\BS{\Sigma}}^{-1}_{\BS{\beta}\BS{\beta}} \|_{\infty}  s_{{\BS{\mu}}_{a0}} s_{\BS{\beta}_0} \sqrt{\log (np)}  \big) = o(1),
\end{align*}
where we have utilized Conditions \ref{coninicox}(iii) and \ref{coninimu}(v). Hence, it follows from \eqref{eqdecompose} that
$$
\sqrt{m} (\check{\theta}_{1, k} - \theta_0) = - J_0^{-1} \sqrt{m} \Phi (\{\mathcal{O}_i\}_{i \in I_k}; \theta_0, \BS{\eta}_{a0}) + o_p(1),
$$
which establishes \eqref{eq-pf-thm1-asym-normal-1}.

\textbf{Step 3}. We next show that
$|\hat \sigma_1^2 - \sigma^2|= o_p(1)$, and the same arguments are applicable to $\hat \sigma_2^2$. Then by Slutsky's lemma, the asymptotic result remains valid after substituting the variance estimator $\hat{\sigma}_1^2$.

To start the variance-consistency argument, let us recall that $${\sigma}^2=J_0^{-1} \sigma_{\Phi}^2 J_0^{-1} =J_0^{-1} \mathbb{E}[m \Phi (\{\mathcal{O}_i\}_{i \in I_k}; \theta_0, \BS{\eta}_{a0})^{2}] J_0^{-1}$$ and
	\begin{align} \label{eq-neww-est-sigma}
		{\hat \sigma}_1^2=\hat {J}_{0,1}^{-1} \bigg\{\frac{1}{n} \sum_{k=1}^K \sum_{i \in I_k} \int_0^{\tau} [(D_i - \bar D_k(t, \check \theta_1,\hat { \BS{\beta}}_k)) - \hat{\BS{\mu}}_{ak}^{\top} ( \BS{Z}_i -\bar {\BS{Z}}_k (t,\check \theta_1, \hat { \BS{\beta}}_k) ) ]^2 dN_i(t) \bigg\}\hat{J}_{0,1}^{-1}.
	\end{align}
	Here, $\hat{J}_{0,1} = -K^{-1}\sum\limits_{k=1}^K( \hat {\Sigma}_{\theta \theta, k,1} - {\hat {\BS{\mu}}_{ak}}^{\top} \hat {\BS{\Sigma}}_{\BS{\beta} \theta, k,1})$, $\hat {\Sigma}_{\theta \theta, k,1} = m^{-1}\sum\limits_{j \in I_k} \delta_j (D_j - \bar D_k (X_j,\check \theta_1, \hat {\BS{\beta}}_k) )^{ 2}$, and $\hat {\BS{\Sigma}}_{\BS{\beta} \theta, k,1} = m^{-1} \sum\limits_{j \in I_k} \delta_j ( \BS{Z}_j - \bar {\BS{Z}}_k (X_j,\check \theta_1, \hat {\BS{\beta}}_k)) (D_j - \bar D_k (X_j,\check \theta_1, \hat {\BS{\beta}}_k) )$. In the proof of Lemma \ref{lmidentification}, we have shown in \eqref{eq-neww-J0} that $J_0 = -{J}_{\theta \theta} + \BS{\mu}_{a0}^{\top} \BS{J}_{\BS{\beta} \theta }$. Then it holds that for each fixed $k \in \{1, \cdots, K\}$,
	\begin{align*}
		&  | -\hat {\Sigma}_{\theta \theta, k,1} + {\hat {\BS{\mu}}_{ak}}^{\top} \hat {\BS{\Sigma}}_{\BS{\beta} \theta, k,1} - J_0| \\
		&= |- \hat {\Sigma}_{\theta \theta, k,1} + {\hat {\BS{\mu}}_{ak}}^{\top} \hat {\BS{\Sigma}}_{\BS{\beta} \theta, k,1} + {J}_{\theta \theta} - \BS{\mu}_{a0}^{\top} \BS{J}_{\BS{\beta} \theta }|  \\
		&\le  |   \hat {\Sigma}_{\theta \theta, k,1} -{J}_{\theta \theta} | +| {\hat {\BS{\mu}}_{ak}}^{\top} \hat {\BS{\Sigma}}_{\BS{\beta} \theta, k,1} - \BS{\mu}_{a0}^{\top} \BS{J}_{\BS{\beta} \theta }|  \\
		&\le   | \hat {\Sigma}_{\theta \theta, k,1} -{J}_{\theta \theta} | +| (\hat {\BS{\mu}}_{ak} -\BS{\mu}_{a0})^{\top} \hat {\BS{\Sigma}}_{\BS{\beta} \theta, k,1}| +| \BS{\mu}_{a0}^{\top} ( \hat {\BS{\Sigma}}_{\BS{\beta} \theta, k,1}-\BS{J}_{\BS{\beta} \theta })|  \\
		&\le   | \hat {\Sigma}_{\theta \theta, k,1} -{J}_{\theta \theta} | + \| \hat {\BS{\mu}}_{ak} -\BS{\mu}_{a0} \|_1 \| \hat {\BS{\Sigma}}_{\BS{\beta} \theta, k,1} \|_{\infty} + \| \BS{\mu}_{a0} \|_1 \|\hat {\BS{\Sigma}}_{\BS{\beta} \theta, k,1}-\BS{J}_{\BS{\beta} \theta }\|_{\infty} \\
		& := U_1 + U_2 + U_3.
	\end{align*}

    We first examine term $\hat{J}_{0,1}$. From the formula for $\hat {\BS{\Sigma}}_{\BS{\beta} \theta, k,1}$ and Condition \ref{coninicox}(i), we have that $ \| \hat {\BS{\Sigma}}_{\BS{\beta} \theta, k,1} \|_{\infty} \le 4 K_{\BS{Z}}$. In addition, using the same techniques as for \eqref{pf-lemma1-claim1} in Section \ref{SecC.4} and \eqref{pf-thm1-eq4} in Section \ref{SecC.5}, we can deduce that 
	\begin{align*}
		|\hat {\Sigma}_{\theta \theta, k, 1} - \Sigma_{\theta \theta} |
		&= O_p(s_{\BS{\beta}_0} \sqrt{\log (np) /n} + \tau_n), \\
		\|\hat {\BS{\Sigma}}_{\BS{\beta}\theta, k, 1} - {\BS{\Sigma}}_{\BS{\beta}\theta} \|_{\infty}
		&= O_p(s_{\BS{\beta}_0} \sqrt{\log (np) /n} + \tau_n), \\
		| J_{\theta \theta} - \Sigma_{\theta \theta} | & = O( \sqrt{\log (np) /n}), \\
		\| {\BS{J}}_{\BS{\beta} \theta} - {\BS{\Sigma}}_{\BS{\beta}\theta} \|_{\infty} & = O( \sqrt{\log (np) /n}),
	\end{align*}
which along with Proposition \ref{lminimut} and Condition \ref{coninimu}(vi) lead to
	\begin{align*}
		U_1 & \le |\hat {\Sigma}_{\theta \theta, k, 1} - \Sigma_{\theta \theta} | + | J_{\theta \theta} - \Sigma_{\theta \theta} | \leq O_p(s_{\BS{\beta}_0} \sqrt{\log (np) /n} + \tau_n) = o_p(1), \\
        U_2 & \leq O_p({\|\BS{\Sigma}}^{-1}_{\BS{\beta}\BS{\beta}} \|_{\infty} s_{{\BS{\mu}}_{a0}} (s_{\BS{\beta}_0} \sqrt{\log (np) /n} + \tau_n) ) = o_p(1),\\
		U_3 & \le \| \BS{\mu}_{a0} \|_1 ( \|\hat {\BS{\Sigma}}_{\BS{\beta}\theta, k, 1} - {\BS{\Sigma}}_{\BS{\beta}\theta} \|_{\infty} + \| {\BS{J}}_{\BS{\beta} \theta} - {\BS{\Sigma}}_{\BS{\beta}\theta} \|_{\infty} ) = O_p(s_{\BS{\beta}_0} \sqrt{\log (np) /n} + \tau_n) = o_p(1).
	\end{align*}
Consequently, it follows that $| -\hat {\Sigma}_{\theta \theta, k,1} + {\hat {\BS{\mu}}_{ak}}^{\top} \hat {\BS{\Sigma}}_{\BS{\beta} \theta, k,1} - J_0| = o_p(1)$ for any fixed $k$. Since $K$ is fixed, we can obtain that $|\hat{J}_{0,1} -J_0|= o_p(1)$.

	It remains to control the inner quadratic term in the sandwich estimator in \eqref{eq-neww-est-sigma}. For this term, it holds that 
	\begin{align*}
		& \frac{1}{n} \sum_{k = 1}^K  \sum_{i \in I_k} \int_0^{\tau} [(D_i - \bar D_k(t, \check \theta_1,\hat { \BS{\beta}}_k))  - \hat{\BS{\mu}}_{ak}^{\top} ( \BS{Z}_i -\bar {\BS{Z}}_k (t,\check \theta_1, \hat { \BS{\beta}}_k) ) ]^2 dN_i(t)  \\
		&= \frac{1}{K} \sum_{k=1}^K \frac{1}{m}\sum_{i \in I_k}\int_0^{\tau} [(D_i - \bar D_k(t, \check \theta_1,\hat { \BS{\beta}}_k))  - \hat{\BS{\mu}}_{ak}^{\top} ( \BS{Z}_i -\bar {\BS{Z}}_k (t,\check \theta_1, \hat { \BS{\beta}}_k) ) ]^2 dN_i(t).
	\end{align*}
	To prove that this term converges to $\sigma_{\Phi}^2 = \mathbb{E}[m \Phi (\{\mathcal{O}_i\}_{i \in I_k}; \theta_0, \BS{\eta}_{a0})^{2}]$, it suffices to show that $ \frac{1}{m}\sum_{i \in I_k}\int_0^{\tau} [(D_i - \bar D_k(t, \check \theta_1,\hat { \BS{\beta}}_k)) - \hat{\BS{\mu}}_{ak}^{\top} ( \BS{Z}_i -\bar {\BS{Z}}_k (t,\check \theta_1, \hat { \BS{\beta}}_k) ) ]^2 dN_i(t) \rightarrow \sigma_{\Phi}^2 $ for each $k= 1,\dots, K$. Indeed, the variance representation in \eqref{eq-neww-sigma-oracle} reduces such problem to bounding
	\begin{align*}
		&\bigg|\frac{1}{m}\sum_{i \in I_k}\int_0^{\tau} [(D_i - \bar D_k(t, \check \theta_1,\hat { \BS{\beta}}_k))  - \hat{\BS{\mu}}_{ak}^{\top} ( \BS{Z}_i -\bar {\BS{Z}}_k (t,\check \theta_1, \hat { \BS{\beta}}_k) ) ]^2 dN_i(t) - \sigma_{\Phi}^2\bigg| \\
		&= \bigg|\frac{1}{m}\sum_{i \in I_k}\int_0^{\tau} [(D_i - \bar D_k(t, \check \theta_1,\hat { \BS{\beta}}_k))  - \hat{\BS{\mu}}_{ak}^{\top} ( \BS{Z}_i -\bar {\BS{Z}}_k (t,\check \theta_1, \hat { \BS{\beta}}_k) ) ]^2 dN_i(t)\\
		&\quad - \mathbb E\bigg\{ \frac{1}{m} \sum_{i \in I_k} \int_0^{\tau} [(D_i - \bar D_k(t, \theta_0, \BS{\beta}_0))  - {\BS{\mu}}_{a0}^{\top} ( \BS{Z}_i -\bar {\BS{Z}}_k (t,\theta_0, \BS{\beta}_0) ) ]^2 dN_i(t)\bigg\}\bigg|\\
		&\leq \bigg|\frac{1}{m}\sum_{i \in I_k}\int_0^{\tau} [(D_i - \bar D_k(t, \check \theta_1,\hat { \BS{\beta}}_k))  - \hat{\BS{\mu}}_{ak}^{\top} ( \BS{Z}_i -\bar {\BS{Z}}_k (t,\check \theta_1, \hat { \BS{\beta}}_k) ) ]^2 dN_i(t)\\
		& \quad -  \frac{1}{m} \sum_{i \in I_k} \int_0^{\tau} [(D_i - \bar D_k(t, \theta_0, \BS{\beta}_0))  - {\BS{\mu}}_{a0}^{\top} ( \BS{Z}_i -\bar {\BS{Z}}_k (t,\theta_0, \BS{\beta}_0) ) ]^2 dN_i(t)\bigg|\\
		& \quad + \bigg|\frac{1}{m} \sum_{i \in I_k} \int_0^{\tau} [(D_i - \bar D_k(t, \theta_0, \BS{\beta}_0))  - {\BS{\mu}}_{a0}^{\top} ( \BS{Z}_i -\bar {\BS{Z}}_k (t,\theta_0, \BS{\beta}_0) ) ]^2 dN_i(t)\\
		&\quad - \mathbb E\bigg\{ \frac{1}{m} \sum_{i \in I_k} \int_0^{\tau} [(D_i - \bar D_k(t, \theta_0, \BS{\beta}_0))  - {\BS{\mu}}_{a0}^{\top} ( \BS{Z}_i -\bar {\BS{Z}}_k (t,\theta_0, \BS{\beta}_0) ) ]^2 dN_i(t)\bigg\}\bigg|\\
		&: = Q_1 +Q_2.
	\end{align*}
	
    The following bounds for terms $Q_1$ and $Q_2$ on the right-hand side of the expression above rely on several auxiliary inequalities. An application of similar arguments as for \eqref{pf-lemma1-eq6} gives that 
    \begin{align*}
		\sup_{t \in [0,\tau]} | \bar D_k(t,\check \theta_1, \hat {\BS{\beta}}_k) - \bar D_k (t,\theta_0, \BS{\beta}_0) | &= O_p(s_{\BS{\beta}_0} \sqrt{\log (np) /n}), \\
		\sup_{t \in [0,\tau]} \| \bar {\BS{Z}}_k (t,\check \theta_1, \hat {\BS{\beta}}_k) - \bar {\BS{Z}}_k (t,\theta_0, \BS{\beta}_0)\|_{\infty} &= O_p(s_{\BS{\beta}_0} \sqrt{\log (np) /n}).
	\end{align*}
In addition, Proposition \ref{lminimut} shows that $\| \hat {\BS{\mu}}_{ak} -\BS{\mu}_{a0} \|_1 = O_p({\|\BS{\Sigma}}^{-1}_{\BS{\beta}\BS{\beta}} \|_{\infty} s_{{\BS{\mu}}_{a0}} (s_{\BS{\beta}_0} \sqrt{\log (np) /n} + \tau_n) )$, and Condition \ref{coninimu}(vi) assumes that $\|\BS{\mu}_{a0} \|_1 = O(1)$.
For term $Q_1$, we can decompose the difference into three terms 
	\begin{align*}
		Q_1 &\leq \bigg|\frac{1}{m}\sum_{i \in I_k}\int_0^{\tau} (D_i - \bar D_k(t, \check \theta_1,\hat { \BS{\beta}}_k))^2-(D_i - \bar D_k(t, \theta_0, \BS{\beta}_0))^2 dN_i(t)\bigg|\\
		& \quad +2\bigg|\frac{1}{m}\sum_{i \in I_k}\int_0^{\tau} (D_i - \bar D_k(t, \check \theta_1,\hat { \BS{\beta}}_k))\hat{\BS{\mu}}_{ak}^{\top} ( \BS{Z}_i -\bar {\BS{Z}}_k (t,\check \theta_1, \hat { \BS{\beta}}_k) )\\
		& \quad -(D_i - \bar D_k(t, \theta_0, \BS{\beta}_0)){\BS{\mu}}_{a0}^{\top} ( \BS{Z}_i -\bar {\BS{Z}}_k (t,\theta_0, \BS{\beta}_0) ) dN_i(t)\bigg|\\
		& \quad + \bigg|\frac{1}{m}\sum_{i \in I_k}\int_0^{\tau} [\hat{\BS{\mu}}_{ak}^{\top} ( \BS{Z}_i -\bar {\BS{Z}}_k (t,\check \theta_1, \hat { \BS{\beta}}_k) ) ]^2- [{\BS{\mu}}_{a0}^{\top} ( \BS{Z}_i -\bar {\BS{Z}}_k (t,\theta_0, \BS{\beta}_0) ) ]^2 dN_i(t)\bigg|\\
		&:= Q_{11} + 2 Q_{12} +Q_{13}.
	\end{align*}
	
    We will bound the three terms on the right-hand side of the expression above in turn.
The first term above can be analyzed by the uniform bound for $\bar D_k$. Specifically, it holds that 
	\begin{align*}
		Q_{11} &\leq \Big|\frac{1}{m}\sum_{i \in I_k}\int_0^{\tau} 2(D_i - \bar D_k(t, \theta_0, \BS{\beta}_0))(\bar D_k(t, \theta_0, \BS{\beta}_0)-\bar D_k(t, \check \theta_1,\hat { \BS{\beta}}_k))dN_i(t)\Big| \\
		& \quad + \Big|\frac{1}{m}\sum_{i \in I_k}\int_0^{\tau}(\bar D_k(t, \theta_0, \BS{\beta}_0)-\bar D_k(t, \check \theta_1,\hat { \BS{\beta}}_k))^2 dN_i(t)\Big| \\
		&\leq \frac{4}{m}\sum_{i \in I_k}\int_0^{\tau} \big|\bar D_k(t, \theta_0, \BS{\beta}_0)-\bar D_k(t, \check \theta_1,\hat { \BS{\beta}}_k) \big| dN_i(t) \\
		& \quad +\frac{1}{m}\sum_{i \in I_k}\int_0^{\tau} \big| \bar D_k(t, \theta_0, \BS{\beta}_0)-\bar D_k(t, \check \theta_1,\hat { \BS{\beta}}_k) \big|^2 dN_i(t) \\
		&= O_p(s_{\BS{\beta}_0} \sqrt{\log (np) /n}).
	\end{align*}

	The cross term $Q_{12}$ can be handled similarly with the error split according to the three sources of estimation error
	\begin{align*}
		Q_{12} &\leq \Big|\frac{1}{m}\sum_{i \in I_k}\int_0^{\tau} (\bar D_k(t, \theta_0, \BS{\beta}_0) - \bar D_k(t, \check \theta_1,\hat { \BS{\beta}}_k))\hat{\BS{\mu}}_{ak}^{\top} ( \BS{Z}_i -\bar {\BS{Z}}_k (t,\check \theta_1, \hat { \BS{\beta}}_k) ) dN_i(t)\Big| \\
		& \quad + \Big|\frac{1}{m}\sum_{i \in I_k}\int_0^{\tau} (D_i - \bar D_k(t, \theta_0, \BS{\beta}_0))(\hat{\BS{\mu}}_{ak}-{\BS{\mu}}_{a0})^{\top} ( \BS{Z}_i -\bar {\BS{Z}}_k (t,\check \theta_1, \hat { \BS{\beta}}_k) )dN_i(t)\Big|\\
		&\quad  + \Big|\frac{1}{m}\sum_{i \in I_k}\int_0^{\tau}(D_i - \bar D_k(t, \theta_0, \BS{\beta}_0)){\BS{\mu}}_{a0}^{\top} (\bar {\BS{Z}}_k (t,\theta_0, \BS{\beta}_0)-\bar {\BS{Z}}_k (t,\check \theta_1, \hat { \BS{\beta}}_k) ) dN_i(t)\Big| \\
		&\leq 2K_{\BS{Z}}\|\hat{\BS{\mu}}_{ak}\|_1\frac{1}{m}\sum_{i \in I_k}\int_0^{\tau} |\bar D_k(t, \theta_0, \BS{\beta}_0) - \bar D_k(t, \check \theta_1,\hat { \BS{\beta}}_k)| dN_i(t) \\
		& \quad + 4K_{\BS{Z}}\|\hat{\BS{\mu}}_{ak}-{\BS{\mu}}_{a0}\|_1\\
		& \quad + 2\|{\BS{\mu}}_{a0}\|_1\frac{1}{m}\sum_{i \in I_k}\int_0^{\tau}\|\bar {\BS{Z}}_k (t,\theta_0, \BS{\beta}_0)-\bar {\BS{Z}}_k (t,\check \theta_1, \hat { \BS{\beta}}_k)\|_\infty dN_i(t) \\
		&=    O_p({\|\BS{\Sigma}}^{-1}_{\BS{\beta}\BS{\beta}} \|_{\infty}  s_{{\BS{\mu}}_{a0}} (s_{\BS{\beta}_0} \sqrt{\log (np) /n} + \tau_n) ).
	\end{align*}
    
	Further, for the quadratic nuisance term $Q_{13}$, we can show that 
	\begin{align*}
		Q_{13} &= \Big|\frac{1}{m}\sum_{i \in I_k}\int_0^{\tau} [\hat{\BS{\mu}}_{ak}^{\top} ( \BS{Z}_i -\bar {\BS{Z}}_k (t,\check \theta_1, \hat { \BS{\beta}}_k) ) + {\BS{\mu}}_{a0}^{\top} ( \BS{Z}_i -\bar {\BS{Z}}_k (t,\theta_0, \BS{\beta}_0) ) ]\\
		& \quad \times[\hat{\BS{\mu}}_{ak}^{\top} ( \BS{Z}_i -\bar {\BS{Z}}_k (t,\check \theta_1, \hat { \BS{\beta}}_k) ) - {\BS{\mu}}_{a0}^{\top} ( \BS{Z}_i -\bar {\BS{Z}}_k (t,\theta_0, \BS{\beta}_0) ) ] dN_i(t)\Big|\\
		&\leq (\|{\hat {\BS{\mu}}_{ak}}^{\top} \|_1 + \|{\BS{\mu}}_{a0}^{\top}\|_1 )K_{\BS{Z}}\frac{1}{m}\sum_{i \in I_k}\int_0^{\tau} \|\hat{\BS{\mu}}_{ak}-{\BS{\mu}}_{a0}\|_1 \|\BS{Z}_i -\bar {\BS{Z}}_k (t,\check \theta_1, \hat { \BS{\beta}}_k)\|_\infty\\
		& \quad +\| {\BS{\mu}}_{a0}\|_1 \|\bar {\BS{Z}}_k (t,\check \theta_1, \hat { \BS{\beta}}_k) -\bar {\BS{Z}}_k (t,\theta_0, \BS{\beta}_0) \|_\infty dN_i(t)\\
		&\leq   O_p({\|\BS{\Sigma}}^{-1}_{\BS{\beta}\BS{\beta}} \|_{\infty}  s_{{\BS{\mu}}_{a0}} (s_{\BS{\beta}_0} \sqrt{\log (np) /n} + \tau_n) ).
	\end{align*}
Hence, combining the three bounds above, we can obtain that 
\begin{equation}
    Q_1 \leq O_p({\|\BS{\Sigma}}^{-1}_{\BS{\beta}\BS{\beta}} \|_{\infty} s_{{\BS{\mu}}_{a0}} (s_{\BS{\beta}_0} \sqrt{\log (np) /n} + \tau_n) ). \label{eq-neww-Q1-bound}
\end{equation}

We next turn to term $Q_2$. By the definition of the empirical process, it holds that 
	\begin{align*}
		Q_2 & = \frac{1}{\sqrt{m}}\Big|\mathbb G_m \bigg(\int_0^{\tau} [(D_i - \bar D_k(t, \theta_0, \BS{\beta}_0)) - {\BS{\mu}}_{a0}^{\top} ( \BS{Z}_i -\bar {\BS{Z}}_k (t,\theta_0, \BS{\beta}_0) ) ]^2 dN_i(t) \Big)\bigg|\\
		&\leq  \frac{1}{\sqrt{m}}\Big|\mathbb G_m \Big(\int_0^{\tau} (D_i - \bar D_k(t, \theta_0, \BS{\beta}_0))^2 dN_i(t) \Big)\Big|\\
		& \quad +\frac{2\|{\BS{\mu}}_{a0}\|_1}{\sqrt{m}}\Big\|\mathbb G_m \Big(\int_0^{\tau} (D_i - \bar D_k(t, \theta_0, \BS{\beta}_0))( \BS{Z}_i -\bar {\BS{Z}}_k (t,\theta_0, \BS{\beta}_0)) dN_i(t) \Big)\Big\|_\infty\\
		& \quad +\frac{\|{\BS{\mu}}_{a0}\|_1^2}{\sqrt{m}}\Big\|\mathbb G_m \Big(\int_0^{\tau}( \BS{Z}_i -\bar {\BS{Z}}_k (t,\theta_0, \BS{\beta}_0) )^{\otimes 2}dN_i(t) \Big)\Big\|_\infty\\
		&:= Q_{21}+Q_{22}+Q_{23}.
	\end{align*}
	The required empirical-process bounds have already been established in Lemma \ref{lmentropy} in Section \ref{SecB.11-new}.  Specifically, we have that 
	\[
	\big|\mathbb G_m(\int_0^{\tau} D_i^2 dN_i(t)) \big|= O_p(\sqrt{\log n}),
	\]
	\[
	\big\|\mathbb G_m(\int_0^{\tau} D_i\BS{Z}_i dN_i(t))\big\|_\infty = O_p(\sqrt{\log (np)}),
	\]
	\[
	\big\|\mathbb G_m(\int_0^{\tau} \BS{Z}_i^{\otimes 2} dN_i(t))\big\|_{\max} = O_p(\sqrt{\log (np^2)}),
	\]
	\[
	\big|\mathbb G_m(\int_0^{\tau} \bar D_k(t, \theta_0, \BS{\beta}_0)^2 dN_i(t))\big| = O_p(\sqrt{\log n}),
	\]
	\[
	\big\|\mathbb G_m(\int_0^{\tau} \bar D_k(t, \theta_0, \BS{\beta}_0)\bar {\BS{Z}}_k (t,\theta_0, \BS{\beta}_0) dN_i(t))\big\|_\infty = O_p(\sqrt{\log (np)}),
	\]
	and
	\[
	\big\|\mathbb G_m(\int_0^{\tau} \bar {\BS{Z}}_k (t,\theta_0, \BS{\beta}_0)^{\otimes 2} dN_i(t))\big\|_{\max} = O_p(\sqrt{\log (np)}).
	\]

	Additionally, by taking $a_i(t) = (D_i - \bar D_k(t,\theta_0,\BS{\beta}_0))\bar D_k(t,\theta_0,\BS{\beta}_0)$ in Lemma \ref{concentrationformg} in Section \ref{SecB.10}, we can show that
	\[
	\big|\mathbb G_m (\int_0^{\tau} (D_i - \bar D_k(t,\theta_0,\BS{\beta}_0))\bar D_k(t,\theta_0,\BS{\beta}_0) dN_i(t))\big| = O_p(\sqrt{\log n}).
	\]
	Then it follows that
	\begin{align*}
		Q_{21} &\leq \frac{1}{\sqrt{m}}\Big\{\big|\mathbb G_m(\int_0^{\tau} D_i^2 dN_i(t)) \big|\\
		& \quad+ 2\big|\mathbb G_m(\int_0^{\tau} (D_i - \bar D_k(t,\theta_0,\BS{\beta}_0))\bar D_k(t,\theta_0,\BS{\beta}_0) dN_i(t))\big|\\
		& \quad+ \big|\mathbb G_m(\int_0^{\tau} \bar D_k(t, \theta_0, \BS{\beta}_0)^2 dN_i(t))\big|\Big\}\\
		&= O_p(\sqrt{(\log n )/ n} ).
	\end{align*}

Similarly, by taking $\BS{a}_i(t)$ in Lemma \ref{concentrationformg} as $(D_i - \bar D_k(t,\theta_0,\BS{\beta}_0))\bar {\BS{Z}}_k(t,\theta_0,\BS{\beta}_0)$ and $(\BS{Z}_i - \bar {\BS{Z}}_k(t,\theta_0,\BS{\beta}_0))\bar {\BS{Z}}_k(t,\theta_0,\BS{\beta}_0)^\top$, respectively, we can deduce that
	 \begin{align*}
	 Q_{22} = O_p( \sqrt{\log (np) /n} ) \quad\mbox{and}\quad
	 Q_{23} = O_p( \sqrt{\log (np^2)/n}) = O_p( \sqrt{\log (np)/n}).
	 \end{align*}
Combining the componentwise bounds above yields that 
\begin{equation} \label{eq-neww-Q2-bound}
    Q_2 \leq O_p( \sqrt{\log (np)/n}).
\end{equation}

Finally, in view of \eqref{eq-neww-Q1-bound} and \eqref{eq-neww-Q2-bound}, we can obtain that 
\begin{align*}
		&  \Big|\frac{1}{n} \sum_{k = 1}^K \sum_{i \in I_k} \int_0^{\tau} [(D_i - \bar D_k(t, \check \theta_1,\hat { \BS{\beta}}_k))  - \hat{\BS{\mu}}_{ak}^{\top} ( \BS{Z}_i -\bar {\BS{Z}}_k (t,\check \theta_1, \hat { \BS{\beta}}_k) ) ]^2 dN_i(t)  - \sigma_{\Phi}^2\Big|  \\
        & =   O_p({\|\BS{\Sigma}}^{-1}_{\BS{\beta}\BS{\beta}} \|_{\infty}  s_{{\BS{\mu}}_{a0}} (s_{\BS{\beta}_0} \sqrt{\log (np) /n} + \tau_n) ) = o_p(1).
	\end{align*}
Consequently, the sandwich variance estimator is consistent
	\[
	|\hat \sigma_1^2 - \sigma^2| = o_p(1).
	\]
Therefore, by invoking Slutsky's lemma, the asymptotic normality also holds when $\sigma^2$ is replaced with $\hat \sigma_1^2 $. Using such asymptotic distribution, the asymptotic confidence interval for $\theta_0$ can be constructed accordingly.

With the aid of the same arguments, we can obtain similar results for $ \check{\theta}_2 $ and $\hat{\sigma}_2^2$.
This completes the proof of Theorem \ref{ththeta}.

\subsection{Proof of Theorem \ref{thbeta}} \label{SecB.2}

The proof of this theorem will expand the debiased estimator into a leading linear influence term plus five remainder terms. The leading term is asymptotically normal after substituting the asymptotically linear representation of $\check{\theta}-\theta_0$ from Theorem \ref{ththeta}. The remaining terms are shown to be asymptotically negligible using the CLIME-type precision matrix bounds, concentration for the Cox Hessian, and the sparsity and rate conditions imposed in Condition \ref{condebeta}. A key technical challenge is to propagate the uncertainty from estimating $\theta_0$ through the debiased estimator for $\BS{\beta}_0$ without losing the first-order normality.

To derive the asymptotic expansion, it follows from the definition of $\check {\BS{\beta}}$ in (\ref{eqestbeta}) and the mean value theorem that there exist some $\tilde \theta$ lying between $\theta_0$ and $\check \theta$, and some $\tilde {\BS{\beta}}$ lying on the line segment connecting $\BS{\beta}_0$ and $\hat {\BS{\beta}}$ such that
	\begin{align*}
		& \sqrt{n} {\BS{c}}^{\top} (\check {\BS{\beta}} -\BS{\beta}_{0}) \\
		&= \sqrt{n} {\BS{c}}^{\top} ( \hat{\BS{\beta}} + \hat {\BS{\Xi}} \dot l_{ \BS{\beta}}(\{\mathcal{O}_i\}_{i=1}^n; \check \theta,\hat{\BS{\beta}})-\BS{\beta}_{0}) \\
		& = \sqrt{n} {\BS{c}}^{\top} \BS{\Sigma}_{\BS{\beta}\BS{\beta}}^{-1}  \dot l_{ \BS{\beta}}(\{\mathcal{O}_i\}_{i=1}^n; \theta_0, \BS{\beta}_0) + \sqrt{n} {\BS{c}}^{\top} (\hat {\BS{\Xi}} - \BS{\Sigma}_{\BS{\beta}\BS{\beta}}^{-1} ) \dot l_{ \BS{\beta}}(\{\mathcal{O}_i\}_{i=1}^n; \theta_0, \BS{\beta}_0) \\
		& \quad + \sqrt{n} {\BS{c}}^{\top} \hat {\BS{\Xi}} ( \dot l_{ \BS{\beta}}(\{\mathcal{O}_i\}_{i=1}^n; \check \theta,\hat{\BS{\beta}})- \dot l_{ \BS{\beta}}(\{\mathcal{O}_i\}_{i=1}^n; \theta_0, \BS{\beta}_0)) + \sqrt{n} {\BS{c}}^{\top} (\hat {\BS{\beta}} -\BS{\beta}_{0}) \\
		&= \sqrt{n} {\BS{c}}^{\top} \BS{\Sigma}_{\BS{\beta}\BS{\beta}}^{-1}  \big\{ \dot l_{ \BS{\beta}}(\{\mathcal{O}_i\}_{i=1}^n; \theta_0, \BS{\beta}_0)   -  \BS{\Sigma}_{\BS{\beta} \theta}  (\check {\theta} -\theta_{0}) \big\} \\
        & \quad + \sqrt{n} {\BS{c}}^{\top} (\hat {\BS{\Xi}} - \BS{\Sigma}_{\BS{\beta}\BS{\beta}}^{-1} ) \dot l_{ \BS{\beta}}(\{\mathcal{O}_i\}_{i=1}^n; \theta_0, \BS{\beta}_0) \\
		& \quad + \sqrt{n} {\BS{c}}^{\top}  (\hat {\BS{\Xi}} \ddot l_{ \BS{\beta}\BS{\beta}}(\{\mathcal{O}_i\}_{i=1}^n; \tilde \theta,\tilde {\BS{\beta}}) + \BS{I}) (\hat {\BS{\beta}}- \BS{\beta}_0)   \\
        & \quad + \sqrt{n} {\BS{c}}^{\top} ( \hat {\BS{\Xi}} - \BS{\Sigma}_{\BS{\beta}\BS{\beta}}^{-1}  ) \ddot l_{ \BS{\beta}\theta}(\{\mathcal{O}_i\}_{i=1}^n; \tilde \theta,\tilde {\BS{\beta}})(\check {\theta} -\theta_{0}) \\
        & \quad + \sqrt{n}  {\BS{c}}^{\top} \BS{\Sigma}_{\BS{\beta}\BS{\beta}}^{-1}  \big\{\ddot l_{ \BS{\beta}\theta}(\{\mathcal{O}_i\}_{i=1}^n; \tilde \theta,\tilde {\BS{\beta}}) - \ddot l_{ \BS{\beta}\theta}(\{\mathcal{O}_i\}_{i=1}^n;  \theta_0, {\BS{\beta}}_0) \big\} (\check {\theta} -\theta_{0}) \\
        & \quad + \sqrt{n} {\BS{c}}^{\top} \BS{\Sigma}_{\BS{\beta}\BS{\beta}}^{-1} \big\{ \ddot l_{ \BS{\beta}\theta}(\{\mathcal{O}_i\}_{i=1}^n;  \theta_0, {\BS{\beta}}_0) +\BS{\Sigma}_{\BS{\beta} \theta}  \big\} (\check {\theta} -\theta_{0}) \\
		& := \Pi_1 +  \Pi_2 +  \Pi_3 +  \Pi_4 + \Pi_5 + \Pi_6.
	\end{align*}
    
	 We will first handle the leading term $ \Pi_1$ on the right-hand side of the expression above. In view of \eqref{eq-pf-thm1-asym-normal-1}, \eqref{eq-neww-check_theta1-normal}, and the definition in \eqref{eqortho}, it holds that 
\begin{align*}
        \Pi_1 & = \sqrt{n} {\BS{c}}^{\top} \BS{\Sigma}_{\BS{\beta}\BS{\beta}}^{-1} \big\{ \dot l_{ \BS{\beta}}(\{\mathcal{O}_i\}_{i=1}^n; \theta_0, \BS{\beta}_0) - \BS{\Sigma}_{\BS{\beta} \theta} (\check {\theta} -\theta_{0}) \big\} \\
        & = \sqrt{n} {\BS{c}}^{\top} \BS{\Sigma}_{\BS{\beta}\BS{\beta}}^{-1} \dot l_{ \BS{\beta}}(\{\mathcal{O}_i\}_{i=1}^n; \theta_0, \BS{\beta}_0) \\
        & \quad -   {\BS{c}}^{\top} \BS{\Sigma}_{\BS{\beta}\BS{\beta}}^{-1}  \BS{\Sigma}_{\BS{\beta} \theta}  J_0^{-1} \frac{1}{\sqrt{K }} \sum_{k=1}^K \sqrt{m}(   \dot l_{\theta}(\{\mathcal{O}_i\}_{i\in I_k}; \theta_0, \BS{\beta}_0)   - \BS{\Sigma}_{ \theta \BS{\beta}} \BS{\Sigma}_{\BS{\beta}\BS{\beta}}^{-1}  \dot l_{ \BS{\beta}}(\{\mathcal{O}_i\}_{i\in I_k}; \theta_0, \BS{\beta}_0)) \\
        & = \sqrt{n} {\BS{c}}^{\top} \BS{\Sigma}_{\BS{\beta}\BS{\beta}}^{-1} \dot l_{ \BS{\beta}}(\{\mathcal{O}_i\}_{i=1}^n; \theta_0, \BS{\beta}_0)  \\
        &  \quad -  \sqrt{n }  {\BS{c}}^{\top} \BS{\Sigma}_{\BS{\beta}\BS{\beta}}^{-1}  \BS{\Sigma}_{\BS{\beta} \theta}  J_0^{-1}   \big\{   \dot l_{\theta}(\{\mathcal{O}_i\}_{i=1}^n; \theta_0, \BS{\beta}_0)   - \BS{\Sigma}_{ \theta \BS{\beta}} \BS{\Sigma}_{\BS{\beta}\BS{\beta}}^{-1}  \dot l_{ \BS{\beta}}(\{\mathcal{O}_i\}_{i = 1}^n; \theta_0, \BS{\beta}_0)\big\} + o_p(1) \\
        & = \sqrt{n} {\BS{c}}^{\top} \BS{\Sigma}_{\BS{\beta}\BS{\beta}}^{-1} \Big\{ \big(I + \BS{\Sigma}_{\BS{\beta} \theta}  J_0^{-1} \BS{\Sigma}_{ \theta \BS{\beta}} \BS{\Sigma}_{\BS{\beta}\BS{\beta}}^{-1} \big) \dot l_{ \BS{\beta}}(\{\mathcal{O}_i\}_{i=1}^n; \theta_0, \BS{\beta}_0) - \BS{\Sigma}_{\BS{\beta} \theta}  J_0^{-1}  \dot l_{\theta}(\{\mathcal{O}_i\}_{i=1}^n; \theta_0, \BS{\beta}_0)    \Big\} \\
        & \quad + o_p(1).
\end{align*}
Notice that $J_0 = - ( \BS{\Sigma}_{\theta \theta} - \BS{\Sigma}_{\theta\BS{\beta}} \BS{\Sigma}_{\BS{\beta}\BS{\beta}}^{-1} \BS{\Sigma}_{\BS{\beta}\theta } ) + o(1)$,
and thus 
\begin{equation}
    \Var(\Pi_1) \rightarrow {\BS{c}}^{\top} \big\{\BS{\Sigma}_{\BS{\beta}\BS{\beta}}^{-1} + \BS{\Sigma}_{\BS{\beta}\BS{\beta}}^{-1} \BS{\Sigma}_{\BS{\beta} \theta} ( \BS{\Sigma}_{\theta \theta} - \BS{\Sigma}_{\theta\BS{\beta}} \BS{\Sigma}_{\BS{\beta}\BS{\beta}}^{-1} \BS{\Sigma}_{\BS{\beta}\theta } )^{-1} \BS{\Sigma}_{\theta\BS{\beta} } \BS{\Sigma}_{\BS{\beta}\BS{\beta}}^{-1} \big\} {\BS{c}}.
\end{equation}
Consequently, by resorting to the proof of Lemma \ref{lemma-asymp-norm-main-term}, similarly we can obtain that under Conditions \ref{coninicox}(i), \ref{coninimu}(ii)--(iii), and \ref{condebeta},
	\begin{eqnarray}\label{eqmainbeta}
		\Pi_1 \stackrel{\mathcal D} {\rightarrow} N(0, \varrho^2)
	\end{eqnarray}
	as $n \rightarrow \infty$, where $\varrho^2$ denotes the limit of ${\BS{c}}^{\top} \big\{\BS{\Sigma}_{\BS{\beta}\BS{\beta}}^{-1} + \BS{\Sigma}_{\BS{\beta}\BS{\beta}}^{-1} \BS{\Sigma}_{\BS{\beta} \theta} ( \BS{\Sigma}_{\theta \theta} - \BS{\Sigma}_{\theta\BS{\beta}} \BS{\Sigma}_{\BS{\beta}\BS{\beta}}^{-1} \BS{\Sigma}_{\BS{\beta}\theta } )^{-1} \BS{\Sigma}_{\theta\BS{\beta} } \BS{\Sigma}_{\BS{\beta}\BS{\beta}}^{-1} \big\} {\BS{c}}$.

    It remains to show that the five remainder terms above are asymptotically negligible.
	Using similar arguments as in the proof of \eqref{pf-lemma1-claim1}, we can show that
    \begin{eqnarray*}
        \|\check {\BS{\Sigma}}_{\BS{\beta}\BS{\beta}} - {\BS{\Sigma}}_{\BS{\beta}\BS{\beta}} \|_{\max} = O_p(s_{\BS{\beta}_0} \sqrt{\log (np) /n} + \tau_n).
    \end{eqnarray*}
	Since $\check {\BS{\Sigma}}_{\BS{\beta}\BS{\beta}}$ and ${\BS{\Sigma}}_{\BS{\beta}\BS{\beta}}$ are symmetric matrices, it holds that
	\begin{align*}
		 &\|\check {\BS{\Sigma}}_{\BS{\beta}\BS{\beta}} {\BS{\Sigma}}_{\BS{\beta}\BS{\beta}}^{-1}- \BS{I} \|_{\max} =  \| {\BS{\Sigma}}_{\BS{\beta}\BS{\beta}}^{-1} \check {\BS{\Sigma}}_{\BS{\beta}\BS{\beta}} - \BS{I} \|_{\max}  \\
	 &\le  \|{\BS{\Sigma}}_{\BS{\beta}\BS{\beta}}^{-1} \|_{\infty} \| \check {\BS{\Sigma}}_{\BS{\beta}\BS{\beta}} - {\BS{\Sigma}}_{\BS{\beta}\BS{\beta}} \|_{\max}
		 = O_p\big( \|{\BS{\Sigma}}_{\BS{\beta}\BS{\beta}}^{-1} \|_{\infty} (s_{\BS{\beta}_0} \sqrt{\log (np) /n} + \tau_n) \big).
	\end{align*}
	Let us define an event $\mathcal{D}=\{ \|\check {\BS{\Sigma}}_{\BS{\beta}\BS{\beta}} {\BS{\Sigma}}_{\BS{\beta}\BS{\beta}}^{-1}- \BS{I} \|_{\infty} \le \lambda_n^* \}$ with $ \lambda_n^* = C \|{\BS{\Sigma}}_{\BS{\beta}\BS{\beta}}^{-1} \|_{\infty} (s_{\BS{\beta}_0} \sqrt{\log (np) /n} + \tau_n) $ for a large constant $C > 0$. Then we have that $P(\mathcal{D}) \rightarrow 1$.

Denote by $({\BS{\Sigma}}_{\BS{\beta}\BS{\beta}}^{-1})_j$ the $j$th column of ${\BS{\Sigma}}_{\BS{\beta}\BS{\beta}}^{-1}$. In light of \eqref{eqClime2}, conditional on event $\mathcal{D}$ it holds that 
	\[
	\|{\hat {\BS{\Xi}}}_j \|_1 \le \| ({\BS{\Sigma}}_{\BS{\beta}\BS{\beta}}^{-1})_j \|_1 \ \mbox{ for each } j=1, \dots, p.
	\]
	It follows that 
	\[
	\| \hat {\BS{\Xi}} \|_{ \infty} = \max_{1 \leq j\leq p} \|{\hat {\BS{\Xi}}}_j \|_1 \le \max_{1 \leq j\leq p} \| ({\BS{\Sigma}}_{\BS{\beta}\BS{\beta}}^{-1})_j \|_1 = \|{\BS{\Sigma}}_{\BS{\beta}\BS{\beta}}^{-1} \|_{ \infty},
	\]
	and by the definition of $\hat {\BS{\Xi}}$,
	\[
	\|\hat {\BS{\Xi}} \check {\BS{\Sigma}}_{\BS{\beta}\BS{\beta}} - \BS{I} \|_{\max} \le \lambda_n^*.
	\]
	Hence, on event $\mathcal{D}$, we can deduce that
	\begin{align*}
		\| \hat {\BS{\Xi}}- {\BS{\Sigma}}_{\BS{\beta}\BS{\beta}}^{-1} \|_{\max} &= \| \hat {\BS{\Xi}} (\BS{I}- \check {\BS{\Sigma}}_{\BS{\beta}\BS{\beta}} {\BS{\Sigma}}_{\BS{\beta}\BS{\beta}}^{-1}) + (\hat {\BS{\Xi}} \check {\BS{\Sigma}}_{\BS{\beta}\BS{\beta}} - \BS{I}) {\BS{\Sigma}}_{\BS{\beta}\BS{\beta}}^{-1} \|_{\infty} \\
		& \le   \|  \hat {\BS{\Xi}}\|_{\infty} \|\BS{I}- \check {\BS{\Sigma}}_{\BS{\beta}\BS{\beta}} {\BS{\Sigma}}_{\BS{\beta}\BS{\beta}}^{-1} \|_{\max} + \| {\BS{\Sigma}}_{\BS{\beta}\BS{\beta}}^{-1} \|_{ \infty}  \|\hat {\BS{\Xi}} \check {\BS{\Sigma}}_{\BS{\beta}\BS{\beta}} - \BS{I} \|_{\max}  \\
		&\le  2 \|{\Sigma}_{\BS{\beta}\BS{\beta}}^{-1} \|_{\infty} \lambda_n^*.
	\end{align*}
    
	By invoking Lemma 9 in \cite{yu2021confidence}, we can show that on event $\mathcal{D}$,
	\[
	\| \hat {\BS{\Xi}}- {\BS{\Sigma}}_{\BS{\beta}\BS{\beta}}^{-1} \|_{ \infty} \le 12 \| \hat {\BS{\Xi}}- {\BS{\Sigma}}_{\BS{\beta}\BS{\beta}}^{-1} \|_{\max} \max_{j=1,\dots, p} h_j.
	\]
	Then it holds that on event $\mathcal{D}$, 
	\begin{align*}
		|\Pi_2| & \le \sqrt{n} \| \hat {\BS{\Xi}} - \BS{\Sigma}_{\BS{\beta}\BS{\beta}}^{-1}\|_{\infty} \| \dot l_{ \BS{\beta}}(\{\mathcal{O}_i\}_{i=1}^n; \theta_0, \BS{\beta}_0)\|_{\infty} \\
		& \le  12 \sqrt{n} \|  \hat {\BS{\Xi}}- {\BS{\Sigma}}_{\BS{\beta}\BS{\beta}}^{-1}  \|_{\max} \max_{j=1,\dots, p} h_j \| \dot l_{ \BS{\beta}}(\{\mathcal{O}_i\}_{i=1}^n; \theta_0, \BS{\beta}_0)\|_{\infty} \\
		& \le  24 \sqrt{n}  \|{\Sigma}_{\BS{\beta}\BS{\beta}}^{-1} \|_{\infty} \lambda_n^* \max_{j=1,\dots, p} h_j \| \dot l_{ \BS{\beta}}(\{\mathcal{O}_i\}_{i=1}^n; \theta_0, \BS{\beta}_0)\|_{\infty}.
	\end{align*}
	Moreover, an application of Lemma 1(i) in \cite{yu2021confidence} leads to 
    $$ \| \dot l_{ \BS{\beta}}(\{\mathcal{O}_i\}_{i=1}^n; \theta_0, \BS{\beta}_0)\|_{\infty} = O_p(\sqrt{\log(np)/n}).$$
	As a result, it follows from Condition \ref{condebeta} that 
    \begin{equation} \label{eq-neww-Pi2}
    |\Pi_2| = o_p(1).
    \end{equation}
With the aid of similar arguments as for Proposition 5 of \cite{yu2021confidence}, we can obtain that 
	\begin{eqnarray} \label{eq-neww-Pi3}
		| \Pi_3 | = o_p(1).
	\end{eqnarray}

	For term $\Pi_4$, in view of the formula of $\ddot l_{ \BS{\beta}\theta}(\{\mathcal{O}_i\}_{i=1}^n; \tilde \theta,\tilde {\BS{\beta}})$, it holds that 
    $$ \| \ddot l_{ \BS{\beta}\theta}(\{\mathcal{O}_i\}_{i=1}^n; \tilde \theta,\tilde {\BS{\beta}})\|_{\infty} \le 4 K_{\BS{Z}}$$ under Condition \ref{coninicox}(i). In addition, $\sqrt{n} (\check{\theta} - \theta_0)$ is asymptotically normal, and thus is $o_p(\sqrt{\log (np)})$.
    Hence, we have that under Condition \ref{condebeta},
	\begin{align} \label{eq-neww-Pi4}
		|\Pi_4 |
		&\le   \sqrt{n} |\check {\theta} -\theta_{0}| \| \hat {\BS{\Xi}} - \BS{\Sigma}_{\BS{\beta}\BS{\beta}}^{-1} \|_{ \infty} \|\ddot l_{ \BS{\beta}\theta}(\{\mathcal{O}_i\}_{i=1}^n; \tilde \theta,\tilde {\BS{\beta}})\|_{\infty} \nonumber\\
        & \leq  O_p( \sqrt{\log (np)} \|  \hat {\BS{\Xi}}- {\BS{\Sigma}}_{\BS{\beta}\BS{\beta}}^{-1}  \|_{\max} \max_{j=1,\dots, p} h_j) \nonumber \\
        & \leq  O_p(\sqrt{\log (np)}  \|{\Sigma}_{\BS{\beta}\BS{\beta}}^{-1} \|_{\infty} \lambda_n^* \max_{j=1,\dots, p} h_j )  \nonumber\\
        & =  O_p( \|{\Sigma}_{\BS{\beta}\BS{\beta}}^{-1} \|_{\infty}^2     \big(s_{\BS{\beta}_0}  \log (np) /\sqrt{n} + \tau_n \sqrt{\log(np) } \big) \max_{j=1,\dots, p} h_j )  = o_p(1).
 \end{align}
For term $\Pi_5$, it follows from the error bound in \eqref{pf-lemma1-eq6} (and a similar result holds for $ \sup_{t \in [0, \tau]} | \bar {D}(t,\hat \theta, \hat {\BS{\beta}}) - \bar {D} (t,\theta_0, \BS{\beta}_0) | $) that under Condition \ref{condebeta},
 \begin{align} \label{eq-neww-Pi5}
    |\Pi_5|
    & \leq  | \sqrt{n}(\check {\theta} -\theta_{0})   |  \|{\Sigma}_{\BS{\beta}\BS{\beta}}^{-1} \|_{\infty} \| \ddot l_{ \BS{\beta}\theta}(\{\mathcal{O}_i\}_{i=1}^n; \tilde \theta,\tilde {\BS{\beta}}) - \ddot l_{ \BS{\beta}\theta}(\{\mathcal{O}_i\}_{i=1}^n;  \theta_0, {\BS{\beta}}_0)  \|_{\infty} \nonumber \\
    & \leq  o_p( \sqrt{\log (np) }    \|{\Sigma}_{\BS{\beta}\BS{\beta}}^{-1} \|_{\infty}   s_{\BS{\beta}_0} \sqrt{\log (np) /n} ) = o_p(1).
\end{align}

Finally, for term $\Pi_6$, it holds that 
\begin{align*}
    \Pi_6
    & \leq | \sqrt{n} (\check {\theta} -\theta_{0})  |  \|{\Sigma}_{\BS{\beta}\BS{\beta}}^{-1} \|_{\infty} \| \ddot l_{ \BS{\beta}\theta}(\{\mathcal{O}_i\}_{i=1}^n;  \theta_0, {\BS{\beta}}_0) +\BS{\Sigma}_{\BS{\beta} \theta}  \|_{\infty}.
\end{align*}
Similar to the proof of Proposition 5 in \cite{yu2021confidence}, we can show that 
\begin{equation*}
    \| \ddot l_{ \BS{\beta}\theta}(\{\mathcal{O}_i\}_{i=1}^n; \theta_0, {\BS{\beta}}_0) +\BS{\Sigma}_{\BS{\beta} \theta} \|_{\infty} \leq O_p( \sqrt{\log (np) / n} ).
\end{equation*}
Consequently, we have that under Condition \ref{condebeta},
\begin{equation} \label{eq-neww-Pi6}
     \Pi_6 \leq O_p(\sqrt{\log(np)} \|{\Sigma}_{\BS{\beta}\BS{\beta}}^{-1} \|_{\infty} \sqrt{\log (np)/n}) = o_p(1).
\end{equation}
Therefore, combining the leading term limit in \eqref{eqmainbeta} with the remainder bounds \eqref{eq-neww-Pi2}--\eqref{eq-neww-Pi6} yields the desired asymptotic normality for $\sqrt{n} {\BS{c}}^{\top} (\check {\BS{\beta}} -\BS{\beta}_{0}) $. This concludes the proof of Theorem \ref{thbeta}.

\subsection{Proof of Proposition \ref{lminilog}} \label{SecC.2}

The proof of this proposition will exploit the 
oracle inequality for sparse logistic regression. We first verify the quadratic margin condition for the logistic loss around the true linear predictor, and then combine it with the logistic Lasso oracle inequality to obtain both the $\ell_1$ coefficient rate and the induced propensity score rate. A main technical point is that bounded covariates convert the coefficient error into a uniform error for the fitted propensity score.

Let us first introduce the excess risk that links the logistic loss to the coefficient error. For any $\BS{\gamma} \in \mathbb{R}^p$, define the excess risk $R(\BS{Z}^{\top} \BS{\gamma} )$ as
	\[
	R(\BS{Z}^{\top} \BS{\gamma} ):= \mathbb{E}[\rho({\BS{Z}}^{\top}\BS{\gamma}, D)-\rho({\BS{Z}}^{\top}\BS{\gamma}_0, D)|\BS{Z}],
	\]
    where $ \rho(u, d)=-du +\log(1+e^u)$ is the logistic loss function.
	To verify the corresponding margin condition, for each $\BS{z} \in \mathcal{\BS{Z}} \subset \mathbb{R}^p$ we set
	\[
	h(u, \BS{z}) := \mathbb{E} [\rho( u, D) | \BS{Z} =\BS{z} ] =-m_0(\BS{z})u+\log(1+e^{u}).
	\]
	Some standard calculations give that 
	\begin{align*}
		\dot h_{u}(u, \BS{z}) &:= \frac{\partial h(u, \BS{z})}{\partial u}=-m_0(\BS{z}) + \frac{e^{u}}{1+e^{u}}, \\
		\ddot h_{u}(u, \BS{z}) &:= \frac{\partial^2 h(u, \BS{z})}{\partial u \partial u}=\frac{e^{u}}{(1+e^{u})^2},
	\end{align*}
	and $h(u, \BS{z}) $ is minimized at $u_0= \log \Big(\frac{m_0(\BS{z}) }{1-m_0(\BS{z}) }\Big) ={\BS{z}}^{\top} \BS{\gamma}_0$.

	Moreover, the second-order derivative $\frac{e^{u}}{(1+e^{u})^2}$ is positive and bounded away from zero for any $u$ near $u_0$. As a result, $h(u, \BS{z})$ behaves quadratically near its minimum $u_0$. Specifically, it holds that 
	\[
	\ddot h_{u}(u, \BS{z}) \ge (e^{|u_0|+\epsilon} +1)^{-2} \ge
	(e^{\epsilon}/{c_0} +1)^{-2}
	\]
	for any $u$ such that $|u -u_0 | \le \epsilon$, where $c_0 = e^{-|u_0|}$.
	This lower bound entails the quadratic margin condition (\cite{buhlmann2011statistics}, p. 120), i.e.,
	\begin{eqnarray}\label{eqexrisk}
	R(\BS{Z}^{\top} \BS{\gamma} ) \ge C \big|{\BS{Z}}^{\top} (\BS{\gamma} - \BS{\gamma}_0) \big|^2
	\end{eqnarray}
	for all $\big|{\BS{Z}}^{\top} (\BS{\gamma} - \BS{\gamma}_0) \big| \le \epsilon$ with $C=(e^{\epsilon}/{c_0} +1)^{-2}$.

	Under Conditions \ref{coninicox}(i) and \ref{coninilog}, it follows from Lemma 6.8 (p.134) in \cite{buhlmann2011statistics} that there exists some positive constant $C^*$ such that
	\begin{align}\label{eqlogorac}
	R(\BS{Z}^{\top} \hat {\BS{\gamma}} )+\lambda_{ \BS{\gamma}} \| \hat {\BS{\gamma}} - \BS{\gamma}_0 \|_1 &\le C^* s_{\BS{\gamma}_0} \lambda_{ \BS{\gamma}}^2.
	\end{align}
	Combining the margin inequality in (\ref{eqexrisk}) with the oracle inequality in (\ref{eqlogorac}) yields that with probability tending to one,
	\begin{eqnarray*}
		\| \hat {\BS{\gamma}} - \BS{\gamma}_0 \|_1 \le C^* s_{\BS{\gamma}_0} \lambda_{ \BS{\gamma}}, ~~ \big| {\BS{Z}}^{\top} ( \hat {\BS{\gamma}} - \BS{\gamma}_0 ) \big|^2 \le \frac{C^*}{C} s_{\BS{\gamma}_0} \lambda_{ \BS{\gamma}}^2.
	\end{eqnarray*}

Finally, observe that
$$
\big| \hat{m}(\BS{Z}_i) - {m}_0(\BS{Z}_i) \big|
 \le ( \exp\{ K_{\BS{Z}} \|\hat {\BS{\gamma}} -\BS{\gamma}_0 \|_1\} -1) = O_p( s_{\BS{\gamma}_0} \sqrt{\log (np) /n}).
$$
Therefore, combining the coefficient and propensity score bounds above yields the desired conclusion, which completes the proof of Proposition \ref{lminilog}.

\subsection{Proof of Proposition \ref{lminicox}} \label{SecC.1}

The result of this proposition will follow by specializing the existing Cox Lasso oracle inequality to the time-independent covariate setting considered here. Since the estimator in \eqref{eqBetaL1} is a special case of the time-dependent estimator investigated in \cite{yu2021confidence}, the required assumptions reduce directly to Conditions \ref{coninicox}(i)--(iii). Thus, an application of Lemma 1(ii) of \cite{yu2021confidence} gives that 
\begin{align*}
    | \hat \theta - \theta_0 | +\| \hat {\BS{\beta}} - \BS{\beta}_0 \|_1= O_p( s_{\BS{\beta}_0} \sqrt{\log (np)/{n}}\,).
\end{align*}
This establishes the stated initial Cox Lasso rate, which concludes the proof of Proposition \ref{lminicox}.

\subsection{Proof of Proposition \ref{lminimut}} \label{SecC.4}

The proof of this proposition will first convert a sup-norm bound for $\hat{\BS{\mu}}_a-\BS{\mu}_{a0}$ into an $\ell_1$-bound using a thresholding argument and the beta-min condition. It will then prove the required sup-norm bound by decomposing the error into sample covariance deviations and the CLIME feasibility error. The main difficulty is that the covariance matrices involve both the Cox risk-set averages and the estimated propensity score. We will handle it by separating the estimation error, empirical fluctuation, and population approximation terms.

We begin with showing \eqref{eq-lminimut-2} assuming that \eqref{eq-lminimut-1} holds. Denote by $ t_n = \|\hat{\BS{\mu}}_{a} - \BS{\mu}_{a0} \|_{\infty} $,
$ h_j = \hat{\BS{\mu}}_{a, j} - \BS{\mu}_{a0, j} $, $ h_j^{(1)} = \hat{\BS{\mu}}_{a, j} \mathbf{1}( \hat{\BS{\mu}}_{a, j} \geq 2 t_n ) - \BS{\mu}_{a0, j}$, and $ h_j^{(2)} = h_j - h_j^{(1)} $. Then it holds that 
$$
\|\BS{\mu}_{a0}\|_1 - \| \BS{h}^{(1)} \|_1 + \| \BS{h}^{(2)} \|_1 \leq \|\BS{\mu}_{a0} + \BS{h}^{(1)} \|_1 + \| \BS{h}^{(2)} \|_1 = \|\hat{\BS{\mu}}_{a}\|_1 \leq \|\BS{\mu}_{a0} \|_{1},
$$
where the last step above is due to \eqref{pf-lemma1-eq2}. Hence, we have that $ \| \BS{h}^{(2)} \|_1 \leq \| \BS{h}^{(1)} \|_1 $, and thus $\| \BS{h} \|_1 = \| \BS{h}^{(1)} \|_1 + \| \BS{h}^{(2)} \|_1 \leq 2 \| \BS{h}^{(1)} \|_1 $. In addition, it follows that 
\begin{equation*}
    \begin{split}
    \| \BS{h}^{(1)} \|_1 & \leq \sum_{j = 1}^p | \hat{\BS{\mu}}_{a, j} - \BS{\mu}_{a0, j} | \mathbf{1} (\hat{\BS{\mu}}_{a, j} \geq 2 t_n) + \sum_{j = 1}^p | \BS{\mu}_{a0, j} | \mathbf{1} (\hat{\BS{\mu}}_{a, j} < 2 t_n).
    \end{split}
\end{equation*}
On event $ \big\{t_n := \| \hat{\BS{\mu}}_{a} - \BS{\mu}_{a0} \|_{\infty} \lesssim \| \BS{\Sigma}_{\BS{\beta} \BS{\beta}}^{-1} \|_{\infty}  s_{\BS{\mu}_{a0} } ( s_{\BS{\beta}_0} \sqrt{\log (np) / n } + \tau_n )  \big\}$, in view of Condition \ref{coninimu}(vi) that $\min_{j \in S_{{\BS{\mu}}_{a0}}} |\bmu_{a0, j}| \geq C {\|\BS{\Sigma}}^{-1}_{\BS{\beta}\BS{\beta}} \|_{\infty} s_{{\BS{\mu}}_{a0}} (s_{\BS{\beta}_0} \sqrt{\log (np) /n} + \tau_n) $ for a sufficiently large constant $C > 0$, we can show that $ \min_{j \in S_{{\BS{\mu}}_{a0}}} |\bmu_{a0, j}|> 3 t_n $ and thus $ \min_{j \in S_{{\BS{\mu}}_{a0}}} |\hat{\bmu}_{a, j}|> 2 t_n $. 

Further, it holds that 
\begin{equation*}
     \sum_{j = 1}^p | \BS{\mu}_{a0, j} | \mathbf{1} (\hat{\BS{\mu}}_{a, j} < 2 t_n) \leq \sum_{j \in S_{{\BS{\mu}}_{a0}}^c } | \BS{\mu}_{a0, j} | = 0.
\end{equation*}
Hence, on event $ \big\{t_n := \| \hat{\BS{\mu}}_{a} - \BS{\mu}_{a0} \|_{\infty} \lesssim \| \BS{\Sigma}_{\BS{\beta} \BS{\beta}}^{-1}  \|_{\infty} s_{\BS{\mu}_{a0} } ( s_{\BS{\beta}_0} \sqrt{\log (np) / n } + \tau_n )\big\}$, we can deduce that 
\begin{equation*}
\begin{split}
    \| \BS{h} \|_1 & \leq 2 \| \BS{h}^{(1)} \|_1 \leq 2 \sum_{j = 1}^p | \hat{\BS{\mu}}_{a, j} - \BS{\mu}_{a0, j} | \mathbf{1} (\hat{\BS{\mu}}_{a, j} \geq 2 t_n) \\
    & \leq 2 t_n \sum_{j = 1}^p \mathbf{1} ( {\BS{\mu}}_{a0, j} \geq  t_n) \leq 2 t_n s_{\BS{\mu}_{a0}},
\end{split}
\end{equation*}
which establishes \eqref{eq-lminimut-2}. It remains to verify the sup-norm bound that underlies the argument.

We now proceed to prove \eqref{eq-lminimut-1}. In light of the definition of estimate $ {\hat {\BS{\mu}}}_a$ in (\ref{eqClime1}) for $\BS{\mu}_{a0}$, we have the following decomposition
	\begin{align} \label{pf-lemma1-eq0}
		\| {\hat {\BS{\mu}}}_a - \BS{\mu}_{a0} \|_{\infty}
		&= \| {\BS{\Sigma}}^{-1}_{\BS{\beta}\BS{\beta}}({\BS{\Sigma}}_{\BS{\beta}\BS{\beta}}{\hat {\BS{\mu}}}_a - {\BS{\Sigma}}_{\BS{\beta} \theta}) \|_{\infty} \nonumber \\
		& =   \big\| {\BS{\Sigma}}^{-1}_{\BS{\beta}\BS{\beta}} \big[ ({\BS{\Sigma}}_{\BS{\beta}\BS{\beta}} -\hat {\BS{\Sigma}}_{\BS{\beta}\BS{\beta}}) {\hat {\BS{\mu}}}_a + (\hat {\BS{\Sigma}}_{\BS{\beta}\BS{\beta}} {\hat {\BS{\mu}}}_a - \hat {\BS{\Sigma}}_{\BS{\beta} \theta} ) + (\hat {\BS{\Sigma}}_{\BS{\beta} \theta} - {\BS{\Sigma}}_{\BS{\beta} \theta})\big] \big\|_{\infty} \nonumber\\
        & \leq  \|{\BS{\Sigma}}^{-1}_{\BS{\beta}\BS{\beta}} \|_{\infty} \big( \| \hat {\BS{\Sigma}}_{\BS{\beta}\BS{\beta}}- {\BS{\Sigma}}_{\BS{\beta}\BS{\beta}} \|_{\max} \|\hat{\BS{\mu}}_{a} \|_{1} + \lambda_n + \|\hat {\BS{\Sigma}}_{\BS{\beta} \theta} - {\BS{\Sigma}}_{\BS{\beta} \theta}\|_{\infty} \big).
	\end{align}
	To derive the desired upper bound for $\| {\hat {\BS{\mu}}}_a - \BS{\mu}_{a0} \|_{\infty}$, we will compare $\BS{\Sigma}$ and its sample counterpart $\hat{\BS{\Sigma}}$. In particular, we claim that with probability tending to one,
    \begin{eqnarray} \label{pf-lemma1-claim1}
		\|\hat {\BS{\Sigma}}_{\BS{\beta}\BS{\beta}} - {\BS{\Sigma}}_{\BS{\beta} \BS{\beta}} \|_{\max} = O_p(s_{\BS{\beta}_0} \sqrt{\log (np) /n} + \tau_n)
	\end{eqnarray}
    and
    \begin{eqnarray} \label{pf-lemma1-claim2}
        \|\hat {\BS{\Sigma}}_{\BS{\beta}\theta} - {\BS{\Sigma}}_{\BS{\beta}\theta} \|_{\infty} = O_p(s_{\BS{\beta}_0} \sqrt{\log (np) /n} + \tau_n).
    \end{eqnarray}
Then it follows from Condition \ref{coninimu}(vi) that with probability tending to one,
	\begin{align} \label{pf-lemma1-eq1}
		\| \hat {\BS{\Sigma}}_{\BS{\beta}\BS{\beta}} \BS{\mu}_{a0} -\hat {\BS{\Sigma}}_{\BS{\beta} \theta} \|_{\infty} & = \|\hat {\BS{\Sigma}}_{\BS{\beta}\BS{\beta}} {\BS{\Sigma}}_{\BS{\beta}\BS{\beta}}^{-1} {\BS{\Sigma}}_{\BS{\beta}\theta} - \hat {\BS{\Sigma}}_{\BS{\beta}\theta} \|_{\infty} \nonumber \\
        &=  \|(\hat {\BS{\Sigma}}_{\BS{\beta}\BS{\beta}} -{\BS{\Sigma}}_{\BS{\beta}\BS{\beta}}) {\BS{\Sigma}}_{\BS{\beta}\BS{\beta}}^{-1} {\BS{\Sigma}}_{\BS{\beta}\theta} +( {\BS{\Sigma}}_{\BS{\beta}\theta}- \hat {\BS{\Sigma}}_{\BS{\beta}\theta})  \|_{\infty}  \nonumber\\
		& \le \| \hat {\BS{\Sigma}}_{\BS{\beta}\BS{\beta}} -{\BS{\Sigma}}_{\BS{\beta}\BS{\beta}}\|_{\max} \|{\BS{\Sigma}}_{\BS{\beta}\BS{\beta}}^{-1} {\BS{\Sigma}}_{\BS{\beta}\theta}\|_{1} +\| {\BS{\Sigma}}_{\BS{\beta}\theta}- \hat {\BS{\Sigma}}_{\BS{\beta}\theta} \|_{\infty} \nonumber\\
		& =  O_p(s_{\BS{\beta}_0} \sqrt{\log (np) /n} + \tau_n).
	\end{align}

    By the construction of ${\hat {\BS{\mu}}}_a$ in (\ref{eqClime1}), when $\lambda_n = C (s_{\BS{\beta}_0} \sqrt{\log (np) /n} + \tau_n)$ with a sufficiently large constant $C > 0$, it holds with probability tending to one that $\|\hat{\BS{\mu}}_a\|_1 \leq \| \BS{\mu}_{a0} \|_1 $ since $\BS{\mu}_{a0}$ satisfies the constraint $ \| \hat {\BS{\Sigma}}_{\BS{\beta}\BS{\beta}} \BS{\mu}_{a0} -\hat {\BS{\Sigma}}_{\BS{\beta} \theta} \|_{\infty} \le \lambda_n$ according to \eqref{pf-lemma1-eq1}. Hence, in view of Condition \ref{coninimu}(vi), we can deduce that with probability tending to one,
	\begin{eqnarray} \label{pf-lemma1-eq2}
		\| {\hat {\BS{\mu}}}_a \|_1 \le \| {\BS{\mu}}_{a0} \|_1 \leq K_{\BS{\mu}}.
	\end{eqnarray}
    Consequently, combining \eqref{pf-lemma1-eq0}--\eqref{pf-lemma1-claim2} and \eqref{pf-lemma1-eq2} yields the desired conclusion of Proposition \ref{lminimut}.

\medskip
    It remains only to establish the two covariance-deviation claims in \eqref{pf-lemma1-claim1} and \eqref{pf-lemma1-claim2}.

\bigskip
\noindent \textbf{Proof of \eqref{pf-lemma1-claim1}}. Let us first prove the deviation bound for the covariance block. In light of the definitions in \eqref{eq-hatSigma_3} and \eqref{eq-Sigma-3}, $\hat {\BS{\Sigma}}_{\BS{\beta}\BS{\beta}}$ and ${\BS{\Sigma}}_{\BS{\beta}\BS{\beta}}$ can be written and decomposed as
	\begin{align*}
		\hat {\BS{\Sigma}}_{\BS{\beta}\BS{\beta}}
		&= \frac{1}{n_1} \sum_{1 \le i \le n, D_i=1}  \int_0^{\tau} ( \BS{Z}_i - \bar {\BS{Z}}(t,\hat \theta, \hat {\BS{\beta}}))^{\otimes 2}  \hat{m}(\BS{Z}_i) dN_i^{(1)}(t) \\
		&\quad + \frac{1}{n_0} \sum_{1 \le i \le n, D_i=0}  \int_0^{\tau} ( \BS{Z}_i - \bar {\BS{Z}}(t,\hat \theta, \hat {\BS{\beta}}))^{\otimes 2}  (1-\hat{m}(\BS{Z}_i) ) dN_i^{(0)}(t)   \\
		&:= \hat{\BS{\Sigma}}_{\BS{\beta}\BS{\beta}}^{(1)} + \hat {\BS{\Sigma}}_{\BS{\beta}\BS{\beta}}^{(0)}
        \end{align*}
        and
        \begin{align*}
		\BS{\Sigma}_{\BS{\beta}\BS{\beta}} &= \mathbb{E} \left\{ \int_0^{\tau} ( \BS{Z} - \BS{U}_{\BS{Z}} (t,\theta_0, \BS{\beta}_0))^{\otimes 2} m_0(\BS{Z}) dN^{(1)}(t) \right\} \\
		&  + \mathbb{E} \left\{ \int_0^{\tau}(\BS{Z} - \BS{U}_{\BS{Z}} (t,\theta_0, \BS{\beta}_0))^{\otimes 2}  (1-m_0(\BS{Z})) dN^{(0)}(t)  \right\} \\
		&:= {\BS{\Sigma}}_{\BS{\beta}\BS{\beta}}^{(1)} + {\BS{\Sigma}}_{\BS{\beta}\BS{\beta}}^{(0)},
	\end{align*}
    respectively. It follows that 
    \begin{eqnarray} \label{pf-lemma1-eq3}
       \|\hat {\BS{\Sigma}}_{\BS{\beta}\BS{\beta}} - {\BS{\Sigma}}_{\BS{\beta}\BS{\beta}} \|_{\max} \leq \|\hat {\BS{\Sigma}}_{\BS{\beta}\BS{\beta}}^{(1)} - {\BS{\Sigma}}_{\BS{\beta}\BS{\beta}}^{(1)} \|_{\max} + \|\hat {\BS{\Sigma}}_{\BS{\beta}\BS{\beta}}^{(0)} - {\BS{\Sigma}}_{\BS{\beta}\BS{\beta}}^{(0)} \|_{\max}.
    \end{eqnarray}
    
    For the first term on the right-hand side of the expression above, an application of the triangle inequality leads to 
    	\begin{align}
		& \big\|\hat {\BS{\Sigma}}_{\BS{\beta}\BS{\beta}}^{(1)} - {\BS{\Sigma}}_{\BS{\beta}\BS{\beta}}^{(1)} \big\|_{\max} \nonumber \\
		& \le   \bigg\| \frac{1}{n_1} \sum_{1 \le i \le n, D_i=1}  \int_0^{\tau} [ ( \BS{Z}_i - \bar {\BS{Z}}(t,\hat \theta, \hat {\BS{\beta}}))^{\otimes 2} - ( \BS{Z}_i - \BS{U}_{\BS{Z}} (t,\theta_0, \BS{\beta}_0))^{\otimes 2}]  \hat{m}(\BS{Z}_i) dN_i^{(1)}(t)  \bigg\|_{\max} \nonumber \\
		&  \quad +  \bigg\|\frac{1}{n_1} \sum_{1 \le i \le n, D_i=1} \int_0^{\tau} ( \BS{Z}_i - \BS{U}_{\BS{Z}} (t,\theta_0, \BS{\beta}_0))^{\otimes 2} (\hat{m}(\BS{Z}_i) -  {m}_0(\BS{Z}_i)) dN_i^{(1)}(t) \bigg\|_{\max} \nonumber \\
		&  \quad +  \bigg\| \frac{1}{n_1} \sum_{1 \le i \le n, D_i=1}  \int_0^{\tau} ( \BS{Z}_i - \BS{U}_{\BS{Z}} (t,\theta_0, \BS{\beta}_0))^{\otimes 2}  {m}_0(\BS{Z}_i) dN_i^{(1)}(t)  \nonumber \\
		&  \quad - \mathbb{E} \left\{ \int_0^{\tau}  ( \BS{Z} - \BS{U}_{\BS{Z}} (t,\theta_0, \BS{\beta}_0))^{\otimes 2} m_0(\BS{Z}) dN^{(1)}(t)  \right\} \bigg\|_{\max} \nonumber\\
		&:= \Delta_1 +\Delta_2 +\Delta_3.  \label{pf-lemma1-eq4}
	\end{align}

We will bound each term on the right-hand side of \eqref{pf-lemma1-eq4} above. For term $\Delta_1$, by the fact of $ 0< \hat{m}(\BS{Z}_i) < 1 $ and Condition \ref{coninicox}(i), it holds that 
\begin{align}
     &  \big\| [ ( \BS{Z}_i - \bar {\BS{Z}}(t,\hat \theta, \hat {\BS{\beta}}))^{\otimes 2} - ( \BS{Z}_i - \BS{U}_{\BS{Z}} (t,\theta_0, \BS{\beta}_0))^{\otimes 2}]  \hat{m}(\BS{Z}_i)  \big\|_{\max} \nonumber \\
       & =     \big\| ( 2\BS{Z}_i - \bar {\BS{Z}}(t,\hat \theta, \hat {\BS{\beta}})- \BS{U}_{\BS{Z}} (t,\theta_0, \BS{\beta}_0))  (  \BS{U}_{\BS{Z}} (t,\theta_0, \BS{\beta}_0)- \bar {\BS{Z}}(t,\hat \theta, \hat {\BS{\beta}}))^{\top}  \hat{m}(\BS{Z}_i) \big\|_{\max} \nonumber \\
       & \leq  4K_{\BS{Z}} \big\|\BS{U}_{\BS{Z}} (t,\theta_0, \BS{\beta}_0)- \bar {\BS{Z}}(t,\hat \theta, \hat {\BS{\beta}})  \big\|_{\infty} \nonumber\\
       & \leq 4K_{\BS{Z}}  \Big( \big\|  \bar {\BS{Z}}(t,\hat \theta, \hat {\BS{\beta}})- \bar {\BS{Z}} (t,\theta_0, \BS{\beta}_0)\|_{\infty} + \|\bar {\BS{Z}}(t,\theta_0, \BS{\beta}_0)- \BS{U}_{\BS{Z}} (t,\theta_0, \BS{\beta}_0) \big\|_{\infty} \Big). \label{pf-lemma1-eq5}
\end{align}
Moreover, a direct calculation shows that for each $t \in [0, \tau]$,
	\begin{equation} \label{eq: new-bound_Z-cont}
    \begin{split}
		&   \|  \bar {\BS{Z}}(t,\hat \theta, \hat {\BS{\beta}}) - \bar {\BS{Z}} (t,\theta_0, \BS{\beta}_0)\|_{\infty}  \\
		&= \Bigg \| \frac{\sum_{j=1}^n \BS{Z}_j Y_j(t) \exp\{D_j \hat \theta + {\BS{Z}_j}^{\top} \hat {\BS{\beta}}\}}{\sum_{j=1}^n Y_j(t) \exp\{D_j \hat \theta + {\BS{Z}_j}^{\top} \hat {\BS{\beta}} \}}-\frac{\sum_{j=1}^n \BS{Z}_j Y_j(t) \exp\{D_j  \theta_0 + {\BS{Z}_j}^{\top} \BS{\beta}_0\}}{\sum_{j=1}^n Y_j(t) \exp\{D_j  \theta_0 + {\BS{Z}_j}^{\top} \BS{\beta}_0\} }\Bigg \|_{\infty} \\
		& \le  2 K_{\BS{Z}}  \sum_{j=1}^n \frac{Y_j(t)\big|\exp\{D_j \hat \theta + {\BS{Z}_j}^{\top} \hat {\BS{\beta}}\} - \exp\{D_j  \theta_0 + {\BS{Z}_j}^{\top}  {\BS{\beta}_0}\}\big|}{\sum_{i=1}^n Y_i(t)\exp\{D_i \hat \theta + {\BS{Z}_i}^{\top} \hat {\BS{\beta}}\}}   \\
		& \leq  2 K_{\BS{Z}} (\exp\{|\hat \theta-\theta_0| +K_{\BS{Z}} \|\hat {\BS{\beta}}- \BS{\beta}_0 \|_1\} -1).
    \end{split}
	\end{equation}
    
    With the aid of Proposition \ref{lminicox} and the fact that $|1 - e^x| = O(x)$ for $|x| < c$ with small constant $c > 0$, we can deduce that 
	\begin{equation} \label{pf-lemma1-eq6}
		\sup_{t \in [0, \tau]} \| \bar {\BS{Z}}(t,\hat \theta, \hat {\BS{\beta}}) - \bar {\BS{Z}} (t,\theta_0, \BS{\beta}_0)\|_{\infty} = O_p(s_{\BS{\beta}_0} \sqrt{\log (np) /n}).
	\end{equation}
In addition, Condition \ref{coninimu}(iii) implies that there exists a sufficiently small positive constant $c_1 >0$ such that $ 1-F_T(\tau) > c_1$.  Then an application of Lemma 7 of \cite{yu2021confidence} gives that 
	\begin{equation} \label{pf-lemma1-eq7}
	 \sup_{t \in [0,\tau]} \| \bar {\BS{Z}}(t,\theta_0, \BS{\beta}_0) - \BS{U}_{\BS{Z}} (t,\theta_0, \BS{\beta}_0) \|_{\infty} = O_p(\sqrt{\log (np) /n}).
	\end{equation}
Since $0 < \hat{m}(\BS{Z}_i) < 1$, combining \eqref{pf-lemma1-eq5}--\eqref{pf-lemma1-eq7} above yields that 
\begin{align} \label{pf-lemma1-eq8}
      \Delta_1 & \leq \frac{1}{n_1} \sum_{1 \le i \le n, D_i=1} \int_0^{\tau} \big\| ( \BS{Z}_i - \bar {\BS{Z}}(t,\hat \theta, \hat {\BS{\beta}}))^{\otimes 2} - ( \BS{Z}_i - \BS{U}_{\BS{Z}} (t,\theta_0, \BS{\beta}_0))^{\otimes 2} \big\|_{\max} \hat{m}(\BS{Z}_i) dN_i^{(1)}(t) \nonumber\\
      & \leq  O_p(s_{\BS{\beta}_0} \sqrt{\log (np) /n}) \cdot  \frac{1}{n_1} \sum_{1 \le i \le n, D_i=1}  \int_0^{\tau}  \hat{m}(\BS{Z}_i)  dN_i^{(1)}(t) \nonumber\\
      & = O_p(s_{\BS{\beta}_0} \sqrt{\log (np) /n}).
\end{align}

Similarly, it follows from Proposition \ref{lminilog} that
	\begin{align} \label{pf-lemma1-eq9}
	 \Delta_2 \nonumber
	&\le \frac{1}{n_1}  \sum_{1 \le i \le n, D_i=1} \int_0^{\tau}  \big\| ( \BS{Z}_i - \BS{U}_{\BS{Z}} (t,\theta_0, \BS{\beta}_0))^{\otimes 2} \big\|_{\max} \big| \hat{m}(\BS{Z}_i) -  {m}_0(\BS{Z}_i) \big|  dN_i^{(1)}(t)  \nonumber \\
	& \le   O_p(4K_{\BS{Z}}^2 \tau_n ) = O_p( \tau_n ) .
	\end{align}
For term $\Delta_3$, an application of Hoeffding's inequality shows that for each $t>0$,
	\begin{eqnarray*}
		P( \Delta_3 > t) \le p(p+1) \exp \Big\{-\frac{n_1 t^2}{128K_{\BS{Z}}^4} \Big\},
	\end{eqnarray*}
	and thus we can obtain that
	\begin{eqnarray} \label{pf-lemma1-eq10}
	\Delta_3 = O_p(\sqrt{\log (np) /n}).
	\end{eqnarray}
    
Combining equations \eqref{pf-lemma1-eq4} and \eqref{pf-lemma1-eq8}--\eqref{pf-lemma1-eq10} leads to
	\begin{equation} \label{pf-lemma1-eq11}
		\|\hat {\BS{\Sigma}}_{\BS{\beta}\BS{\beta}}^{(1)} - {\BS{\Sigma}}_{\BS{\beta}\BS{\beta}}^{(1)} \|_{\max} = O_p(s_{\BS{\beta}_0} \sqrt{\log (np) /n} + \tau_n).
	\end{equation}
Using similar arguments, we can show that 
\begin{equation} \label{pf-lemma1-eq12}
    \|\hat {\BS{\Sigma}}_{\BS{\beta}\BS{\beta}}^{(0)} - {\BS{\Sigma}}_{\BS{\beta}\BS{\beta}}^{(0)} \|_{\max} = O_p(s_{\BS{\beta}_0} \sqrt{\log (np) /n} + \tau_n),
\end{equation}
which together with \eqref{pf-lemma1-eq3} and \eqref{pf-lemma1-eq11} yields the desired claim \eqref{pf-lemma1-claim1}.

\bigskip
\noindent \textbf{Proof of \eqref{pf-lemma1-claim2}}. Let us rewrite $\hat {\Sigma}_{\BS{\beta} \theta}$ and $ \Sigma_{\BS{\beta} \theta} $ as

	\begin{align*}
		\hat {\BS{\Sigma}}_{\BS{\beta} \theta} &= \frac{1}{n_1} \sum_{1 \le i \le n, D_i=1} \int_0^{\tau} ( \BS{Z}_i - \bar {\BS{Z}}(t,\hat \theta, \hat {\BS{\beta}})) (1 - \bar{D}(t, \hat{\theta}, \hat{\BS{\beta}})) \hat{m}(\BS{Z}_i) dN_i^{(1)}(t) \\
		& + \frac{1}{n_0} \sum_{1 \le i \le n, D_i=0}  \int_0^{\tau} ( \BS{Z}_i - \bar {\BS{Z}}(t,\hat \theta, \hat {\BS{\beta}}))\bar{D}(t, \hat{\theta}, \hat{\BS{\beta}}) (1-\hat{m}(\BS{Z}_i) ) dN_i^{(0)}(t)   \\
		&:= \hat {\BS\Sigma}_{\BS{\beta} \theta}^{(1)}   - \hat {\BS\Sigma}_{\BS{\beta} \theta}^{(0)}
        \end{align*}
        and
        \begin{align*}
		\BS{\Sigma}_{\BS{\beta} \theta} &= \mathbb{E} \left\{ \int_0^{\tau} ( \BS{Z} - U_{\BS{Z}} (t,\theta_0, \BS{\beta}_0)) (1 - U_D(t,\theta_0, \BS{\beta}_0)) m_{0}(\BS{Z}) dN^{(1)}(t) \right\} \\
		& - \mathbb{E} \left\{   \int_0^{\tau}  ( \BS{Z} - U_{\BS{Z}} (t,\theta_0, \BS{\beta}_0)) U_D(t,\theta_0, \BS{\beta}_0) (1-m_{0}(\BS{Z})) dN^{(0)}(t) \right\} \\
		&:= \BS\Sigma_{\BS{\beta} \theta}^{(1)} -\BS\Sigma_{\BS{\beta} \theta}^{(0)}.
	\end{align*}
    
	As in the proof of \eqref{pf-lemma1-eq11} and \eqref{pf-lemma1-eq12}, we can deduce that 
	\begin{align*}
		& \big\| \hat {\BS\Sigma}_{\BS{\beta} \theta}^{(1)} - \BS\Sigma_{\BS{\beta} \theta}^{(1)} \big\|_{\infty} \\
		&\le  \bigg\| \frac{1}{n_1} \sum_{1 \le i \le n: D_i=1}  \int_0^{\tau} [(\BS{Z}_i - \bar {\BS{Z}}(t,\hat \theta, \hat {\BS{\beta}})) (U_D(t,\theta_0, \BS{\beta}_0) - \bar D(t,\hat \theta, \hat {\BS{\beta}})) \\
		&\quad + (U_{\BS{Z}}(t,\theta_0, \BS{\beta}_0)-\bar {\BS{Z}} (X_i,\hat \theta, \hat {\BS{\beta}})) (1-U_D(t,\theta_0, \BS{\beta}_0) )]  \hat{m}(\BS{Z}_i) dN_i^{(1)}(t) \bigg\|_{\infty}  \\
		&\quad + \bigg\| \frac{1}{n_1} \sum_{1 \le i \le n: D_i=1}  \int_0^{\tau} ( \BS{Z}_i - U_{\BS{Z}}(t,\theta_0, \BS{\beta}_0)) (1 - U_D(t,\theta_0, \BS{\beta}_0)) (\hat {m}(\BS{Z}_i) - {m}_0(\BS{Z}_i) )dN_i^{(1)}(t) \bigg\|_{\infty} \\
		&\quad + \bigg\| \frac{1}{n_1} \sum_{1 \le i \le n: D_i=1}  \int_0^{\tau} ( \BS{Z}_i - U_{\BS{Z}}(t,\theta_0, \BS{\beta}_0)) (1 - U_D(t,\theta_0, \BS{\beta}_0)) {m}_0(\BS{Z}_i) dN_i^{(1)}(t) \\
		& \quad- \mathbb{E} \left\{   \int_0^{\tau}  ( \BS{Z} - U_{\BS{Z}} (t,\theta_0, \BS{\beta}_0)) (1 - U_D(t,\theta_0, \BS{\beta}_0)) m_{0}(\BS{Z}) dN^{(1)}(t)  \right\} \bigg\|_{\infty} \\
		&=  O_p(s_{\BS{\beta}_0} \sqrt{\log (np) /n} + \tau_n)
	\end{align*}
	and
	\begin{align*}
		&  \| \hat {\BS\Sigma}_{\BS{\beta} \theta}^{(0)} - \BS\Sigma_{\BS{\beta} \theta}^{(0)}\|_{\infty} \\
		& \le   \bigg\| \frac{1}{n_0} \sum_{1 \le i \le n: D_i=0}  \int_0^{\tau} [(\BS{Z}_i - \bar {\BS{Z}}(t,\hat \theta, \hat {\BS{\beta}})) ( \bar D(t,\hat \theta, \hat {\BS{\beta}})-U_D(t,\theta_0, \BS{\beta}_0)) \\
		& \quad + (U_{\BS{Z}}(t,\theta_0, \BS{\beta}_0)-\bar {\BS{Z}} (X_i,\hat \theta, \hat {\BS{\beta}})) U_D(t,\theta_0, \BS{\beta}_0) ] (1- \hat{m}(\BS{Z}_i)) dN_i^{(0)}(t) \bigg\|_{\infty}  \\
		& \quad + \bigg\| \frac{1}{n_0} \sum_{1 \le i \le n: D_i=0}  \int_0^{\tau} ( \BS{Z}_i - U_{\BS{Z}}(t,\theta_0, \BS{\beta}_0)) U_D(t,\theta_0, \BS{\beta}_0) (\hat {m}(\BS{Z}_i) - {m}_0(\BS{Z}_i) )dN_i^{(0)}(t) \bigg\|_{\infty} \\
		& \quad + \bigg\| \frac{1}{n_0} \sum_{1 \le i \le n: D_i=0}  \int_0^{\tau} ( \BS{Z}_i - U_{\BS{Z}}(t,\theta_0, \BS{\beta}_0)) U_D(t,\theta_0, \BS{\beta}_0)(1- {m}_0(\BS{Z}_i)) dN_i^{(0)}(t) \\
		& \quad - \mathbb{E} \left\{   \int_0^{\tau}  ( \BS{Z} - U_{\BS{Z}} (t,\theta_0, \BS{\beta}_0)) U_D(t,\theta_0, \BS{\beta}_0)(1- m_{0}(\BS{Z})) dN^{(0)}(t)  \right\} \bigg\|_{\infty} \\
		&= O_p(s_{\BS{\beta}_0} \sqrt{\log (np) /n} + \tau_n).
	\end{align*}
    
	Consequently, we can obtain that 
	\begin{align} \label{pf-lemma1-eq13}
		\|\hat {\BS{\Sigma}}_{\BS{\beta}\theta} -{\BS{\Sigma}}_{\BS{\beta}\theta} \|_{\infty} &\leq \|\hat {\BS{\Sigma}}_{\BS{\beta}\theta}^{(1)} -{\BS{\Sigma}}_{\BS{\beta}\theta}^{(1)} \|_{\infty} + \|\hat {\BS{\Sigma}}_{\BS{\beta}\theta}^{(0)} -{\BS{\Sigma}}_{\BS{\beta}\theta}^{(0)} \|_{\infty} \nonumber \\
		&= O_p(s_{\BS{\beta}_0} \sqrt{\log (np) /n} + \tau_n).
	\end{align}
This together with \eqref{pf-lemma1-claim1} establishes \eqref{pf-lemma1-claim2}, which completes the proof of Proposition \ref{lminimut}.

\subsection{Proof of Proposition \ref{lmnearortho}} \label{SecC.5}

The proof of this proposition will calculate the Gateaux derivative of the population score with respect to the nuisance parameter, and show that it is of smaller order than $n^{-1/2}$ on the nuisance realization set $\mathcal{T}_N$. The cancellation comes from the definition of $\BS{\mu}_{a0}=\BS{\Sigma}_{\BS{\beta}\BS{\beta}}^{-1}\BS{\Sigma}_{\BS{\beta}\theta}$, while the residual terms are controlled by the uniform risk-set concentration. A key challenge is to compare finite-sample risk-set matrices to their population counterparts at a rate sharp enough to preserve the root-$n$ inference.

Similar to \cite{chernozhukov2018double}, we can show that the pathwise derivative map
\[
D_r[\BS{\eta} - \BS{\eta}_{a0}] = \frac{\partial \mathbb{E}\Phi (\{\mathcal{O}_i\}_{i=1}^n;\theta_0, \BS{\eta}_{a0}+ r( \BS{\eta}- \BS{\eta}_{a0}))}{\partial r}
\]
exists for all $r \in [0,1)$ and $\BS{\eta}$ in the nuisance realization set $\mathcal T_N$, and the near-orthogonality condition holds at $r = 0$; that is,
\[
\big| \partial_{\BS{\eta}} \mathbb{E}\Phi (\{\mathcal{O}_i\}_{i=1}^n;\theta_0, \BS{\eta}_{a0}) [\BS{\eta} -\BS{\eta}_{a0}] \big| := \big| D_0[\BS{\eta} - \BS{\eta}_{a0}] \big| = o(n^{-1/2})
\]
holds for all $\BS\eta \in \mathcal T_N$. In view of \eqref{eqpartiallogparlik}, we have that $\mathbb{E} [\Phi (\{\mathcal{O}_i\}_{i=1}^n;\theta_0, \BS{\eta}_{a0})] =0$. From the formula of $\Phi (\{\mathcal{O}_i\}_{i=1}^n;\theta, \BS{\eta})$, we can verify that the Gateaux derivative map of $\mathbb{E}[\Phi (\{\mathcal{O}_i\}_{i=1}^n;\theta_0, \BS{\eta}_{a0}+r(\BS{\eta}-\BS{\eta}_{a0}))]$ exists for all $r \in [0,1)$ and $\BS{\eta}\in \mathcal{T}_N$.

To calculate such derivative, observe that
\begin{eqnarray*}
	\frac{\partial \Phi (\{\mathcal{O}_i\}_{i=1}^n; \theta, \BS{\eta}) }{\partial {\BS{\eta}} }= \left( \begin{array}{c}
		\ddot l_{ \BS{\beta}\theta}(\{\mathcal{O}_i\}_{i=1}^n;\theta, \BS{\beta})- \ddot l_{ \BS{\beta}\BS{\beta}}(\{\mathcal{O}_i\}_{i=1}^n;\theta, \BS{\beta}) \BS{\mu}_a \\
		- \dot l_{ \BS{\beta}}(\{\mathcal{O}_i\}_{i=1}^n;\theta, \BS{\beta})
	\end{array} \right).
\end{eqnarray*}
By interchanging the differentiation and expectation via the dominated convergence theorem, we can deduce that 
\begin{align}
	\frac{\partial \mathbb{E}\Phi (\{\mathcal{O}_i\}_{i=1}^n;\theta, \BS{\eta})}{ \partial {\BS{\eta}}} \Big|_{\theta=\theta_0, \BS{\eta}=\BS{\eta}_{a0}}
	&= \left( \begin{array}{c}
		-\BS{J}_{\BS{\beta} \theta}+\BS{J}_{\BS{\beta} \BS{\beta}} \BS{\mu}_{a0} \\
		\BS{0}
	\end{array} \right). \label{eq-new-E-Phi}
\end{align}
Since the nuisance parameter $\BS{\eta}$ is finite-dimensional, it follows that 
\begin{align}
	\partial_{\BS{\eta}} \mathbb{E}\Phi (\{\mathcal{O}_i\}_{i=1}^n;\theta_0, \BS{\eta}_{a0}) [\BS{\eta} -\BS{\eta}_{a0}] &:= \frac{\partial \mathbb{E}\Phi (\{\mathcal{O}_i\}_{i=1}^n;\theta_0, \BS{\eta}_{a0}+ r( \BS{\eta}- \BS{\eta}_{a0}))}{\partial r} \Big|_{r=0} \nonumber\\
	&=\Big( \frac{\partial \mathbb{E}\Phi (\{\mathcal{O}_i\}_{i=1}^n;\theta, \BS{\eta})}{ \partial  {\BS{\eta}}} \Big|_{\theta=\theta_0, \BS{\eta}=\BS{\eta}_{a0}} \Big)^{\top} (\BS{\eta} -\BS{\eta}_{a0}) \nonumber \\
	&= {\left( \begin{array}{c}
			-\BS{J}_{\BS{\beta} \theta} +\BS{J}_{\BS{\beta} \BS{\beta}} \BS{\mu}_{a0}\\
			\BS{0}
		\end{array} \right) }^{\top} \left( \begin{array}{c}
		\BS{\beta} -\BS{\beta}_0 \nonumber \\
		\BS{\mu}_a-\BS{\mu}_{a0}
	\end{array} \right) \\
	&= (-\BS{J}_{\BS{\beta} \theta} + \BS{J}_{\BS{\beta} \BS{\beta}}\BS{\mu}_{a0})^{\top} (\BS{\beta} -\BS{\beta}_0). \label{pf-thm1-eq00}
\end{align}

As a result, it holds that 
\begin{eqnarray*}
	\big| \partial_{\BS{\eta}} \mathbb{E}\Phi (\{\mathcal{O}_i\}_{i=1}^n;\theta_0, \BS{\eta}_{a0}) [\BS{\eta} -\BS{\eta}_{a0}] \big| \le \big\| \BS{J}_{\BS{\beta} \theta}- \BS{J}_{\BS{\beta} \BS{\beta}} \BS{\mu}_{a0} \big\|_{\infty} \big\| \BS{\beta} -\BS{\beta}_0 \big\|_1.
\end{eqnarray*}
To quantify the residual deviation, denote by  $\BS{E}_{\BS{\beta}\BS{\beta}} := {\BS{J}}_{\BS{\beta}\BS{\beta}} -{\BS{\Sigma}}_{\BS{\beta}\BS{\beta}}$ and $\BS{E}_{\BS{\beta}\theta} := {\BS{J}}_{\BS{\beta} \theta} - {\BS{\Sigma}}_{\BS{\beta}\theta}$. Since $ {\BS{\mu}}_{a0} = {\BS{\Sigma}}^{-1}_{\BS{\beta}\BS{\beta}} \BS{\Sigma}_{\BS{\beta} \theta}$, we can show that 
\begin{align}
	\|\BS{J}_{\BS{\beta} \theta}- \BS{J}_{\BS{\beta} \BS{\beta}} \BS{\mu}_{a0} \|_{\infty}
	&= \|\BS{J}_{\BS{\beta} \theta}- \BS{J}_{\BS{\beta} \BS{\beta}}{\BS{\Sigma}}^{-1}_{\BS{\beta}\BS{\beta}} \BS{\Sigma}_{\BS{\beta} \theta} \|_{\infty} \nonumber \\
	&\le  \|\BS{E}_{\BS{\beta} \theta}\|_{\infty}+\| \BS{E}_{\BS{\beta}\BS{\beta}}{\BS{\Sigma}}^{-1}_{\BS{\beta}\BS{\beta}} \BS{\Sigma}_{\BS{\beta} \theta} \|_{\infty} \nonumber  \\
	&\le  \| \BS{E}_{\BS{\beta} \theta}\|_{\infty}+\| \BS{E}_{\BS{\beta}\BS{\beta}}\|_{\max}  \| {\BS{\Sigma}}^{-1}_{\BS{\beta}\BS{\beta}} \BS{\Sigma}_{\BS{\beta} \theta} \|_{1}. \label{pf-thm1-eq0}
\end{align}

For a sufficiently large constant $\widetilde{C} > 0$, let us define an event
\begin{align} \label{eq-neww-event-eps}
  \mathcal{E}=&\left\{ \sup_{t \in [0,\tau]} \big| \bar D(t,\theta_0, \BS{\beta}_0) - U_D(t,\theta_0, \BS{\beta}_0)\big| \le \tilde C \sqrt{\log (np) /n}\,\,\,\mbox{and}\right. \nonumber\\
  &  \quad \quad\quad\quad \left.\sup_{t \in [0,\tau]} \| \bar {\BS{Z}}(t,\theta_0, \BS{\beta}_0) - \BS{U}_{\BS{Z}} (t,\theta_0, \BS{\beta}_0) \|_{\infty}\le \tilde C \sqrt{\log (np) /n}
  \right\}.
\end{align}
By invoking Lemma 7 in \cite{yu2021confidence}, it follows from Conditions \ref{coninicox}(i), \ref{coninicox}(iv), and \ref{coninimu}(ii)--(iii) that $\mathbb{P}(\mathcal{E}^c) = O(n^{-1})$, and thus $\mathbb{P}(\mathcal{E}) \rightarrow 1$.

Further, in light of the definitions of  $\BS{E}_{\BS{\beta}\theta} := {\BS{J}}_{\BS{\beta} \theta} - {\BS{\Sigma}}_{\BS{\beta}\theta}$, ${\BS{J}}_{\BS{\beta} \theta}$, and ${\BS{\Sigma}}_{\BS{\beta}\theta}$, it holds that 
	\begin{align}
		\| \BS{E}_{\BS{\beta} \theta} \|_{\infty}
		&= \Big\| \mathbb{E} \Big \{\frac{1}{n} \sum_{i=1}^n \int_0^{\tau}  \big[( \BS{Z}_i - \bar {\BS{Z}} (t,\theta_0, \BS{\beta}_0) ) (\bar D(t,\theta_0, \BS{\beta}_0) -U_D(t,\theta_0, \BS{\beta}_0)) \nonumber\\
        & \quad + ( \bar {\BS{Z}}(t,\theta_0, \BS{\beta}_0)- \BS{U}_{\BS{Z}} (t,\theta_0, \BS{\beta}_0)) ( D_i - U_D(t,\theta_0, \BS{\beta}_0) ) \big]  dN_i(t) \Big \} \Big\|_{\infty} \nonumber\\
		&\leq  \int_{\mathcal{E}}  \Big\| \frac{1}{n} \sum_{i=1}^n \int_0^{\tau}  \big[( \BS{Z}_i - \bar {\BS{Z}} (t,\theta_0, \BS{\beta}_0) ) (\bar D(t,\theta_0, \BS{\beta}_0) -U_D(t,\theta_0, \BS{\beta}_0)) \nonumber \\
		&\quad  + ( \bar {\BS{Z}}(t,\theta_0, \BS{\beta}_0)- \BS{U}_{\BS{Z}} (t,\theta_0, \BS{\beta}_0)) (D_i-  U_D(t,\theta_0, \BS{\beta}_0) )  \big] dN_i(t) \Big\|_{\infty} dP  + \mathbb{P} (\mathcal{E}^c) \cdot 4  K_{\BS{Z}}  \nonumber\\
        & \leq   \int_{\mathcal{E}}  \Big\| \frac{1}{n} \sum_{i=1}^n \int_0^{\tau}  ( \BS{Z}_i - \bar {\BS{Z}} (t,\theta_0, \BS{\beta}_0) ) (\bar D(t,\theta_0, \BS{\beta}_0) -U_D(t,\theta_0, \BS{\beta}_0)) dN_i(t)\Big\|_{\infty} dP \nonumber\\
        & \quad + \int_{\mathcal{E}}  \Big\| \frac{1}{n} \sum_{i=1}^n  ( \bar {\BS{Z}}(t,\theta_0, \BS{\beta}_0)- \BS{U}_{\BS{Z}} (t,\theta_0, \BS{\beta}_0)) (D_i-  U_D(t,\theta_0, \BS{\beta}_0) )    dN_i(t) \Big\|_{\infty} dP \nonumber \\
        & \quad +  O(n^{-1}). \label{pf-thm1-eq1}
\end{align}

On event $\mathcal{E}$, the first two terms on the right-hand side of the expression above can be bounded as 
\begin{align}
	&  \Big\| \frac{1}{n} \sum_{i=1}^n \int_0^{\tau} (\BS{Z}_i - \bar {\BS{Z}} (t,\theta_0, \BS{\beta}_0))(\bar D(t,\theta_0, \BS{\beta}_0) - U_D(t,\theta_0, \BS{\beta}_0)) dN_i(t)\Big\|_{\infty}  \nonumber\\
	& \le   \frac{1}{n} \sum_{i=1}^n \Big\| \int_0^{\tau} (\BS{Z}_i - \bar {\BS{Z}} (t,\theta_0, \BS{\beta}_0))(\bar D(t,\theta_0, \BS{\beta}_0) - U_D(t,\theta_0, \BS{\beta}_0)) dN_i(t)\Big\|_{\infty}  \nonumber\\
	& =   \frac{1}{n} \sum_{i=1}^n \Big\| \delta_i (\BS{Z}_i - \bar {\BS{Z}} (X_i,\theta_0, \BS{\beta}_0))(\bar D(X_i,\theta_0, \BS{\beta}_0) - U_D(X_i,\theta_0, \BS{\beta}_0)) \Big\|_{\infty} \nonumber \\
	& \le   \frac{1}{n} \sum_{i=1}^n \Big\|  \BS{Z}_i - \bar {\BS{Z}} (X_i,\theta_0, \BS{\beta}_0) \Big\|_{\infty} \Big|\bar D(X_i,\theta_0, \BS{\beta}_0) - U_D(X_i,\theta_0, \BS{\beta}_0) \Big|   \nonumber\\
	& \le  2 K_{\BS{Z}} \tilde C \sqrt{\log (np) /n}  \label{pf-thm1-eq2}
	\end{align}
	and
	\begin{align}
	&   \Big\| \frac{1}{n} \sum_{i=1}^n \int_0^{\tau} (\bar {\BS{Z}} (t,\theta_0, \BS{\beta}_0)  - \BS{U}_{\BS{Z}} (t,\theta_0, \BS{\beta}_0)) (D_i-  U_D(X_i,\theta_0, \BS{\beta}_0) )  dN_i(t) \Big\|_{\infty}  \nonumber \\
	& \le   \frac{1}{n} \sum_{i=1}^n \Big\| \int_0^{\tau} (\bar {\BS{Z}} (t,\theta_0, \BS{\beta}_0)  - \BS{U}_{\BS{Z}} (t,\theta_0, \BS{\beta}_0)) (D_i-  U_D(X_i,\theta_0, \BS{\beta}_0))  dN_i(t)\Big\|_{\infty} \nonumber \\
	& =   \frac{1}{n} \sum_{i=1}^n \Big\| \delta_i (\bar {\BS{Z}} (X_i,\theta_0, \BS{\beta}_0)  - \BS{U}_{\BS{Z}} (X_i,\theta_0, \BS{\beta}_0)) (D_i-  U_D(X_i,\theta_0, \BS{\beta}_0) )  \Big\|_{\infty} \nonumber \\
	& \le   \frac{1}{n} \sum_{i=1}^n \Big\|  \bar {\BS{Z}} (X_i,\theta_0, \BS{\beta}_0)  - \BS{U}_{\BS{Z}} (X_i,\theta_0, \BS{\beta}_0) \Big\|_{\infty} \Big|D_i-  U_D(X_i,\theta_0, \BS{\beta}_0) \Big|  \nonumber \\
	& \le   \tilde C \sqrt{\log (np) /n},\label{pf-thm1-eq3}
\end{align}
respectively.

Combining \eqref{pf-thm1-eq1}--\eqref{pf-thm1-eq3}, we can derive the desired bound for $\BS{E}_{\BS{\beta} \theta}$
\begin{align}
	\| \BS{E}_{\BS{\beta} \theta} \|_{\infty} &\le (2 K_{\BS{Z}} +1) \tilde C \sqrt{\log (np) /n} + O( n^{-1}) \nonumber\\
	&= O( \sqrt{\log (np) /n}). \label{pf-thm1-eq4}
\end{align}
An application of similar arguments gives that 
\begin{align}
		\| \BS{E}_{\BS{\beta} \BS{\beta}} \|_{\max}
		&= \Big\| \mathbb{E} \Big \{ \frac{1}{n} \sum_{i=1}^n \int_0^{\tau} [2 ( \BS{Z}_i - \bar {\BS{Z}}(t,\theta_0, \BS{\beta}_0)) ( \bar {\BS{Z}}(t,\theta_0, \BS{\beta}_0)- \BS{U}_{\BS{Z}} (t,\theta_0, \BS{\beta}_0))^{\top}  \nonumber \\
		& \quad + (\bar {\BS{Z}}(t,\theta_0, \BS{\beta}_0) - \BS{U}_{\BS{Z}} (t,\theta_0, \BS{\beta}_0)) ^{\otimes 2}] dN_i(t) \Big \} \Big\|_{\max}  \nonumber\\
		&\le    4 K_{\BS{Z}}  \tilde C \sqrt{\log (np) /n} + \tilde C^2  \log (np) /n + O(n^{-1})  \nonumber \\
		&= O( \sqrt{\log (np) /n}),  \label{pf-thm1-eq42}
\end{align}
which together with \eqref{pf-thm1-eq0}, \eqref{pf-thm1-eq4}, and Condition \ref{coninimu}(vi) leads to
\begin{eqnarray*}
	\|\BS{J}_{\BS{\beta} \theta}- \BS{J}_{\BS{\beta} \BS{\beta}} \BS{\mu}_{a0} \|_{\infty} = O( \sqrt{\log (np) /n}).
\end{eqnarray*}

Consequently, it follows from Condition \ref{coninicox}(iii) that
\begin{align*}
	\big| \partial_{\BS{\eta}} \mathbb{E}\Phi (\{\mathcal{O}_i\}_{i=1}^n;\theta_0, \BS{\eta}_{a0}) [\BS{\eta} -\BS{\eta}_{a0}] \big|
	&\le   \big\| \BS{J}_{\BS{\beta} \theta}- \BS{J}_{\BS{\beta} \BS{\beta}} \BS{\mu}_{a0}  \big\|_{\infty}  \big\| \BS{\beta} -\BS{\beta}_0 \big\|_1 \\
	& =   O(s_{\BS{\beta}_0} \log (np)/n )  \\
	& =  o(n^{-1/2}).
\end{align*}
Therefore, the Neyman near-orthogonality condition holds in view of the definition 2.2 of \cite{chernozhukov2018double}. This concludes the proof of Proposition \ref{lmnearortho}.

\section{Some key lemmas and their proofs} \label{SecC}

\subsection{Proof of Lemma \ref{lemma-equi-Sigma}} \label{Sec:pf-newlemma1}

The proof of this lemma will condition on $\BS{Z}$, apply the propensity score model, and then exploit Fubini's theorem to exchange the conditional expectations and time integration. A technical step is to express the observed counting-process contribution through the potential counting processes $N^{(1)}$ and $N^{(0)}$ while preserving the Cox weights.

Recall that $\tilde w(t, \theta_0, \BS{\beta}_0) = Y(t) \exp\{D \theta_0+\BS{Z}^{\top} \BS{\beta_0} \}$ and $ Y(t) = D Y^{(1)}(t) + (1 - D) Y^{(0)}(t) $. Then we can condition on $\BS{Z}$ and use the propensity score model in \eqref{eqmodellog} to deduce that 
\begin{align*}
	&  \mathbb{E} [( \BS{Z} - \BS{U}_{\BS{Z}} (t,\theta_0, \BS{\beta}_0)) (D - U_D(t,\theta_0, \BS{\beta}_0))  \tilde w(t, \theta_0, \BS{\beta}_0) | \BS{Z} ] \\
	&=   \mathbb{E} [( \BS{Z} - \BS{U}_{\BS{Z}} (t,\theta_0, \BS{\beta}_0)) (D - U_D(t,\theta_0, \BS{\beta}_0))  \tilde w(t, \theta_0, \BS{\beta}_0) | \BS{Z}, D=1 ]P(D=1 |\BS{Z})  \\
	& \quad +  \mathbb{E} [( \BS{Z} - \BS{U}_{\BS{Z}} (t,\theta_0, \BS{\beta}_0)) (D - U_D(t,\theta_0, \BS{\beta}_0))  \tilde w(t, \theta_0, \BS{\beta}_0) | \BS{Z}, D=0 ]P(D=0 |\BS{Z}) \\
	&=  ( \BS{Z} - \BS{U}_{\BS{Z}} (t,\theta_0, \BS{\beta}_0)) (1 - U_D(t,\theta_0, \BS{\beta}_0)) \mathbb{E} [  Y(t)| \BS{Z}, D=1 ] \exp\{\theta_0+\BS{Z}^{\top} \BS{\beta}_0 \} m_0(\BS{Z}) \\
	& \quad - ( \BS{Z} - \BS{U}_{\BS{Z}} (t,\theta_0, \BS{\beta}_0)) U_D(t,\theta_0, \BS{\beta}_0) \mathbb{E} [ Y(t)| \BS{Z}, D=0 ]\exp\{\BS{Z}^{\top} \BS{\beta}_0 \} (1-m_0(\BS{Z})) \\
	&=  ( \BS{Z} - \BS{U}_{\BS{Z}} (t,\theta_0, \BS{\beta}_0)) (1 - U_D(t,\theta_0, \BS{\beta}_0)) \mathbb{E} [  Y^{(1)}(t)| \BS{Z}] \exp\{\theta_0+\BS{Z}^{\top} \BS{\beta}_0 \} m_0(\BS{Z}) \\
	& \quad - ( \BS{Z} - \BS{U}_{\BS{Z}} (t,\theta_0, \BS{\beta}_0)) U_D(t,\theta_0, \BS{\beta}_0) \mathbb{E} [ Y^{(0)}(t) | \BS{Z} ]\exp\{\BS{Z}^{\top} \BS{\beta}_0 \} (1-m_0(\BS{Z})).
    \end{align*}
    
Hence, it follows from Fubini's theorem and the law of iterated expectation that $\BS{\Sigma}_{\BS{\beta} \theta}$ can be written as
\begin{align}
     \BS{\Sigma}_{\BS{\beta} \theta}
	&=  \mathbb{E} \left\{   \int_0^{\tau}  ( \BS{Z} - \BS{U}_{\BS{Z}} (t,\theta_0, \BS{\beta}_0)) (1 - U_D(t,\theta_0, \BS{\beta}_0)) Y^{(1)}(t) \exp\{\theta_0+\BS{Z}^{\top} \BS{\beta}_0 \} m_0(\BS{Z}) \lambda_0(t) dt  \right\} \nonumber\\
	&  \quad- \mathbb{E} \left\{   \int_0^{\tau}  ( \BS{Z} -  \BS{U}_{\BS{Z}} (t,\theta_0, \BS{\beta}_0)) U_D(t,\theta_0, \BS{\beta}_0) Y^{(0)}(t) \exp\{\BS{Z}^{\top} \BS{\beta}_0 \} (1-m_0(\BS{Z})) \lambda_0(t) dt  \right\} \nonumber \\
	& =  \mathbb{E} \left\{   \int_0^{\tau}  ( \BS{Z} -  \BS{U}_{\BS{Z}} (t,\theta_0, \BS{\beta}_0)) (1 - U_D(t,\theta_0, \BS{\beta}_0)) m_0(\BS{Z}) dN^{(1)}(t)  \right\} \nonumber \\
	&  \quad - \mathbb{E} \left\{   \int_0^{\tau}  ( \BS{Z} -  \BS{U}_{\BS{Z}} (t,\theta_0, \BS{\beta}_0)) U_D(t,\theta_0, \BS{\beta}_0) (1-m_0(\BS{Z})) dN^{(0)}(t) \right\}.
\end{align}
Similarly, we can obtain \eqref{eq-Sigma-2} and \eqref{eq-Sigma-3}. This completes the proof of Lemma \ref{lemma-equi-Sigma}.

\subsection{Lemma \ref{lmidentification} and its proof} \label{SecC.6}

Lemma \ref{lmidentification} below establishes the local identification of the scalar target parameter. It shows that the population orthogonal score has a stable, nonzero derivative at $\theta_0$, and the score is locally separated from zero away from $\theta_0$.

\begin{lemma}\label{lmidentification}
	Assume that Conditions \ref{coninicox}(i), \ref{coninimu}(ii)--(iv), and \ref{coninimu}(vi) are satisfied. Let $J^* := \frac{\partial \mathbb{E}[\Phi (\{\mathcal{O}_i\}_{i =1}^n; \theta, \BS{\eta}_{a0})]}{\partial \theta} \Big|_{\theta=\theta_0}$. Then there exist some finite constants $c_2 > c_1 >0$ such that $c_1 < -J^*<c_2$. Moreover, we have that for each $\theta \in \mathcal{N}_{\theta_0}^*$, 
	\[
	2 \big| \mathbb{E} [\Phi (\{\mathcal{O}_i\}_{i=1}^n;\theta, \BS{\eta}_{a0})] \big| \ge |J^*(\theta-\theta_0)|.
	\]
\end{lemma}

\noindent\textit{Proof}. The proof of this lemma will calculate the Jacobian $J^*$ by differentiating the population score and rewriting it in terms of the population covariance blocks. The Schur-complement-type quantity is bounded away from zero by the model conditions, while finite-sample to population deviations are controlled by the risk-set concentration. The local separation statement then follows from the mean value theorem and continuity of the derivative.

By resorting to Fubini's theorem and the dominated convergence theorem, the Jacobian matrix $J^*$ (a scalar in this case) is given by 
\begin{align*}
	J^* &= \frac{\partial \mathbb{E}[\Phi (\{\mathcal{O}_i\}_{i =1}^n; \theta, \BS{\eta}_{a0})]}{\partial \theta} \bigg|_{\theta=\theta_0} = \mathbb{E} \bigg[\frac{\partial [\Phi (\{\mathcal{O}_i\}_{i =1}^n; \theta, \BS{\eta}_{a0})]}{\partial \theta} \bigg|_{\theta=\theta_0} \bigg]\\
	&= \mathbb{E} \bigg[- \frac{1}{n} \sum_{i=1}^n \int_0^{\tau}  \sum_{j=1}^n
	( D_j - \bar D(t,\theta_0, \BS{\beta}_0))^{2} w_j(t, \theta_0, \BS{\beta}_0) dN_i(t) \bigg] \\
	& \quad + \BS{\mu}_{a0}^{\top} \mathbb{E} \bigg[ \frac{1}{n} \sum_{i=1}^n \int_0^{\tau}  \sum_{j=1}^n  ( D_j - \bar D(t,\theta_0, \BS{\beta}_0)) (\BS{Z}_j - \bar {\BS{Z}} (t,\theta_0, \BS{\beta}_0)) w_j(t, \theta_0, \BS{\beta}_0)  dN_i(t) \bigg].
\end{align*}
Moreover, it follows from the Doob--Meyer decomposition for the counting process $ \{N_i(t)\}_{t \in [0, \tau]} $ that
\begin{equation} \label{doob-mayer-decom}
    N_i(t) = M_i(t) + \int_{0}^ {t} \widetilde{w}_i(s, \theta_0, \BS{\beta}_0) \lambda_0(s) ds,
\end{equation}
where $ \{M_i(t)\}_{t \in [0, \tau]} $ is a mean-zero martingale. 

Substituting the decomposition in (\ref{doob-mayer-decom}) above into the oracle score leads to 
\begin{align} \label{eq-neww-J0}
    J^* &= - \mathbb{E} \bigg[\frac{1}{n} \sum_{i=1}^n \int_0^{\tau} ( D_i - \bar D(t,\theta_0, \BS{\beta}_0))^{2} \tilde w_i(t, \theta_0, \BS{\beta}_0) \lambda_0(t) dt \bigg] \nonumber\\
	& \quad + \BS{\mu}_{a0}^{\top} \mathbb{E} \bigg[ \frac{1}{n} \sum_{i=1}^n \int_0^{\tau}  ( D_i - \bar D(t,\theta_0, \BS{\beta}_0)) (\BS{Z}_i - \bar {\BS{Z}} (t,\theta_0, \BS{\beta}_0)) \tilde w_i(t, \theta_0, \BS{\beta}_0) \lambda_0(t) dt \bigg] \nonumber \\
    	&= - \mathbb{E} \Big\{\frac{1}{n} \sum_{i=1}^n  \int_0^{\tau} (D_i- \bar D(t,\theta_0, \BS{\beta}_0))^{2} dN_i(t) \Big\} \nonumber \\
	&\quad + \BS{\mu}_{a0}^{\top} \mathbb{E} \Big\{\frac{1}{n} \sum_{i=1}^n \int_0^{\tau}  ( D_i - \bar D(t,\theta_0, \BS{\beta}_0)) (\BS{Z}_i - \bar {\BS{Z}} (t,\theta_0, \BS{\beta}_0)) dN_i(t) \Big\}  \nonumber\\
	&= -{J}_{\theta \theta} + \BS{\mu}_{a0}^{\top} \BS{J}_{\BS{\beta} \theta }.
\end{align}

Denote by ${J}_{\theta \theta} -{ \Sigma}_{\theta \theta} := {E}_{\theta \theta}$. Using the same notation as in the proof of Proposition \ref{lmnearortho} (Cf. Section \ref{SecC.5}), we can show that
\begin{align*}
	-J^* &= { \Sigma}_{\theta \theta} + {E}_{\theta \theta}- (\BS{ \Sigma}_{ \BS{\beta}\theta}+\BS{E}_{ \BS{\beta}\theta})^{\top} \BS{ \Sigma}_{\BS{\beta}\BS{\beta}} ^{-1}\BS{ \Sigma}_{\BS{\beta} \theta}\\
	&\geq   { \Sigma}_{\theta \theta} - (\BS{ \Sigma}_{ \BS{\beta}\theta})^{\top} \BS{ \Sigma}_{\BS{\beta}\BS{\beta}} ^{-1}\BS{ \Sigma}_{\BS{\beta} \theta} -| {E}_{\theta \theta}| - |(\BS{E}_{ \BS{\beta}\theta})^{\top} \BS{ \Sigma}_{\BS{\beta}\BS{\beta}} ^{-1} \BS{ \Sigma}_{\BS{\beta} \theta} | \\
	&\ge   { \Sigma}_{\theta \theta} - (\BS{ \Sigma}_{ \BS{\beta}\theta})^{\top} \BS{ \Sigma}_{\BS{\beta}\BS{\beta}} ^{-1}\BS{ \Sigma}_{\BS{\beta} \theta} -|{E}_{\theta \theta} | - \|\BS{E}_{ \BS{\beta}\theta}\|_{\infty} \|\BS{ \Sigma}_{\BS{\beta}\BS{\beta}} ^{-1}\BS{ \Sigma}_{\BS{\beta} \theta} \|_1,
\end{align*}
and similarly,
\begin{align*}
	- J^*
	&\le    { \Sigma}_{\theta \theta} - (\BS{ \Sigma}_{ \BS{\beta}\theta})^{\top} \BS{ \Sigma}_{\BS{\beta}\BS{\beta}} ^{-1}\BS{ \Sigma}_{\BS{\beta} \theta}  + | {E}_{\theta \theta} | + \|\BS{E}_{ \BS{\beta}\theta}\|_{\infty} \| \BS{ \Sigma}_{\BS{\beta}\BS{\beta}} ^{-1}\BS{ \Sigma}_{\BS{\beta} \theta} \|_1.
\end{align*}
In view of Condition \ref{coninimu}(iv), the quantity $ {\Sigma}_{\theta \theta} - (\BS{ \Sigma}_{ \BS{\beta}\theta})^{\top} \BS{ \Sigma}_{\BS{\beta}\BS{\beta}} ^{-1}\BS{ \Sigma}_{\BS{\beta} \theta}$ is invertible and positive. Thus, there exists a constant $\tilde c >0$ such that $ { \Sigma}_{\theta \theta} - (\BS{ \Sigma}_{ \BS{\beta}\theta})^{\top} \BS{ \Sigma}_{\BS{\beta}\BS{\beta}} ^{-1}\BS{ \Sigma}_{\BS{\beta} \theta} > \tilde c$. On the other hand, Conditions \ref{coninicox}(i) and \ref{coninimu}(vi) entail that $ { \Sigma}_{\theta \theta} - (\BS{ \Sigma}_{ \BS{\beta}\theta})^{\top} \BS{ \Sigma}_{\BS{\beta}\BS{\beta}} ^{-1}\BS{ \Sigma}_{\BS{\beta} \theta} \le 4(1+K_{\BS{Z}}K_{\BS{\mu}})$. 

Recall that we have derived in \eqref{pf-thm1-eq4} that $\| \BS{E}_{\BS{\beta} \theta} \|_{\infty} = O(\sqrt{\log (np) /n} )$. Similarly, we can obtain that $| {E}_{\theta \theta}| \leq O(\sqrt{\log (np) /n})$, and thus 
\[
| {E}_{\theta \theta}| + \|\BS{E}_{ \BS{\beta}\theta}\|_{\infty} \|\BS{ \Sigma}_{\BS{\beta}\BS{\beta}} ^{-1}\BS{ \Sigma}_{\BS{\beta} \theta} \|_1 = O(\sqrt{\log (np) /n}).
\]
As a result, there exist some constants $c_1 < c_2$ such that
\begin{equation} \label{eq: new-bound-J0}
0< c_1 < \tilde c - O(\sqrt{\log (np) /n}) \le - J^* \le 4(1+K_{\BS{Z}}K_{\BS{\mu}}) + O(\sqrt{\log (np) /n}) < c_2.
\end{equation}
Hence, the first part of Lemma \ref{lmidentification} is established.

To prove the second part of Lemma \ref{lmidentification}, observe that for all $\theta \in \mathcal{N}_{\theta_0}^*$,
\begin{align}
	&  \frac{\partial  \mathbb{E}[\Phi (\{\mathcal{O}_i\}_{i =1}^n; \theta, \BS{\eta}_{a0})]}{ \partial \theta}  \nonumber\\
	&= \mathbb{E} \bigg[\frac{\partial  [\Phi (\{\mathcal{O}_i\}_{i =1}^n; \theta, \BS{\eta}_{a0})]}{\partial \theta} \bigg] \nonumber\\
	&= \mathbb{E} \bigg[- \frac{1}{n} \sum_{i=1}^n \int_0^{\tau}  \sum_{j=1}^n
	( D_j - \bar D(t,\theta, \BS{\beta}_0))^{ 2} w_j(t, \theta, \BS{\beta}_0) dN_i(t) \nonumber\\
	& \quad + {\BS{\mu}_{a0}}^{\top} \frac{1}{n} \sum_{i=1}^n \int_0^{\tau}  \sum_{j=1}^n  ( D_j - \bar D(t,\theta, \BS{\beta}_0)) (\BS{Z}_j - \bar {\BS{Z}} (t,\theta, \BS{\beta}_0)) w_j(t, \theta, \BS{\beta}_0)  dN_i(t)  \bigg]. \label{eq-new-Phi-deri-theta}
\end{align}
By the mean value theorem and recalling that $\mathbb{E} [\Phi (\{\mathcal{O}_i\}_{i=1}^n;\theta_0, \BS{\eta}_0)] =0$, it holds that 
\begin{eqnarray*}
	\mathbb{E} [\Phi (\{\mathcal{O}_i\}_{i=1}^n;\theta, \BS{\eta}_{a0})]
	= \frac{\partial \mathbb{E}[\Phi (\{\mathcal{O}_i\}_{i =1}^n; \theta, \BS{\eta}_{a0})]}{\partial \theta} \Big|_{\theta=\tilde \theta} \cdot (\theta-\theta_0),
\end{eqnarray*}
where $\tilde \theta$ lies between $\theta$ and $\theta_0$. Since $ \theta \in \mathcal{N}_{\theta_0}^*$, it follows that $ \tilde\theta \in \mathcal{N}_{\theta_0}^*$ and
\begin{eqnarray*}
	\big| \mathbb{E} [\Phi (\{\mathcal{O}_i\}_{i=1}^n;\theta, \BS{\eta}_{a0})] \big| = \bigg|\frac{\partial \mathbb{E}[\Phi (\{\mathcal{O}_i\}_{i =1}^n; \theta, \BS{\eta}_{a0})]}{\partial \theta} \Big|_{\theta=\tilde \theta} \bigg| \cdot \big| \theta-\theta_0 \big|.
\end{eqnarray*}

The local separation now follows from the continuity. Since $\frac{\partial \mathbb{E}[\Phi (\{\mathcal{O}_i\}_{i =1}^n; \theta, \BS{\eta}_{a0})]}{\partial \theta}$ is continuous with respect to $\theta$ and is equal to $J^*$ when $\theta=\theta_0$, combining this with the fact that $c_1 < - J^* < c_2$, we can obtain that when $n$ is sufficiently large,
\begin{equation} \label{eq-Phi-derivative}
\inf_{\theta \in \mathcal{N}_{\theta_0}^*} \bigg( - \frac{\partial \mathbb{E}[\Phi (\{\mathcal{O}_i\}_{i =1}^n; \theta, \BS{\eta}_{a0})]}{\partial \theta} \bigg) \ge |J^*|/2,
\end{equation}
and thus 
\begin{equation*}
\inf_{\theta \in \mathcal{N}_{\theta_0}^*} \bigg| - \frac{\partial \mathbb{E}[\Phi (\{\mathcal{O}_i\}_{i =1}^n; \theta, \BS{\eta}_{a0})]}{\partial \theta} \bigg| \ge |J^*|/2.
\end{equation*}
Consequently, it holds that 
\[
2 \big| \mathbb{E} [\Phi (\{\mathcal{O}_i\}_{i=1}^n;\theta, \BS{\eta}_{a0})] \big| \ge |J^*(\theta-\theta_0)|, ~\forall~ \theta \in \mathcal{N}_{\theta_0}^*.
\]
This concludes the proof of Lemma \ref{lmidentification}.

\subsection{Lemma \ref{lmentropy1} and its proof} \label{SecB.11}

Lemma \ref{lmentropy1} below controls the empirical error in the derivative of the orthogonal score with respect to $\theta$. This bound is employed in the linearization of the estimating equation for $\check\theta_{1,k}$.

\begin{lemma}\label{lmentropy1}
	Assume that Conditions \ref{coninicox}--\ref{contheta} are satisfied. Then we have that 
	\begin{align*}
		& \big|  \sqrt{m}  \big(J_0 -   \partial_{\theta}  \Phi (\{\mathcal{O}_i\}_{i \in I_k}; \theta_0, \BS{\eta}_{a0})\big)   ( \check{\theta}_{1, k}- \theta_0 ) \big|  =    O_p \big(  s_{\BS{\beta}_0} n^{-1/2}   {\log (np)}   \big).
	\end{align*}
\end{lemma}

\noindent\textit{Proof}. The proof of this lemma will decompose the derivative error into a scalar Hessian fluctuation and a cross-Hessian fluctuation weighted by $\BS{\mu}_{a0}$. The concentration bounds for the Cox Hessian and the boundedness of $\|\BS{\mu}_{a0}\|_1$ then yield the stated order after multiplying by the local rate of $\check\theta_{1,k}-\theta_0$.

Recall that $ J_0 = -\Sigma_{\theta \theta} + \BS{\mu}_{a0}^{\top} \BS{\Sigma}_{\BS{\beta} \theta}$. It follows from the definition of $\Phi(\{\mathcal{O}_i\}_{i \in I_k}; \theta, \BS{\eta} )$ in \eqref{eqortho} that
$$
 \partial_{\theta} \Phi (\{\mathcal{O}_i\}_{i \in I_k}; \theta_0, \BS{\eta}_{a0}) = \ddot l_{\theta \theta} (\{\mathcal{O}_i\}_{i \in I_k}; \theta_0, \BS{\eta}_{a0}) - \bmu_{a0} ^{\top} \ddot{l}_{\BS{\beta} \theta} (\{\mathcal{O}_i\}_{i \in I_k}; \theta_0, \BS{\eta}_{a0}).
$$
Then it holds that 
\begin{equation} \label{eq-new-pf-lemma3-1}
    \begin{split}
        & \big|  \sqrt{m}  \big(J_0 -   \partial_{\theta}  \Phi (\{\mathcal{O}_i\}_{i \in I_k}; \theta_0, \BS{\eta}_{a0})\big)   ( \check{\theta}_{1, k}- \theta_0 ) \big| \\
        & \leq \sqrt{m}\big|   \check{\theta}_{1, k}- \theta_0   \big|    \cdot \big|\big(J_0 -   \partial_{\theta}  \Phi (\{\mathcal{O}_i\}_{i \in I_k}; \theta_0, \BS{\eta}_{a0})\big)   \big| \\
        & \leq \sqrt{m}\big|   \check{\theta}_{1, k}- \theta_0   \big|   \cdot \Big(\big|  \ddot l_{\theta \theta} (\{\mathcal{O}_i\}_{i \in I_k}; \theta_0, \BS{\eta}_{a0}) + \Sigma_{\theta \theta} \big| \\
        &\quad+ \big| \bmu_{a0}^{\top}\big( \ddot l_{\BS{\beta} \theta} (\{\mathcal{O}_i\}_{i \in I_k}; \theta_0, \BS{\eta}_{a0}) + \BS{\Sigma}_{\BS{\beta}  \theta} \big) \big|\Big).
    \end{split}
\end{equation}

In addition, from the proof of Proposition 3 in \cite{yu2021confidence}, we have that 
\begin{equation} \label{eq-new-pf-lemma3-2}
    \big| \ddot{l}_{\theta \theta} + \Sigma_{\theta \theta} \big| = O_p( n^{-1/2} \sqrt{\log (np)} ),
\end{equation}
\begin{equation} \label{eq-new-pf-lemma3-3}
    \| \ddot{l}_{\BS{\beta} \theta} + \BS{\Sigma}_{\BS{\beta} \theta} \|_{\infty} = O_p( n^{-1/2} \sqrt{\log (np)} ),
\end{equation}
and
\begin{equation} \label{eq-new-pf-lemma3-4}
    \| \ddot{l}_{\BS{\beta} \BS{\beta}} + \BS{\Sigma}_{\BS{\beta} \BS{\beta}} \|_{\max} = O_p( n^{-1/2} \sqrt{\log (np)} ).
\end{equation}
Consequently, since $ \|\bmu_{a0}\|_1 \leq K_{\mu} $, we can obtain from \eqref{eq-new-pf-lemma3-1}--\eqref{eq-new-pf-lemma3-3} that for $ \check{\theta}_{1, k} \in \mathcal{N}_{\theta_0}^*$,
\begin{equation*}
    \begin{split}
         & \big|  \sqrt{m}  \big(J_0 -   \partial_{\theta}  \Phi (\{\mathcal{O}_i\}_{i \in I_k}; \theta_0, \BS{\eta}_{a0})\big)   ( \check{\theta}_{1, k}- \theta_0 ) \big| \\
         & \leq O_p( s_{\BS{\beta}_0} \sqrt{\log (np)} ( n^{-1/2} \sqrt{\log(np)} +  K_{\mu}n^{-1/2} \sqrt{\log(np)} ) ) \\
         & = O_p \big(  s_{\BS{\beta}_0} n^{-1/2}   {\log (np)}   \big).
    \end{split}
\end{equation*}
This completes the proof of Lemma \ref{lmentropy1}.

\subsection{Lemma \ref{newlemma1} and its proof} \label{Sec:pf-newLemma1}

Lemma \ref{newlemma1} below bounds the empirical fluctuation of the nuisance derivative of the orthogonal score. This result will be needed to show that estimating $\BS{\eta}_0$ does not create a first-order remainder in the expansion of the score.

\begin{lemma} \label{newlemma1}
Assume that Conditions \ref{coninicox}--\ref{contheta} are satisfied. Then we have that 
    \begin{align*}
        & \big| \sqrt{m} \big( \partial_{\BS{\eta}}   \Phi (\{\mathcal{O}_i\}_{i \in I_k}; \theta_0, \BS{\eta}_{a0}) -  \partial_{\BS{\eta}}  \mathbb{E}\Phi (\{\mathcal{O}_i\}_{i \in I_k}; \theta_0, \BS{\eta}_{a0})\big)^{\top} ( \hat{\BS{\eta}}_k - \BS{\eta}_{a0} )  \big| \\
        &= O_p \Big(  {\|\BS{\Sigma}}^{-1}_{\BS{\beta}\BS{\beta}} \|_{\infty}  s_{{\BS{\mu}}_{a0}} \big(s_{\BS{\beta}_0} n^{-1/2} \log (np)  +  \tau_n \sqrt{\log (np)} \big)  \Big).
    \end{align*}
\end{lemma}

\noindent\textit{Proof}. The proof of this lemma will write the nuisance derivative explicitly and subtract its expectation. The resulting terms are controlled by the Hessian concentration, the near-orthogonality population bounds, and the $\ell_1$ rates for $\hat{\BS{\eta}}_k-\BS{\eta}_{a0}$. A key point is that the multiplier $\sqrt m$ is offset by the high-dimensional rates imposed in the assumptions.

In light of the definition of $\Phi(\{\mathcal{O}_i\}_{i \in I_k}; \theta, \BS{\eta} )$ in \eqref{eqortho}, it holds that
\begin{equation}
    \begin{split}
        \partial_{\BS{\eta}} \Phi (\{\mathcal{O}_i\}_{i \in I_k}; \theta_0, \BS{\eta}_{a0}) = \begin{pmatrix}
         \ddot{l}_{ \BS{\beta} \theta} (\{\mathcal{O}_i\}_{i \in I_k}; \theta_0, \BS{\beta}_0) - \ddot{l}_{ \BS{\beta} \BS{\beta}} ( \{\mathcal{O}_i\}_{i \in I_k}; \theta_0, \BS{\beta}_0) \bmu_{a0} \\
          - \dot{l}_{ \BS{\beta}} (\{\mathcal{O}_i\}_{i \in I_k}; \theta_0, \BS{\beta}_0)
        \end{pmatrix}.
    \end{split}
\end{equation}
Then it follows from \eqref{eq-new-E-Phi} that 
\begin{equation} \label{eq-new-pf-lemma4-0}
    \begin{split}
        & \partial_{\BS{\eta}}   \Phi (\{\mathcal{O}_i\}_{i \in I_k}; \theta_0, \BS{\eta}_{a0}) -  \partial_{\BS{\eta}}  \mathbb{E}\Phi (\{\mathcal{O}_i\}_{i \in I_k}; \theta_0, \BS{\eta}_{a0})  \\
        & = \begin{pmatrix}
         \ddot{l}_{ \BS{\beta} \theta} (\{\mathcal{O}_i\}_{i \in I_k}; \theta_0, \BS{\beta}_0) + \BS{J}_{\BS{\beta} \theta} - (\ddot{l}_{ \BS{\beta} \BS{\beta}} ( \{\mathcal{O}_i\}_{i \in I_k}; \theta_0, \BS{\beta}_0) + \BS{J}_{\BS{\beta} \BS{\beta}}) \bmu_{a0} \\
          - \dot{l}_{ \BS{\beta}} (\{\mathcal{O}_i\}_{i \in I_k}; \theta_0, \BS{\beta}_0)
        \end{pmatrix}.
    \end{split}
\end{equation}
Using the definition of $ \bmu_{a0} = \BS{\Sigma}_{\BS{\beta} \BS{\beta}}^{-1} \BS{\Sigma}_{\BS{\beta} \theta} $, we can deduce that 
\begin{equation} \label{eq-new-pf-lemma4-1}
    \begin{split}
        &  \ddot{l}_{ \BS{\beta} \theta} (\{\mathcal{O}_i\}_{i \in I_k}; \theta_0, \BS{\beta}_0) + \BS{J}_{\BS{\beta} \theta} - (\ddot{l}_{ \BS{\beta} \BS{\beta}} ( \{\mathcal{O}_i\}_{i \in I_k}; \theta_0, \BS{\beta}_0) + \BS{J}_{\BS{\beta} \BS{\beta}}) \bmu_{a0}  \\
        & = \ddot{l}_{ \BS{\beta} \theta} (\{\mathcal{O}_i\}_{i \in I_k}; \theta_0, \BS{\beta}_0)+ \BS{\Sigma}_{\BS{\beta} \theta} - (\ddot{l}_{ \BS{\beta} \BS{\beta}} ( \{\mathcal{O}_i\}_{i \in I_k}; \theta_0, \BS{\beta}_0) + \BS{\Sigma}_{\BS{\beta} \BS{\beta}}) \bmu_{a0} \\
        & \quad + (\BS{J}_{\BS{\beta} \theta} -\BS{\Sigma}_{\BS{\beta} \theta}  ) - (\BS{J}_{\BS{\beta} \BS{\beta}} - \BS{\Sigma}_{\BS{\beta} \BS{\beta}})  \bmu_{a0}.
    \end{split}
\end{equation}

Substituting the bounds in \eqref{eq-new-pf-lemma3-2}--\eqref{eq-new-pf-lemma3-4} and \eqref{pf-thm1-eq4}--\eqref{pf-thm1-eq42} into \eqref{eq-new-pf-lemma4-1} yields that 
\begin{equation} \label{eq-new-pf-lemma4-2}
    \begin{split}
         & \| \ddot{l}_{ \BS{\beta} \theta} (\{\mathcal{O}_i\}_{i \in I_k}; \theta_0, \BS{\beta}_0) + \BS{J}_{\BS{\beta} \theta} - (\ddot{l}_{ \BS{\beta} \BS{\beta}} ( \{\mathcal{O}_i\}_{i \in I_k}; \theta_0, \BS{\beta}_0) + \BS{J}_{\BS{\beta} \BS{\beta}}) \bmu_{a0} \|_{\infty} \\
         & \leq O_p( n^{-1/2} \sqrt{\log (np)}  ).
    \end{split}
\end{equation}
Moreover, it follows from Lemma 1 in \cite{yu2021confidence} that
\begin{equation} \label{eq-new-pf-lemma4-3}
   \| \dot{l}_{ \BS{\beta}} (\{\mathcal{O}_i\}_{i \in I_k}; \theta_0, \BS{\beta}_0) \|_{\infty} = O_p(n^{-1/2} \sqrt{\log (np)}).
\end{equation}
Consequently, combining \eqref{eq-new-pf-lemma4-0}, \eqref{eq-new-pf-lemma4-2}, and \eqref{eq-new-pf-lemma4-3}, we can obtain that 
\begin{equation}
    \begin{split}
        &  \big| \sqrt{m} \big( \partial_{\BS{\eta}}   \Phi (\{\mathcal{O}_i\}_{i \in I_k}; \theta_0, \BS{\eta}_{a0}) -  \partial_{\BS{\eta}}  \mathbb{E}\Phi (\{\mathcal{O}_i\}_{i \in I_k}; \theta_0, \BS{\eta}_{a0})\big)^{\top} ( \hat{\BS{\eta}}_k - \BS{\eta}_{a0} )  \big| \\
         & \leq  \sqrt{m} \big\|\partial_{\BS{\eta}}   \Phi (\{\mathcal{O}_i\}_{i \in I_k}; \theta_0, \BS{\eta}_{a0}) -  \partial_{\BS{\eta}}  \mathbb{E}\Phi (\{\mathcal{O}_i\}_{i \in I_k}; \theta_0, \BS{\eta}_{a0})\big\|_{\infty} \big\| \hat{\BS{\eta}}_k - \BS{\eta}_{a0}\big\|_{1} \\
         & = O_p(\sqrt{\log (np)} {\|\BS{\Sigma}}^{-1}_{\BS{\beta}\BS{\beta}} \|_{\infty}  s_{{\BS{\mu}}_{a0}} (s_{\BS{\beta}_0} \sqrt{\log (np) /n} + \tau_n)  ).
    \end{split}
\end{equation}
This concludes the proof of Lemma \ref{newlemma1}.

\subsection{Lemma \ref{lmderivative} and its proof} \label{SecB.12}

Lemma \ref{lmderivative} below controls the second-order Taylor remainder of the orthogonal score uniformly over the local parameter set. It ensures that the Taylor expansion used for the asymptotic normality has an asymptotically negligible nonlinear remainder. Recall that $\mathcal{T}_N =\{\BS{\eta}= (\BS{\beta}^{\top}, \BS{\mu}_a^{\top})^{\top} \in \mathcal{B} \times \mathcal{U} \subset \mathbb{R}^{2p}: \, \|\BS{\beta}-\BS{\beta}_0\|_1 \le \tilde C s_{\BS{\beta}_0} \sqrt{\log (np)/n}, \, \|\BS{\mu}_a-\BS{\mu}_{a0} \|_1 \le \tilde C {\|\BS{\Sigma}}^{-1}_{\BS{\beta}\BS{\beta}} \|_{\infty} s_{{\BS{\mu}}_{a0}} (s_{\BS{\beta}_0} \sqrt{\log (np) /n} + \tau_n) \}$.

\begin{lemma}\label{lmderivative}
	Assume that Conditions \ref{coninicox}(i) and \ref{coninimu}(v)--(vi) are satisfied. Then we have that for any $k \in \{1, \cdots, K\}$, 
	\begin{align*}
		&  \sup_{r \in (0, 1), \theta \in \mathcal{N}_{\theta_0}^*, \BS{\eta} \in \mathcal{T}_{N}} \bigg| \frac{\partial^2 \mathbb{E}[  \Phi (\{\mathcal{O}_i\}_{i \in I_k};\theta_0+r(\theta-\theta_0), \BS{\eta}_{a0}+r(\BS{\eta}-\BS{\eta}_{a0})) ]}{\partial r^2 }  \bigg| \\
        &\le  O \big( {\|\BS{\Sigma}}^{-1}_{\BS{\beta}\BS{\beta}} \|_{\infty}  s_{{\BS{\mu}}_{a0}} s_{\BS{\beta}_0}^2   \log (np)/n  +  \tau_n  {\|\BS{\Sigma}}^{-1}_{\BS{\beta}\BS{\beta}} \|_{\infty}  s_{{\BS{\mu}}_{a0}} s_{\BS{\beta}_0} \sqrt{\log (np) /n}  \big).
	\end{align*}
\end{lemma}

\noindent\textit{Proof}. The proof of this lemma will differentiate the score twice along a path from $(\theta_0,\BS{\eta}_{a0})$ to $(\theta,\BS{\eta})$. Each term in the derivative is bounded using the maximum deviation $|D_j(\theta-\theta_0)+\BS{Z}_j^\top(\BS{\beta}-\BS{\beta}_0)|$, the normalization of the risk-set weights, and the $\ell_1$-bound on $\BS{\mu}_a-\BS{\mu}_{a0}$. The resulting expression is of smaller order under Condition \ref{coninimu}(v).

For any $r \in (0,1)$, $\theta \in \mathcal{N}_{\theta_0}^*$, and $\BS{\eta}=(\BS{\beta}^{\top}, \BS{\mu}_{a}^{\top})^{\top} \in \mathcal{T}_{N}$, denote by 
$d_j(\theta, \BS{\beta}) =D_j(\theta-\theta_0)+\BS{Z}_j^{\top}(\BS{\beta}-\BS{\beta}_0)$ and $v_j(r, \theta, \BS{\beta}) = w_{k, j}(t, \theta_0+r(\theta-\theta_0),\BS{\beta}_0+r(\BS{\beta}-\BS{\beta}_0))$ for $j \in I_k$ and $k \in \{1, \cdots, K\}$. Then it holds that 
\begin{align*}
	\max_{j \in I_k} \sup_{\theta \in \mathcal{N}_{\theta_0}^*, \BS{\eta} \in \mathcal{T}_{N}} |d_j(\theta, \BS{\beta})|
    & \le   \sup_{\theta \in \mathcal{N}_{\theta_0}^*, \BS{\eta} \in \mathcal{T}_{N}} (|\theta-\theta_0|+ K_{\BS{Z}} \| \BS{\beta}-\BS{\beta}_0 \|_1) \\
	& \le  (2C_0+ K_{\BS{Z}} \tilde C)s_{\BS{\beta}_0} \sqrt{\log (np)/n},
\end{align*}
and by definition,
\[
\sum_{j \in I_k} v_j(r, \theta, \BS{\beta}) =1.
\]

A direct but lengthy calculation gives that 
\begin{align*}
	&  \frac{\partial^2 \Phi (\{\mathcal{O}_i\}_{i \in I_k};\theta_0+r(\theta-\theta_0), \BS{\eta}_{a0}+r(\BS{\eta}-\BS{\eta}_{a0}))}{\partial r^2} \\
	&= - \frac{1}{m} \sum_{i \in I_k} \int_0^{\tau} \sum_{j \in I_k} D_j d_j(\theta, \BS{\beta})^2  v_j(r, \theta, \BS{\beta})  dN_i(t) \\
	& \quad +\frac{2}{m} \sum_{i \in I_k} \int_0^{\tau}  \sum_{j \in I_k} D_j d_j(\theta, \BS{\beta}) v_j(r, \theta, \BS{\beta}) \times \sum_{j=1}^n d_j(\theta, \BS{\beta})  v_j(r, \theta, \BS{\beta})  dN_i(t)  \\
	&\quad +\frac{1}{m} \sum_{i \in I_k} \int_0^{\tau}  \sum_{j \in I_k} D_j  v_j(r, \theta, \BS{\beta}) \times \sum_{j=1}^n d_j(\theta, \BS{\beta})^2  v_j(r, \theta, \BS{\beta})  dN_i(t)  \\
	& \quad-\frac{2}{m} \sum_{i \in I_k} \int_0^{\tau} \sum_{j \in I_k} D_j  v_j(r, \theta, \BS{\beta}) \times \Big\{\sum_{j \in I_k} d_j(\theta, \BS{\beta}) v_j(r, \theta, \BS{\beta}) \Big\}^2 dN_i(t)  \\
	& \quad+(\BS{\mu}_a -\BS{\mu}_{a0})^{\top} \frac{2}{m} \sum_{i \in I_k} \int_0^{\tau}  \sum_{j \in I_k} \BS{Z}_j d_j(\theta, \BS{\beta}) v_j(r, \theta, \BS{\beta})  dN_i(t) \\
	& \quad -(\BS{\mu}_a -\BS{\mu}_{a0})^{\top} \frac{2}{m} \sum_{i \in I_k} \int_0^{\tau} \sum_{j \in I_k} \BS{Z}_j  v_j(r, \theta, \BS{\beta}) \times \sum_{j \in I_k} d_j(\theta, \BS{\beta}) v_j(r, \theta, \BS{\beta})  dN_i(t)  \\
	& \quad +(\BS{\mu}_{a0}+ r(\BS{\mu}_a -\BS{\mu}_{a0}))^{\top} \frac{1}{m} \sum_{i \in I_k} \int_0^{\tau} \sum_{j \in I_k} \BS{Z}_j d_j(\theta, \BS{\beta})^2  v_j(r, \theta, \BS{\beta})  dN_i(t) \\
	& \quad-(\BS{\mu}_{a0}+ r(\BS{\mu}_a -\BS{\mu}_{a0}))^{\top}\frac{2}{m} \sum_{i \in I_k} \int_0^{\tau} \sum_{j \in I_k} \BS{Z}_j d_j(\theta, \BS{\beta})  v_j(r, \theta, \BS{\beta})  \times \sum_{j \in I_k} d_j(\theta, \BS{\beta})  v_j(r, \theta, \BS{\beta})  dN_i(t) \\
	& \quad-(\BS{\mu}_{a0}+ r(\BS{\mu}_a -\BS{\mu}_{a0}))^{\top} \frac{1}{m} \sum_{i \in I_k}\int_0^{\tau}  \sum_{j \in I_k} \BS{Z}_j  v_j(r, \theta, \BS{\beta}) \times \sum_{j \in I_k} d_j(\theta, \BS{\beta})^2 v_j(r, \theta, \BS{\beta})  dN_i(t) \\
	& \quad +(\BS{\mu}_{a0}+ r(\BS{\mu}_a -\BS{\mu}_{a0}))^{\top} \frac{2}{m} \sum_{i \in I_k} \int_0^{\tau} \sum_{j \in I_k} \BS{Z}_j  v_j(r, \theta, \BS{\beta})   \times \Big\{\sum_{j \in I_k} d_j(\theta, \BS{\beta})  v_j(r, \theta, \BS{\beta})  \Big\}^2 dN_i(t) .
\end{align*}

Consequently, it follows from the representation above that 
\begin{align*}
	&  \sup_{r \in (0, 1), \theta \in \mathcal{N}_{\theta_0}^*, \BS{\eta} \in \mathcal{T}_{N}} \bigg| \frac{\partial^2  \Phi (\{\mathcal{O}_i\}_{i \in I_k};\theta_0+r(\theta-\theta_0), \BS{\eta}_{a0}+r(\BS{\eta}-\BS{\eta}_{a0}))}{\partial r ^2}  \bigg| \\
	& \le  6(1+ K_{\BS{Z}} (\|\BS{\mu}_{a0}\|_1 +\| \BS{\mu}_a -\BS{\mu}_{a0}\|_1)) (\max_{j \in I_k} \sup_{\theta \in \mathcal{N}_{\theta_0}^*, \BS{\eta} \in \mathcal{T}_{N}} |d_j(\theta, \BS{\beta})|)^2 \\
	& \quad +4 K_{\BS{Z}} \| \BS{\mu}_a -\BS{\mu}_{a0}\|_1 \max_{j \in I_k} \sup_{\theta \in \mathcal{N}_{\theta_0}^*, \BS{\eta} \in \mathcal{T}_{N}} |d_j(\theta, \BS{\beta})| \\
	& \le   6 \big(1+ K_{\BS{Z}} [ K_{\BS{\mu}}+ {\|\BS{\Sigma}}^{-1}_{\BS{\beta}\BS{\beta}} \|_{\infty}  s_{{\BS{\mu}}_{a0}} (s_{\BS{\beta}_0} \sqrt{\log (np) /n} + \tau_n)  ] \big) (2C_0+ K_{\BS{Z}} \tilde C)^2 s_{\BS{\beta}_0}^2 \log (np)/n \\
	& \quad + 4 K_{\BS{Z}}
	(2C_0+ K_{\BS{Z}} \tilde C) s_{\BS{\beta}_0}\sqrt{\log (np) / n } {\|\BS{\Sigma}}^{-1}_{\BS{\beta}\BS{\beta}} \|_{\infty} s_{{\BS{\mu}}_{a0}} (s_{\BS{\beta}_0} \sqrt{\log (np) /n} + \tau_n) \\
	& = O_p \big( {\|\BS{\Sigma}}^{-1}_{\BS{\beta}\BS{\beta}} \|_{\infty}  s_{{\BS{\mu}}_{a0}} s_{\BS{\beta}_0}^2   \log (np)/n  +  \tau_n  {\|\BS{\Sigma}}^{-1}_{\BS{\beta}\BS{\beta}} \|_{\infty}  s_{{\BS{\mu}}_{a0}} s_{\BS{\beta}_0} \sqrt{\log (np) /n}  \big).
\end{align*}
Taking an expectation gives the same upper bound order for $ \frac{\partial^2 \mathbb{E} [\Phi (\{\mathcal{O}_i\}_{i \in I_k};\theta_0+r(\theta-\theta_0), \BS{\eta}_{a0}+r(\BS{\eta}-\BS{\eta}_{a0})) ]}{\partial r^2}$.  This completes the proof of Lemma \ref{lmderivative}.

\subsection{Lemma \ref{lmentropy0-1} and its proof} \label{SecB.7}

\begin{lemma}\label{lmentropy0-1}
	Assume Conditions \ref{coninicox} -  \ref{contheta} are satisfied. Then we have
	\begin{eqnarray}\label{ieqB}
		& & \sup_{\substack{\theta \in \mathcal{N}_{\theta_0}^* \\  \| \BS{\beta} -\BS{\beta}_0\|_1 \le \tilde C s_{\BS{\beta}_0} \sqrt{\log(np) /n} }} \bigg| \frac{1}{\sqrt{m}} \sum_{i \in I_k} \bigg\{ \int_0^{\tau} \sum_{j \in I_k} D_j (D_j - \bar D_k(t,  \theta, \BS{\beta})) w_{k,j} (t,  \theta,  \BS{\beta}) dN_i(t) \nonumber \\
		& & - \mathbb{E} \Big[ \int_0^{\tau} \sum_{j \in I_k} D_j (D_j - \bar D_k(t, \theta,  \BS{\beta})) w_{k,j} (t,  \theta,  \BS{\beta})  dN_i(t) \Big] \bigg\}  \bigg|   \le O_p(s_{\BS{\beta}_0 }\sqrt{\log (np)}) .
	\end{eqnarray} 
\end{lemma}

\noindent \textit{Proof}. 
 To prove the bound, we decompose the objective into four parts as follows.  Specifically, for each $i \in I_k$, we can decompose the integrand in the inequality  as
\begin{align*}
	&  \int_0^\tau \sum_{j\in I_k} D_j (D_j - \bar D_k(t,\theta,\BS{\beta})) w_{k,j}(t,\theta,\BS{\beta}) dN_i(t)  \\
	&= \int_0^\tau \sum_{j\in I_k} [D_j^2w_{k,j}(t,\theta,\BS{\beta}) - D_j^2w_{k,j}(t,\theta_0,\BS{\beta}_0)] dN_i(t)\\
	&  + \int_0^\tau \Big \{\sum_{j\in I_k} [D_j^2 w_{k,j}(t,\theta_0,\BS{\beta}_0)] - D_i^2 \Big \}dN_i(t) \\
	&  + \int_0^\tau [D_i^2 - \bar D_k(t,\theta_0,\BS{\beta}_0)^2] dN_i(t)
	+ \int_0^\tau [\bar D_k(t,\theta_0,\BS{\beta}_0)^2- \bar D_k(t,\theta,\BS{\beta})^2] dN_i(t) \\
	&= B_1 +B_2 +B_3 + B_4.
\end{align*}
Intuitively, since $w_{k,j}(t,\theta,\BS{\beta})$ is continuous with respect to $\theta$ and $\BS{\beta}$, and we consider the supremum only on a small space around $(\theta_0, \BS{\beta}_0)$, thus both the first term $B_1$ and the last term $B_4$ are small. We borrow a fundamental lasso-type result based on Cox model from  Lemma 1 (i) in \cite{yu2021confidence} to tackle the second term $B_2$. Moreover,  the third term $B_3$ can be addressed by applying Lemma \ref{lmentropy} in Section \ref{SecB.11-new}.  The detailed arguments are provided below.


For any $k \in \{1, \cdots, K\}$ and $j \in I_k$, we have
\begin{align*}
	\tilde w_{j}(t,\theta,\BS{\beta}) - \tilde w_{j}(t,\theta_0,\BS{\beta}_0) &= Y_j(t)(\exp\{D_j\theta + \BS{Z}_j^T\BS{\beta}\}- \exp\{D_j\theta_0 + \BS{Z}_j^T\BS{\beta}_0\})\\
	& =   \tilde w_{j}(t,\theta_0,\BS{\beta}_0) (\exp\{D_j(\theta - \theta_0) + \BS{Z}_j^T(\BS{\beta}-\BS{\beta}_0)\})-1).
\end{align*}
Since we focus on the region for $(\theta,\BS{\beta})$ such that $\theta \in \mathcal{N}_{\theta_0}^*$ and 
$\| \BS{\beta} -\BS{\beta}_0\|_1 \le \tilde C s_{\BS{\beta}_0} \sqrt{\log(np) /n}$, we have 
   \begin{align*}
      |D_j(\theta - \theta_0) + \BS{Z}_j^T(\BS{\beta}-\BS{\beta}_0)| \leq (2C_0 + K_{\BS{Z}}\tilde{C})s_{\BS{\beta}_0} \sqrt{\log(np) /n},
   \end{align*}
which entails 
   \begin{align*}
     \frac{|\tilde w_{j}(t,\theta,\BS{\beta}) - \tilde w_{j}(t,\theta_0,\BS{\beta}_0)|}{\tilde w_{j}(t,\theta_0,\BS{\beta}_0)} = O(s_{\BS{\beta}_0} \sqrt{\log(np) /n}).
   \end{align*}
For each $j \in I_k$, let $\alpha_j=\tilde w_{j}(t,\theta,\BS{\beta})/\tilde w_{j}(t,\theta_0,\BS{\beta}_0)$ (in case that both the numerator and denominator are zero, define $\alpha_j$ as 1). Then we have that
\[
\alpha_j \geq 1 - C_1s_{\BS{\beta}_0} \sqrt{\log(np) /n},
\]
holds for some universal constant $C_1 >0$, 
and the bound is positive for $n$ large enough.

Now we have
\begin{align}
	 & \frac{|w_{k, j}(t,\theta,\BS{\beta}) - w_{k, j}(t,\theta_0,\BS{\beta}_0)|}{w_{k, j}(t,\theta_0,\BS{\beta}_0)} 
	=  \frac{|\sum_{l \in I_k} (\alpha_j-\alpha_l)\tilde  w_l(t,\theta_0,\BS{\beta}_0)|}{\sum_{l \in I_k} \alpha_l\tilde w_l(t,\theta_0,\BS{\beta}_0)} \nonumber\\
	\leq &  \frac{\sum_{l \in I_k} |\alpha_j-\alpha_l|\tilde w_l(t,\theta_0,\BS{\beta}_0)}{\sum_{l \in I_k} \alpha_l\tilde w_l(t,\theta_0,\BS{\beta}_0)} \nonumber\\
	\leq & \frac{\sum_{l \in I_k} 2(\exp( (2C_0 + K_{\BS{Z}}\tilde{C}) s_{\BS{\beta}_0} \sqrt{\log(np) /n})-1)\tilde w_l(t,\theta_0,\BS{\beta}_0)}{\sum_{l \in I_k} (1 - C_1s_{\BS{\beta}_0} \sqrt{\log(np) /n})\tilde w_l(t,\theta_0,\BS{\beta}_0)} \nonumber\\
	= & \frac{2(\exp( (2C_0 + K_{\BS{Z}}\tilde{C}) s_{\BS{\beta}_0} \sqrt{\log(np) /n})-1)}{1 - C_1s_{\BS{\beta}_0} \sqrt{\log(np) /n}} \nonumber\\
	= & O(s_{\BS{\beta}_0}\sqrt{\log(np)/n}) \label{pf-lemma3-new-eq1}.
\end{align}
Thus, for the first term $B_1$, we have
\begin{align*}
	|B_1| &=  \bigg|\int_0^\tau \sum_{j\in I_k} [D_j^2w_{k, j}(t,\theta,\BS{\beta}) - D_j^2w_{k,j}(t,\theta_0,\BS{\beta}_0)] dN_i(t) \bigg |\\
	&\leq  \int_0^\tau \sum_{j\in I_k} [D_j^2w_{k, j}(t,\theta_0,\BS{\beta}_0)\frac{|w_{k, j}(t,\theta,\BS{\beta}) - w_{k, j}(t,\theta_0,\BS{\beta}_0)|}{w_{k, j}(t,\theta_0,\BS{\beta}_0)}] dN_i(t)\\
	&\leq  \frac{2(\exp( (2C_0 + K_{\BS{Z}}\tilde{C})  s_{\BS{\beta}_0} \sqrt{\log(np) /n})-1)}{1 - C_1s_{\BS{\beta}_0} \sqrt{\log(np) /n}}\int_0^\tau \sum_{j\in I_k} [D_j^2w_{k, j}(t,\theta_0,\BS{\beta}_0)] dN_i(t)\\ 
	&\leq  \frac{2(\exp( (2C_0 + K_{\BS{Z}}\tilde{C}) s_{\BS{\beta}_0} \sqrt{\log(np) /n})-1)}{1 - C_1s_{\BS{\beta}_0} \sqrt{\log(np) /n}}\\
	& = O(s_{\BS{\beta}_0}\sqrt{\log(np)/n}).
\end{align*}

Furthermore, the expectation of $B_1$ is also of order $O(s_{\BS{\beta}_0}\sqrt{\log(np)/n})$. Summing over all $i \in I_k$ and dividing by $\sqrt{m}$ gives
\begin{eqnarray}\label{ieqB1}
& & \sup_{\substack{\theta \in \mathcal{N}_{\theta_0}^* \\  \| \BS{\beta} -\BS{\beta}_0\|_1 \le \tilde C s_{\BS{\beta}_0} \sqrt{\log(np) /n}} } \bigg| \frac{1}{\sqrt{m}} \sum_{i \in I_k} \bigg\{ \int_0^{\tau} \sum_{j\in I_k} [D_j^2w_{k, j}(t,\theta,\BS{\beta}) - D_j^2w_{k,j}(t,\theta_0,\BS{\beta}_0)]  dN_i(t)  \nonumber \\
& & - \mathbb{E} \Big[ \int_0^{\tau} \sum_{j\in I_k} [D_j^2w_{k, j}(t,\theta,\BS{\beta}) - D_j^2w_{k,j}(t,\theta_0,\BS{\beta}_0)]   dN_i(t) \Big] \bigg\}  \bigg|   \le O(s_{\BS{\beta}_0}\sqrt{\log(np) }).
\end{eqnarray} 

Next, we deal with $B_2$. Note that $\int_0^\tau \big \{\sum_{j\in I_k} [D_j^2 w_{k, j}(t,\theta_0,\BS{\beta}_0)] - D_i^2 \big \}dN_i(t)$ is actually $\int_0^\tau [\bar D_k(t,\theta_0,\BS{\beta}_0) - D_i ]dN_i(t)$ since $D_j^2 = D_j$.
We can conclude from Lemma 1 (i) in \cite{yu2021confidence} that 
\[
\bigg|\frac{1}{m}\sum_{i \in I_k} \int_0^\tau(D_i-\bar D_k(t,\theta_0,\BS{\beta}_0))dN_i(t)\bigg| = O_p(\sqrt{\log(np)/n}).
\]
Thus summing the second term $B_2$ over all $i \in I_k$ and dividing by $\sqrt{m}$ yields
\[
\bigg| \frac{1}{\sqrt{m}}\sum_{i \in I_k} \int_0^\tau(\bar D_k(t,\theta_0,\BS{\beta}_0)-D_i)dN_i(t) \bigg| = O_p(\sqrt{\log(np)}).
\]
 The expectation of $B_2$ is also of order $ O(\sqrt{\log(np)})$. Therefore, we have
\begin{eqnarray}\label{ieqB2}
& & \sup_{\substack{\theta \in \mathcal{N}_{\theta_0}^* \\  \| \BS{\beta} -\BS{\beta}_0\|_1 \le \tilde C s_{\BS{\beta}_0} \sqrt{\log(np) /n}} } \bigg| \frac{1}{\sqrt{m}} \sum_{i \in I_k} \bigg\{  \int_0^\tau \big \{\sum_{j\in I_k} [D_j^2 w_{k,j}(t,\theta_0,\BS{\beta}_0)] - D_i^2 \big \}dN_i(t)  \nonumber \\
& & - \mathbb{E} \Big[ \int_0^{\tau} \big \{\sum_{j\in I_k} [D_j^2 w_{k,j}(t,\theta_0,\BS{\beta}_0)] - D_i^2 \big \} dN_i(t) \Big] \bigg\}  \bigg|   \le O_p(\sqrt{\log(np)}).
\end{eqnarray} 


For the last term $B_{4}$, since
\begin{align*}
	|B_{4}| &= \bigg| \int_0^\tau (\bar D_k(t,\theta_0,\BS{\beta}_0) + \bar D_k(t,\theta,\BS{\beta}))(\bar D_k(t,\theta_0,\BS{\beta}_0) - \bar D_k(t,\theta,\BS{\beta})) dN_i(t) \bigg| \\
	& \le  \int_0^\tau \big|(\bar D_k(t,\theta_0,\BS{\beta}_0) + \bar D_k(t,\theta,\BS{\beta}))(\bar D_k(t,\theta_0,\BS{\beta}_0) - \bar D_k(t,\theta,\BS{\beta}))\big| dN_i(t) \\
	& \le  2 \int_0^\tau \big| (\bar D_k(t,\theta_0,\BS{\beta}_0) - \bar D_k(t,\theta,\BS{\beta}))\big| dN_i(t) \\
	& \le  2  \int_0^\tau \sum_{j\in I_k} [D_j w_{k, j}(t,\theta_0,\BS{\beta}_0)\frac{|w_{k, j}(t,\theta,\BS{\beta}) - w_{k, j}(t,\theta_0,\BS{\beta}_0)|}{w_{k, j}(t,\theta_0,\BS{\beta}_0)}] dN_i(t),
\end{align*}
we can use the bound of $\frac{|w_{k, j}(t,\theta,\BS{\beta}) - w_{k, j}(t,\theta_0,\BS{\beta}_0)|}{w_{k, j}(t,\theta_0,\BS{\beta}_0)}$ in \eqref{pf-lemma3-new-eq1} again to yield that $$|B_{4}| = O_p(s_{\BS{\beta}_0}\sqrt{\log(np)/n}).$$
Similar to previous discussion, we can obtain that
\begin{eqnarray}\label{ieqB3}
& & \sup_{\substack{\theta \in \mathcal{N}_{\theta_0}^* \\  \| \BS{\beta} -\BS{\beta}_0\|_1 \le \tilde C s_{\BS{\beta}_0} \sqrt{\log(np) /n}} } \bigg| \frac{1}{\sqrt{m}} \sum_{i \in I_k} \bigg\{  \int_0^\tau [\bar D_k(t,\theta_0,\BS{\beta}_0)^2- \bar D_k(t,\theta,\BS{\beta})^2] dN_i(t)  \nonumber \\
& & - \mathbb{E} \Big[ \int_0^{\tau} [\bar D_k(t,\theta_0,\BS{\beta}_0)^2- \bar D_k(t,\theta,\BS{\beta})^2] dN_i(t) \Big] \bigg\}  \bigg|   \le O_p(s_{\BS{\beta}_0}\sqrt{\log(np)}).
\end{eqnarray} 
%

For the third term $B_{3}= \int_0^\tau [D_i^2 - \bar D_k(t,\theta_0,\BS{\beta}_0)^2] dN_i(t)$, it follows from \eqref{ieqB4} and \eqref{ieqB5} in Lemma \ref{lmentropy} that 

\begin{equation} \label{eq: new-B3-0}
\begin{split}
    & \sup_{\substack{\theta \in \mathcal{N}_{\theta_0}^* \\  \| \BS{\beta} -\BS{\beta}_0\|_1 \le \tilde C s_{\BS{\beta}_0} \sqrt{\log(np) /n}} } \bigg| \frac{1}{\sqrt{m}} \sum_{i \in I_k} \bigg\{  \int_0^\tau  (D_i^2 - \bar D_k(t,\theta_0,\BS{\beta}_0)^2)  dN_i(t)  \\
    & \quad - \mathbb{E}  \int_0^\tau  (D_i^2 - \bar D_k(t,\theta_0,\BS{\beta}_0)^2)  dN_i(t) \bigg\} \bigg| \leq O_p(\sqrt{\log m }) = O_p(\sqrt{\log n }).
\end{split}    
\end{equation}

 Therefore, combining the results (\ref{ieqB1})--(\ref{eq: new-B3-0}), we can finally prove the desired result (\ref{ieqB}). This completes the proof of Lemma \ref{lmentropy0-1}.

\subsection{Lemma \ref{lmentropy0-2} and its proof} \label{SecC.8}

\begin{lemma}\label{lmentropy0-2}
	Assume Conditions \ref{coninicox} -  \ref{contheta} are satisfied. Then we have
	\begin{eqnarray}\label{ieqC}
		& & \sup_{\substack{\theta \in \mathcal{N}_{\theta_0^*} \\  \| \BS{\beta} -\BS{\beta}_0\|_1 \le \tilde C s_{\BS{\beta}_0} \sqrt{\log(np) /n}} }  \bigg\| \frac{1}{\sqrt{m}} \sum_{i \in I_k} \bigg\{ \int_0^{\tau}  \sum_{j \in I_k} D_j (\BS{Z}_j - \bar {\BS{Z}}_k(t, \theta,  {\BS{\beta}})) w_{k,j} (t, \theta, {\BS{\beta}}) dN_i(t) \nonumber \\
		& & - \mathbb{E} \Big[\int_0^{\tau} \sum_{j \in I_k} D_j  (\BS{Z}_j - \bar {\BS{Z}}_k(t, \theta, {\BS{\beta}})) w_{k,j} (t, \theta, {\BS{\beta}})  dN_i(t) \Big] \bigg\} \bigg\|_{\infty}   \le O_p(s_{\BS{\beta}_0}\sqrt{\log(np) }).
	\end{eqnarray} 
\end{lemma}

\noindent \textit{Proof}. This proof is similar to the proof of Lemma \ref{lmentropy0-1} in Section \ref{SecB.7}.  Specifically, for each $i \in I_k$, we decompose the integrand in the inequality  as
\begin{align*}
	&  \int_0^{\tau}  \sum_{j \in I_k} D_j (\BS{Z}_j - \bar {\BS{Z}}_k(t, \theta,  {\BS{\beta}})) w_{k,j} (t, \theta, {\BS{\beta}}) dN_i(t)  \\
	&=  \int_0^\tau \sum_{j\in I_k} [D_j \BS{Z}_j  w_{k,j}(t,\theta,\BS{\beta}) - D_j  \BS{Z}_j  w_{k,j}(t,\theta_0,\BS{\beta}_0)] dN_i(t)\\
	&  + \int_0^\tau \bigg \{\sum_{j\in I_k} [D_j \BS{Z}_j  w_{k,j}(t,\theta_0,\BS{\beta}_0)] - D_i \BS{Z}_i  \bigg \}dN_i(t) \\
	&   +  \int_0^\tau [D_i \BS{Z}_i  - \bar D_k(t,\theta_0,\BS{\beta}_0) \bar {\BS{Z}}_k(t, \theta_0,  {\BS{\beta}_0}) ] dN_i(t) \\
	&   +   \int_0^\tau [\bar D_k(t,\theta_0,\BS{\beta}_0) \bar {\BS{Z}}_k(t, \theta_0,  {\BS{\beta}_0})  - \bar D_k(t,\theta,\BS{\beta}) \bar {\BS{Z}}_k(t, \theta,  {\BS{\beta}}) ] dN_i(t) \\
	&= C_1 +C_2 +C_3+C_4.
\end{align*} 

For the first term $C_1$, we have
\begin{align*}
	\|C_1 \|_{\infty} &=  \bigg\| \int_0^\tau \sum_{j\in I_k} [D_j \BS{Z}_j  w_{k,j}(t,\theta,\BS{\beta}) - D_j  \BS{Z}_j  w_{k,j}(t,\theta_0,\BS{\beta}_0)] dN_i(t)\bigg\|_{\infty}  \\
	&\leq  \int_0^\tau \sum_{j\in I_k} \| D_j \BS{Z}_j\|_{\infty} w_{k, j}(t,\theta_0,\BS{\beta}_0) \frac{|w_{k, j}(t,\theta,\BS{\beta}) - w_{k, j}(t,\theta_0,\BS{\beta}_0)|}{w_{k, j}(t,\theta_0,\BS{\beta}_0)} dN_i(t)\\
	&\leq \frac{2(\exp( (2C_0 + K_{\BS{Z}}\tilde{C})  s_{\BS{\beta}_0} \sqrt{\log(np) /n})-1)}{1 - C_1s_{\BS{\beta}_0} \sqrt{\log(np) /n}}\int_0^\tau \sum_{j\in I_k} [ \| D_j \BS{Z}_j\|_{\infty} w_{k,j}(t,\theta_0,\BS{\beta}_0)] dN_i(t)\\ 
	&\leq  K_{\BS{Z}} \frac{2(\exp( (2C_0 + K_{\BS{Z}}\tilde{C}) s_{\BS{\beta}_0} \sqrt{\log(np) /n})-1)}{1 - C_1s_{\BS{\beta}_0} \sqrt{\log(np) /n}}\\
	& = O(s_{\BS{\beta}_0}\sqrt{\log(np)/n}), 
\end{align*}
where the second inequality follows from the bound of $\frac{|w_{k, j}(t,\theta,\BS{\beta}) - w_{k, j}(t,\theta_0,\BS{\beta}_0)|}{w_{k, j}(t,\theta_0,\BS{\beta}_0)}$ in the proof of Lemma \ref{lmentropy0-1}, and the third inequality follows from Condition \ref{coninicox} (i).
Thus, the expectation of each component of $C_1$ is also of order $O(s_{\BS{\beta}_0}\sqrt{\log(np)/n})$. Summing over all $i \in I_k$ and divided by $\sqrt{m}$ gives
\begin{align}\label{ieqC1}
& \sup_{\substack{\theta \in \mathcal{N}_{\theta_0}^* \\  \| \BS{\beta} -\BS{\beta}_0\|_1 \le \tilde C s_{\BS{\beta}_0} \sqrt{\log(np) /n}} } \bigg\| \frac{1}{\sqrt{m}} \sum_{i \in I_k} \bigg\{ \int_0^\tau \sum_{j\in I_k} [D_j \BS{Z}_j  w_{k,j}(t,\theta,\BS{\beta}) - D_j  \BS{Z}_j  w_{k,j}(t,\theta_0,\BS{\beta}_0)] dN_i(t)  \nonumber \\
& - \mathbb{E} \Big[  \int_0^\tau \sum_{j\in I_k} [D_j \BS{Z}_j  w_{k,j}(t,\theta,\BS{\beta}) - D_j  \BS{Z}_j  w_{k,j}(t,\theta_0,\BS{\beta}_0)] dN_i(t) \Big] \bigg\}  \bigg\|_{\infty}   \le O(s_{\BS{\beta}_0}\sqrt{\log(np) }).
\end{align} 

For the second term $C_2$,  let us define
\begin{eqnarray*}
	\BS{a}_i(t) = D_i\BS{Z}_i - \sum_{j\in I_k} [D_j \BS{Z}_j  w_{k,j}(t,\theta_0,\BS{\beta}_0)]
\end{eqnarray*}
for all $i\in I_k$. Then we have 
\begin{align*}
	&\sum_{i\in I_k} \BS{a}_i(t)w_{k,i}(t,\theta_0,\BS{\beta}_0) \\
	&= \sum_{i\in I_k}\{D_i\BS{Z}_iw_{k,i}(t,\theta_0,\BS{\beta}_0) - \sum_{j\in I_k} [D_j \BS{Z}_j  w_{k,j}(t,\theta_0,\BS{\beta}_0)]w_{k,i}(t,\theta_0,\BS{\beta}_0)\} \\
	& = \sum_{i\in I_k}[D_i\BS{Z}_iw_{k,i}(t,\theta_0,\BS{\beta}_0)] - \sum_{j\in I_k} [D_j \BS{Z}_j  w_{k,j}(t,\theta_0,\BS{\beta}_0)] \sum_{i\in I_k}w_{k,i}(t,\theta_0,\BS{\beta}_0) \\
	& =  0.
\end{align*}
Thus, it follows from  Lemma \ref{concentrationformg} that 
\begin{eqnarray*}
	\frac{1}{\sqrt{m}}\bigg \| \sum_{i\in I_k}\int_0^\tau \BS{a}_i(t)dN_i(t)- \mathbb{E}\sum_{i\in I_k}\int_0^\tau \BS{a}_i(t)dN_i(t)\bigg \|_\infty= O_p(\sqrt{\log(np)}).
\end{eqnarray*}
This implies 
\begin{eqnarray}\label{ieqC2}
& & \sup_{\substack{\theta \in \mathcal{N}_{\theta_0}^* \\  \| \BS{\beta} -\BS{\beta}_0\|_1 \le \tilde C s_{\BS{\beta}_0} \sqrt{\log(np) /n}}} \bigg\| \frac{1}{\sqrt{m}} \sum_{i \in I_k} \bigg\{ \int_0^\tau \big \{\sum_{j\in I_k} [D_j \BS{Z}_j  w_{k,j}(t,\theta_0,\BS{\beta}_0)] - D_i \BS{Z}_i  \big \}dN_i(t)  \nonumber \\
& & - \mathbb{E} \Big[  \int_0^\tau \big \{\sum_{j\in I_k} [D_j \BS{Z}_j  w_{k,j}(t,\theta_0,\BS{\beta}_0)] - D_i \BS{Z}_i  \big \} dN_i(t) \Big] \bigg\}  \bigg\|_{\infty}   \le O_p(\sqrt{\log(np)}).
\end{eqnarray} 

For the third term $C_3$, it holds that 
\begin{align*}
	& \frac{1}{\sqrt{m}}\bigg\| \sum_{i\in I_k} \int_0^\tau [D_i \BS{Z}_i  - \bar D_k(t,\theta_0,\BS{\beta}_0) \bar {\BS{Z}}_k(t, \theta_0,  {\BS{\beta}_0}) ] dN_i(t)\\
	& - \sum_{i\in I_k} \mathbb E \int_0^\tau [D_i \BS{Z}_i  - \bar D_k(t,\theta_0,\BS{\beta}_0) \bar {\BS{Z}}_k(t, \theta_0,  {\BS{\beta}_0}) ] dN_i(t)\bigg\|_\infty\\
	&\leq \frac{1}{\sqrt{m}}\bigg\| \sum_{i\in I_k} \Big\{\int_0^\tau D_i \BS{Z}_i dN_i(t) - \mathbb E\int_0^\tau D_i \BS{Z}_i dN_i(t) \Big\}\bigg\|_\infty\\
	&+  \frac{1}{\sqrt{m}}\bigg\| \sum_{i\in I_k} \Big\{\int_0^\tau \bar D_k(t,\theta_0,\BS{\beta}_0) \bar {\BS{Z}}_k(t, \theta_0,  {\BS{\beta}_0}) dN_i(t) \\
	& - \mathbb E\int_0^\tau \bar D_k(t,\theta_0,\BS{\beta}_0) \bar {\BS{Z}}_k(t, \theta_0,  {\BS{\beta}_0})dN_i(t) \Big\}\bigg\|_\infty 
\end{align*} 
According to  the results \eqref{eq-neww-Gm-DZ} and \eqref{eq-neww-G-1} in Lemma \ref{lmentropy}, we have 
\begin{eqnarray}\label{ieqC3}
& & \sup_{\substack{\theta \in \mathcal{N}_{\theta_0}^* \\  \| \BS{\beta} -\BS{\beta}_0\|_1 \le \tilde C s_{\BS{\beta}_0} \sqrt{\log(np) /n}} } \bigg\| \frac{1}{\sqrt{m}} \sum_{i \in I_k} \bigg\{ \int_0^\tau [D_i \BS{Z}_i  - \bar D_k(t,\theta_0,\BS{\beta}_0) \bar {\BS{Z}}_k(t, \theta_0,  {\BS{\beta}_0}) ] dN_i(t)  \nonumber \\
& & - \mathbb{E} \Big[  \int_0^\tau [D_i \BS{Z}_i  - \bar D_k(t,\theta_0,\BS{\beta}_0) \bar {\BS{Z}}_k(t, \theta_0,  {\BS{\beta}_0}) ] dN_i(t) \Big] \bigg\}  \bigg\|_{\infty}   \le O_p(\sqrt{\log(np)}) .
\end{eqnarray}

For $C_4$, 
we may apply the same technique used in the proof of Lemma \ref{lmentropy0-1}. To be more specific, we have
\begin{align*}
	&\bigg\| \int_0^\tau [\bar D_k(t,\theta_0,\BS{\beta}_0) \bar {\BS{Z}}_k(t, \theta_0,  {\BS{\beta}_0})  - \bar D_k(t,\theta,\BS{\beta}) \bar {\BS{Z}}_k(t, \theta,  {\BS{\beta}}) ] dN_i(t)\bigg\|_\infty \\
	&\leq \frac{1}{2}\bigg\|\int_0^\tau (\bar D_k(t,\theta_0,\BS{\beta}_0) - \bar D_k(t,\theta,\BS{\beta}))(\bar {\BS{Z}}_k(t, \theta_0,  {\BS{\beta}_0})+\bar {\BS{Z}}_k(t, \theta,  {\BS{\beta}}) ) dN_i(t)\bigg\|_\infty\\
	& + \frac{1}{2}\bigg\|\int_0^\tau (\bar D_k(t,\theta_0,\BS{\beta}_0) + \bar D_k(t,\theta,\BS{\beta}))(\bar {\BS{Z}}_k(t, \theta_0,  {\BS{\beta}_0})-\bar {\BS{Z}}_k(t, \theta,  {\BS{\beta}}) ) dN_i(t) \bigg\|_\infty\\
	& \leq   K_Z\int_0^\tau |\bar D_k(t,\theta_0,\BS{\beta}_0) - \bar D_k(t,\theta,\BS{\beta})| dN_i(t)  + \int_0^\tau \|\bar {\BS{Z}}_k(t, \theta_0,  {\BS{\beta}_0})-\bar {\BS{Z}}_k(t, \theta,  {\BS{\beta}})\|_\infty dN_i(t) \\
	& =  K_Z\int_0^\tau  \sum_{j\in I_k} [D_j w_{k, j}(t,\theta_0,\BS{\beta}_0)\frac{|w_{k, j}(t,\theta,\BS{\beta}) - w_{k, j}(t,\theta_0,\BS{\beta}_0)|}{w_{k, j}(t,\theta_0,\BS{\beta}_0)}] dN_i(t) \\
	&+ \int_0^\tau  \sum_{j\in I_k} [\|Z_j\|_\infty w_{k, j}(t,\theta_0,\BS{\beta}_0)\frac{|w_{k, j}(t,\theta,\BS{\beta}) - w_{k, j}(t,\theta_0,\BS{\beta}_0)|}{w_{k, j}(t,\theta_0,\BS{\beta}_0)}]dN_i(t).
\end{align*}

The upper bound of $\frac{|w_{k, j}(t,\theta,\BS{\beta}) - w_{k, j}(t,\theta_0,\BS{\beta}_0)|}{w_{k, j}(t,\theta_0,\BS{\beta}_0)}$ established in \eqref{pf-lemma3-new-eq1} yields that $\|C_{4}\|_\infty = O(s_{\BS{\beta}_0}\sqrt{\log(np)/n}) $. Hence, for the empirical-process term associated with $C_4$, we have    
\begin{align}\label{ieqC4}
& \sup_{\substack{\theta \in \mathcal{N}_{\theta_0}^* \\  
\| \BS{\beta} -\BS{\beta}_0\|_1 \le \tilde C s_{\BS{\beta}_0} \sqrt{\log(np) /n}} } \bigg| \frac{1}{\sqrt{m}} \sum_{i \in I_k} \bigg\{  \int_0^\tau [\bar D_k(t,\theta_0,\BS{\beta}_0) \bar {\BS{Z}}_k(t, \theta_0,  {\BS{\beta}_0})  - \bar D_k(t,\theta,\BS{\beta}) \bar {\BS{Z}}_k(t, \theta,  {\BS{\beta}})] dN_i(t)  \nonumber \\
& - \mathbb{E} \Big[ \int_0^{\tau} [\bar D_k(t,\theta_0,\BS{\beta}_0) \bar {\BS{Z}}_k(t, \theta_0,  {\BS{\beta}_0})  - \bar D_k(t,\theta,\BS{\beta}) \bar {\BS{Z}}_k(t, \theta,  {\BS{\beta}})] dN_i(t) \Big] \bigg\}  \bigg|   \le O_p(s_{\BS{\beta}_0}\sqrt{\log(np) }) .
\end{align} 
 Finally, we can  obtain (\ref{ieqC}) by combining the results in (\ref{ieqC1})--(\ref{ieqC4}). This concludes the proof of Lemma \ref{lmentropy0-2}.

\subsection{Lemma \ref{lmentropy0-3} and its proof} \label{SecB.9}

\begin{lemma}\label{lmentropy0-3}
	Assume Conditions \ref{coninicox} -  \ref{contheta} are satisfied. Then we have
	\begin{align}\label{ieqG}
		&  \sup_{\substack{\theta \in \mathcal{N}_{\theta_0^*} \\  \| \BS{\beta} -\BS{\beta}_0\|_1 \le \tilde C s_{\BS{\beta}_0} \sqrt{\log(np) /n} }} \bigg\| \frac{1}{\sqrt{m}} \sum_{i \in I_k} \bigg\{ \int_0^{\tau}  \sum_{j \in I_k} \BS{Z}_j   (\BS{Z}_j - \bar {\BS{Z}}_k(t,  \theta, {\BS{\beta}}))^{\top} w_{k,j} (t,  \theta, {\BS{\beta}})  dN_i(t) \nonumber \\
		& - \mathbb{E} \Big[ \int_0^{\tau}  \sum_{j \in I_k} \BS{Z}_j (\BS{Z}_j - \bar {\BS{Z}}_k(t,  \theta,  {\BS{\beta}}))^{\top} w_{k,j} (t,  \theta,  {\BS{\beta}})  dN_i(t) \Big] \bigg\} \bigg\|_{\max}   \le O_p(s_{\BS{\beta}_0}\sqrt{\log(np) }).
	\end{align} 
\end{lemma}

\noindent \textit{Proof}. This proof is also similar to the proof of Lemma \ref{lmentropy0-1} in Section \ref{SecB.7} and the proof of Lemma \ref{lmentropy0-2} in Section \ref{SecC.8}.  For each $i \in I_k$, we decompose the integrand in the inequality  as
\begin{align*}
	&  \int_0^{\tau}  \sum_{j \in I_k} \BS{Z}_j   (\BS{Z}_j - \bar {\BS{Z}}_k(t,  \theta, {\BS{\beta}}))^{\top} w_{k,j} (t,  \theta, {\BS{\beta}})  dN_i(t) \\
	&=  \int_0^\tau \sum_{j\in I_k} [\BS{Z}_j^{\otimes 2} w_{k,j}(t,\theta,\BS{\beta}) - \BS{Z}_j^{\otimes 2} w_{k,j}(t,\theta_0,\BS{\beta}_0)] dN_i(t)\\
	&  + \int_0^\tau \big \{\sum_{j\in I_k} [\BS{Z}_j^{\otimes 2}  w_{k,j}(t,\theta_0,\BS{\beta}_0)] -\BS{Z}_i^{\otimes 2} \big \}dN_i(t) \\
	&  +  \int_0^\tau [\BS{Z}_i^{\otimes 2}  - \bar {\BS{Z}}_k(t,\theta_0,\BS{\beta}_0)^{\otimes 2} ] dN_i(t) \\
	&  +  \int_0^\tau [\bar {\BS{Z}}_k(t,\theta_0,\BS{\beta}_0)^{\otimes 2}  - \bar {\BS{Z}}_k(t,\theta,\BS{\beta})^{\otimes 2} ] dN_i(t) \\
	&= G_1 +G_2 +G_3 +G_4.
\end{align*} 

For the first term $G_1$, we have
\begin{align*}
	\|G_1 \|_{\max} = & \bigg\|  \int_0^\tau \sum_{j\in I_k} [\BS{Z}_j^{\otimes 2} w_{k,j}(t,\theta,\BS{\beta}) - \BS{Z}_j^{\otimes 2} w_{k,j}(t,\theta_0,\BS{\beta}_0)] dN_i(t) \bigg\|_{\max}  \\
	\leq & \int_0^\tau \sum_{j\in I_k} \|  \BS{Z}_j^{\otimes 2}  \|_{\max} w_{k, j}(t,\theta_0,\BS{\beta}_0) \frac{|w_{k, j}(t,\theta,\BS{\beta}) - w_{k, j}(t,\theta_0,\BS{\beta}_0)|}{w_{k, j}(t,\theta_0,\BS{\beta}_0)}] dN_i(t)\\
	\leq & \frac{2(\exp( (2C_0 + K_{\BS{Z}}\tilde{C})  s_{\BS{\beta}_0} \sqrt{\log(np) /n})-1)}{1 - C_1s_{\BS{\beta}_0} \sqrt{\log(np) /n}}\int_0^\tau \sum_{j\in I_k} [ \|  \BS{Z}_j^{\otimes 2} \|_{\max} w_{k, j}(t,\theta_0,\BS{\beta}_0)] dN_i(t)\\ 
	\leq & K_{\BS{Z}}^2 \frac{2(\exp( (2C_0 + K_{\BS{Z}}\tilde{C}) s_{\BS{\beta}_0} \sqrt{\log(np) /n})-1)}{1 - C_1s_{\BS{\beta}_0} \sqrt{\log(np) /n}}\\
	 = & O_p(s_{\BS{\beta}_0}\sqrt{\log(np)/n}).
\end{align*}
Thus, the expectation of each component of $G_1$ is also of order $O(s_{\BS{\beta}_0}\sqrt{\log(np)/n})$. Summing over all $i \in I_k$ and divided by $\sqrt{m}$ derives
\begin{align}\label{ieqG1}
&  \sup_{\substack{\theta \in \mathcal{N}_{\theta_0}^* \\  \| \BS{\beta} -\BS{\beta}_0\|_1 \le \tilde C s_{\BS{\beta}_0} \sqrt{\log(np) /n}} } \bigg\| \frac{1}{\sqrt{m}} \sum_{i \in I_k} \bigg\{  \int_0^\tau \sum_{j\in I_k} [\BS{Z}_j^{\otimes 2} w_{k,j}(t,\theta,\BS{\beta}) - \BS{Z}_j^{\otimes 2} w_{k,j}(t,\theta_0,\BS{\beta}_0)] dN_i(t)  \nonumber \\
&  - \mathbb{E} \Big[   \int_0^\tau \sum_{j\in I_k} [\BS{Z}_j^{\otimes 2} w_{k,j}(t,\theta,\BS{\beta}) - \BS{Z}_j^{\otimes 2} w_{k,j}(t,\theta_0,\BS{\beta}_0)] dN_i(t) \Big] \bigg\}  \bigg\|_{\max}   \le O_p(s_{\BS{\beta}_0}\sqrt{\log(np) }).
\end{align} 



For the second term $G_2$, let us define
\begin{eqnarray*}
	\BS{a}_i(t) = \BS{Z}_i^{\otimes 2} - \sum_{j\in I_k} [\BS{Z}_j^{\otimes 2}   w_{k,j}(t,\theta_0,\BS{\beta}_0)]
\end{eqnarray*}
for all $i\in I_k$. Observe that 
\begin{align*}
	&\sum_{i\in I_k} \BS{a}_i(t)w_{k,i}(t,\theta_0,\BS{\beta}_0) \\
	&= \sum_{i\in I_k} [\BS{Z}_i^{\otimes 2}w_{k,i}(t,\theta_0,\BS{\beta}_0) - \sum_{j\in I_k} \BS{Z}_j^{\otimes 2}  w_{k,j}(t,\theta_0,\BS{\beta}_0)]w_{k,i}(t,\theta_0,\BS{\beta}_0) \\
	& = \sum_{i\in I_k}[\BS{Z}_i^{\otimes 2}w_{k,i}(t,\theta_0,\BS{\beta}_0)] - \sum_{j\in I_k} [\BS{Z}_j^{\otimes 2}  w_{k,j}(t,\theta_0,\BS{\beta}_0)] \sum_{i\in I_k}w_{k,i}(t,\theta_0,\BS{\beta}_0) \\
	& = 0.
\end{align*}
Thus, an application of Lemma \ref{concentrationformg} again leads to 
\begin{eqnarray*}
	\frac{1}{\sqrt{m}}\bigg \| \sum_{i\in I_k}\int_0^\tau \BS{a}_i(t)dN_i(t)- \mathbb{E}\sum_{i\in I_k}\int_0^\tau \BS{a}_i(t)dN_i(t) \bigg\|_{\max}= O_p(\sqrt{\log(np^2)}).
\end{eqnarray*}
Therefore, we obtain that
\begin{eqnarray}\label{ieqG2}
& & \sup_{\substack{\theta \in \mathcal{N}_{\theta_0}^* \\  \| \BS{\beta} -\BS{\beta}_0\|_1 \le \tilde C s_{\BS{\beta}_0} \sqrt{\log(np) /n}} } \bigg\| \frac{1}{\sqrt{m}} \sum_{i \in I_k} \bigg\{  \int_0^\tau \big \{\sum_{j\in I_k} [\BS{Z}_j^{\otimes 2}  w_{k,j}(t,\theta_0,\BS{\beta}_0)] -\BS{Z}_i^{\otimes 2}  \big \}dN_i(t)  \nonumber \\
& & - \mathbb{E} \Big[ \int_0^{\tau} \big \{\sum_{j\in I_k} [\BS{Z}_j^{\otimes 2}  w_{k,j}(t,\theta_0,\BS{\beta}_0)] -\BS{Z}_i^{\otimes 2} \big \} dN_i(t) \Big] \bigg\}  \bigg\|_{\max}   \le O_p(\sqrt{\log(np^2)}).
\end{eqnarray} 

For the third term $G_3 = \int_0^\tau \BS{Z}_i^{\otimes 2}  dN_i(t) - \int_0^\tau  \bar {\BS{Z}}_k(t,\theta_0,\BS{\beta}_0)^{\otimes 2}  dN_i(t)$,  it follows from \eqref{eq-neww-Gm-Z2} and \eqref{eq-neww-G-matrix} in Lemma \ref{lmentropy} that 
\begin{eqnarray}\label{ieqG3}
& & \sup_{\substack{\theta \in \mathcal{N}_{\theta_0}^* \\  \| \BS{\beta} -\BS{\beta}_0\|_1 \le \tilde C s_{\BS{\beta}_0} \sqrt{\log(np) /n}} } \big\| \mathbb G_m (G_3)\big\|_{\infty}   \le O_p( \sqrt{\log(np)} ).
\end{eqnarray}

For $G_4$, similarly we rewrite 
\begin{align*}
	\big\|G_4\big\|_{\max}&=\bigg\| \int_0^\tau [\bar {\BS{Z}}_k(t,\theta_0,\BS{\beta}_0)^{\otimes 2}  - \bar {\BS{Z}}_k(t,\theta,\BS{\beta})^{\otimes 2} ] dN_i(t)\bigg\|_{\max}\\
	&= \bigg\|\int_0^\tau (\bar {\BS{Z}}_k(t, \theta_0,  {\BS{\beta}_0}) - \bar {\BS{Z}}_k(t, \theta,  {\BS{\beta}}))(\bar {\BS{Z}}_k(t, \theta_0,  {\BS{\beta}_0})+\bar {\BS{Z}}_k(t, \theta,  {\BS{\beta}}) )^\top dN_i(t)\bigg\|_{\max}\\
	& \leq   2K_Z\int_0^\tau \|\bar {\BS{Z}}_k(t, \theta_0,  {\BS{\beta}_0})-\bar {\BS{Z}}_k(t, \theta,  {\BS{\beta}})\|_{\infty} dN_i(t) \\
	& =  2K_Z\int_0^\tau  \sum_{j\in I_k} [\|\BS{Z}_j\|_\infty w_{k, j}(t,\theta_0,\BS{\beta}_0)\frac{|w_{k, j}(t,\theta,\BS{\beta}) - w_{k, j}(t,\theta_0,\BS{\beta}_0)|}{w_{k, j}(t,\theta_0,\BS{\beta}_0)}]dN_i(t).
\end{align*}
Taking into account the upper bound of $\frac{|w_{k, j}(t,\theta,\BS{\beta}) - w_{k, j}(t,\theta_0,\BS{\beta}_0)|}{w_{k, j}(t,\theta_0,\BS{\beta}_0)}$ obtained in \eqref{pf-lemma3-new-eq1} gives us that
\begin{eqnarray}\label{ieqG4}
	\big\|\mathbb G_m(G_4)\big\|_\infty = O(s_{\BS{\beta}_0}\sqrt{\log(np) }).
\end{eqnarray}
Summing up all four statements (\ref{ieqG1})--(\ref{ieqG4}) leads to the desired conclusion. This completes the proof of Lemma \ref{lmentropy0-3}.

\subsection{Lemma \ref{lmentropy0-4} and its proof} \label{SecB.9-new}
Recall that $\mathcal{T}_N =\{\BS{\eta}_a= (\BS{\beta}^{\top}, \BS{\mu}_a^{\top})^{\top} \in \mathcal{B} \times \mathcal{U} \subset  \mathbb{R}^{2p}: \, \|\BS{\beta}-\BS{\beta}_0\|_1 \le \tilde C s_{\BS{\beta}_0} \sqrt{\log (np)/n}, \, 
     \|\BS{\mu}_a-\BS{\mu}_{a0} \|_1 \le \tilde C {\|\BS{\Sigma}}^{-1}_{\BS{\beta}\BS{\beta}} \|_{\infty}  s_{{\BS{\mu}}_{a0}} (s_{\BS{\beta}_0} \sqrt{\log (np) /n} + \tau_n) \}$.
\begin{lemma}\label{lmentropy0-4}
	Assume Conditions \ref{coninicox} -  \ref{contheta} are satisfied. Then for any $k \in \{1, \cdots, K\}$ and conditional on $\{\mathcal{O}_i\}_{i \in I_k^c}$ and the event $\hat{\BS{\eta}}_k \in \mathcal{T}_N $, 
	we have
	\begin{align*}
		&\sup_{ \theta \in \mathcal{N}_{\theta_0^*} } \big| \mathbb{G}_m [\Phi (\{\mathcal{O}_i\}_{i \in I_k};\theta, \hat{\BS{\eta}}_k) -\Phi (\{\mathcal{O}_i\}_{i \in I_k};\theta_0, \BS{\eta}_{a0})] \big| \\
        &\leq O_p\big(  (s_{\BS{\beta}_0}^2 + {\|\BS{\Sigma}}^{-1}_{\BS{\beta}\BS{\beta}} \|_{\infty} s_{{\BS{\mu}}_{a0}}  s_{\BS{\beta}_0}  ) \log(np) /\sqrt{n}  +  \tau_n  {\|\BS{\Sigma}}^{-1}_{\BS{\beta}\BS{\beta}} \|_{\infty} s_{{\BS{\mu}}_{a0}}   \sqrt{\log (np)}   \big)
	\end{align*} 
\end{lemma}

\noindent \textit{Proof}. First, noting that conditional on $\{\mathcal{O}_i\}_{i \in I_k^c}$ and $I_k^c$, the estimator $\hat {\BS{\eta}}_k$ can be regarded as fixed, thus the expectation $\mathbb{E}$ is only with respect to $\{\mathcal{O}_i\}_{i \in I_k}$. By definition,  we have
\begin{align*}
	\Phi (\{\mathcal{O}_i\}_{i \in I_k}; \theta, \BS{\eta}) &= \dot l_{ \theta}(\{\mathcal{O}_i\}_{i \in I_k}; \theta, \BS{\beta})-  \BS{\mu}_{a}^{\top} \dot l_{ \BS{\beta}}(\{\mathcal{O}_i\}_{i \in I_k}; \theta, \BS{\beta}) \\
	&=  \frac{1}{m} \sum_{i \in I_k} \int_0^{\tau} (D_i - \bar D_k(t,  \theta, \BS{\beta})) dN_i(t) \\
	& -  \frac{1}{m} \sum_{i \in I_k} \int_0^{\tau} \BS{\mu}_{a}^{\top} ( \BS{Z}_i -\bar {\BS{Z}}_k (t, \theta, \BS{\beta}) ) dN_i(t),
\end{align*}
 It follows by the mean value theorem that  for any $\theta \in \mathcal{N}_{\theta_0}^*$, there exist  $\tilde \theta$ lying between $\theta$ and $\theta_0$, and $\tilde {\BS{\beta}}$ lying  between $\hat {\BS{\beta}}_k$ and $\BS{\beta}_0$ such that 
\begin{align*}
	&  \Phi (\{\mathcal{O}_i\}_{i \in I_k};\theta, \hat{\BS{\eta}}_k) -\Phi (\{\mathcal{O}_i\}_{i \in I_k};\theta_0, \BS{\eta}_{a0}) \\
	&=  -\frac{1}{m} \sum_{i \in I_k} \int_0^{\tau} (\bar D_k(t, \theta, \hat{ \BS{\beta}}_k) - \bar D_k(t,  \theta_0, \BS{\beta}_0)) dN_i(t) \\
	&  + \frac{1}{m} \sum_{i \in I_k} \int_0^{\tau} \hat {\BS{\mu}}_{ak}^{\top} ( \bar {\BS{Z}}_k (t, \theta, \hat {\BS{\beta}}_k) -\bar {\BS{Z}}_k (t, \theta_0, \BS{\beta}_0) ) dN_i(t) \\
	&  -   \frac{1}{m} \sum_{i \in I_k} \int_0^{\tau} (\hat {\BS{\mu}}_{ak} -\BS{\mu}_{a0})^{\top} ( \BS{Z}_i -\bar {\BS{Z}}_k (t, \theta_0, \BS{\beta}_0) ) dN_i(t) \\
	&= -\frac{1}{m} \sum_{i \in I_k} \int_0^{\tau} \sum_{j \in I_k} D_j (D_j - \bar {D}_k(t, \tilde \theta, \tilde {\BS{\beta}})) w_{k,j} (t, \tilde \theta, \tilde {\BS{\beta}}) ({\theta}  -\theta_0) dN_i(t) \\
	&-\frac{1}{m} \sum_{i \in I_k} \int_0^{\tau} \sum_{j \in I_k} D_j (\BS{Z}_j - \bar {\BS{Z}}_k(t, \tilde \theta, \tilde {\BS{\beta}}))^{\top} w_{k,j} (t, \tilde \theta, \tilde {\BS{\beta}}) (\hat {\BS{\beta}}_k  -\BS{\beta}_0) dN_i(t)\\
	&  + \frac{1}{m} \sum_{i \in I_k} \int_0^{\tau} \hat {\BS{\mu}}_{ak}^{\top} \sum_{j \in I_k} D_j (\BS{Z}_j - \bar {\BS{Z}}_k(t, \tilde \theta, \tilde {\BS{\beta}}))  w_{k,j} (t, \tilde \theta, \tilde {\BS{\beta}}) (\theta -\theta_0) dN_i(t) \\
	&  + \frac{1}{m} \sum_{i \in I_k} \int_0^{\tau} \hat {\BS{\mu}}_{ak}^{\top} \sum_{j \in I_k} \BS{Z}_j   (\BS{Z}_j - \bar {\BS{Z}}_k(t, \tilde \theta, \tilde {\BS{\beta}}))^{\top} w_{k,j} (t, \tilde \theta, \tilde {\BS{\beta}}) (\hat {\BS{\beta}}_k  -\BS{\beta}_0) dN_i(t)  \\
	&  -   \frac{1}{m} \sum_{i \in I_k} \int_0^{\tau} (\hat {\BS{\mu}}_{ak} -\BS{\mu}_{a0})^{\top} ( \BS{Z}_i -\bar {\BS{Z}}_k (t, \theta_0, \BS{\beta}_0) ) dN_i(t). \\
    & := H_1 + H_2 + H_3 + H_4 + H_5.
\end{align*}
Next we proceed to deal with the empirical measure of the above five terms separately. All the following arguments are conditional on $\{\mathcal{O}_i\}_{i \in I_k^c}$. First, for $H_1$, we have
\begin{align*}
 	 |\mathbb{G}_m [H_1] | & =  \bigg| \frac{1}{\sqrt{m}} \sum_{i \in I_k} \bigg\{ \int_0^{\tau} \sum_{j \in I_k} D_j (D_j - \bar {D}_k(t, \tilde \theta, \tilde {\BS{\beta}})) w_{k,j} (t, \tilde \theta, \tilde {\BS{\beta}}) ({\theta}  -\theta_0) dN_i(t) \\
	&  - \mathbb{E} \Big[ \int_0^{\tau} \sum_{j \in I_k} D_j (D_j - \bar D_k(t, \tilde \theta, \tilde {\BS{\beta}})) w_{k,j} (t, \tilde \theta, \tilde {\BS{\beta}}) (\theta- \theta_0) dN_i(t) \Big] \bigg\} \bigg| \\
	& \le   \sup_{\substack{\theta \in \mathcal{N}_{\theta_0^*} \\  \| \BS{\beta} -\BS{\beta}_0\|_1 \le \tilde C s_{\BS{\beta}_0} \sqrt{\log(np) /n}} } \bigg| \frac{1}{\sqrt{m}} \sum_{i \in I_k} \bigg\{ \int_0^{\tau} \sum_{j \in I_k} D_j (D_j - \bar D_k(t,  \theta, \BS{\beta})) w_{k,j} (t,  \theta,  \BS{\beta}) dN_i(t)  \\
	&  - \mathbb{E} \Big[ \int_0^{\tau} \sum_{j \in I_k} D_j (D_j - \bar D_k(t, \theta,  \BS{\beta})) w_{k,j} (t,  \theta,  \BS{\beta})  dN_i(t) \Big] \bigg\}  \bigg|   \sup_{\theta \in \mathcal{N}_{\theta_0^*}} |\theta- \theta_0| \\
	& \le    \sup_{\substack{\theta \in \mathcal{N}_{\theta_0^*} \\  \| \BS{\beta} -\BS{\beta}_0\|_1 \le \tilde C s_{\BS{\beta}_0} \sqrt{\log(np) /n}} } \bigg| \frac{1}{\sqrt{m}} \sum_{i \in I_k} \bigg\{ \int_0^{\tau} \sum_{j \in I_k} D_j (D_j - \bar D_k(t,  \theta, \BS{\beta})) w_{k,j} (t,  \theta,  \BS{\beta}) dN_i(t) \\
	&  - \mathbb{E} \Big[ \int_0^{\tau} \sum_{j \in I_k} D_j (D_j - \bar D_k(t, \theta,  \BS{\beta})) w_{k,j} (t,  \theta,  \BS{\beta})  dN_i(t) \Big] \bigg\}  \bigg|  \cdot  O(s_{\BS{\beta}_0} \sqrt{\log (np)/{n}} ).
\end{align*}
Therefore, it follows from Lemma \ref{lmentropy0-1} that $ |\mathbb{G}_m [H_1] | = O_p ( s_{\BS{\beta}_0}^2 \log(np) /\sqrt {n} )$.

 Similarly, it follows from Lemmas \ref{lmentropy0-2} and \ref{lmentropy0-3} that 
  
\begin{align*}
	  |\mathbb{G}_m [H_2] | & =   \bigg| \frac{1}{\sqrt{m}} \sum_{i \in I_k}  \bigg\{ \int_0^{\tau} \sum_{j \in I_k} D_j (\BS{Z}_j - \bar {\BS{Z}}_k(t, \tilde \theta, \tilde {\BS{\beta}}))^{\top} w_{k,j} (t, \tilde \theta, \tilde {\BS{\beta}}) (\hat {\BS{\beta}}_k  -\BS{\beta}_0) dN_i(t)\\
	&  - \mathbb{E} \Big[  \int_0^{\tau} \sum_{j \in I_k} D_j (\BS{Z}_j - \bar {\BS{Z}}_k(t, \tilde \theta, \tilde {\BS{\beta}}))^{\top} w_{k,j} (t, \tilde \theta, \tilde {\BS{\beta}})(\hat {\BS{\beta}}_k -\BS{\beta}_0) dN_i(t) \Big] \bigg\} \bigg| \\
	& \le   \sup_{\substack{\theta \in \mathcal{N}_{\theta_0^*} \\  \| \BS{\beta} -\BS{\beta}_0\|_1 \le \tilde C s_{\BS{\beta}_0} \sqrt{\log(np) /n}} } \bigg\| \frac{1}{\sqrt{m}} \sum_{i \in I_k}  \bigg\{ \int_0^{\tau} \sum_{j \in I_k} D_j (\BS{Z}_j - \bar {\BS{Z}}_k(t,  \theta,  {\BS{\beta}}))^{\top} w_{k,j} (t,  \theta,  {\BS{\beta}}) dN_i(t)\\
	&  - \mathbb{E} \Big[  \int_0^{\tau} \sum_{j \in I_k} D_j (\BS{Z}_j - \bar {\BS{Z}}_k(t, \theta,  {\BS{\beta}}))^{\top} w_{k,j} (t, \theta,  {\BS{\beta}})dN_i(t) \Big] \bigg\} \bigg\|_{\infty}  \cdot \| \hat {\BS{\beta}}_k -\BS{\beta}_0\|_1  \\
    & \leq  O_{p} ( s_{\BS{\beta}_0}^2 \log(np) /\sqrt {n} ),
\end{align*}
\begin{align*}
	 |\mathbb{G}_m [H_3] |  &=  \bigg| \frac{1}{\sqrt{m}} \sum_{i \in I_k} \bigg\{ \int_0^{\tau} \hat {\BS{\mu}}_{ak}^{\top} \sum_{j \in I_k} D_j (\BS{Z}_j - \bar {\BS{Z}}_k(t, \tilde \theta, \tilde {\BS{\beta}}))  w_{k,j} (t, \tilde \theta, \tilde {\BS{\beta}}) (\theta -\theta_0) dN_i(t) \\
	&  - \mathbb{E} \Big[\int_0^{\tau} \hat {\BS{\mu}}_{ak}^{\top} \sum_{j \in I_k} D_j  (\BS{Z}_j - \bar {\BS{Z}}_k(t, \tilde \theta, \tilde {\BS{\beta}})) w_{k,j} (t, \tilde \theta, \tilde {\BS{\beta}}) (\theta -\theta_0) dN_i(t) \Big] \bigg\} \bigg| \\
	& \le  \sup_{\substack{\theta \in \mathcal{N}_{\theta_0^*} \\  \| \BS{\beta} -\BS{\beta}_0\|_1 \le \tilde C s_{\BS{\beta}_0} \sqrt{\log(np) /n}} }  \bigg\| \frac{1}{\sqrt{m}} \sum_{i \in I_k} \bigg\{ \int_0^{\tau}  \sum_{j \in I_k} D_j (\BS{Z}_j - \bar {\BS{Z}}_k(t, \theta,  {\BS{\beta}})) w_{k,j} (t, \theta, {\BS{\beta}}) dN_i(t) \\
	&  - \mathbb{E} \Big[\int_0^{\tau} \sum_{j \in I_k} D_j  (\BS{Z}_j - \bar {\BS{Z}}_k(t, \theta, {\BS{\beta}})) w_{k,j} (t, \theta, {\BS{\beta}})  dN_i(t) \Big] \bigg\} \bigg\|_{\infty} \cdot \|\hat {\BS{\mu}}_{ak} \|_1 \cdot  \sup_{\theta \in \mathcal{N}_{\theta_0^*}} |\theta -\theta_0| \\
    & \leq  O_p( s_{\BS{\beta}_0}^2 \log(np) /\sqrt {n} )
\end{align*}
\begin{align*}
	 |\mathbb{G}_m [H_4] | &=  \bigg| \frac{1}{\sqrt{m}} \sum_{i \in I_k} \bigg\{ \int_0^{\tau} \hat {\BS{\mu}}_{ak}^{\top} \sum_{j \in I_k} \BS{Z}_j   (\BS{Z}_j - \bar {\BS{Z}}_k(t, \tilde \theta, \tilde {\BS{\beta}}))^{\top} w_{k,j} (t, \tilde \theta, \tilde {\BS{\beta}}) (\hat {\BS{\beta}}_k  -\BS{\beta}_0) dN_i(t)  \\
	&  - \mathbb{E} \Big[ \int_0^{\tau} \hat {\BS{\mu}}_{ak}^{\top} \sum_{j \in I_k} \BS{Z}_j (\BS{Z}_j - \bar {\BS{Z}}_k(t, \tilde \theta, \tilde {\BS{\beta}}))^{\top} w_{k,j} (t, \tilde \theta, \tilde {\BS{\beta}}) (\hat {\BS{\beta}}_k  -\BS{\beta}_0) dN_i(t) \Big] \bigg\} \bigg| \\
	& \le  \sup_{\substack{\theta \in \mathcal{N}_{\theta_0^*} \\  \| \BS{\beta} -\BS{\beta}_0\|_1 \le \tilde C s_{\BS{\beta}_0} \sqrt{\log(np) /n}} } \bigg\| \frac{1}{\sqrt{m}} \sum_{i \in I_k} \bigg\{ \int_0^{\tau}  \sum_{j \in I_k} \BS{Z}_j   (\BS{Z}_j - \bar {\BS{Z}}_k(t,  \theta, {\BS{\beta}}))^{\top} w_{k,j} (t,  \theta, {\BS{\beta}})  dN_i(t)  \\
	&  - \mathbb{E} \Big[ \int_0^{\tau}  \sum_{j \in I_k} \BS{Z}_j (\BS{Z}_j - \bar {\BS{Z}}_k(t,  \theta,  {\BS{\beta}}))^{\top} w_{k,j} (t,  \theta,  {\BS{\beta}})  dN_i(t) \Big] \bigg\} \bigg\|_{\infty} \cdot \|\hat {\BS{\mu}}_{ak} \|_1 \cdot \| \hat {\BS{\beta}}_k -\BS{\beta}_0\|_1 \\
    & \leq  O_p( s_{\BS{\beta}_0}^2 \log(np) /\sqrt {n} ),
\end{align*}
where we have used the result that 
\begin{align*}
\| \hat{\BS{\mu}}_{a k} \|_1  \leq \| \hat{\BS\mu}_{a k}  -  \BS{\mu}_{a0}\|_1 + \|  \BS{\mu}_{a0} \|_1 \leq C {\|\BS{\Sigma}}^{-1}_{\BS{\beta}\BS{\beta}} \|_{\infty}  s_{{\BS{\mu}}_{a0}} (s_{\BS{\beta}_0} \sqrt{\log (np) /n} + \tau_n) + K_{\BS{\mu}} = O(1)
\end{align*}
by Proposition \ref{lminimut} and Condition \ref{coninimu} (v)--(vi). 

Regarding $ H_5 $, it follows from Proposition \ref{lminimut} and Lemma 2 (i) in \cite{yu2021confidence}  that 
\begin{align*}
    & \bigg| \frac{1}{\sqrt m} \sum_{i \in I_k} \int_0^{\tau} (\hat {\BS{\mu}}_{ak} -\BS{\mu}_{a0})^{\top} ( \BS{Z}_i -\bar {\BS{Z}}_k (t, \theta_0, \BS{\beta}_0) ) dN_i(t) \bigg| \\
    &\le  \| \hat {\BS{\mu}}_{ak} -\BS{\mu}_{a0} \|_1  \cdot \Big\| \frac{1}{\sqrt{m}} \sum_{i \in I_k} \int_0^{\tau}  ( \BS{Z}_i -\bar {\BS{Z}}_k (t, \theta_0, \BS{\beta}_0) ) dN_i(t)  \Big\|_{\infty} \\
	& \le O_p({\|\BS{\Sigma}}^{-1}_{\BS{\beta}\BS{\beta}} \|_{\infty}  s_{{\BS{\mu}}_{a0}} (s_{\BS{\beta}_0} \sqrt{\log (np) /n} + \tau_n) ) O_p(\sqrt{\log (np)}) \\
	& = O_p\big(  {\|\BS{\Sigma}}^{-1}_{\BS{\beta}\BS{\beta}} \|_{\infty} s_{{\BS{\mu}}_{a0}}  s_{\BS{\beta}_0}   \log(np) /\sqrt{n}  +  \tau_n  {\|\BS{\Sigma}}^{-1}_{\BS{\beta}\BS{\beta}} \|_{\infty} s_{{\BS{\mu}}_{a0}}   \sqrt{\log (np)}   \big) .
\end{align*}

Therefore, combining all the above results together, we can obtain that
\begin{align*}
	& \sup_{\theta \in \mathcal{N}_{\theta_0}^*} \big| \mathbb{G}_m [\Phi (\{\mathcal{O}_i\}_{i \in I_k};\theta, \hat{\BS{\eta}}_k) -\Phi (\{\mathcal{O}_i\}_{i \in I_k};\theta_0, \BS{\eta}_{a0})] \big| \\
    & \le  O_p\big(  (s_{\BS{\beta}_0}^2 + {\|\BS{\Sigma}}^{-1}_{\BS{\beta}\BS{\beta}} \|_{\infty} s_{{\BS{\mu}}_{a0}}  s_{\BS{\beta}_0}  ) \log(np) /\sqrt{n}  +  \tau_n  {\|\BS{\Sigma}}^{-1}_{\BS{\beta}\BS{\beta}} \|_{\infty} s_{{\BS{\mu}}_{a0}}   \sqrt{\log (np)}   \big).
\end{align*} 
This concludes the proof of Lemma \ref{lmentropy0-4}.

\subsection{Lemma \ref{lemma-asymp-norm-main-term} and its proof} \label{new.sec.lem6}

Lemma \ref{lemma-asymp-norm-main-term} below establishes the central limit theorem for the leading oracle score term. It is the probabilistic core of Theorem \ref{ththeta}.

\begin{lemma} \label{lemma-asymp-norm-main-term}
    Assume that Conditions \ref{coninicox} and \ref{coninimu} are satisfied. Then we have that 
    \begin{equation} \label{eq-lemma-8-result-normal}
        \sqrt{m} J_0^{-1} \Phi (\{\mathcal{O}_i\}_{i \in I_k}; \theta_0, \BS{\eta}_{a0}) \stackrel{\mathcal D} {\rightarrow} N(0, J_0^{-1} \sigma_{\Phi}^2 J_0^{-1})
    \end{equation}
    as $m \rightarrow \infty$.
\end{lemma}

\noindent\textit{Proof}. The proof of this lemma will decompose the score into a sum of independent and identically distributed (i.i.d.) martingale-integral terms plus two terms caused by replacing the population risk-set averages with the empirical ones. The leading martingale sum satisfies a standard central limit theorem, while the two replacement terms have vanishing variance by the uniform risk-set concentration. This establishes the stated asymptotic Gaussian distribution with variance $J_0^{-1}\sigma_\Phi^2J_0^{-1}$.

In view of \eqref{eq-neww-expand-Phi-0}, we can derive the following decomposition
\begin{equation} \label{eq-neww-Phi-decomposition}
\begin{split}
     & \sqrt{m}    \Phi (\{\mathcal{O}_i\}_{i \in I_k}; \theta_0, \BS{\eta}_{a0} ) \\
     & = \frac{1}{\sqrt{m}} \sum_{i \in I_k}  \int_0^{\tau} [(D_i - \bar D_k(t, \theta_0, \BS{\beta}_0))  - {\BS{\mu}}_{a0}^{\top} ( \BS{Z}_i -\bar {\BS{Z}}_k (t,\theta_0, \BS{\beta}_0) ) ]dM_i(t) \\
     & = \frac{1}{\sqrt{m}} \sum_{i \in I_k}  \int_0^{\tau} [(D_i -  U_D(t, \theta_0, \BS{\beta}_0))  - {\BS{\mu}}_{a0}^{\top} ( \BS{Z}_i - U_{\BS{Z}} (t,\theta_0, \BS{\beta}_0) ) ]dM_i(t) \\
     & \quad - \frac{1}{\sqrt{m}} \sum_{i \in I_k}  \int_0^{\tau} ( \bar D_k(t, \theta_0, \BS{\beta}_0)  -  U_D(t, \theta_0, \BS{\beta}_0)) dM_i(t) \\
     & \quad + \frac{1}{\sqrt{m}} \sum_{i \in I_k}  \int_0^{\tau} {\BS{\mu}}_{a0}^{\top} ( \bar {\BS{Z}}_k(t, \theta_0, \BS{\beta}_0)  -  U_{\BS{Z}} (t, \theta_0, \BS{\beta}_0)) dM_i(t) \\
     & :=  \frac{1}{\sqrt{m}} \sum_{i \in I_k} U_{i} -  \frac{1}{\sqrt{m}} \sum_{i \in I_k}  V_{D, i} +  \frac{1}{\sqrt{m}} \sum_{i \in I_k}  V_{Z, i},
\end{split}
\end{equation}
where $ \{M_i(t)\}_{t \in [0, \tau]} $ is a mean-zero martingale. We first show that the first term on the right-hand side of the expression above is asymptotically normal. Note that $\{U_i\}_{i \in I_k}$ are i.i.d. mean-zero random variables with finite variance $ \Var(U_i) $. Hence, it follows that
\begin{equation} \label{eq-neww-main-normal}
     \frac{1}{\sqrt{m}} \sum_{i \in I_k} U_{i} \stackrel{D}{\rightarrow} N(0, \Var(U_i)).
\end{equation}

We next show that the last two terms above are asymptotically negligible. Observe that $\{V_{D, i}\}_{i \in I_k}$ are exchangeable with $\mathbb{E} [V_{D, i}] = 0$. As in the proof of Proposition 3 in \cite{yu2021confidence}, it holds that 
\begin{equation*}
\begin{split}
    \Var \Big( \frac{1}{\sqrt{m}} \sum_{i \in I_k} V_{D, i} \Big) & = \Var(V_{D, i}) \\
    & = \mathbb{E} \Big[ \int_{0}^{\tau} ( \bar D_k(t, \theta_0, \BS{\beta}_0)  -  U_D(t, \theta_0, \BS{\beta}_0))^{ 2} \tilde{w}_{i} (t, \theta_0, \BS{\beta}_0) \lambda_0 (t) d t\Big].
\end{split}
\end{equation*}
In addition, we have that 
$$ \tilde{w}_{i} (t, \theta_0, \BS{\beta}_0) \leq \exp ( D_i \theta_0 + \BS{Z}_{i}^{\top} \BS{\beta}_0 ) \leq \exp ( |\theta_0| + K_{\BS{Z}} \|\BS{\beta}_0 \|_1 ) \leq M $$ for some constant $M > 0$, and $ \int_{0}^{\tau} \lambda_0 (t) < \infty $ by Condition \ref{coninimu}(iii). Then it follows that 
\begin{equation*}
    \begin{split}
        \Var \Big( \frac{1}{\sqrt{m}} \sum_{i \in I_k} V_{D, i} \Big) \leq C \mathbb{E} \Big[\sup_{t \in [0, \tau]} (\bar{D}_k(t, \theta_0, \BS{\beta}_0) - U_D(t, \theta_0, \BS{\beta}_0))^{ 2} \Big].
    \end{split}
\end{equation*}
By resorting to Lemma 7 in \cite{yu2021confidence}, using Conditions \ref{coninicox}(i), \ref{coninicox}(iv), and \ref{coninimu}(ii)--(iii), we can show that with probability $1 - O(n^{-1})$, event in \eqref{eq-neww-event-eps} gives that 
\begin{align*}
   &  \sup_{t \in [0,\tau]} \big| \bar D_k(t,\theta_0, \BS{\beta}_0) - U_D(t,\theta_0, \BS{\beta}_0)\big| \leq C \sqrt{\log (np) /n}, \\
    & \sup_{t \in [0,\tau]} \| \bar {\BS{Z}}_k(t,\theta_0, \BS{\beta}_0) - \BS{U}_{\BS{Z}} (t,\theta_0, \BS{\beta}_0) \|_{\infty}\le   C \sqrt{\log (np) /n}.
\end{align*}
Consequently, we can obtain that 
\begin{equation} \label{eq-neww-VD-convergence}
    \Var \Big( \frac{1}{\sqrt{m}} \sum_{i \in I_k} V_{D, i} \Big) \leq C \Big( \frac{\log (np)}{n} + \frac{1}{n} \Big) \to 0.
\end{equation}

An application of similar arguments leads to 
\begin{equation} \label{eq-neww-VZ-convergence}
    \Var \Big( \frac{1}{\sqrt{m}} \sum_{i \in I_k} V_{Z, i} \Big) \to 0.
\end{equation}
Recall the definition of $\sigma_{\Phi}^2 = \Var(\sqrt{m} \Phi (\{\mathcal{O}_i\}_{i \in I_k}; \theta_0, \BS{\eta}_{a0}) )$ and the fact that $\Var(\frac{1}{\sqrt{m}} \sum_{i \in I_k} U_{i} ) = \Var(U_i)$. Then it follows from the decomposition in \eqref{eq-neww-Phi-decomposition} and the results in \eqref{eq-neww-VD-convergence}--\eqref{eq-neww-VZ-convergence} that
\begin{equation} \label{eq-neww-var-appro}
    \Big| \sigma_{\Phi}^2 - \Var(\frac{1}{\sqrt{m}} \sum_{i \in I_k} U_{i} ) \Big| \to 0.
\end{equation}
Therefore, combining \eqref{eq-neww-Phi-decomposition}, \eqref{eq-neww-main-normal}, \eqref{eq-neww-VD-convergence}, and \eqref{eq-neww-VZ-convergence} yields the desired conclusion in \eqref{eq-lemma-8-result-normal}. This concludes the proof of Lemma \ref{lemma-asymp-norm-main-term}.

\subsection{Lemma \ref{lmentropy} and its proof} \label{SecB.11-new}

Lemma \ref{lmentropy} below collects the concentration bounds for several empirical process terms that appear repeatedly in the variance and remainder analyses. These bounds keep the proof modular and avoid repeating similar Hoeffding and martingale arguments.

\begin{lemma}\label{lmentropy}
	Assume that Conditions \ref{coninicox}--\ref{contheta} are satisfied. Then we have that 
	\begin{align}\label{ieqB4}
\bigg| \frac{1}{\sqrt{m}} \sum_{i \in I_k} \bigg\{ \int_0^{\tau} D_i^2dN_i(t) - \mathbb{E} \Big[ \int_0^{\tau} D_i^2 dN_i(t) \Big] \bigg\} \bigg| = O_p(\sqrt{\log n}),
\end{align}
\begin{align} \label{eq-neww-Gm-DZ}
	 \bigg \| \frac{1}{\sqrt{m}} \sum_{i\in I_k}\Big[\int_0^\tau D_i \BS{Z}_i dN_i(t)- \mathbb E \int_0^\tau D_i \BS{Z}_i dN_i(t)\Big]\bigg \|_\infty = O_p(\sqrt{\log (np)}),
\end{align}
\begin{align} \label{eq-neww-Gm-Z2}
	 \bigg\|\sum_{i\in I_k} \frac{1}{\sqrt{m}} \Big[\int_0^\tau \BS{Z}_i^{\otimes 2} dN_i(t)- \mathbb E \int_0^\tau \BS{Z}_i^{\otimes 2} dN_i(t) \Big]\bigg\|_{\max} = O_p(\sqrt{\log (np)}),
\end{align}
\begin{align} \label{ieqB5}
	\bigg| \frac{1}{\sqrt{m}} \sum_{i \in I_k} \Big[\int_0^{\tau} \bar D_k(t,\theta_0,\BS{\beta}_0)^2 dN_i(t) - \mathbb{E} \int_0^{\tau} \bar D_k(t,\theta_0,\BS{\beta}_0)^2 dN_i(t) \Big] \bigg| = O_p(\sqrt{\log n}),
\end{align}
\begin{align} \label{eq-neww-G-1}
   & \bigg\| \frac{1}{\sqrt{m}}\sum_{i \in I_k} \Big[ \int_0^\tau \bar D_k(t,\theta_0,\BS{\beta}_0) \bar {\BS{Z}}_k(t, \theta_0,  {\BS{\beta}_0}) dN_i(t) \nonumber\\
   & \quad - \mathbb{E} \int_0^\tau \bar D_k(t,\theta_0,\BS{\beta}_0) \bar {\BS{Z}}_k(t, \theta_0,  {\BS{\beta}_0}) dN_i(t) \Big] \bigg\|_\infty = O_p(\sqrt{\log (np)}),
\end{align}
and
\begin{align} \label{eq-neww-G-matrix}
    & \bigg\| \frac{1}{\sqrt{m}}\sum_{i \in I_k} \Big[ \int_0^\tau  \bar {\BS{Z}}_k(t, \theta_0,  {\BS{\beta}_0})^{\otimes 2} dN_i(t)  \nonumber\\
    & \quad - \mathbb{E} \int_0^\tau  \bar {\BS{Z}}_k(t, \theta_0,  {\BS{\beta}_0})^{\otimes 2} dN_i(t) \Big]\bigg\|_{\max} = O_p(\sqrt{\log (np)}).
\end{align}
\end{lemma}

\noindent\textit{Proof}. The proof of this lemma will first apply Hoeffding's inequality to bounded covariate terms. It will then handle terms involving empirical risk-set averages by replacing them with their population limits and bounding the resulting differences using the uniform concentration. For integrals with predictable compensators, the Doob--Meyer decomposition separates the martingale and predictable parts, each of which is analyzed separately.

Since $\{\int_0^{\tau} D_i^2dN_i(t)\}_{i\in I_k}$ is a sequence of bounded independent random variables and $m$ is of the same order as $n$, it follows from Hoeffding's inequality that for some constant $C > 0$,
\begin{eqnarray*}
	\mathbb P \bigg(\bigg| \frac{1}{\sqrt{m}} \sum_{i \in I_k} \bigg\{ \int_0^{\tau} D_i^2dN_i(t) - \mathbb{E} \Big[ \int_0^{\tau} D_i^2 dN_i(t) \Big] \bigg\} \bigg|> \sqrt{C\log(n)} \bigg) \leq 2 n^{-1},
\end{eqnarray*}
which entails bound \eqref{ieqB4}.
Similarly, an application of Hoeffding's inequality also gives that 
\begin{eqnarray*}
	\mathbb P \bigg(\bigg|\sum_{i\in I_k} \frac{1}{\sqrt{m}} \Big[\int_0^\tau D_i \BS{Z}_{i,l} dN_i(t)- \mathbb E \int_0^\tau D_i \BS{Z}_{i,l} dN_i(t) \Big] \bigg| > \sqrt{K_Z\log(np)} \bigg) \leq \frac{2}{np}
\end{eqnarray*}
and
\begin{eqnarray*}
	\mathbb P \bigg( \bigg|\sum_{i\in I_k} \frac{1}{\sqrt{m}} \bigg[\int_0^\tau \BS{Z}_{i,l} \BS{Z}_{i,k} dN_i(t)- \mathbb E \int_0^\tau \BS{Z}_{i,l} \BS{Z}_{i,k} dN_i(t) \bigg] \bigg| > C \sqrt{K_Z\log(np^2)} \bigg) \leq \frac{2}{np^2}.
\end{eqnarray*}
Then applying the union bound, we can deduce that 
\begin{eqnarray*}
	\mathbb P \bigg(\bigg|\sum_{i\in I_k} \frac{1}{\sqrt{m}} \Big[\int_0^\tau D_i \BS{Z}_{i} dN_i(t)- \mathbb E \int_0^\tau D_i \BS{Z}_{i} dN_i(t) \Big] \bigg| > \sqrt{K_Z\log(np)} \bigg) \leq \frac{2}{n}
\end{eqnarray*}
and
\begin{eqnarray*}
	\mathbb P \bigg( \bigg|\sum_{i\in I_k} \frac{1}{\sqrt{m}} \bigg[\int_0^\tau \BS{Z}_i^{\otimes 2} dN_i(t)- \mathbb E \int_0^\tau \BS{Z}_i^{\otimes 2} dN_i(t) \bigg] \bigg| > C \sqrt{K_Z\log(np^2)} \bigg) \leq \frac{2}{n}.
\end{eqnarray*}
Hence, we can establish \eqref{eq-neww-Gm-DZ} and \eqref{eq-neww-Gm-Z2}.

Next, we will show \eqref{ieqB5}. Let us define 
\begin{equation*}
\begin{split}
& \mathbb G_m \Big(\int_0^{\tau} \bar D_k(t,\theta_0,\BS{\beta}_0)^2 dN_i(t)\Big) \\
& := \frac{1}{\sqrt{m}}\sum_{i \in I_k} \Big[ \int_0^{\tau} \bar D_k(t,\theta_0,\BS{\beta}_0)^2 dN_i(t) - \mathbb{E} \int_0^{\tau} \bar D_k(t,\theta_0,\BS{\beta}_0)^2 dN_i(t) \Big].
\end{split}
\end{equation*}
Approximating $\bar D_k(t,\theta_0,\BS{\beta}_0)$ with its population counterpart $U_D(t,\theta_0,\BS{\beta}_0)$, it holds that 
\begin{align*}
	&\bigg|\mathbb G_m \Big(\int_0^{\tau} \bar D_k(t,\theta_0,\BS{\beta}_0)^2 dN_i(t) \Big)\bigg| \\
	&\leq \bigg | \mathbb G_m \Big(\int_0^\tau [\bar D_k(t,\theta_0,\BS{\beta}_0)-U_D(t,\theta_0,\BS{\beta}_0)] ^2dN_i(t) \Big)\bigg| \\
	&\quad  + \bigg | \mathbb G_m \Big(\int_0^\tau 2(\bar D_k(t,\theta_0,\BS{\beta}_0)-U_D(t,\theta_0,\BS{\beta}_0))U_D(t,\theta_0,\BS{\beta}_0)dN_i(t) \Big) \bigg| \\
	& \quad  + \bigg | \mathbb G_m \Big(\int_0^\tau U_D(t,\theta_0,\BS{\beta}_0)^2dN_i(t) \Big)\bigg|.
\end{align*}
For the first term in the decomposition above, by the triangle inequality it can be divided into a random part and the expectation of the random part 
\begin{align*}
	&\bigg | \mathbb G_m(\int_0^\tau (\bar D_k(t,\theta_0,\BS{\beta}_0)-U_D(t,\theta_0,\BS{\beta}_0)) ^2dN_i(t))\bigg| \\
	&\leq\frac{1}{\sqrt{m}}\sum_{i\in I_k} \int_0^\tau (\bar D_k(t,\theta_0,\BS{\beta}_0)-U_D(t,\theta_0,\BS{\beta}_0)) ^2dN_i(t) \\
	& \quad +  \frac{1}{\sqrt{m}}\sum_{i\in I_k} \mathbb E\int_0^\tau (\bar D_k(t,\theta_0,\BS{\beta}_0)-U_D(t,\theta_0,\BS{\beta}_0)) ^2dN_i(t).
\end{align*}

The random part above can be bounded by the fact that
\[
\sup_{t\in[0,\tau]}|\bar D_k(t,\theta_0,\BS{\beta}_0) - U_D(t,\theta_0,\BS{\beta}_0)| = O_p(\sqrt{ \log(np) / n}),
\]
which is based on Lemma 7 in \cite{yu2021confidence} that
\[
\mathbb P \Big(\sup_{t\in[0,\tau]}|\bar D_k(t,\theta_0,\BS{\beta}_0) - U_D(t,\theta_0,\BS{\beta}_0)| > C\sqrt{ \log(np) / n} \Big) \leq O(n^{-1})
\]
for some absolute constant $C$ when $n$ is large enough\footnote{This requires that $F_T(\tau) < 1$ from Condition \ref{coninimu}(iii), and technically speaking, we should use $m$ here, but since $m/n$ is of constant order, this statement still holds.}. Consequently, we can deduce that 
\begin{align*}
	&\frac{1}{\sqrt{m}}\sum_{i\in I_k} \int_0^\tau [(\bar D_k(t,\theta_0,\BS{\beta}_0)-U_D(t,\theta_0,\BS{\beta}_0)) ^2dN_i(t) \\
	&\leq  \frac{1}{\sqrt{m}}\sum_{i\in I_k}\int_0^\tau \sup_{t\in[0,\tau]}|\bar D_k(t,\theta_0,\BS{\beta}_0) - U_D(t,\theta_0,\BS{\beta}_0)|^2 dN_i(t)\\
	& \leq  \sqrt{m}\sup_{t\in[0,\tau]}|\bar D_k(t,\theta_0,\BS{\beta}_0) - U_D(t,\theta_0,\BS{\beta}_0)|^2\\
	& =  O_p( \log(np) /\sqrt{n}).
\end{align*}

For the deterministic part, we will work on an event
  \begin{align*}
      A_n= \bigg\{\sup_{t\in[0,\tau]}|\bar D_k(t,\theta_0,\BS{\beta}_0) - U_D(t,\theta_0,\BS{\beta}_0)| \leq c\sqrt{\frac{\log(np)}{n}}\bigg\}.
  \end{align*}
On event $A_n$, the expected value is indeed bounded, and the probability of $A_n^c$ is fairly small. Specifically, it holds that 
\begin{align*}
	& \frac{1}{\sqrt{m}}\sum_{i\in I_k} \mathbb E\int_0^\tau [(\bar D_k(t,\theta_0,\BS{\beta}_0)-U_D(t,\theta_0,\BS{\beta}_0)) ^2dN_i(t) \\
	&=  \frac{1}{\sqrt{m}}\sum_{i\in I_k} \mathbb E\int_0^\tau [(\bar D_k(t,\theta_0,\BS{\beta}_0)-U_D(t,\theta_0,\BS{\beta}_0))^2I_{A_n}dN_i(t) \\
	&+ \frac{1}{\sqrt{m}}\sum_{i\in I_k} \mathbb E\int_0^\tau [(\bar D_k(t,\theta_0,\BS{\beta}_0)-U_D(t,\theta_0,\BS{\beta}_0))^2I_{A_n^c} dN_i(t) \\
	&\leq  \frac{1}{\sqrt{m}}\sum_{i\in I_k} \mathbb E\int_0^\tau \sup_{t\in[0,\tau]}|\bar D_k(t,\theta_0,\BS{\beta}_0) - U_D(t,\theta_0,\BS{\beta}_0)|^2I_{A_n}dN_i(t) +  4 \frac{1}{\sqrt{m}}\sum_{i\in I_k} \mathbb P(I_{A_n^c}) \\
	&=  O(\log(np)/ \sqrt{n}) + O( n^{-1/2}) \\
	& = O(\log(np)/ \sqrt{n}).
\end{align*}
Since both terms are controlled at the required order, we can obtain that 
\[
\bigg | \mathbb G_m(\int_0^\tau [\bar D_k(t,\theta_0,\BS{\beta}_0)-U_D(t,\theta_0,\BS{\beta}_0)] ^2dN_i(t))\bigg| = O_p(\log(np)/ \sqrt{n}).
\]

For the second term $\big | \mathbb G_m(\int_0^\tau 2(\bar D_k(t,\theta_0,\BS{\beta}_0)-U_D(t,\theta_0,\BS{\beta}_0))U_D(t,\theta_0,\BS{\beta}_0)dN_i(t))\big|$, an application of similar arguments leads to 
\[
\bigg | \mathbb G_m(\int_0^\tau 2(\bar D_k(t,\theta_0,\BS{\beta}_0)-U_D(t,\theta_0,\BS{\beta}_0))U_D(t,\theta_0,\BS{\beta}_0)dN_i(t))\bigg | = O_p(\sqrt{\log(np)}).
\]
For the last term $\big | \mathbb G_m(\int_0^\tau U_D(t,\theta_0,\BS{\beta}_0)^2dN_i(t))\big|$, since there is no randomness in the integrand $U_D(t,\theta_0,\BS{\beta}_0)$, we need only to focus on the submartingale $N_i(t)$. In light of the Doob--Meyer decomposition $N_i(t) = M_i(t) + A_i(t)$, we have that 
\begin{align*}
	\bigg | \mathbb G_m(\int_0^\tau U_D(t,\theta_0,\BS{\beta}_0)^2dN_i(t))\big| &\leq \bigg | \int_0^\tau U_D(t,\theta_0,\BS{\beta}_0)^2d\mathbb G_m(M_i(t))\bigg|\\
	&+ \bigg | \int_0^\tau U_D(t,\theta_0,\BS{\beta}_0)^2d\mathbb G_m(A_i(t))\bigg|.
\end{align*}
In particular, we can bound the second moment of the first term above for martingales by It\^{o}'s formula using the i.i.d. assumption, that is, 
\begin{align*}
	\mathbb E \big | \int_0^\tau U_D(t,\theta_0,\BS{\beta}_0)^2d\mathbb G_m(M_i(t))\big|^2 & = \mathbb E \int_0^\tau U_D(t,\theta_0,\BS{\beta}_0)^4 d\langle \mathbb G_m(M_i), \mathbb G_m(M_i)\rangle(t)\\
	&= \mathbb E \int_0^\tau U_D(t,\theta_0,\BS{\beta}_0)^4 d\frac{\sum_{i\in I_k}A_i(t)}{m}\\
	&\leq \mathbb E \frac{\sum_{i\in I_k}A_i(\tau)}{m}\\
	& = \mathbb EA(\tau),
\end{align*}
which is a constant. Then dividing this term by any sequence diverging to infinity, say $\{\log n\}_{n=1}^\infty$, yields a degenerate random term. Indeed, since 
\begin{align*}
	\mathbb E \Big[ \frac{1}{\log(n)}\big | \int_0^\tau U_D(t,\theta_0,\BS{\beta}_0)^2d\mathbb G_m(M_i(t))\big|^2 \Big]= o(1),
\end{align*}
we have that $\big | \int_0^\tau U_D(t,\theta_0,\BS{\beta}_0)^2d\mathbb G_m(M_i(t))\big|/\sqrt{\log n}$ converges to some constant in probability. Hence, it is easy to show that the first term satisfies that 
$$\bigg | \int_0^\tau U_D(t,\theta_0,\BS{\beta}_0)^2d\mathbb G_m(M_i(t))\bigg| = O_p(\sqrt{\log n}).$$

For the second term with predictable processes, we exploit the fact that $\mathbb G_m(A_i(t))$ converges in some suitable sense to a tight process. Indeed, we can expand it
\begin{align*}
	&\bigg | \int_0^\tau U_D(t,\theta_0,\BS{\beta}_0)^2d\mathbb G_m(A_i(t))\bigg| \\
	&= \bigg|\int_0^\tau U_D(t,\theta_0,\BS{\beta}_0)^2\frac{\sum_{j\in I_k}\tilde w_{k,j}(t,\theta_0,\BS{\beta}_0) - \mathbb E\tilde w_{k,j}(t,\theta_0,\BS{\beta}_0)}{\sqrt{m}} \lambda_0(t)dt\bigg|,
\end{align*}
and according to Example 2.11.16 from \cite{van1996weak}, we have that \[
\frac{\sum_{j\in I_k}\tilde w_{k,j}(t,\theta_0,\BS{\beta}_0) - \mathbb E\tilde w_{k,j}(t,\theta_0,\BS{\beta}_0)}{\sqrt{m}}\]
converges in $l^\infty([0,\tau])$ to a tight Gaussian process $\mathbb G(t)$ with filtration $\{\mathcal{F}_t: t \in [0, \tau]\}$, where $\mathcal{F}_t=\sigma \{Y_i(s+), N_i(s), D_i, \BS{Z}_i, 0 \le s \le t, i=1, 2, \dots \}$. More precisely, it holds that 
\[\sup_{t\in [0,\tau]}\bigg|\frac{\sum_{j\in I_k}\tilde w_{k,j}(t,\theta_0,\BS{\beta}_0) - \mathbb E\tilde w_{k,j}(t,\theta_0,\BS{\beta}_0)}{\sqrt{m}} - \mathbb G(t)\bigg| \to 0\] in probability.

As a result, we can show that 
\begin{align*}
	& \bigg|\int_0^\tau U_D(t,\theta_0,\BS{\beta}_0)^2\frac{\sum_{j\in I_k}\tilde w_{k,j}(t,\theta_0,\BS{\beta}_0) - \mathbb E\tilde w_{k,j}(t,\theta_0,\BS{\beta}_0)}{\sqrt{m}} \lambda_0(t)dt\bigg|\\
	&\leq  \int_0^\tau U_D(t,\theta_0,\BS{\beta}_0)^2\Big|\frac{\sum_{j\in I_k}\tilde w_{k,j}(t,\theta_0,\BS{\beta}_0) - \mathbb E\tilde w_{k,j}(t,\theta_0,\BS{\beta}_0)}{\sqrt{m}}-\mathbb G(t)\Big| \lambda_0(t)dt\\
	& \quad+ \int_0^\tau U_D(t,\theta_0,\BS{\beta}_0)^2|\mathbb G(t)| \lambda_0(t)dt\\
	&= \quad O_p(1) + O_p(\sqrt{\log n}),
\end{align*}
where the last step above holds since $\int_0^\tau U_D(t,\theta_0,\BS{\beta}_0)^2|\mathbb G(t)| \lambda_0(t)dt$ does not depend on $n$. Hence, it follows that 
\begin{align*}
	\Big | \int_0^\tau U_D(t,\theta_0,\BS{\beta}_0)^2d\mathbb G_m(A_i(t))\Big| = O_p(\sqrt{\log n}).
\end{align*}
This together with the bounds for the martingale term implies that
$$\Big | \mathbb G_m(\int_0^\tau U_D(t,\theta_0,\BS{\beta}_0)^2dN_i(t))\Big| = O_p(\sqrt{\log n}).
$$
Consequently, combining the above results establishes \eqref{ieqB5}.

To show \eqref{eq-neww-G-1}, we can apply a similar decomposition strategy as in the proof for \eqref{ieqB5}. Specifically, it holds that 
\begin{align*}
  & \bigg\| \mathbb G_m(\int_0^\tau \bar D_k(t,\theta_0,\BS{\beta}_0) \bar {\BS{Z}}_k(t, \theta_0,  {\BS{\beta}_0}) dN_i(t))\bigg\|_\infty\\
	&\leq  \bigg \| \mathbb G_m(\int_0^\tau (\bar D_k(t,\theta_0,\BS{\beta}_0)-U_D(t,\theta_0,\BS{\beta}_0))(\bar{\BS{Z}}_k(t,\theta_0,\BS{\beta}_0)-\BS{U}_{\BS{Z}}(t,\theta_0,\BS{\beta}_0))dN_i(t))\bigg\|_\infty \\
	& \quad+ \bigg \| \mathbb G_m(\int_0^\tau (\bar D_k(t,\theta_0,\BS{\beta}_0)-U_D(t,\theta_0,\BS{\beta}_0))\BS{U}_{\BS{Z}}(t,\theta_0,\BS{\beta}_0)dN_i(t))\bigg\|_\infty\\
	& \quad + \bigg \| \mathbb G_m(\int_0^\tau U_D(t,\theta_0,\BS{\beta}_0)(\bar{\BS{Z}}_k(t,\theta_0,\BS{\beta}_0)-\BS{U}_{\BS{Z}}(t,\theta_0,\BS{\beta}_0))dN_i(t))\bigg\|_\infty\\
	& \quad +\bigg \| \mathbb G_m(\int_0^\tau U_D(t,\theta_0,\BS{\beta}_0)\BS{U}_{\BS{Z}}(t,\theta_0,\BS{\beta}_0)dN_i(t))\bigg\|_\infty.
\end{align*}
For each term on the right-hand side of the inequality above, an application of similar arguments as in the proof for \eqref{ieqB5} yields that 
\begin{eqnarray*}
	 \bigg \| \mathbb G_m(\int_0^\tau (\bar D_k(t,\theta_0,\BS{\beta}_0)-U_D(t,\theta_0,\BS{\beta}_0))(\bar{\BS{Z}}_k(t,\theta_0,\BS{\beta}_0)-\BS{U}_{\BS{Z}}(t,\theta_0,\BS{\beta}_0))dN_i(t)) \bigg\|_\infty = O_p \Big(\frac{\log(np)}{\sqrt{n}} \Big),
\end{eqnarray*}
\[
 \bigg \| \mathbb G_m(\int_0^\tau (\bar D_k(t,\theta_0,\BS{\beta}_0)-U_D(t,\theta_0,\BS{\beta}_0))\BS{U}_{\BS{Z}}(t,\theta_0,\BS{\beta}_0)dN_i(t)) \bigg \|_\infty = O_p(\sqrt{\log(np)}),
\]
\[
 \bigg \| \mathbb G_m(\int_0^\tau U_D(t,\theta_0,\BS{\beta}_0)(\bar{\BS{Z}}_k(t,\theta_0,\BS{\beta}_0)-\BS{U}_{\BS{Z}}(t,\theta_0,\BS{\beta}_0))dN_i(t)) \bigg \|_\infty = O_p(\sqrt{\log(np)}),
\]
and
\begin{align*}
	&\bigg   \| \mathbb G_m(\int_0^\tau U_D(t,\theta_0,\BS{\beta}_0)\BS{U}_{\BS{Z}}(t,\theta_0,\BS{\beta}_0)dN_i(t))\big\|_\infty \\
	&\leq    \bigg  \| \int_0^\tau U_D(t,\theta_0,\BS{\beta}_0)\BS{U}_{\BS{Z}}(t,\theta_0,\BS{\beta}_0)d\mathbb G_m(M_i(t))\big\|_\infty\\
	&\quad + \bigg   \| \int_0^\tau U_D(t,\theta_0,\BS{\beta}_0)\BS{U}_{\BS{Z}}(t,\theta_0,\BS{\beta}_0)d\mathbb G_m(A_i(t))  \bigg  \|_\infty\\
	&= O_p(\sqrt{\log(n)}).
\end{align*}
Hence, bound \eqref{eq-neww-G-1} holds.

Finally, we prove \eqref{eq-neww-G-matrix}. We again follow a similar decomposition strategy as in the proof for (\ref{ieqB5}). We can write
\begin{align*}
	 &   \bigg  \| \mathbb G_m(\int_0^\tau  \bar {\BS{Z}}_k(t, \theta_0,  {\BS{\beta}_0})^{\otimes 2} dN_i(t))  \bigg  \|_{\max}\\
	&\leq  \bigg   \| \mathbb G_m(\int_0^\tau(\bar{\BS{Z}}_k(t,\theta_0,\BS{\beta}_0)-\BS{U}_{\BS{Z}}(t,\theta_0,\BS{\beta}_0))^{\otimes 2}dN_i(t) )  \bigg  \|_{\max} \\
	& \quad +   \bigg   \| \mathbb G_m(\int_0^\tau 2(\bar{\BS{Z}}_k(t,\theta_0,\BS{\beta}_0)-\BS{U}_{\BS{Z}}(t,\theta_0,\BS{\beta}_0))\BS{U}_{\BS{Z}}(t,\theta_0,\BS{\beta}_0)^\top dN_i(t))  \bigg  \|_{\max}\\
	& \quad  +  \bigg   \| \mathbb G_m(\int_0^\tau \BS{U}_{\BS{Z}}(t,\theta_0,\BS{\beta}_0)^{\otimes 2}dN_i(t))  \bigg  \|_{\max}.
\end{align*}
By similar arguments, we can obtain that 
\begin{eqnarray*}
	 \bigg \| \mathbb G_m(\int_0^\tau(\bar{\BS{Z}}_k(t,\theta_0,\BS{\beta}_0)-\BS{U}_{\BS{Z}}(t,\theta_0,\BS{\beta}_0))^{\otimes 2}dN_i(t)) \bigg \|_{\max} = O_p \Big(\frac{\log(np)}{\sqrt{n}} \Big),
\end{eqnarray*}
\[
 \bigg \| \mathbb G_m(\int_0^\tau 2(\bar{\BS{Z}}_k(t,\theta_0,\BS{\beta}_0)-\BS{U}_{\BS{Z}}(t,\theta_0,\BS{\beta}_0))\BS{U}_{\BS{Z}}(t,\theta_0,\BS{\beta}_0)^\top dN_i(t)) \bigg \|_{\max}= O_p(\sqrt{\log(np)}),
\]
and
\begin{align*}
	 \bigg \| \mathbb G_m(\int_0^\tau \BS{U}_{\BS{Z}}(t,\theta_0,\BS{\beta}_0)^{\otimes 2}dN_i(t))\big\|_{\max} &\leq \big \| \int_0^\tau \BS{U}_{\BS{Z}}(t,\theta_0,\BS{\beta}_0)^{\otimes 2}d\mathbb G_m(M_i(t)) \bigg \|_{\max}\\
	& \quad + \bigg   \| \int_0^\tau \BS{U}_{\BS{Z}}(t,\theta_0,\BS{\beta}_0)^{\otimes 2}d\mathbb G_m(A_i(t))  \bigg  \|_{\max}\\
	&= O_p(\sqrt{\log(n)}).
\end{align*}
Thus, bound \eqref{eq-neww-G-matrix} follows. This completes the proof of Lemma \ref{lmentropy}.

\subsection{Lemma \ref{concentrationformg} and its proof} \label{SecB.10}

Lemma \ref{concentrationformg} below provides a reusable concentration inequality for predictable integrands with zero weighted risk-set mean. It is employed to control several scalar and matrix empirical process terms in the previous lemmas.

\begin{lemma}\label{concentrationformg}
Let $\{a_i(t)\}_{i=1}^n$ be a sequence of predictable processes on $[0,\tau]$ uniformly bounded by some constant $M>0$, i.e., $|a_i(t)| \leq M $ for all $t\in [0,\tau]$ and $i = 1,2,\dots,n$. Assume that $\sum_{i=1}^{n} a_i(t)w_i(t,\theta_0,\BS{\beta}_0) = 0$ holds for any $t \in [0,\tau]$. Then we have that 
	\begin{eqnarray*}
		\frac{1}{\sqrt{n}}\bigg | \sum_{i=1}^n\int_0^\tau a_i(t)dN_i(t)- \mathbb{E}\sum_{i=1}^n\int_0^\tau a_i(t)dN_i(t)\bigg | = O_p(\sqrt{\log(n)}).
	\end{eqnarray*}
\end{lemma}

\noindent\textit{Proof}. The proof of this lemma will evaluate the counting-process integral at jump times and construct a martingale with bounded increments. The Azuma--Hoeffding concentration gives the stochastic bound, while the zero weighted-mean condition makes the compensator contribution vanish. A coordinatewise union bound extends the result from scalar predictable processes to matrix-valued ones.

Specifically, let $t_j$ be the time of the $j$th jump of the process $\sum_{i=1}^n \int_0^\tau Y_i(t)dN_i(t)$, $j=1,\dots,m$ and $t_0 = 0$. Then $t_j$'s are stopping times, and notice that $t_j \leq \tau$ holds almost surely. We define
\begin{eqnarray*}
	X_j = \sum_{i=1}^n\int_0^{t_j} a_i(t)dN_i(t).
\end{eqnarray*}
The assumption of $ \sum_{i = 1}^n a_i(t) w_i (t, \theta_0, \BS{\beta}_0) = 0$ entails that 
$$\sum_{i = 1}^n a_i(t) \widetilde{w}_i (t, \theta_0, \BS{\beta}_0) = 0. $$ Using the Doob--Meyer decomposition $N_i(t)= M_i(t)+ A_i(t)$ for all $i=1,\dots,n$, we can deduce that 
\begin{align*}
	X_j & =\sum_{i=1}^n\int_0^{t_j} a_i(t)dN_i(t) \\
	&=\sum_{i=1}^n\int_0^{t_j} a_i(t)dM_i(t) +\sum_{i=1}^n\int_0^{t_j} a_i(t)dA_i(t)\\
	& =  \sum_{i=1}^n\int_0^{t_j} a_i(t)dM_i(t) +\sum_{i=1}^n\int_0^{t_j} a_i(t)\tilde w_i(t,\theta_0,\BS{\beta}_0)\lambda_0(t)dt\\
	& =  \sum_{i=1}^n\int_0^{t_j} a_i(t)dM_i(t).
\end{align*}

Since each $M_i(t)$ is a martingale and $a_i(t)$'s are predictable, $\{X_j,j=0,1,\dots\}$ is a martingale with difference 
$$|X_j - X_{j-1}| \leq \max_{t,i} |a_i(t)| \leq M.$$ Let $m_0  $ be the greatest integer lower bound of $C^2_0n$. By the martingale version of  Hoeffding's inequality \citep{azuma1967weighted}, it holds that 
\begin{eqnarray*}
	\mathbb{P}(|X_{m_0}| > MnC_0x) \leq 2\exp(-n^2C_0^2x^2/(2m_0))\leq 2e^{-nx^2/2}.
\end{eqnarray*}
Observe that $X_{m_0} = \sum_{i=1}^n\int_0^{\tau} a_i(t)dN_i(t)$ \textit{if and only if} there is no jump after $t_{m_0}$, equivalently, if there are at most $m_0$ jumps of the process $\sum_{i=1}^n \int_0^\tau Y_i(t)dN_i(t)$. In our setting, we have that $m_0 \leq n$, and thus
\begin{eqnarray*}
	\mathbb P\bigg( \bigg|\frac{1}{n}\sum_{i=1}^n \int_0^\tau a_i(t)dN_i(t) \bigg| > MC_0x \bigg) \leq 2e^{-nx^2/2}.
\end{eqnarray*}
Taking $x$ as $\sqrt{\frac{2\log(n)}{n}}$ gives that 
\begin{eqnarray*}
	\mathbb P \bigg( \bigg|\frac{1}{n}\sum_{i=1}^n \int_0^\tau a_i(t)dN_i(t) \bigg| > M\sqrt{\frac{\log(n)}{n}} \bigg) \leq \frac{2}{n},
\end{eqnarray*}
which implies $|\frac{1}{n}\sum_{i=1}^n \int_0^\tau a_i(t)dN_i(t)| = O_p(\sqrt{\frac{\log{n}}{n}})$.

On the other hand, by invoking Fubini's theorem, we can rewrite the expected value as
\begin{align*}
	\frac{1}{\sqrt n } \bigg |\mathbb E\sum_{i=1}^n \int_0^\tau a_i(t)dN_i(t)\bigg |
    &\leq \mathbb E\bigg |\sum_{i=1}^n \int_0^\tau a_i(t)dN_i(t)\bigg |\\
	& =  \frac{1}{\sqrt n } \int_0^\infty \mathbb P \bigg(\bigg |\sum_{i=1}^n \int_0^\tau a_i(t)dN_i(t)\bigg | \geq M C_0 n x \bigg) d(M C_0 n x) \\
	&  \leq   \int_0^\infty {2C_0 M \sqrt{n} e^{-nx^2/2}}  d x\\
	& =  O( 1 ) .
\end{align*}
Combining the two bounds above, we can obtain that 
\begin{eqnarray*}
	\frac{1}{\sqrt{n}}\big \| \sum_{i=1}^n\int_0^\tau a_i(t)dN_i(t)- \mathbb{E}\sum_{i=1}^n\int_0^\tau a_i(t)dN_i(t)\big\| = O_p(\sqrt{\log(n)}),
\end{eqnarray*}
which gives the desired conclusion.

Finally, the same reasoning extends to matrix-valued predictable processes. For a sequence of matrices $\{\BS{A}_i(t)\}_{i=1}^n$ of size $p\times q$ for some positive integers $p$ and $q$ instead of a sequence of scalars $\{a_i(t)\}_{i=1}^n$, if the condition of $\sum_{i=1  }^n \BS{A}_i(t)w_i(t,\theta_0,\BS{\beta}_0) = 0$ is satisfied, the same concentration argument leads to 
\begin{eqnarray*}
	\mathbb P\bigg( \bigg|\frac{1}{n}\sum_{i=1}^n \int_0^\tau (\BS{A}_i)_{l,k}(t)dN_i(t) \bigg| > M C_0 x \bigg) \leq 2e^{-nx^2/2}
\end{eqnarray*}
for any $l=1,2,\dots,p$ and $k = 1,2,\dots,q$. Consequently, choosing $x = \sqrt{\frac{2\log(pqn)}{n}}$ and applying the union bound yield that 
\begin{eqnarray*}
	\mathbb P\bigg(\bigg\|\frac{1}{ n}\sum_{i=1}^n \int_0^\tau \BS{A}_i(t)dN_i(t)\bigg\|_{\max} > MC_0 \sqrt{\frac{2\log(pqn)}{n}} \bigg) \leq \frac{2}{n},
\end{eqnarray*}
which entails that $\|\frac{1}{\sqrt n}\sum_{i=1}^n \int_0^\tau \BS{A}_i(t)dN_i(t)\|_{\max} = O_p(\sqrt{\log(pqn)}  )$. This concludes the proof of Lemma \ref{concentrationformg}.

\section{Additional simulation results for Studies 2 and 3 in Section \ref{Sec5}} \label{SecA}

In this section, we report simulation results for $\bbeta$ in Studies 2 and 3 in Section \ref{Sec5}. As mentioned earlier, based on two estimators $\check{\theta}_1$ and $\check{\theta}_2$ for $\theta$,
one can obtain two different debiased estimators $\check{\bbeta}_1$ and $\check{\bbeta}_2$ from \eqref{eqestbeta}. We compare them to the penalized maximum partial likelihood estimator $\hat{\bbeta}$, which is defined in \eqref{eqBetaL1} and can be obtained using all observations and the R package \texttt{glmnet} with regularization parameter $\lambda_{\theta\BS{\beta}}$ selected by the $10$-fold cross-validation. We also include the oracle estimator $\hat{\bbeta}_{\mathrm{or}}$ as a benchmark for comparisons. The oracle estimator $\hat{\bbeta}_{\mathrm{or}}$ is computed from $(\hat{\theta}_{\mathrm{or}}, \hat{\bbeta}_{\mathrm{or}})$, which is based on the true underlying sparse model and obtained using the  \texttt{coxph} function in the R package \texttt{survival} with all the observations.

Tables \ref{Table-beta-n500-p10-N1000} and \ref{Table-beta-n500-p100-N1000} summarize the simulation results for $\bbeta$ over $1000$ repetitions under different
settings in Studies 2 and 3, respectively. CR represents the censoring rate, BIAS is the bias of the estimator, SE.T denotes the estimate of the asymptotic standard deviation of the estimator, and SE stands for the empirical standard error of the estimator. The results reveal that the Lasso estimator has a larger bias, whereas our proposed method reduces substantially the bias. $\check{\bbeta}_1$ and $\check{\bbeta}_2$, the two versions of our estimator of $\bbeta$, are nearly identical. Moreover, the proposed method identifies all signal variables. Since $p=100$ in Study 3, we report only the estimates $\beta_j$ for all the signal variables and the first $5$ noise variables for conciseness.

\begin{table}[H]
	\caption{Simulation results for $\bbeta$ in Study 2.}
	\vspace{1mm}
	\centering

  \scalebox{0.63}{
	\begin{tabular}{clrrrr rrrrrr}
		\toprule
		\hline
		Approach & Measure & \multicolumn{1}{c}{$\beta_1$} & \multicolumn{1}{c}{$\beta_2$} & \multicolumn{1}{c}{$\beta_3$} & \multicolumn{1}{c}{$\beta_4$} & \multicolumn{1}{c}{$\beta_5$} & \multicolumn{1}{c}{$\beta_6$} & \multicolumn{1}{c}{$\beta_7$} & \multicolumn{1}{c}{$\beta_8$} & \multicolumn{1}{c}{$\beta_9$} & \multicolumn{1}{c}{$\beta_{10}$} \\
		\hline
		\multicolumn{12}{c}{Setting 1 (CR: 0.1325)} \\
		\midrule
        \multirow{2}{*}{$\hat{\bbeta}$} & $\mbox{BIAS}$ &-0.0311 & 0.0295 &-0.0309 & 0.0278 &-0.0249 & -0.0047 & 0.0004 & 0.0007 & 0.0010 & -0.0013 \\
                           & $\mbox{SE}$       &  0.0540 & 0.0535 & 0.0538 & 0.0551 & 0.0584 & 0.0387 & 0.0365 & 0.0364 & 0.0330 & 0.0345\\
   \hline
\multirow{3}{*}{$\check{\bbeta}_1$} & $\mbox{BIAS}$ & -0.0110 & 0.0095 & -0.0109 & 0.0076 & -0.0095 & -0.0019 & 0.0013 & 0.0010 & 0.0016 &-0.0016 \\
     & $\mbox{SE.T}$  & 0.0474 & 0.0473 & 0.0472 & 0.0473 & 0.0529 & 0.0500 & 0.0439 & 0.0439 & 0.0439 & 0.0439 \\
     & $\mbox{SE}$   & 0.0531 & 0.0527 & 0.0530 & 0.0542 & 0.0584 & 0.0489 & 0.0463 & 0.0463 & 0.0422 & 0.0441 \\   
 \hline
\multirow{3}{*}{$\check{\bbeta}_2$} & $\mbox{BIAS}$ & -0.0112 & 0.0096 & -0.0110 & 0.0078 & -0.0088 & -0.0026 & 0.0013 & 0.0010 & 0.0016 & -0.0016 \\
     & $\mbox{SE.T}$ &  0.0474 & 0.0473 & 0.0472 & 0.0473 & 0.0529 & 0.0500 & 0.0439 & 0.0439 & 0.0439 & 0.0439 \\
     & $\mbox{SE}$   &  0.0531 & 0.0527 & 0.0530 &0.0542 & 0.0584 & 0.0489 & 0.0463 & 0.0463 & 0.0422 & 0.0441 \\ 
		\midrule     
   \multirow{3}{*}{$\hat{\bbeta}_{or}$} & $\mbox{BIAS}$ & 0.0041 & -0.0056 & 0.0043 & -0.0078 & 0.0043 & 0 & 0 & 0 & 0 & 0\\
                           & $\mbox{SE.T}$ &  0.0529 & 0.0527 & 0.0526 & 0.0528 & 0.0560 & 0 & 0 & 0 & 0 & 0\\
                           & $\mbox{SE}$   &  0.0529 & 0.0532 & 0.0532 & 0.0542 & 0.0581 & 0 & 0 & 0 & 0 & 0\\
		\midrule
		\multicolumn{12}{c}{Setting 2 (CR: 0.3937)} \\
		\midrule

        \multirow{2}{*}{$\hat{\bbeta}$} & $\mbox{BIAS}$ & -0.0373 & 0.0327 & -0.0368 & 0.0325 & -0.0280 &-0.0054 &0.0004 & 0 &0.0003 & -0.0020\\
                           & $\mbox{SE}$       & 0.0618 & 0.0626 & 0.0652 & 0.0643 & 0.0688 & 0.0462 & 0.0413 & 0.0435 & 0.0371 & 0.0419 \\
  \hline
     \multirow{3}{*}{$\check{\bbeta}_1$} & $\mbox{BIAS}$ & -0.0134 & 0.0089 & -0.013 & 0.0087 & -0.0110 & -0.0016 & 0.0007 & 0.0010 & 0.0007 & -0.0021 \\
                & $\mbox{SE.T}$ & 0.0561 & 0.0561 & 0.0559 & 0.0560 & 0.0628 & 0.0601 & 0.0528 & 0.0528 & 0.0527 & 0.0528 \\
                & $\mbox{SE}$   & 0.0612 & 0.0619 & 0.0641 & 0.0633 & 0.0690 & 0.0588 & 0.0535 & 0.0555 & 0.0490 & 0.0541 \\             
   \hline
     \multirow{3}{*}{$\check{\bbeta}_2$} & $\mbox{BIAS}$ & -0.0136 & 0.0091 & -0.0132 & 0.0088 & -0.0102 & -0.0024 & 0.0006 & 0.0010 & 0.0007 & -0.0021 \\
                & $\mbox{SE.T}$ & 0.0561 & 0.0561 & 0.0559 & 0.0560 & 0.0628 & 0.0601 & 0.0528 & 0.0528 & 0.0527 & 0.0528 \\
                & $\mbox{SE}$   & 0.0612 & 0.0618 & 0.0641 & 0.0633 & 0.0689 & 0.0588 & 0.0535 & 0.0555 & 0.0490 & 0.0540 \\               
   \hline
   \multirow{3}{*}{$\hat{\bbeta}_{or}$} & $\mbox{BIAS}$ &  0.0036 & -0.0080 & 0.0038 & -0.0085 & 0.0052 & 0 & 0 & 0 & 0 & 0\\
                           & $\mbox{SE.T}$ & 0.0623 & 0.0622 & 0.0620 & 0.0622 & 0.0663 & 0 & 0 & 0 & 0 & 0 \\
                           & $\mbox{SE}$   & 0.0622 & 0.0624 & 0.0646 & 0.0638 & 0.0687 & 0 & 0 & 0 & 0 & 0\\
		\midrule
		\multicolumn{12}{c}{Setting 3 (CR: 0.1149)} \\
		\midrule

        \multirow{2}{*}{$\hat{\bbeta}$} & $\mbox{BIAS}$ & -0.0330 & 0.0423 &-0.0337 & 0.0465 &-0.0283 &-0.0019 &-0.0004 &-0.0003 & 0.0003 & -0.0021\\
                           & $\mbox{SE}$  &      0.0664 & 0.0694 & 0.0661 & 0.0677 &0.0700 & 0.0429 & 0.0429 & 0.0448 & 0.0439 & 0.0408\\                    
        \hline
        \multirow{3}{*}{$\check{\bbeta}_1$} & $\mbox{BIAS}$ & -0.0156 & 0.0213 &-0.0164 & 0.0256 &-0.0130 &-0.0002 & 0 & 0.0001 & 0.0001 & -0.0024\\
                           & $\mbox{SE.T}$ &  0.0587 & 0.0588 & 0.0588 & 0.0588 & 0.0632 & 0.0554 & 0.0562 & 0.0561 & 0.0561 & 0.0502\\
                           & $\mbox{SE}$   & 0.0659 & 0.0681 & 0.0653 & 0.0666 & 0.0696 & 0.0511 & 0.0505 & 0.0524 & 0.0518 & 0.0486\\
    \hline
     \multirow{3}{*}{$\check{\bbeta}_2$} & $\mbox{BIAS}$ & -0.0157 & 0.0215 &-0.0164 & 0.0257 & -0.0124 & -0.0010 &-0.0001 & 0.0001 &0.0001 & -0.0025\\
                           & $\mbox{SE.T}$ &  0.0587 & 0.0588 & 0.0588 & 0.0588 & 0.0632 & 0.0554 & 0.0562 & 0.0561 & 0.0561 & 0.0502 \\
                           & $\mbox{SE}$   & 0.0659 & 0.0681 & 0.0653 & 0.0665 & 0.0695 & 0.0510 & 0.0505 & 0.0523 & 0.0518 & 0.0486 \\
   \hline
   \multirow{3}{*}{$\hat{\bbeta}_{or}$} & $\mbox{BIAS}$ & 0.0051 & -0.0063 & 0.0040 &-0.0019 & 0.0052 & 0 & 0 & 0 & 0 & 0\\
                           & $\mbox{SE.T}$ & 0.0655 & 0.0657 & 0.0657 & 0.0656 & 0.0684 & 0 & 0 & 0 & 0 & 0\\
                           & $\mbox{SE}$   & 0.0662 & 0.0677 & 0.0650 & 0.0667 & 0.0700 & 0 & 0 & 0 & 0 & 0\\                        
		\midrule
		\multicolumn{12}{c}{Setting 4 (CR: 0.3845)} \\
		\midrule

        \multirow{2}{*}{$\hat{\bbeta}$} & $\mbox{BIAS}$ & -0.0354 & 0.0474 &-0.0367 & 0.0526 &-0.0326 &-0.0003 &-0.0020 &-0.0010 & 0.0009 & -0.0025\\
                           & $\mbox{SE}$       & 0.0806 & 0.0817 & 0.0815 & 0.0826 & 0.0839 & 0.0526 & 0.0500 & 0.0546 & 0.0532 & 0.0500 \\                
        \hline
        \multirow{3}{*}{$\check{\bbeta}_1$} & $\mbox{BIAS}$ & -0.0169 & 0.0260 &-0.0183 & 0.0312 &-0.0170 & 0.0017 &-0.0010 &-0.0006 & 0.0011 & -0.0029\\
                           & $\mbox{SE.T}$ &  0.0700 & 0.0700 & 0.0701 & 0.0700 & 0.0755 & 0.0668 & 0.0678 & 0.0676 & 0.0675 & 0.0606\\
                           & $\mbox{SE}$   & 0.0802 & 0.0802 & 0.0809 & 0.0812 & 0.0834 & 0.0626 & 0.0586 & 0.0633 & 0.0624 & 0.0592 \\
    \hline
     \multirow{3}{*}{$\check{\bbeta}_2$} & $\mbox{BIAS}$ & -0.0169 & 0.0262 & -0.0182 & 0.0314 & -0.0163 & 0.0008 &-0.0013 &-0.0007 & 0.0010 &-0.0029\\
                           & $\mbox{SE.T}$ & 0.0700 & 0.0700 & 0.0701 & 0.0700 & 0.0755 & 0.0668 & 0.0678 & 0.0676 & 0.0675 & 0.0606\\
                           & $\mbox{SE}$   & 0.0802 & 0.0803 & 0.0809 & 0.0811 & 0.0834 & 0.0625 & 0.0586 & 0.0632 & 0.0623 & 0.0591\\                      
   \hline
   \multirow{3}{*}{$\hat{\bbeta}_{or}$} & $\mbox{BIAS}$ & 0.0068 & -0.0076 & 0.0056 & -0.0021 & 0.0045 & 0 & 0 & 0 & 0 & 0\\
                           & $\mbox{SE.T}$ & 0.0778 & 0.0778 & 0.0779 & 0.0778 & 0.0811 & 0 & 0 & 0 & 0 & 0\\
                           & $\mbox{SE}$   & 0.0800 & 0.0801 & 0.0807 & 0.0815 & 0.0835 & 0 & 0 & 0 & 0 & 0\\                     
       \hline
		\bottomrule
	\end{tabular}\label{Table-beta-n500-p10-N1000}
    }
\end{table}

\begin{table}[H]
	\caption{Simulation results for $\bbeta$ in Study 3. Since $p=100$ in this study, we report only the estimates $\hat{\beta}_j$ for all the signal variables and the first $5$ noise variables for conciseness.}
	\vspace{1mm}
	\centering
  \scalebox{0.63}{
	\begin{tabular}{clrrrrr rrrrr}
		\toprule
		\midrule
		Approach & Measure & \multicolumn{1}{c}{$\beta_1$} & \multicolumn{1}{c}{$\beta_2$} & \multicolumn{1}{c}{$\beta_3$} & \multicolumn{1}{c}{$\beta_4$} & \multicolumn{1}{c}{$\beta_5$} & \multicolumn{1}{c}{$\beta_6$} & \multicolumn{1}{c}{$\beta_7$} & \multicolumn{1}{c}{$\beta_8$} & \multicolumn{1}{c}{$\beta_9$} & \multicolumn{1}{c}{$\beta_{10}$} \\
		\midrule
		\multicolumn{12}{c}{Setting 1 (CR: 0.1333)} \\
        \midrule
        \multirow{2}{*}{$\hat{\bbeta}$} & $\mbox{BIAS}$ & -0.1013 & 0.1019 & -0.1002 & 0.1005 & -0.0829 & -0.0022 & 0.0002 & 0.0004 & -0.0003 & 0\\
                          & $\mbox{SE}$       &  0.0514 & 0.0520 & 0.0518 & 0.0505 & 0.0553 & 0.0121 & 0.0110 & 0.0132 & 0.0114 & 0.0133\\     
        \hline
        \multirow{3}{*}{$\check{\bbeta}_1$} & $\mbox{BIAS}$ & -0.0449 & 0.0458 &-0.0439 & 0.0444 &-0.0439 & 0.0012 & 0.0015 & 0.0002 &-0.0001 &-0.0003\\
                          & $\mbox{SE.T}$ &   0.0540 & 0.0539 & 0.0539 & 0.0539 & 0.0596 & 0.0565 & 0.0504 & 0.0503 & 0.0503 & 0.0503\\
                          & $\mbox{SE}$   &  0.0512 & 0.0514 & 0.0510 & 0.0498 & 0.0558 & 0.0354 & 0.0366 & 0.0384 & 0.0375 & 0.0388\\                           
    \hline
     \multirow{3}{*}{$\check{\bbeta}_2$} & $\mbox{BIAS}$ & -0.0450 & 0.0460 &-0.0441 & 0.0445 &-0.0433 & 0.0005 & 0.0015 & 0.0002 &-0.0001 &-0.0003\\
                          & $\mbox{SE.T}$ &  0.0540 & 0.0539 & 0.0539 & 0.0539 & 0.0596 & 0.0565 & 0.0504 & 0.0503 & 0.0503 & 0.0503 \\
                          & $\mbox{SE}$   &  0.0512 & 0.0514 & 0.0510 & 0.0498 & 0.0557 & 0.0354 & 0.0366 & 0.0384 & 0.0375 & 0.0388\\  
  \hline
  \multirow{3}{*}{$\hat{\bbeta}_{or}$} & $\mbox{BIAS}$ & 0.0038 & -0.0029 & 0.0050 & -0.0044 & 0.0044 & 0 & 0 & 0 & 0 & 0\\
                          & $\mbox{SE.T}$ &  0.0527 &  0.0527 & 0.0527 &  0.0527 & 0.0560 & 0 & 0 & 0 & 0 & 0\\
                          & $\mbox{SE}$   &  0.0523 &  0.0529 & 0.0518 &  0.0508 & 0.0553 & 0 & 0 & 0 & 0 & 0  \\ 
                          
		\midrule
		\multicolumn{12}{c}{Setting 2 (CR: 0.3952)} \\
		\midrule

        \multirow{2}{*}{$\hat{\bbeta}$} & $\mbox{BIAS}$ & -0.1118 & 0.1131 &-0.1110 & 0.1119 & -0.0912 & -0.0022 & 0.0002 & 0.0005 &-0.0004 & 0\\
                          & $\mbox{SE}$       &  0.0627 & 0.0628 & 0.0617 & 0.0581 &  0.0631 &  0.0140 & 0.0121 & 0.0162 & 0.0136 & 0.0158\\                   
        \hline
        \multirow{3}{*}{$\check{\bbeta}_1$} & $\mbox{BIAS}$ & -0.0443 & 0.0461 &-0.0440 & 0.0448 &-0.0456 & 0.0029 & 0.0005 & 0.0005 &-0.0014 & 0.0007\\
                          & $\mbox{SE.T}$ & 0.0687 & 0.0685 & 0.0684 & 0.0685 & 0.0757 & 0.0728 & 0.0651 & 0.0649 & 0.0650 & 0.0650\\
                          & $\mbox{SE}$   & 0.0622 & 0.0623 & 0.0615 & 0.0579 & 0.0636 & 0.0430 & 0.0427 & 0.0478 & 0.0442 & 0.0467\\                                
    \hline
     \multirow{3}{*}{$\check{\bbeta}_2$} & $\mbox{BIAS}$ & -0.0444 & 0.0462 & -0.0441 & 0.0449 &-0.0449 & 0.0022 & 0.0005 & 0.0005 &-0.0014 & 0.0007\\
                          & $\mbox{SE.T}$ & 0.0687 & 0.0685 &  0.0684 & 0.0685 & 0.0757 & 0.0728 & 0.0651 & 0.0649 & 0.0650 & 0.0650\\
                          & $\mbox{SE}$   & 0.0622 & 0.0623 &  0.0615 & 0.0579 & 0.0635 & 0.0430 & 0.0427 & 0.0477 & 0.0441 & 0.0467\\                            
  \hline
  \multirow{3}{*}{$\hat{\bbeta}_{or}$} & $\mbox{BIAS}$ & 0.0043 &-0.0030 & 0.0062 &-0.0043 & 0.0040 & 0 & 0 & 0 & 0 & 0\\
                          & $\mbox{SE.T}$ & 0.0622 & 0.0621 & 0.0622 & 0.0622 & 0.0662 & 0 & 0 & 0 & 0 & 0\\
                          & $\mbox{SE}$   & 0.0631 & 0.0629 & 0.0622 & 0.0595 & 0.0633 & 0 & 0 & 0 & 0 & 0\\                             
		\midrule
		\multicolumn{12}{c}{Setting 3 (CR: 0.1149)} \\
		\midrule

        \multirow{2}{*}{$\hat{\bbeta}$} & $\mbox{BIAS}$ & -0.1211 & 0.1604 &-0.1200 & 0.1568 &-0.1048 &-0.0038 &-0.0016 &0.0002 &-0.0009 & 0.0003\\
                          & $\mbox{SE}$       & 0.0655 & 0.0711 & 0.0671 & 0.0673 & 0.0712 & 0.0157 & 0.0159 &0.0152 & 0.0146 & 0.0145\\                       
        \hline
        \multirow{3}{*}{$\check{\bbeta}_1$} & $\mbox{BIAS}$ & -0.0779 & 0.1147 & -0.0768 & 0.1110 &-0.0696 &-0.0021 & 0.0001 & 0.0019 &-0.0002 & 0.0021\\
                          & $\mbox{SE.T}$ & 0.0668 & 0.0665 &  0.0670 & 0.0667 & 0.0714 & 0.0626 & 0.0641 & 0.0640 & 0.0640 & 0.0640\\
                          & $\mbox{SE}$   & 0.0650 & 0.0690 &  0.0671 & 0.0655 & 0.0715 & 0.0379 & 0.0370 & 0.0371 & 0.0364 & 0.0365\\                              
    \hline
     \multirow{3}{*}{$\check{\bbeta}_2$} & $\mbox{BIAS}$ & -0.0776 & 0.1152 &-0.0765 & 0.1115 & -0.0688 & -0.0030 &-0.0003 & 0.0017 &-0.0003 & 0.0020\\
                          & $\mbox{SE.T}$ &  0.0668 & 0.0665 & 0.0670 & 0.0667 &  0.0714 &  0.0626 & 0.0641 & 0.0640 & 0.0640 & 0.0640\\
                          & $\mbox{SE}$   &  0.0650 & 0.0689 & 0.0670 & 0.0655 &  0.0714 &  0.0378 & 0.0370 & 0.0370 & 0.0364 & 0.0365\\                              
  \hline
  \multirow{3}{*}{$\hat{\bbeta}_{or}$} & $\mbox{BIAS}$ & 0.0023 &-0.0007 & 0.0056 &-0.0050 & 0.0051 & 0 & 0 & 0 & 0 & 0\\
                          & $\mbox{SE.T}$ & 0.0655 & 0.0654 & 0.0657 & 0.0657 & 0.0685 & 0 & 0 & 0 & 0 & 0\\
                          & $\mbox{SE}$   & 0.0638 & 0.0684 & 0.0657 & 0.0651 & 0.0707 & 0 & 0 & 0 & 0 & 0\\                        
		\midrule
		\multicolumn{12}{c}{Setting 4 (CR: 0.3871)} \\
		\midrule

        \multirow{2}{*}{$\hat{\bbeta}$} & $\mbox{BIAS}$ & -0.1391 & 0.1846 &-0.1357 & 0.1822 &-0.1175 &-0.0042 &-0.0016 & 0.0002 &-0.0008 & 0.0002\\
                          & $\mbox{SE}$       &  0.0752 & 0.0848 & 0.0809 & 0.0810 & 0.0845 & 0.0200 & 0.0176 & 0.0172 & 0.0187 & 0.0170\\                           
        \hline
        \multirow{3}{*}{$\check{\bbeta}_1$} & $\mbox{BIAS}$ & -0.0886 & 0.1281 & -0.0851 & 0.1257 &-0.0769 &-0.0013 & 0.0004 & 0.0015 & 0.0003 & 0.0024\\
                          & $\mbox{SE.T}$ & 0.0855 & 0.0851 &  0.0857 & 0.0852 & 0.0912 & 0.0811 & 0.0829 & 0.0828 & 0.0830 & 0.0828\\
                          & $\mbox{SE}$   & 0.0746 & 0.0824 &  0.0810 & 0.0790 & 0.0847 & 0.0468 & 0.0445 & 0.0447 & 0.0448 & 0.0441\\                               
    \hline
     \multirow{3}{*}{$\check{\bbeta}_2$} & $\mbox{BIAS}$ & -0.0882 & 0.1286 &-0.0848 & 0.1262 &-0.0761 &-0.0023 &-0.0001 & 0.0012 & 0.0001 & 0.0023\\
                          & $\mbox{SE.T}$ & 0.0855 & 0.0851 & 0.0857 & 0.0852 & 0.0913 & 0.0811 & 0.0829 & 0.0828 & 0.0830 & 0.0828\\
                          & $\mbox{SE}$   & 0.0746 & 0.0824 & 0.0810 & 0.0790 & 0.0847 & 0.0467 & 0.0444 & 0.0447 & 0.0448 & 0.0440\\                               
  \hline
  \multirow{3}{*}{$\hat{\bbeta}_{or}$} & $\mbox{BIAS}$ & 0.0006 &-0.0015 & 0.0068 &-0.0051 & 0.0061 & 0 & 0 & 0 & 0 & 0\\
                          & $\mbox{SE.T}$ & 0.0777 & 0.0777 & 0.0781 & 0.0780 & 0.0814 & 0 & 0 & 0 & 0 & 0\\
                          & $\mbox{SE}$   & 0.0730 & 0.0826 & 0.0785 & 0.0785 & 0.0840 & 0 & 0 & 0 & 0 & 0\\                       
		\hline
		\bottomrule
	\end{tabular}\label{Table-beta-n500-p100-N1000}
    }
\end{table}

\end{document}